\documentclass[letterpaper]{article}
\usepackage[compat2,dvips,nohead]{geometry}
\usepackage{graphicx}
\usepackage{psfrag}
\usepackage{amsmath}
\usepackage{amsthm}
\usepackage{amssymb}
\usepackage{natbib}
\usepackage{url}
\usepackage{pstricks,pst-node,pst-text,pst-3d}  
\usepackage{fltpage}                            

\newcommand{\A}{\ensuremath{\mathbf{A}}}
\newcommand{\B}{\ensuremath{\mathbf{B}}}
\newcommand{\C}{\ensuremath{\mathbf{C}}}
\newcommand{\D}{\ensuremath{\mathbf{D}}}
\newcommand{\E}{\ensuremath{\mathbf{E}}}
\newcommand{\F}{\ensuremath{\mathbf{F}}}
\newcommand{\G}{\ensuremath{\mathbf{G}}}

\newcommand{\I}{\ensuremath{\mathbf{I}}}

\newcommand{\LL}{\ensuremath{\mathbf{L}}}
\newcommand{\M}{\ensuremath{\mathbf{M}}}
\newcommand{\N}{\ensuremath{\mathbf{N}}}

\newcommand{\PP}{\ensuremath{\mathbf{P}}}

\newcommand{\RR}{\ensuremath{\mathbf{R}}}
\renewcommand{\SS}{\ensuremath{\mathbf{S}}}

\newcommand{\U}{\ensuremath{\mathbf{U}}}
\newcommand{\V}{\ensuremath{\mathbf{V}}}
\newcommand{\W}{\ensuremath{\mathbf{W}}}
\newcommand{\X}{\ensuremath{\mathbf{X}}}
\newcommand{\Y}{\ensuremath{\mathbf{Y}}}

\renewcommand{\b}{\ensuremath{\mathbf{b}}}
\renewcommand{\c}{\ensuremath{\mathbf{c}}}
\newcommand{\dd}{\ensuremath{\mathbf{d}}}

\newcommand{\f}{\ensuremath{\mathbf{f}}}

\newcommand{\bk}{\ensuremath{\mathbf{k}}}

\newcommand{\m}{\ensuremath{\mathbf{m}}}
\newcommand{\n}{\ensuremath{\mathbf{n}}}

\newcommand{\sss}{\ensuremath{\mathbf{s}}}  

\newcommand{\uu}{\ensuremath{\mathbf{u}}}
\renewcommand{\v}{\ensuremath{\mathbf{v}}}
\newcommand{\w}{\ensuremath{\mathbf{w}}}
\newcommand{\x}{\ensuremath{\mathbf{x}}}
\newcommand{\y}{\ensuremath{\mathbf{y}}}
\newcommand{\z}{\ensuremath{\mathbf{z}}}
\newcommand{\0}{\ensuremath{\mathbf{0}}}
\newcommand{\1}{\ensuremath{\mathbf{1}}}

\newcommand{\bkappa}{\ensuremath{\boldsymbol{\kappa}}}
\newcommand{\blambda}{\ensuremath{\boldsymbol{\lambda}}}
\newcommand{\bmu}{\ensuremath{\boldsymbol{\mu}}}

\newcommand{\bLambda}{\ensuremath{\boldsymbol{\Lambda}}}

\newcommand{\bSigma}{\ensuremath{\boldsymbol{\Sigma}}}

\newcommand{\bbR}{\ensuremath{\mathbb{R}}}

\newcommand{\bbZ}{\ensuremath{\mathbb{Z}}}

\newcommand{\calD}{\ensuremath{\mathcal{D}}}
\newcommand{\calE}{\ensuremath{\mathcal{E}}}
\newcommand{\calF}{\ensuremath{\mathcal{F}}}

\newcommand{\calL}{\ensuremath{\mathcal{L}}}

\newcommand{\calN}{\ensuremath{\mathcal{N}}}
\newcommand{\calO}{\ensuremath{\mathcal{O}}}

\newcommand{\bydef}{\stackrel{\mathrm{\scriptscriptstyle def}}{=}}

\newcommand{\abs}[1]{\left\lvert#1\right\rvert}
\newcommand{\norm}[1]{\left\lVert#1\right\rVert}

\newcommand{\leftexp}[2]{{\vphantom{#2}}^{#1}{#2}}

\newcommand{\caja}[4][1]{{%
    \renewcommand{\arraystretch}{#1}%
    \begin{tabular}[#2]{@{}#3@{}}%
      #4%
    \end{tabular}%
    }}

{%
\begin{list}{#1}{
\vspace{-\topsep}
\vspace{-\partopsep}
\setlength{\itemindent}{0cm}
\setlength{\rightmargin}{0cm}
\setlength{\listparindent}{0cm}
\settowidth{\labelwidth}{#1}
\setlength{\leftmargin}{\labelwidth}
\addtolength{\leftmargin}{\labelsep}
\setlength{\itemsep}{0cm}
}%
}%
{%
\end{list}
\vspace{-\topsep}
\vspace{-\partopsep}
}

{\begin{enumerate}%
}%
{\end{enumerate}}

\newcommand{\Volume}{\operatorname{vol}}
\newcommand{\volume}[1]{\ensuremath{\Volume\left(#1\right)}}

\newcommand{\sgnop}{\operatorname{sgn}}
\newcommand{\sgn}[1]{\ensuremath{\sgnop\left(#1\right)}}

\newcommand{\im}{\operatorname{im}}

\newcommand{\diagop}{\operatorname{diag}}
\newcommand{\diag}[1]{\ensuremath{\diagop\left(#1\right)}}
\newcommand{\traceop}{\operatorname{tr}}
\newcommand{\trace}[1]{\ensuremath{\traceop\left(#1\right)}}
\newcommand{\rankop}{\operatorname{rank}}
\newcommand{\rank}[1]{\ensuremath{\rankop\left(#1\right)}}

\newcommand{\vectop}{\operatorname{vec}}
\newcommand{\vect}[1]{\ensuremath{\vectop\left(#1\right)}}

\theoremstyle{plain}
\newtheorem{thm}{Theorem}[section]

\newtheorem*{lemma*}{Lemma}
\newtheorem{prop}[thm]{Proposition}
\newtheorem*{prop*}{Proposition}
\newtheorem{cor}[thm]{Corollary}

\theoremstyle{definition}
\newtheorem{defn}{Definition}[section]
\newtheorem*{defn*}{Definition}

\newtheorem*{exmp*}{Example}

\newtheorem*{conj*}{Conjecture}
\newtheorem{rmk}[thm]{Remark}

\theoremstyle{remark}
\newtheorem*{rmk*}{Remark}

\graphicspath{{grf/}}

\bibpunct[, ]{(}{)}{;}{a}{,}{,} 

\title{Generalised elastic nets\thanks{This paper was written on August 14th, 2003 and has not been updated since then.}}
\author{Miguel {\'A}.\ Carreira-Perpi{\~n}{\'a}n%
  \thanks{Current address: Dept.\ of Computer Science, University of Toronto. 6 King's College Road. Toronto, ON M5S 3H5, Canada. Email: \texttt{miguel@cs.toronto.edu}.}
  \quad \& \quad Geoffrey J.\ Goodhill \\
  Department of Neuroscience \\
  Georgetown University Medical Center \\
  3900 Reservoir Road NW, Washington, DC 20007, US \\
  \texttt{miguel@cns.georgetown.edu}, \texttt{geoff@georgetown.edu}}
\date{}

\AtBeginDvi{}

\begin{document}

\maketitle

\begin{abstract}
The elastic net was introduced as a heuristic algorithm for
combinatorial optimisation and has been applied, among other problems,
to biological modelling. It has an energy function which trades off a
fitness term against a tension term. In the original formulation of
the algorithm the tension term was implicitly based on a first-order
derivative. In this paper we generalise the elastic net model to an
arbitrary quadratic tension term, e.g.\ derived from a discretised
differential operator, and give an efficient learning algorithm. We
refer to these as generalised elastic nets (GENs). We
give a theoretical analysis of the tension term for 1D nets with
periodic boundary conditions, and show that the model is sensitive to
the choice of finite difference scheme that represents the discretised
derivative. We illustrate some of these issues in the context of
cortical map models, by relating the choice of tension term to a
cortical interaction function. In particular, we prove that this
interaction takes the form of a Mexican hat for the original elastic
net, and of progressively more oscillatory Mexican hats for
higher-order derivatives. The results apply not only to generalised
elastic nets but also to other methods using discrete differential
penalties, and are expected to be useful in other areas, such as
data analysis, computer graphics and optimisation problems.
\end{abstract}

The elastic net was first proposed as a method to obtain good
solutions to the travelling salesman problem (TSP;
\citealp{DurbinWillsha87a}) and was subsequently also found to be a
very successful cortical map model
\citep{DurbinMitchis90a,GoodhilWillsh90a,Erwin_95a,Swindal96a,WolfGeisel98a}. Essentially, it trades
off the desire to represent a set of data (cities in the TSP, feature
preferences in cortical maps) with the desire for this representation
to be smooth in the sense of minimising the sum of squared distances
of neighbouring centroids. The relative success of the elastic net in
these applications, its flavour of wirelength minimisation, and its
ease of implementation, are probable reasons why there has so far been
little attempt to generalise the model beyond its original
formulation, specifically in terms of using more complex tension
terms. In the TSP context this is perhaps understandable in that the
original tension term closely approximates the true cost, the sum of
nonsquared distances (some investigations of the effect of using
alternative exponents for the distance function have been performed;
\citealp{DurbinMitchis90a,Yuille_96a}). However, in the context of
cortical map modelling the tension term represents an abstraction of
lateral connections, whose functional form is unknown
biologically. Although the sum-of-square-distances form was motivated
as wirelength of neuronal connections, assumed large between neurons
having very dissimilar stimulus preferences \citep{DurbinMitchis90a},
this begs the question of examining other types of tension terms to
see what difference, if any, they make, as well as to relate the
tension term with biologically significant parameters, such as a
description of the intracortical connectivity
pattern. \citet{Dayan93a} (see also \citealp{Yuille90a,Yuille_96a})
was the first to suggest this and considered the relation of quadratic
tension terms---just like those considered here---with other cortical
map models. In this paper we investigate generalized tension terms in
more detail. Although our motivation is primarily biological, the
extended model is also expected to be generally useful in areas such
as computer vision, computer graphics, image processing, unsupervised
learning and TSP-type problems.

The paper consists of three parts. In the first (sections~\ref{s:GEN}--\ref{s:ann}) we define the
generalised elastic net (GEN) via the introduction of a quadratic penalty
(tension term) in the energy of the net and give efficient learning
algorithms for it. In the second (sections~\ref{s:S}--\ref{s:tension}) we consider a class of quadratic
penalties defined as discretised differential operators, and give a
theoretical analysis that is applicable not only to generalised
elastic nets but also to any model with a Gaussian prior. In the
third part (section~\ref{s:cmap}) we demonstrate the results of the previous parts in the
problem of cortical map modelling and report additional results via
simulation.

Throughout the paper we will often refer to the ``original elastic net'' meaning the original formulation of the elastic net by \citet{DurbinWillsha87a}, with its sum-of-square-distances tension term. The equations for the original elastic net are given in section~\ref{s:origEN}.

\section{Probabilistic formulation of the generalised elastic net (GEN)}
\label{s:GEN}

Given a collection of centroids $\{\y_m\}^M_{m=1} \subset \bbR^D$ that we express as a $D \times M$ matrix $\Y \bydef (\y_1,\dots,\y_M)$ and a scale, or variance, parameter $\sigma \in \bbR^+$, consider a Gaussian-mixture density $p(\x|\Y;\sigma) \bydef \sum^M_{m=1}{\frac{1}{M} p(\x|m)}$ with $M$ components, where $\x|m \sim \calN(\y_m,\sigma^2\I_D)$, i.e., all covariance matrices are equal and isotropic. Consider a Gaussian prior on the centroids $p(\Y;\beta) \propto e^{-\frac{\beta}{2} \trace{\Y^T\Y\SS}}$ where $\beta$ is a regularisation, or inverse variance, parameter and \SS\ is a positive definite or positive semidefinite $M \times M$ matrix---in the latter case the prior being improper. The normalisation constant of this density will not be relevant for our purposes and is given in appendix~\ref{s:density-norm-cte}.

We define a \emph{generalised elastic net (GEN)} by $p(\x|\Y;\sigma) p(\Y;\beta)$. The original elastic net of \citet{DurbinWillsha87a} and \citet{Durbin_89a} is recovered for a particular choice of the matrix \SS\ (section~\ref{s:origEN}). It is also possible to define a prior over the scale $\sigma$, but we will not do so here, since $\sigma$ will play the role of the temperature in a deterministic annealing algorithm.

Without the prior over centroids, the centroids could be permuted at will with no change in the model, since the variable $m$ is just an index. The prior can be used to convey the topologic (dimension and shape) and geometric (e.g.\ curvature) structure of a manifold implicitly defined by the centroids (as if the centroids were a sample from a continuous latent variable model; \citealp{Carreir01a}). This prior can also be seen as a Gaussian process prior (e.g.\ \citealp{Bishop_98b}, p.~219), where our matrix \SS\ is the inverse of the Gaussian process covariance matrix. The semantics of \SS\ is not necessary to develop a learning algorithm and so its study is postponed to section~\ref{s:S}. However, it will be convenient to keep in mind that \SS\ will be typically derived from a discretised derivative based on a finite difference scheme (or \emph{stencil}) and that \SS\ will be a sparse matrix.

\section{Parameter estimation with annealing}
\label{s:ann}

From a statistical learning point of view, one might wish to find the values of the parameters \Y\ and $\sigma$ that maximise the objective function
\begin{equation}
  \label{e:ann:MAP}
  \log{p(\Y,\sigma|\X)} = \log{p(\X|\Y,\sigma)} + \log{p(\Y)} - \log{p(\X)}
\end{equation}
given a training set $\{\x_n\}^N_{n=1} \subset \bbR^D$ expressed as a $D \times N$ matrix $\X \bydef (\x_1,\dots,\x_n)$---that is, maximum-a-posteriori (MAP) estimation. For iid data and ignoring a term independent of \Y\ and $\sigma$, this equation reduces to:
\begin{equation}
  \label{e:ann:MAP2}
  \calL_{\text{MAP}}(\Y,\sigma) = \sum^N_{n=1}{ \log{ \sum^M_{m=1}{e^{-\frac{1}{2}\norm{\frac{\x_n-\y_m}{\sigma}}^2}} } } - ND\log{\sigma} - \frac{\beta}{2} \trace{\Y^T\Y\SS}.
\end{equation}
However, as in \citet{DurbinWillsha87a} and \citet{Durbin_89a}, we are more interested in deterministic annealing algorithms that minimise the energy function
\begin{equation}
  \label{e:ann:E}
  E(\Y,\sigma) \bydef \underbrace{\rule[-3ex]{0pt}{0pt}-\alpha\sigma\sum^N_{n=1}{ \log{ \sum^M_{m=1}{e^{-\frac{1}{2}\norm{\frac{\x_n-\y_m}{\sigma}}^2}} } }}_{\text{fitness term}} + \underbrace{\rule[-3ex]{0pt}{0pt}\frac{\beta}{2} \trace{\Y^T\Y\SS}}_{\text{tension term}}
\end{equation}
over \Y\ alone, starting with a large $\sigma$ (for which the tension term dominates) and tracking the minimum to a small value of $\sigma$ (for which the fitness term dominates). This is so because (1) one can find good solutions to combinatorial optimisation problems such as the TSP (which require $\sigma \rightarrow 0$) and to dimension-reduction problems such as cortical map modelling (which do not require $\sigma \rightarrow 0$); and (2) if considered as a dynamical system for a continuous latent space, the evolution of the net as a function of $\sigma$ and the iteration index may model the temporal evolution of cortical maps. To attain good generalisation to unseen data, the parameter $\beta$ can be considered a hyperparameter and one can look for an optimal value for it given training data, by cross-validation, Bayesian inference (e.g.\ \citealp{Utsugi97a}) or some other means. However, in this paper we will be interested in investigating the behaviour of the model for a range of $\beta$ values, rather than fixing $\beta$ according to some criterion.

We will call the $\alpha$\nobreakdash-term the \emph{fitness term}, arising from the Gaussian mixture $p(\X|\Y,\sigma)$, and the $\beta$\nobreakdash-term the \emph{tension term}, arising from the prior $p(\Y)$.

Eq.~\eqref{e:ann:E} differs from the MAP objective function of eq.~\eqref{e:ann:MAP2} in an immaterial change of sign, in the deletion of a now constant term $N D \log{\sigma}$, and in the multiplication of the fitness term by $\alpha\sigma$, where $\alpha$ is a positive constant. The latter reduces the influence of the fitness term with respect to the tension term as $\sigma$ decreases. Since $\alpha$ can be absorbed in $\beta$, we will take $\alpha = 1$ hereafter. However, it can be useful to have a different $\alpha_n$ for each training point $\x_n$ in order to simulate a training set that overrepresents some data points over others (see section~\ref{s:alpha_n}).

We derive three iterative algorithms for minimising $E$ (gradient descent, matrix iteration and Cholesky factorisation), all based on the gradient of $E$:
\begin{equation}
  \label{e:ann:grad}
  \frac{\partial E}{\partial \Y} = - \frac{\alpha}{\sigma} (\X\W - \Y\G) + \beta \Y \left( \frac{\SS+\SS^T}{2} \right)
\end{equation}
where we define the weight matrix $\W_{N\times M} = (w_{nm})$ and the invertible diagonal matrix $\G_{M\times M} = \diag{g_m}$ as
\begin{equation*}
  w_{nm} \bydef \frac{e^{-\frac{1}{2}\norm{\frac{\x_n-\y_m}{\sigma}}^2}}{\sum^M_{m'=1}{e^{-\frac{1}{2}\norm{\frac{\x_n-\y_{m'}}{\sigma}}^2}}} \qquad \qquad g_m \bydef \sum^N_{n=1}{w_{nm}}.
\end{equation*}
The weight $w_{nm}$ is also the responsibility $p(m|\x_n)$ of centroid $\bmu_m$ for generating point $\x_n$, and so $g_m$ is the total responsibility of centroid $\bmu_m$ (or the average number of training points assigned to it). The matrix $\X\W$ is then a list of average centroids. All three algorithms, particularly the matrix-iteration and Cholesky-factorisation ones, benefit considerably from the fact that the matrix \SS\ will typically be sparse (with a banded or block-banded structure).

\subsection{Gradient descent}

We simply iterate $\Y^{(\tau+1)} \bydef \Y^{(\tau)} + \Delta\Y^{(\tau)}$ with $\Delta\Y^{(\tau)} \bydef - \sigma \left.\frac{\partial E}{\partial \Y}\right|_{\Y^{(\tau)}}$. However, this only converges if the ratio $\beta/\alpha$ is small (for large problems, it needs to be very small). For the original elastic net this is easily seen for a net with two centroids: the gradient of the tension term at each centroid points towards the other, and if the gradient step is larger than the distance between both centroids, the algorithm diverges exponentially (as we have confirmed by simulation). This could be alleviated by taking a shorter gradient step (by taking $\eta \Delta\Y^{(\tau)}$ with $\eta < 1$), but for large $\beta/\alpha$ the step would become very small.

\subsection{Iterative matrix method}

Equating the gradient of eq.~\eqref{e:ann:grad} to zero we obtain the nonlinear system
\begin{equation}
  \label{e:ann:zerograd}
  \Y \A = \X \W \qquad \text{with} \qquad \A \bydef \G + \sigma \frac{\beta}{\alpha} \left( \frac{\SS+\SS^T}{2} \right).
\end{equation}
This equation is similar to others obtained under related modelling assumptions \citep{Yuille_96a,Dayan93a,Bishop_98a,Bishop_98b,Utsugi97a}. \A\ is a symmetric positive definite $M \times M$ matrix (since \SS\ is positive (semi)definite), and both \A\ (through \G) and \W\ depend on \Y. This equation is the basis for a fixed-point iteration, where we solve for \Y\ with \G\ and \W\ fixed, then update \G\ and \W\ for the new \Y, and repeat till convergence. Convergence is guaranteed since we can derive an analogous procedure as an EM algorithm that optimises over \Y\ for fixed $\sigma$ as in \citet{Yuille_94a} and \citet{Utsugi97a}. Essentially, eq.~\eqref{e:ann:zerograd} becomes the familiar Gaussian-mixture update for the means with $\beta = 0$. Thus, from the algorithmic point of view, it is the addition of $\sigma \frac{\beta}{\alpha} \left( \frac{\SS+\SS^T}{2} \right)$ that determines the topology.

For the class of matrices \SS\ studied in this paper, \A\ is a large sparse matrix (of the order of $10^4 \times 10^4$) so that explicitly computing its inverse is out of the question: besides being a numerically ill-posed operation, it is computationally very costly in time and memory, since $\A^{-1}$ is a large, nonsparse matrix. Instead, we can solve the system for \Y\ (with \A\ and \W\ fixed) by an iterative matrix method (see e.g.\ \citealp{IsaacsKeller66,Smith85a}). Rearrange the system~\eqref{e:ann:zerograd} as $\A\x = \b$ (calling \x\ the unknowns) with $\A = \D - \LL - \U$ (diagonal $-$ lower triangular $-$ upper triangular) to give an iterative procedure $\x^{(\tau+1)} = \C \x^{(\tau)} + \c$, a few iterations of which will usually be enough to compute \x\ (or \Y\ in eq.~\eqref{e:ann:grad}) approximately, assuming the procedure converges. The following procedures are common:
\begin{description}
\item[Jacobi] Decomposing the equation as $\D\x = (\LL+\U)\x + \b$ results in the iterative procedure $\x^{(\tau+1)} = \D^{-1} (\LL+\U) \x^{(\tau)} + \D^{-1} \b$. Since \D\ is diagonal, $\D^{-1}(\LL+\U)$ can be efficiently computed and remains sparse.
\item[Gauss-Seidel] Decomposing the equation as $(\D-\LL)\x = \U\x + \b$ results in the iterative procedure $\x^{(\tau+1)} = (\D-\LL)^{-1} \U \x^{(\tau)} + (\D-\LL)^{-1} \b$, which can be implemented without explicitly computing $(\D-\LL)^{-1}$ if instead of computing $\x^{(\tau+1)}$ from $\x^{(\tau)}$, we use the already computed elements $x^{(\tau+1)}_1,\dots,x^{(\tau+1)}_i$ as well as the old ones $x^{(\tau)}_{i+2},\dots$ to compute $x^{(\tau+1)}_{i+1}$; see e.g.\ eq.~\eqref{e:origEN:GS}.
\item[Successive overrelaxation (SOR)] is also possible and can be faster, but requires setting of the relaxation parameter by trial and error.
\end{description}
For sparse matrices, both Jacobi and Gauss-Seidel are particularly fast and respect the sparsity structure at each iteration, without introducing extra nonzero elements. Both methods are quite similar, though for the kind of sparse matrix \A\ that we have with the elastic net, Gauss-Seidel should typically be about twice as fast than Jacobi
and requires keeping just one copy of \Y\ in memory.

The matrix iterates can be interleaved with the updates for \G\ and \W.

\subsection{Direct method by Cholesky factorisation}

We can obtain \Y\ efficiently and robustly by solving the sparse system of equations~\eqref{e:ann:zerograd} by Gaussian elimination via Cholesky decomposition%
\footnote{David Willshaw (pers.~comm.) has also developed independently the idea of the Cholesky factorisation for the original elastic net.},
since \A\ is symmetric and positive definite. Specifically, the procedure consists of \citep{GeorgeLiu81a,Duff_86a}:
\begin{enumerate}
\item \emph{Ordering}: find a good permutation \PP\ of the matrix \A.
\item \emph{Cholesky factorisation}: factorise the permuted matrix $\PP\A\PP^T$ into $\LL\LL^T$ where \LL\ is lower triangular with nonnegative diagonal elements.
\item \emph{Triangular system solution}: solve $\LL\Y_0 = \PP\X\W$ and $\LL^T\Y_1 = \Y_0$ both by Gaussian elimination in $\calO(M)$ and then set $\Y = \PP^T\Y_1$.
\end{enumerate}
Step 1 is not strictly necessary but usually accelerates the procedure for sparse matrices. This is because, although the Cholesky factorisation does not add zeroes outside the bands of \A\ (and thus preserves its banded structure), it may add new zeroes inside (i.e., add new ``fill''), and it is possible to reduce the number of zeroes in \LL\ by reordering the rows of \A. However, the exact minimisation of the fill is NP-complete and so one has to use a heuristic ordering method, a number of which exist, such as minimum degree ordering \citep[pp.~115--137]{GeorgeLiu81a}.

\subsection{Method comparison}

The gradient descent method as defined earlier converges only for $\beta/\alpha$ values smaller than a certain threshold that is extremely small for large nets; e.g.\ \citet{DurbinWillsha87a} used $\beta/\alpha = 10$ for a 2D TSP problem of $N = 100$ cities and a net with $M = 250$ centroids; but \citet{DurbinMitchis90a} used $\beta/\alpha = 0.00025$ for a 4D cortical map model of $M = 1\,600$ centroids. For large problems, a tiny ratio $\beta/\alpha$ means that---particularly for fast annealing---the tension term has very little influence on the final net, and as a result there is no benefit in using one matrix \SS\ over another.

A sufficient condition for the convergence of both matrix iteration methods (Jacobi and Gauss-Seidel) is that the matrix \A\ be positive definite (which it is). Also, the Cholesky factorisation is stable without pivoting for all symmetric positive definite matrices, although pivoting is advisable for positive semidefinite matrices \citep{GolubLoan96a}. Thus, all three methods are generally appropriate. However, for high values of $\beta/\alpha$ and if \SS\ is positive semidefinite, then \A\ becomes numerically close to being positive semidefinite (even in later stages of training, when \G\ has large diagonal elements). In this case, the Jacobi and Gauss-Seidel methods may diverge, as we have observed in practice (and converged if $\beta/\alpha$ was lowered). In contrast, we have never found the Cholesky factorisation method to diverge, at least for the largest problems we have tried with $N,M \approx 2 \cdot 10^4$, $D = 5$, stencils of up to order $p = 4$ and $\beta/\alpha$ up to $10^6$.

The Cholesky factorisation is a direct method, computed in a finite number of operations $\calO(\frac{1}{6}M^3)$ for dense \A\ but much less for sparse \A. This is unlike the Jacobi or Gauss-Seidel methods, which are iterative and in principle require an infinite number of iterations to converge (although in practice a few may be enough). With enough iterations (a few tens, for the problems we tried), the Jacobi method converges to the solution of the Cholesky method; with as few as $5$ to $10$, it gets a reasonably good one. In general, the computation time required for the solution of eq.~\eqref{e:ann:zerograd} by the Cholesky factorisation is usually about twice as high as that of the Jacobi method. The bottleneck of all elastic net learning algorithms is the computation of the weight matrix \W\ of squared distances between cities and centroids, which typically takes several times longer than solving eq.~\eqref{e:ann:zerograd}. Thus, using the Cholesky factorisation only means an increase of around $10$--$20$\% of the total computation time.

Regarding the quality of the iteration, the Cholesky method (and, for enough iterations, the Jacobi and Gauss-Seidel methods) goes deeper down the energy surface at each annealing iteration; this is particularly noticeable when the net collapses into its centre of mass for large $\sigma$. In contrast, gradient descent takes tiny steps and requires many more iterations for a noticeable change to occur. It might then be possible to use faster annealing than with the gradient method, with considerable speedups. On the other hand, the Cholesky method is not appropriate for online learning, where training points come one at a time, because the net would change drastically from one training point to the next. However, this is not a problem practically since the stream of data can always be split into chunks of appropriate size.

In summary, the robustness and efficiency of the (sparse) Cholesky factorisation make it our method of choice and allow us to investigate the behaviour of the model for a larger range of $\beta$ values than has previously been possible. All the simulations in this paper use this method.

\subsection{Practical extensions}

Here we describe two practically convenient modifications of the basic elastic net model.

\subsubsection{Weighting points of the training set}
\label{s:alpha_n}

For some applications it may be convenient to define a separate $\alpha_n \ge 0$ for each data point $\x_n$ (e.g.\ to overrepresent some data points without having to add extra copies of each point, which would make $N$ larger). In this case, the energy becomes
\begin{equation*}
  E(\Y,\sigma) = -\sigma\sum^N_{n=1}{ \alpha_n\log{ \sum^M_{m=1}{e^{-\frac{1}{2}\norm{\frac{\x_n-\y_m}{\sigma}}^2}} } } + \frac{\beta}{2} \trace{\Y^T\Y\SS}
\end{equation*}
and all other equations remain the same by defining the weights as
\begin{equation*}
  w_{nm} \bydef \alpha_n \frac{e^{-\frac{1}{2}\norm{\frac{\x_n-\y_m}{\sigma}}^2}}{\sum^M_{m'=1}{e^{-\frac{1}{2}\norm{\frac{\x_n-\y_{m'}}{\sigma}}^2}}} = \alpha_n p(m|\x_n),
\end{equation*}
that is, multiplying the old $w_{nm}$ times $\alpha_n$.

Over- or underrepresenting training set points is useful in cortical map modelling to simulate deprivation conditions (e.g.\ monocular deprivation, by reducing $\alpha_n$ for the points associated with one eye) or nonuniform feature distributions (e.g.\ cardinal orientations are overrrepresented in natural images). It is also possible to make $\alpha_n$ dependent on $\sigma$, so that e.g.\ the overrepresentation may take place at specific times during learning; this is useful to model critical periods \citep{Carreir_03a}.

\subsubsection{Introducing zero-valued mixing proportions}

We defined the fitness term of the elastic net as a Gaussian mixture with equiprobable components. Instead, we can associate a mixing proportion $\pi_m \bydef p(m)$ with each component $m$, subject to $\pi_m \in [0,1]$ $\forall m$ and $\sum^M_{m=1}{\pi_m} = 1$. In an unsupervised learning setting, we could take $\{\pi_m\}^M_{m=1}$ as parameters and learn them from the training set, just as we do with the centroids $\{\y_m\}^M_{m=1}$. However, we can also use them to disable centroids from the model selectively during training (by setting the corresponding $\pi_m$ to zero). This is a computationally convenient strategy to use non-rectangular grid shapes in 2D nets. Specifically, for each component $m$ to be disabled, we need to:
\begin{itemize}
\item Set $\pi_m = 0$ (and renormalise all $\pi_m$).
\item If using a symmetric matrix $\SS = \D^T\D$, set to zero column $m$ of \D\ and all rows $m'$ of \D\ that had a nonzero in the element corresponding to column $m$. This is equivalent to eliminating from the tension term  all linear combinations that involved $\y_m$. This implicitly assumes a specific type of boundary condition; other types of b.c.\ may be used by appropriately modifying such rows (rather than zeroing them). If \SS\ is not symmetric, the manipulations are more complicated.
\end{itemize}
It is easy to see that the energy is now independent of the value of $\y_m$: no ``force'' is exerted on $\y_m$ either from the fitness or the tension term, and likewise $\y_m$ exerts no tension on any other centroid. Thus, in the training algorithm, we can simply remove it (and the appropriate parts of \SS, etc.) from the update equations. We could also simply leave it there, since its gradient component is zero and its associated equation~\eqref{e:ann:zerograd} is zero both in the RHS and the LHS (for all $d = 1,\dots,D$). However, the latter option is computationally wasteful, since the operations associated with $\y_m$ are still being carried out, and can lead to numerical instability in the iterative-matrix and Cholesky-factorisation methods, since the matrices involved become singular (although this can be easily overcome by setting $g_m$ to some nonzero value).

When mixing proportions $\pi_m$ and training set weights $\alpha_n$ are used, the energy becomes:
\begin{equation*}
  E(\Y,\sigma) = -\sigma\sum^N_{n=1}{ \alpha_n\log{ \sum^M_{m=1}{\pi_m e^{-\frac{1}{2}\norm{\frac{\x_n-\y_m}{\sigma}}^2}} } } + \frac{\beta}{2} \trace{\Y^T\Y\SS}
\end{equation*}
and all other equations remain the same by defining the weights as
\begin{equation*}
  w_{nm} \bydef \alpha_n \frac{\pi_m e^{-\frac{1}{2}\norm{\frac{\x_n-\y_m}{\sigma}}^2}}{\sum^M_{m'=1}{\pi_{m'} e^{-\frac{1}{2}\norm{\frac{\x_n-\y_{m'}}{\sigma}}^2}}} = \alpha_n p(m|\x_n).
\end{equation*}

Selectively disabling centroids is useful in cortical map modelling to use a 2D net that approximates the shape of primary visual cortex and may include lesions (patches of inactive neurons in the cortex). It can also be used to train separate nets on the same training set and so force them to compete with each other (both with 1D or 2D nets); in fact, the central-difference stencil (section~\ref{s:cendiff-family}) leads to a similar situation by separating the tension term into decoupled subterms.

\subsection{Comparison with the original elastic net model}
\label{s:origEN}

The original elastic net results from using the matrix \D\ corresponding to a stencil $(0,\ -1,\ 1)$ and $\SS = \D^T\D$ (see section~\ref{s:D}). For example, for a 1D net with $M = 9$ centroids (where $a = 0$ for nonperiodic b.c.\ and $a = 1$ for periodic b.c.):
\begin{equation}
  \label{e:origEN:D}
  \D = \left(
    \begin{smallmatrix}
      -1 & 1  & 0  & 0  & 0  & 0  & 0  & 0  & 0  \\
      0  & -1 & 1  & 0  & 0  & 0  & 0  & 0  & 0  \\
      0  & 0  & -1 & 1  & 0  & 0  & 0  & 0  & 0  \\
      0  & 0  & 0  & -1 & 1  & 0  & 0  & 0  & 0  \\
      0  & 0  & 0  & 0  & -1 & 1  & 0  & 0  & 0  \\
      0  & 0  & 0  & 0  & 0  & -1 & 1  & 0  & 0  \\
      0  & 0  & 0  & 0  & 0  & 0  & -1 & 1  & 0  \\
      0  & 0  & 0  & 0  & 0  & 0  & 0  & -1 & 1  \\
      a  & 0  & 0  & 0  & 0  & 0  & 0  & 0  & -a
    \end{smallmatrix}
  \right) \qquad
  \SS = \left(
    \begin{smallmatrix}
      a^2+1 & -1 & 0  & 0  & 0  & 0  & 0  & 0  & -a^2  \\
      -1    & 2  & -1 & 0  & 0  & 0  & 0  & 0  & 0     \\
      0     & -1 & 2  & -1 & 0  & 0  & 0  & 0  & 0     \\
      0     & 0  & -1 & 2  & -1 & 0  & 0  & 0  & 0     \\
      0     & 0  & 0  & -1 & 2  & -1 & 0  & 0  & 0     \\
      0     & 0  & 0  & 0  & -1 & 2  & -1 & 0  & 0     \\
      0     & 0  & 0  & 0  & 0  & -1 & 2  & -1 & 0     \\
      0     & 0  & 0  & 0  & 0  & 0  & -1 & 2  & -1    \\
      -a^2  & 0  & 0  & 0  & 0  & 0  & 0  & -1 & a^2+1
    \end{smallmatrix}
  \right).  
\end{equation}
We obtain the original elastic net equations \citep{DurbinWillsha87a,Durbin_89a} as follows:
\begin{subequations}
  \label{e:origEN}
  \begin{alignat}{3}
    \label{e:origEN:E}
    \text{Energy:} & \quad & E(\Y,\sigma) & = -\alpha\sigma\sum^N_{n=1}{ \log{ \sum^M_{m=1}{e^{-\frac{1}{2}\norm{\frac{\x_n-\y_m}{\sigma}}^2}} } } + \frac{\beta}{2} \sum^M_{m=1}{\norm{\y_{m+1} - \y_m}^2} \\
    \label{e:origEN:grad}
    \text{Gradient:} & & \Delta \y_m & = \alpha \sum^N_{n=1}{w_{nm} (\x_n - \y_m)} + \beta \sigma (\y_{m+1} - 2 \y_m + \y_{m-1}) \\
    \label{e:origEN:J}
    \text{Jacobi:} & & \displaystyle \y^{(\tau+1)}_m & = \frac{\alpha \sum^N_{n=1}{w_{nm}\x_n} + \beta \sigma (\y^{(\tau)}_{m+1}+\y^{(\tau)}_{m-1})}{\alpha \sum^N_{n=1}{w_{nm}} + 2 \beta \sigma} \\
    \label{e:origEN:GS}
    \text{Gauss-Seidel:} & & \displaystyle \y^{(\tau+1)}_m & = \frac{\alpha \sum^N_{n=1}{w_{nm}\x_n} + \beta \sigma (\y^{(\tau)}_{m+1}+\y^{(\tau+1)}_{m-1})}{\alpha \sum^N_{n=1}{w_{nm}} + 2 \beta \sigma}.
  \end{alignat}
\end{subequations}
The original elastic net papers used annealing with either gradient descent \citep{DurbinWillsha87a,Durbin_89a} or Gauss-Seidel iteration \citep{DurbinMitchis90a}. For the latter, the updates of $\y_1,\dots,\y_m$ must be done sequentially on the index $m$.

\section{Construction of the matrix \SS}
\label{s:S}

The definition of the model and the optimisation algorithms given depend on the matrix \SS; we are left now with the search for a meaningful matrix \SS\ that incorporates some knowledge of the problem being modelled. For example, for the original elastic net, the tension term embodies an approximation to the length of the net, which is the basis for a heuristic solution of the TSP and for a cortical map model. In this section we discuss general properties of the matrix \SS\ and then, in the next section, concentrate on differential operators.

Firstly, note that the formulation of the tension term (ignoring the $\beta/2$ factor) as $\trace{\Y^T \Y \SS} = \sum_{m,m',d}{s_{mm'} y_{dm} y_{dm'}}$ is not the most general quadratic form of the parameters $\{y_{dm}\}^{M,D}_{m,d=1}$. This is because the cross-terms $y_{dm} y_{d'm'}$ for $d \neq d'$ have weight zero and the terms $y_{dm} y_{dm'}$ have a weight $s_{mm'}$ independent of the dimension $d$. Thus, the different dimensions of the net $d = 1,\dots,D$ (or maps) are independent and formally identical in the tension term (though not so in the fitness term). The most general quadratic form could be represented as $\vect{\Y}^T \SS \vect{\Y}$ where $\vect{\Y}$ is the column concatenation of all elements of \Y\ and \SS\ is now $MD \times MD$, or as $\sum_{m,m',d,d'}{s_{m d m' d'} y_{md} y_{m' d'}}$ where $s_{m d m' d'}$ is a 4D tensor. However, our more particular class of penalties still has a rich behaviour and we restrict ourselves to it.

Secondly, note that it is enough to consider symmetric matrices \SS\ of the form $\D^T \D$ where \D\ is arbitrary. We have that, for a nonsymmetric matrix \SS:
\begin{equation*}
  \trace{\Y^T\Y\left( \frac{\SS+\SS^T}{2} \right)} = \frac{1}{2} \left( \trace{\Y^T\Y\SS} + \trace{\Y^T\Y\SS^T} \right) = \frac{1}{2} \left( \trace{\Y^T\Y\SS} + \trace{(\Y^T\Y\SS^T)^T} \right) = \trace{\Y^T\Y\SS}
\end{equation*}
and so using \SS\ is equivalent to using the symmetric form $\frac{\SS+\SS^T}{2}$. Further, since \SS\ must be positive (semi)definite for the energy to be lower bounded, we can always write $\SS = \D^T \D$ for some real \D\ matrix, without loss of generality, and so the tension term can be written $\trace{\Y^T \Y \SS} = \norm{\D\Y^T}^2$ where $\norm{\A} \bydef \sqrt{\trace{\A^T\A}} = \sqrt{\sum_{ij}{a^2_{ij}}}$ is the Frobenius matrix norm. However, the learning algorithms given earlier can be used for any \SS.

The \D\ matrix has two functions. As \emph{neighbourhood relation}, it specifies the strength of the tension between centroids and thus the expected \emph{metric properties} of the net, such as its curvature. As \emph{adjacency matrix}, it specifies what centroids are directly connected%
\footnote{However, the notion of adjacency becomes blurred when we consider that a sparse stencil such as $(0,\ -1,\ 1)$ can be equivalent (i.e., produce the same \SS) to a nonsparse one, as we show in section~\ref{s:circ2sten}.}
and so the \emph{topology} of the net. Thus, by changing \D, we can turn a given collection of centroids into a line, a closed line, a plane sheet, a torus, etc. Practically, we typically concern ourselves with a fixed topology and are interested in the effects of changing the ``elasticity'' or ``stiffness'' of the net via the tension term.

Thirdly, for net topologies where the neighbourhood relations do not depend on the actual region in question of the net, we can represent the matrix \D\ in a compact form via a \emph{stencil} and suitable boundary conditions (b.c.). Multiplication by \D\ then becomes convolution by the stencil. Further, in section~\ref{s:circ2sten} we will show that with periodic b.c.\ we can always take \D\ to be a symmetric matrix and so $\SS = \D^2$.

In summary, we can represent quadratic tension terms in full generality through an operator matrix \D, and we will concentrate on a particular but important class of matrices \D, where different dimensions are independent and identical in form, and where the operator is assumed translationally invariant so that \D\ results from convolving with a stencil. Before considering appropriate types of stencil, we note some important invariance properties.

\subsection{Invariance of the penalty term with respect to rigid motions of the net}

Consider an orthogonal $D \times D$ matrix \RR\ and an arbitrary $D \times 1$ vector \bmu. In order that the matrix \D\ represent an operator invariant to rotations and translations of the net centroids, we must have the same penalty for \Y\ and for $\RR (\Y + \bmu \1^T)$ (where \1\ is the $M \times 1$ vector of ones):
\begin{equation*}
  \begin{split}
    \trace{\Y^T\Y\D^T\D} &= \trace{(\Y + \bmu \1^T)^T \RR^T \RR (\Y + \bmu \1^T) \D^T \D} \\
    &= \trace{\Y^T\Y\D^T\D} + \trace{\Y^T \bmu \1^T \D^T \D} + \trace{\1 \bmu^T \Y \D^T \D} + \trace{\1 \bmu^T \bmu \1^T \D^T \D}.
  \end{split}
\end{equation*}
Thus, any \D\ will provide invariance to rotations, but invariance to translations requires $\D \1 = \0$, i.e., a differential operator matrix \D\ (see later).

Another way to see this, for the circulant matrix case discussed later, is as follows: a circulant positive definite matrix $\SS = \D^T\D$ can always be decomposed as $\tilde{\SS} + k \1 \1^T$ with $\tilde{\SS}$ circulant positive semidefinite (verifying $\tilde{\SS} \1 = \0$, or $\sum^M_{m=1}{\tilde{\sss}_{1m}} = 0$) and $k = \frac{1}{M^2} \1^T \SS \1 = \frac{1}{M} \sum^M_{m=1}{\tilde{\sss}_{1m}} > 0$. This corresponds to the product of two priors, one given by the differential operator matrix $\tilde{\SS}$ and the other one by the matrix $k \1 \1^T$. The effect of the latter is to penalise non-zero-mean nets, since $\trace{\Y^T\Y(k \1 \1^T)} = k \sum^D_{d=1}{\left( \smash{\sum^M_{m=1}{y_{dm}}} \right)^2}$, which we do not want in the present context, since it depends on the choice of origin for the net. Such matrices can be desirable for smoothing applications, naturally, and the training algorithms still work in this case because \SS\ remains positive semidefinite (and so the energy function is bounded below).

\subsection{Invariance of \SS\ with respect to transformations of \D\ or the stencil}
\label{s:D:inv}

There can be different matrices \D\ that result in the same $\SS = \D^T\D$; in fact, $\D^T$ is a square root of \SS. Any matrix $\D' \bydef \RR \D$ will produce the same \SS\ if \RR\ is a $k' \times k$ matrix verifying $\RR^T \RR = \I_k$. This is the same sort of unidentifiability as in factor analysis (by rotation of the factors; \citealp{Barthol87a}), and in our case it means that the same nonnegative quadratic form can be obtained for many different matrices \D. Particular cases include:
\begin{itemize}
\item Orthogonal rotation of \D: \RR\ is square with $\RR^{-1} = \RR^T$.
\item Sign reversal of any number of rows of \D: \RR\ is diagonal with $\pm 1$ elements.
\item Permutation of rows of \D\ (i.e., reordering of the summands in the tension term): \RR\ is a permutation matrix.
\item Insertion of rows of zeroes to \D: $\RR = \RR_{k' \times k}$ is the identity $\I_k$ with $k'-k$ intercalated rows of zeroes.
\end{itemize}
When \D\ is derived from a stencil, \SS\ is invariant with respect to the following stencil transformations:
\begin{itemize}
\item Sign reversal, e.g.\ $(1,\ -2,\ 1)$ is equivalent to $(-1,\ 2,\ -1)$.
\item Shifting the stencil, since this is equivalent to permuting rows of \D. In particular, padding with zeroes the borders of the stencil (but not inserting zeroes), e.g.\ the forward difference $(0,\ -1,\ 1)$ is equivalent to the backward difference $(-1,\ 1,\ 0)$ but not to the central difference $(-1,\ 0,\ 1)$.
\end{itemize}
These invariance properties are useful in the analysis of section~\ref{s:tension} and in the implementation of code to construct the matrix \SS\ given the stencil.

\section{Construction of the matrix \D\ from a differential stencil $\varsigma$}
\label{s:D}

To represent the matrix \D\ by a stencil consider for example the original elastic net, in which the tension term consists of summing terms of the form $\norm{\y_{m+1} - \y_m}^2$ over the whole net. In this case, the stencil is $(0,\ -1,\ 1)$ and contains the coefficients that multiply $\y_{m-1}$, $\y_m$ and $\y_{m+1}$. In general, a stencil $\varsigma = (\varsigma_{-k},\varsigma_{-k+1},\dots,\varsigma_{0},\dots,\varsigma_{k-1},\varsigma_{k})$ represents a linear combination $\sum^k_{i=-k}{\varsigma_i \y_{m+i}}$ of which then we take the square. A given row of matrix \D\ is obtained by centring the stencil on the corresponding column, and successive rows by shifting the stencil. If the stencil is \emph{sparse}, i.e., has few nonzero elements, then \D\ and \SS\ are sparse (with a banded or block-banded structure).

It is necessary to specify boundary conditions when one of the stencil ends overspills near the net boundaries. In this paper we will consider only the simplest types of boundary conditions: \emph{periodic}, which uses modular arithmetic; and \emph{nonperiodic} or \emph{open}, which simply discards linear combinations (rows of \D) that overspill. In both cases the resulting \D\ matrix is a structured matrix: circulant for periodic b.c., since it is obtained by successively rotating the stencil (see section~\ref{s:tension}); or quasi-Toeplitz for nonperiodic b.c., being almost completely defined by its top row and left column. The following example illustrates these ideas for a stencil $\varsigma = (a,b,c,d,e)$ and a 1D net with $M = 7$:
\begin{equation*}
  \text{Nonperiodic b.c.: } \D = \left(
    \begin{smallmatrix}
      a & b & c & d & e & 0 & 0 \\
      0 & a & b & c & d & e & 0 \\
      0 & 0 & a & b & c & d & e
    \end{smallmatrix}
  \right) \qquad
  \text{Periodic b.c.: } \D = \left(
    \begin{smallmatrix}
      c & d & e & 0 & 0 & a & b \\
      b & c & d & e & 0 & 0 & a \\
      a & b & c & d & e & 0 & 0 \\
      0 & a & b & c & d & e & 0 \\
      0 & 0 & a & b & c & d & e \\
      e & 0 & 0 & a & b & c & d \\
      d & e & 0 & 0 & a & b & c
    \end{smallmatrix}
  \right) \qquad
  \Y^T = \left(
    \begin{smallmatrix}
      \y^T_1 \\
      \y^T_2 \\
      \vdots \\
      \y^T_7
    \end{smallmatrix}
  \right).
\end{equation*}
These ideas generalise directly to elastic nets of two or more dimensions, but the structure of \D\ is more complicated, being circulant or quasi-Toeplitz by blocks. In the analysis of section~\ref{s:tension} we consider only periodic b.c.\ for simplicity, but the examples of section~\ref{s:cmap} will include both periodic and nonperiodic b.c.

Some notational remarks. As in the earlier example, we will assume the convention that the first row of \D\ contains the stencil with its central coefficient in the first column, the coefficients to the right of the stencil occupy columns $2$, $3$, etc.\ of the matrix and the coefficients to the left occupy columns $M$, $M-1$, etc.\ respectively. The remaining rows of \D\ are obtained by successively rotating the stencil. Without loss of generality, the stencil must have an odd number of elements to avoid ambiguity in the assignation of these elements to the net points. Occasionally we will index the stencil starting from the left, e.g.\ $\varsigma = (\varsigma_{0},\varsigma_{1},\varsigma_{2},\varsigma_{3},\varsigma_{4})$. We will always assume that the number of centroids $M$ in the net (the dimension of the vector \y) is larger than the number of coefficients in the stencil and so will often need to pad the stencil with zeroes. Rather than explicitly notating all these circumstances, which would be cumbersome, we will make clear in each context which version of the indexing we are using. This convention will make easy various operations we will need later, such as rotating the stencil to build a circulant matrix \D, convolving the stencil with the net or computing the Fourier transform of the stencil. The tension term separates additively into $D$ terms, one for each row of \Y, so we consider the case $D = 1$; call $\y \bydef \Y^T$ the resulting $M \times 1$ vector. In this formulation, the vector $\D^T \y$ is the discrete convolution of the net and the stencil, while the vector $\D \y$ (which is what we use) is the discrete convolution of the net and the \emph{reversed} stencil $\overleftarrow{\varsigma} \bydef (\dots \varsigma_{2},\varsigma_{1},\varsigma_{0},\varsigma_{-1},\varsigma_{-2} \dots)$. In both cases the tension value is the same. This can be readily seen by noting that $\sum^M_{k=1}{d_{ki} d_{kj}} = \sum^M_{k=1}{d_{ik} d_{jk}}$ if \D\ is circulant, which implies $\D^T\D = \D\D^T$ and so $\norm{\D\Y^T}^2 = \norm{\D^T\Y^T}^2$; or by noting that the stencil and the reverse stencil have the same power spectrum ($\abs{\hat{\varsigma}_k}^2 = \abs{\smash{\hat{\overleftarrow{\varsigma}}_{\!\! k}}}^2$).

\subsection{Types of discretised operator}
\label{s:D:types-op}

Turning now to the choice of stencil, there are two basic types:
\begin{description}
\item[Differential, or roughening] This results when the stencil is a finite-difference approximation to a continuous differential operator, such as the forward difference approximation to the first derivative \citep{ConteBoor80a}:
  \begin{equation*}
    y'(u) \approx \frac{y(u+h) - y(u)}{h} \Longrightarrow \varsigma = (0,\ -1,\ 1)
  \end{equation*}
  where the grid constant $h$ is small. Differential operators characterise the metric properties of the function $y$, such as its curvature, and are \emph{local}, their value being given at the point in question. Consequently, the algebraic sum of the elements of a finite-difference stencil must be zero (otherwise, the operator's value value would diverge in the continuum limit, as $h \rightarrow 0$). This \emph{zero-sum condition} can also be expressed as $\D\1 = \0$, where \1\ is a vector of ones, and has the consequences that: \D\ is rank-defective (having a zero eigenvalue with eigenvector \1); \SS\ is positive semidefinite and the prior on the centroids is improper; and the tension term is invariant to rigid motions of the centroids (translations and rotations). Note the fitness term is invariant to permutations of \Y\ but not to rigid motions.
\item[Integral, or smoothing] This results when the stencil is a Newton-Cotes quadrature formula, such as Simpson's rule \citep{ConteBoor80a}:
  \begin{equation*}
    \int^{u+2h}_{u}{y(v)\, dv} \approx \left(\frac{y(u) + 4y(u+h) + y(u+2h)}{3}\right) h \Longrightarrow \varsigma = \frac{1}{3} (0,\ 0,\ 1,\ 4,\ 1).
  \end{equation*}
Integral operators are \emph{not local}, their value being given by an interval. Thus, the corresponding stencil is not zero-sum.
\end{description}
Integral operators depend on the choice of origin (i.e., the ``DC value'') and so do not seem useful to constrain the geometric form of an elastic net, nor would they have an obvious biological interpretation (although, of course, they are useful as smoothing operators). Differential operators, in contrast, can be readily interpreted in terms of continuity, smoothness and other geometric properties of curves---which are one of the generally accepted principles that govern cortical maps, together with uniformity of coverage of the stimuli space \citep{Swindal91a,Swindal96a,Swindal_00a,CarreirGoodhil02b}. Note that the effect of the tension term is opposite to that of the operator, e.g.\ a differential operator will result in a penalty for nonsmooth nets.

Also, let us briefly consider a prior that is particularly simple and easy to implement: $\SS = \I$, or a tension term proportional to $\sum_{md}{y^2_{md}}$, resulting from a stencil $\varsigma = (1)$. By placing it on the parameters of a mapping, this prior has often been used to regularise mappings, as in neural nets (weight decay, or ridge regression; \citealp{Bishop95a}) or GTM \citep{Bishop_98a}. In these cases, such a prior performs well because even if the weights are forced to small magnitudes, the class of functions they represent is still large and so the data can be modeled well. In contrast, in the elastic net the prior is not on the weights of a parameterized mapping but on the values of the mapping itself, and so the effect is disastrous: first, it biases the centroids towards the origin irrespectively of the location of the training set in the coordinate system; second, there is no topology because the prior factorizes over all $y_{md}$.

We now concentrate on \emph{differential} stencils, i.e., $\sum_n{\varsigma_n f_n}$ will approximate a given derivative of a continuous function $f(u)$ through a regularly spaced sample of $f$. In essence, the stencil is just a finite difference scheme. Particular cases of such differential operators have been applied in related work. \citet{Dayan93a} proposed to construct topology matrices by selecting different ways of reconstructing a net node from its neighbours (via a matrix \calE). This strategy is equivalent to using a \D\ matrix defined as $\D = \I - \calE$. For the examples he used, $\calE = (0,\ 0,\ 1)$ corresponds to $\D = (-1,\ 1,\ 0)$ (the original elastic net: forward difference of order $1$); and $\calE = \left(\frac{1}{2},\ 0,\ \frac{1}{2}\right)$ corresponds to $\D = \left(-\frac{1}{2},\ 1,\ -\frac{1}{2}\right)$ (forward difference of order $2$). The latter was also used by \citet{Utsugi97a} and \citet{Pascual_01a}. Here we seek a more systematic way of deriving stencils that approximate a derivative of order $p$. Tables~\ref{t:D:diff:1D}-\ref{t:D:diff:2D} give several finite-difference schemes for functions (nets) of one and two dimensions, respectively.

\subsection{Truncation error of a differential stencil}

Methods for obtaining finite-difference schemes such as truncated Taylor expansions (the method of undetermined coefficients), interpolating polynomials or symbolic techniques can be found in numerical analysis textbooks (see e.g.\ \citealp[section~7.1]{ConteBoor80a}; \citealp[chapter~4 and appendix~B]{GeraldWheatl94a}; \citealp[\S~11.6]{GodunovRyaben87a}). Here we consider Taylor expansions. Given a differential stencil $\varsigma$, we can determine what derivative it approximates and compute its truncation error by using a Taylor expansion around $u$:
\begin{equation*}
  f(u + mh) = \sum^R_{r=0}{f^{(r)}(u) \frac{(m h)^r}{r!}} + f^{(R+1)}(\xi) \frac{(m h)^{R+1}}{(R+1)!} \qquad m \in \bbZ, \ \xi \in (u,u+mh).
\end{equation*}
Consider a centre-aligned stencil $\varsigma = (\dots,\varsigma_{-1},\varsigma_0,\varsigma_1,\dots)$ and define $\overleftarrow{\varsigma}(u) = \sum_m{\overleftarrow{\varsigma}_{\!\! m} \delta(u - mh)}$ over \bbR\ where $\overleftarrow{\varsigma}_{\!\! m} = \varsigma_{-m}$. Then:
\begin{equation}
  \begin{split}
    \label{e:D:Taylor}
    (\overleftarrow{\varsigma} \ast f)(u) &= \int^{\infty}_{-\infty}{\sum_m{\overleftarrow{\varsigma}_{\!\! m} \delta(v - mh)} f(u-v) \, dv} = \sum_m{\varsigma_{-m} \int^{\infty}_{-\infty}{\delta(v - mh) f(u-v) \, dv}} = \sum_m{\varsigma_{-m} f(u - mh)} \\
    &= \sum_m{\varsigma_m f(u + mh)} = \sum^R_{r=0}{\alpha_r h^r f^{(r)}(u)} + \alpha_{R+1} h^{R+1} f^{(R+1)}(\xi) \qquad \text{with } \alpha_r = \frac{1}{r!} \sum_m{\varsigma_m m^r}.
  \end{split}
\end{equation}
Call $p$ and $q$ the smallest integers such that $\alpha_p, \alpha_q \neq 0$ and $\alpha_r = 0$ for $r < p$ and $p < r < q$. Then, doing $q = R+1$ gives
\begin{equation}
  \label{e:D:trunc}
  f^{(p)}(u) = \frac{(\overleftarrow{\varsigma} \ast f)(u)}{\alpha_p h^p} + \epsilon(h) \qquad \text{with} \qquad \epsilon(h) \bydef -f^{(q)}(\xi) \frac{\alpha_q}{\alpha_p} h^{q-p}
\end{equation}
where $\xi$ is ``near'' $u$. That is, $\varsigma$ represents a derivative of order $p$ with a truncation error $\epsilon(h)$ of order $q-p \ge 1$ (note that for any differential stencil $p \ge 1$ and so $\alpha_0 = \sum_m{\varsigma_m} = 0$, the zero-sum condition). Conversely, to construct a stencil that approximates a derivative of order $p$ with truncation error of order $q - p$ we simply solve for the coefficients $\varsigma_m$ given the values $\alpha_r$.

Obviously, a derivative of order $p$ can be approximated by many different stencils (with different or the same truncation error).
From a numerical analysis point of view, one seeks stencils that have a high error order (so that the approximation is more accurate) and as few nonzero coefficients as possible (so that the convolution $\overleftarrow{\varsigma} \ast f$ can be efficiently computed); for example, for the first-order derivative the central-difference stencil $\left(-\frac{1}{2},\ 0,\ \frac{1}{2}\right)$, which has quadratic error, is preferable to the forward-difference stencil $(0,\ -1,\ 1)$, which has linear error---in fact, the central difference is routinely used by mathematical software to approximate gradients.

However, from the point of view of the GEN, the nonuniqueness of the stencil raises an important question: can we expect different stencils of the same derivative order to behave similarly? Surprisingly, the answer is no (see section~\ref{s:tension:families}).

Before proceeding, it is important to make a clarification. It is well known that estimating derivatives from noisy data (e.g.\ to locate edges in an image) is an ill-posed problem. The derivatives we compute are not ill-posed because the net (given by the location of the centroids) is not a noisy function---the value of the tension term is computed exactly.

\begin{table}[th!]
  \begin{center}
    \renewcommand{\arraystretch}{1.5}
    \begin{tabular}{|c|c|l|c|}
\hline
Order $p$ & Stencil $\times\ h^p$ & \multicolumn{1}{c|}{Error term} & Key \\
\hline
\hline
$1$ & $(0,\ -1,\ 1)$ & $- y''(\xi) \frac{h}{2}$ & 1A \\
$1$ & $(-1,\ 1,\ 0)$ & $\phantom{-} y''(\xi) \frac{h}{2}$ & 1A \\
$1$ & $\frac{1}{2} (0,\ 0,\ -3,\ 4,\ -1)$ & $\phantom{-} y'''(\xi) \frac{h^2}{3}$ & 1C \\
$1$ & $\frac{1}{2} (-1,\ 0,\ 1)$ & $- y'''(\xi) \frac{h^2}{6}$ & 1B \\
$1$ & $\frac{1}{12} (0,\ 0,\ 0,\ 0,\ -25,\ 48,\ -36,\ 16,\ -3)$ & $\phantom{-} y^{\mathrm{v}}(\xi) \frac{h^4}{5}$ & 1D \\
$1$ & $\frac{1}{12} (0,\ 0,\ -3,\ -10,\ 18,\ -6,\ 1)$ & $- y^{\mathrm{v}}(\xi) \frac{h^4}{20}$ & 1E \\
$1$ & $\frac{1}{12} (1,\ -8,\ 0,\ 8,\ -1)$ & $\phantom{-} y^{\mathrm{v}}(\xi) \frac{h^4}{30}$ & 1F \\
$1$ & $\frac{1}{60} (0,\ 0,\ 2,\ -24,\ -35,\ 80,\ -30,\ 8,\ -1)$ & $\phantom{-} y^{\mathrm{vii}}(\xi) \frac{h^6}{105}$ & 1G \\
$1$ & $\frac{1}{60} (-1,\ 9,\ -45,\ 0,\ 45,\ -9,\ 1)$ & $- y^{\mathrm{vii}}(\xi) \frac{h^6}{140}$ & 1H \\
\hline
$2$ & $(0,\ 0,\ 1,\ -2,\ 1)$ & $- y'''(\xi) h$ & 2A \\
$2$ & $(0,\ 0,\ 0,\ 2,\ -5,\ 4,\ -1)$ & $\phantom{-} y^{\mathrm{iv}}(\xi) \frac{11 h^2}{12}$ & 2C \\
$2$ & $\frac{1}{4} (1,\ 0,\ -2,\ 0,\ 1)$ & $- y^{\mathrm{iv}}(\xi) \frac{h^2}{3}$ & 2B \\ 
$2$ & $(1,\ -2,\ 1)$ & $- y^{\mathrm{iv}}(\xi) \frac{h^2}{12}$ & 2A \\
$2$ & $\frac{1}{12} (0,\ 0,\ 0,\ 0,\ 35,\ -104,\ 114,\ -56,\ 11)$ & $- y^{\mathrm{v}}(\xi) \frac{5h^3}{6}$ & 2D \\
$2$ & $\frac{1}{12} (0,\ 0,\ 11,\ -20,\ 6,\ 4,\ -1)$ & $\phantom{-} y^{\mathrm{v}}(\xi) \frac{h^3}{12}$ & 2E \\
$2$ & $\frac{1}{12} (-1,\ 16,\ -30,\ 16,\ -1)$ & $\phantom{-} y^{\mathrm{vi}}(\xi) \frac{h^4}{90}$ & 2F \\
$2$ & $\frac{1}{180} (0,\ 0,\ -13,\ 228,\ -420,\ 200,\ 15,\ -12,\ 2)$ & $- y^{\mathrm{vii}}(\xi) \frac{h^5}{90}$ & 2G \\
$2$ & $\frac{1}{180} (2,\ -27,\ 270,\ -490,\ 270,\ -27,\ 2)$ & $- y^{\mathrm{viii}}(\xi) \frac{h^6}{560}$ & 2H \\
\hline
$3$ & $(0,\ 0,\ 0,\ -1,\ 3,\ -3,\ 1)$ & $- y^{\mathrm{iv}}(\xi) \frac{3 h}{2}$ & 3A \\
$3$ & $\frac{1}{8} (-1,\ 0,\ 3,\ 0,\ -3,\ 0,\ 1)$ & $- y^{\mathrm{v}}(\xi) \frac{h^2}{2}$ & 3B \\ 
$3$ & $\frac{1}{2} (-1,\ 2,\ 0,\ -2,\ 1)$ & $- y^{\mathrm{v}}(\xi) \frac{h^2}{4}$ & 3C \\
$3$ & $\frac{1}{2} (0,\ 0,\ -3,\ 10,\ -12,\ 6,\ -1)$ & $\phantom{-} y^{\mathrm{v}}(\xi) \frac{h^2}{4}$ & 3D \\
$3$ & $\frac{1}{8} (0,\ 0,\ -1,\ -8,\ 35,\ -48,\ 29,\ -8,\ 1)$ & $- y^{\mathrm{vii}}(\xi) \frac{h^4}{15}$ & 3E \\
$3$ & $\frac{1}{8} (1,\ -8,\ 13,\ 0,\ -13,\ 8,\ -1)$ & $\phantom{-} y^{\mathrm{vii}}(\xi) \frac{7h^4}{120}$ & 3F \\
\hline
$4$ & $(0,\ 0,\ 0,\ 0,\ 1,\ -4,\ 6,\ -4,\ 1)$ & $- y^{\mathrm{v}}(\xi) 2h$ & 4A \\ 
$4$ & $(1,\ 0,\ -4,\ 0,\ 6,\ 0,\ -4,\ 0,\ 1)$ & $- y^{\mathrm{vi}}(\xi) \frac{2h^2}{3}$ & 4B \\ 
$4$ & $(1,\ -4,\ 6,\ -4,\ 1)$ & $- y^{\mathrm{vi}}(\xi) \frac{h^2}{6}$ & 4A \\
$4$ & $\frac{1}{6} (-1,\ 12,\ -39,\ 56,\ -39,\ 12,\ -1)$ & $\phantom{-} y^{\mathrm{viii}}(\xi) \frac{7h^4}{240}$ & 4C \\
\hline
    \end{tabular}
    \caption{A gallery of 1D discretised differential operators obtained from finite difference schemes. $y$ is a 1D function of a 1D independent variable $u$ evaluated at points in a regular grid with interpoint separation $h \bydef u_{m+1} - u_m$, so that $y_m \bydef y(u_m)$, and $\xi$ is an unknown point near $u_m$. The key in the last column refers to figure~\ref{f:stencil-families1Db}. Example: for 3A we have $y'''_m = \frac{y_{m+3} - 3y_{m+2} + 3y_{m+1} - y_m}{h^3} - y^{\mathrm{iv}}(\xi) \frac{3 h}{2}$.}
    \label{t:D:diff:1D}
  \end{center}
\end{table}

\begin{table}[th!]
  \begin{center}
    \renewcommand{\arraystretch}{2.5}
    \begin{tabular}{|c|c|c|c|c|}
\hline
Order $p$ & Operator & Stencil $\times\ h^p$ & Error term & Key \\
\hline
\hline
2 & Laplacian $\nabla^2{y}$ & $\left(\begin{smallmatrix}   & 1 &   \\ 1 & -4 & 1 \\   & 1 &   \end{smallmatrix}\right)$ & $\calO(h^2)$ & $\nabla^2_{+}$ \\
2 & Laplacian $\nabla^2{y}$ & $\frac{1}{2} \left(\begin{smallmatrix} 1 &   & 1 \\   & -4 &   \\ 1 &   & 1 \end{smallmatrix}\right)$ & $\calO(h^2)$ & $\nabla^2_{\times}$ \\
2 & Laplacian $\nabla^2{y}$ & $\frac{1}{6} \left(\begin{smallmatrix} 1 & 4 & 1 \\ 4 & -20 & 4 \\ 1 & 4 & 1 \end{smallmatrix}\right)$ & $\calO(h^6)$ & $\nabla^2_9$ \\
2 & Laplacian $\nabla^2{y}$ & $\frac{1}{12} \left(\begin{smallmatrix} & & -1 & & \\ & & 16 & & \\ -1 & 16 & -60 & 16 & -1 \\ & & 16 & & \\ & & -1 & & \end{smallmatrix}\right)$ & $\calO(h^4)$ & $\nabla^2_{\text{a}}$ \\[2ex]
\hline
4 & Biharmonic $\nabla^4{y}$ & $\left(\begin{smallmatrix} & & 1 & & \\ & 2 & -8 & 2 & \\ 1 & -8 & 20 & -8 & 1 \\ & 2 & -8 & 2 & \\ & & 1 & & \end{smallmatrix}\right)$ & $\calO(h^2)$ & $\nabla^4_{\text{a}}$ \\[2ex]
4 & Biharmonic $\nabla^4{y}$ & $\frac{1}{6} \left(\begin{smallmatrix} & & & -1 & & & \\ & & -1 & 14 & -1 & & \\ & -1 & 20 & -77 & 20 & -1 & \\ -1 & 14 & -77 & 184 & -77 & 14 & -1 \\ & -1 & 20 & -77 & 20 & -1 & \\ & & -1 & 14 & -1 & & \\ & & & -1 & & & \end{smallmatrix}\right)$ & $\calO(h^4)$ & $\nabla^4_{\text{b}}$ \\[4ex]
\hline
    \end{tabular}
    \caption{A gallery of 2D discretised differential operators obtained from finite difference schemes. $y$ is a 1D function of a 2D independent variable \uu\ evaluated at points in a regular square grid with interpoint separation $h \bydef \norm{\uu_{i+1,j} - \uu_{ij}} = \norm{\uu_{i,j+1} - \uu_{ij}}$, so that $y_{ij} \bydef y(\uu_{ij})$. The 2D Laplacian operator is $\nabla^2{y} \bydef \frac{\partial^2 y}{\partial u^2_1} + \frac{\partial^2 y}{\partial u^2_2}$ and the 2D biharmonic operator is $\smash[t]{\nabla^4{y} \bydef (\nabla^2)^2{y} = \frac{\partial^4 y}{\partial u^4_1} + 2 \frac{\partial^4 y}{\partial u^2_1\partial u^2_2} + \frac{\partial^4 y}{\partial u^4_2}}$. The key in the last column refers to figure~\ref{f:stencil-families2Db}. Example: for $\nabla^2_{+}$ we have $\nabla^2{y_{ij}} = \frac{y_{i+1,j} + y_{i-1,j} + y_{i,j+1} + y_{i,j-1} - 4y_{ij}}{h^2} + \calO(h^2)$. Zero-valued coefficients are omitted in the stencil representation.}
    \label{t:D:diff:2D}
  \end{center}
\end{table}

\section{Analysis of the tension term in 1D}
\label{s:tension}

We have a way of constructing stencils (or convolution kernels) that represent an arbitrary differential operator, and from a stencil the corresponding \D\ matrix that represents the discrete convolution with the net. The tension term value is then the summed value in norm $L_2$ of this convolution. In this section we theoretically analyse the character of the different stencils. We begin by comparing the discrete net with the continuum limit and noting the extra degree of freedom that the choice of stencil introduces in the discrete case as opposed to the unique definition of derivative in the continuous one. The subsequent analysis is based on the Fourier spectrum of the stencil (or equivalently the eigenspectrum of the \SS\ matrix) for the case of circulant matrices \D\ (i.e., periodic b.c.). We characterise the behaviour of families of stencils, based on the forward and central difference, respectively, and show how the former but not the latter matches the behaviour in the continuous case. In particular, we show that the frequency content of the net (the stripe width for cortical maps) moves towards higher frequencies as the stencil order increases in the forward-difference family, while for the central-difference family the net has the highest frequency for any order. We also show that the different stencils can be rewritten as Mexican-hat kernels with progressively more oscillations as the order increases. In particular, this means that the original elastic net, motivated by wirelength arguments, is equivalent to an excitatory-inhibitory Mexican hat. For the most part, we will consider 1D nets for simplicity, although many of the results carry over to the $L$\nobreakdash-\hspace{0pt}dimensional case.

While the behaviour of the GEN is given by the joint effect of the fitness and tension terms of its energy function, a separate analysis of the tension term gives insight into the character of the minima of the energy. Besides, it makes explicit the differences between the continuous and the discrete formulations of the net.

\subsection{Continuous case vs discrete case}
\label{s:tension:cont}

Let us consider the tension term of the energy function~\eqref{e:ann:E} of the GEN with $\SS \bydef \D^T\D$ and again ignoring the $\frac{\beta}{2}$ factor, $\trace{\Y^T\Y\SS} = \norm{\D\Y^T}^2$. Consider a 1D continuous net $\y(u) \bydef (y_1(u),\dots,y_D(u))^T \in \bbR^D$ depending on a continuous variable $u$ that takes values in \bbR. We can express the tension term as a sum of $D$ terms, one for each function $y_d$, where each term is of the form:
\begin{equation}
  \label{e:tension-cont}
  \int^{\infty}_{-\infty}{(\calD y(u))^2 \, du} = \norm{\calD y}^2_2
\end{equation}
where $\norm{\cdot}_2$ is the $L_2$\nobreakdash-norm in the space of square-integrable functions. \calD\ represents a linear differential operator, such as the derivative of order $p$, $\calD_p \bydef \frac{d^p}{du^p}$. Since the tension term must be kept small, a function $y$ having much local variation over \bbR\ will incur a high penalty and will not likely result in a mimimum of the energy. Any function $y_0$ belonging to the nullspace of the operator \calD, i.e., satisfying $\calD y_0(u) = 0$ for all $u \in \bbR$, incurs zero penalty, and so $y+y_0$ incurs the same penalty as $y$. For example, the nullspace of $\calD_p$ consists of the polynomials of degree less than $p$. However, the penalty due to the fitness term must also be taken into account, so that the minima of the energy generally are not in the nullspace of \calD.

When the fitness term is also quadratic, such as $\int{(y-g)^2 \, du}$ for fixed $g$ in a regression, regularisation problems like this can be approached from the point of view of function approximation in Hilbert spaces. Under suitable conditions, there is a unique minimiser given by a spline \citep{Wahba90a}. However, in our case the fitness term is not quadratic but results from Gaussian-mixture density estimation. In general, the variational problem of density estimation subject to derivative penalties is not solved analytically \citep[pp.~110ff]{Silver86a}.

We can still get insight by working in the Fourier domain. Call $\hat{y}$ the Fourier transform of $y$ (see appendix~\ref{s:fourier}). By applying Parseval's theorem to the continuous tension term~\eqref{e:tension-cont} with $p$th-order derivative $\calD_p$, we obtain that the tension energy is the same in both domains:
\begin{equation*}
  \int^{\infty}_{-\infty}{\left(\frac{d^p y}{du^p}\right)^2 \, du} = \int^{\infty}_{-\infty}{\abs{(i 2 \pi k)^p \hat{y}(k)}^2 \, dk} = \int^{\infty}_{-\infty}{\underbrace{(2 \pi k)^{2p}}_{\text{filter}} \underbrace{\abs{\hat{y}(k)}^2}_{\text{power}} \, dk}
\end{equation*}
since the Fourier transform of $\frac{d^p y}{du^p}$ is $(i 2 \pi k)^p \hat{y}(k)$. This means that $\calD_p$ is acting as a high-pass filter, where the cutoff frequency increases monotonically with $p$; see fig.~\ref{f:1Dcderiv}. Therefore, high-frequency functions will incur a high penalty and the minima of the energy will likely have low frequencies---again, subject to the effect of the fitness term.

Using the convolution theorem, we obtain that the inverse Fourier transform of $(i 2 \pi k)^p$ is $\varsigma = \frac{d^p \delta}{du^p}$, where $\delta$ is the delta function, and so we can write the tension term as a convolution with this filter, $\calD y = \varsigma \ast y$. This makes explicit the relation with the discrete case, which is the one we actually implement in the GEN, by discretising both $\varsigma$ and $y$. However, the discrete case is not an exact correlate of the continuous one because the choice of discrete differential operator (stencil) introduces an extra degree of freedom. This choice can result in unexpected results. In this paper we consider finite-difference stencils $\varsigma$ such that the discrete convolution approximates a derivative of order $p$. In section~\ref{s:ext:design-power} we mention an alternative definition based on the delta function.

\begin{figure}
  \begin{center}
    \begin{tabular}{@{\hspace{.03\textwidth}}c@{\hspace{.07\textwidth}}c@{}}
      \psfrag{p1}[cr][r]{$p = 1$}
      \psfrag{p2}[cr][r]{$p = 2$}
      \psfrag{p3}[cr][r]{$p = 3$}
      \psfrag{p4}[cr][r]{$p = 4$}
      \psfrag{1}{}
      \psfrag{2}{}
      \psfrag{3}{}
      \psfrag{4}{}
      \psfrag{0}[t][]{$0$}
      \psfrag{k}[t][]{$k$}
      \psfrag{P(k)}[b][t]{$P(k)$}
      \psfrag{b}[t][]{$\frac{1}{2\pi}$}
      \includegraphics[width=0.45\textwidth]{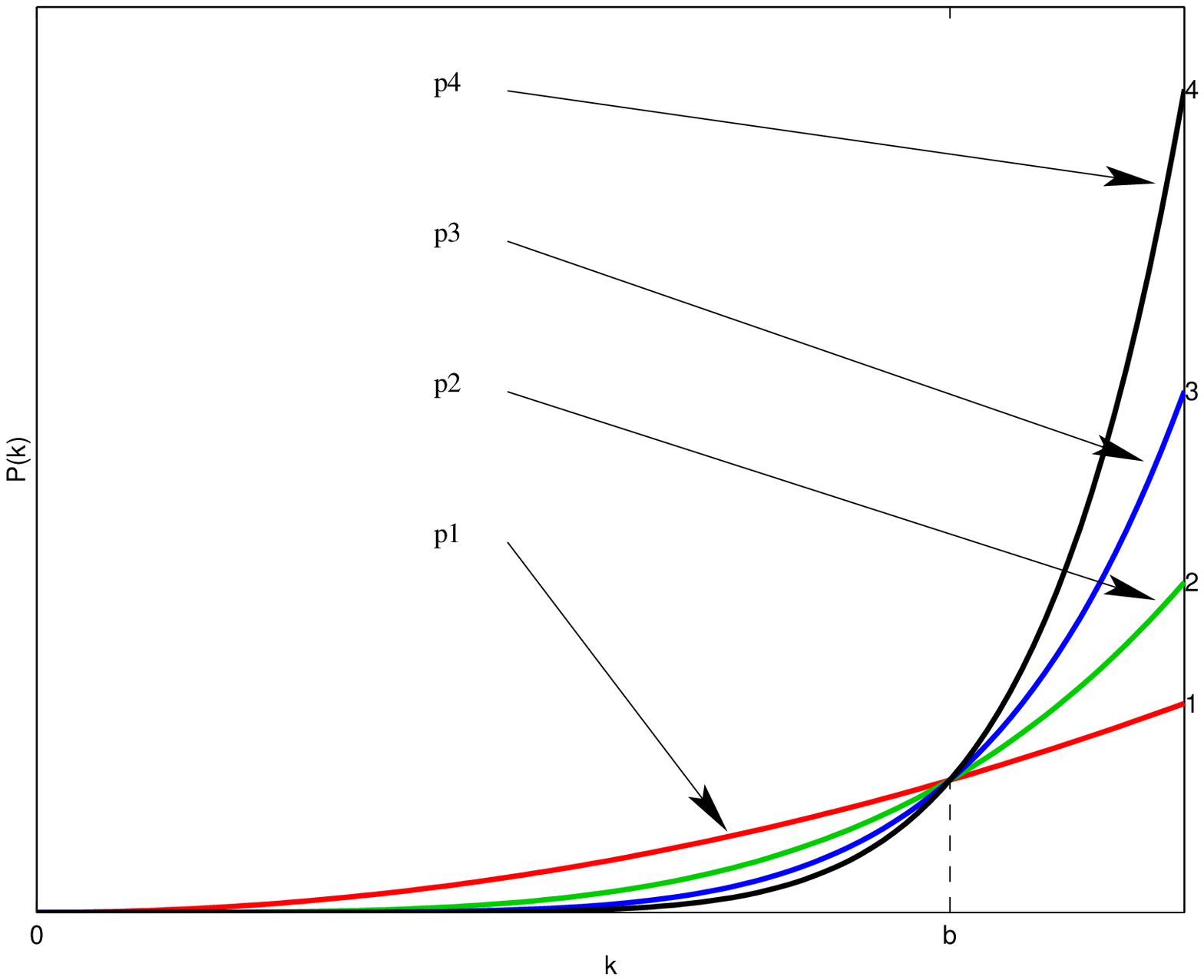} &
      \psfrag{p1}[cr][r]{$(2 \pi k)^{2p}$}
      \psfrag{p2}[r][r]{$\abs{\hat{y}(k)}^2$}
      \psfrag{p3}[r][r]{$(2 \pi k)^{2p} \abs{\hat{y}(k)}^2$}
      \psfrag{0}[t][]{$0$}
      \psfrag{k}[t][]{$k$}
      \psfrag{P(k)}[b][t]{$P(k)$}
      \includegraphics[width=0.45\textwidth]{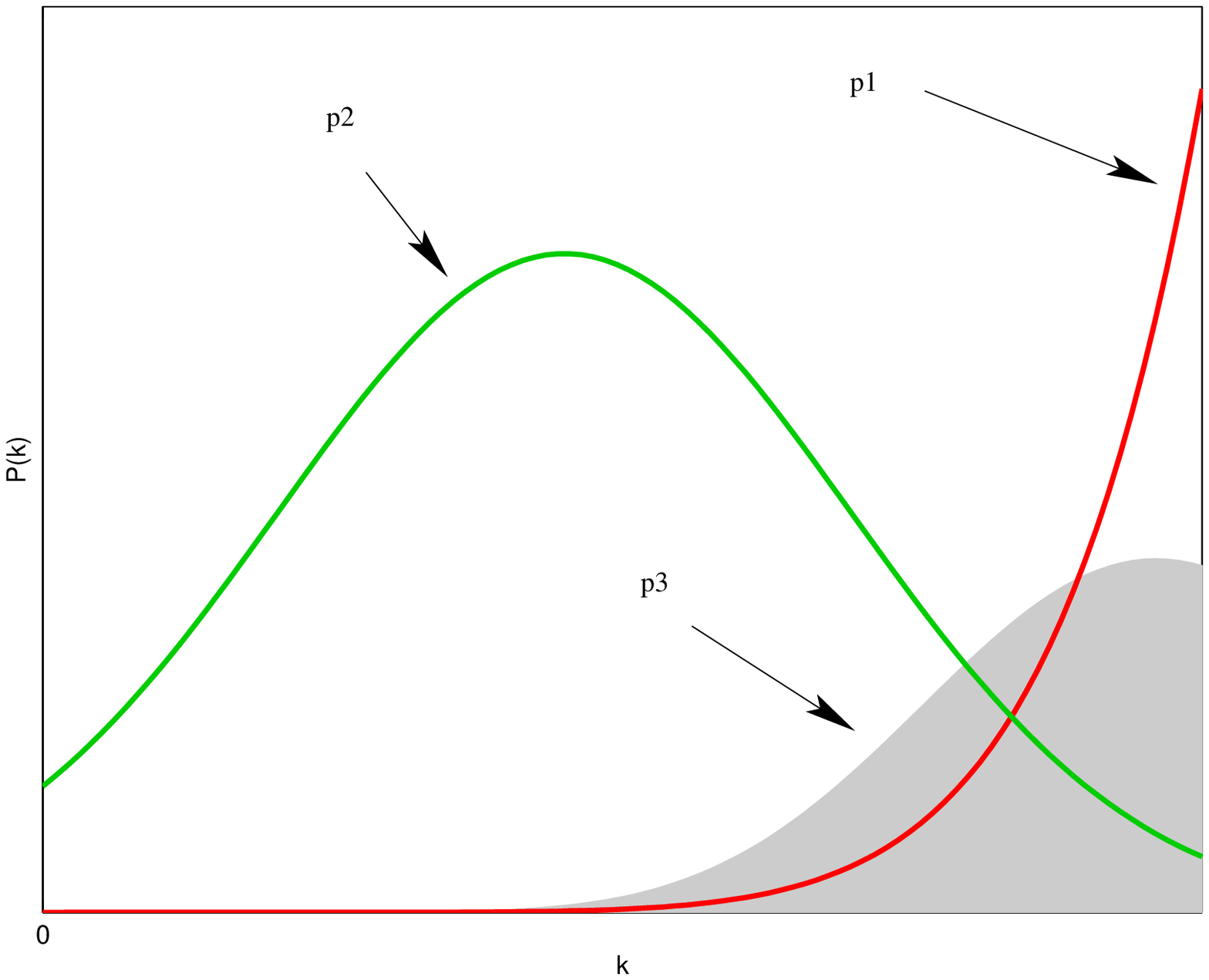}
    \end{tabular}
    \caption{\emph{Left}: power spectrum associated with the derivative of order $p$ of a continuous function, $P(k) = (2 \pi k)^{2p}$, where $k \ge 0$ is the continuous frequency. \emph{Right}: the area under the curve $(2 \pi k)^{2p} \abs{\hat{y}(k)}^2$ is the penalty of the function $y$ (with power spectrum $\abs{\hat{y}(k)}^2$). High frequencies are penalised more than low ones.}
    \label{f:1Dcderiv}
  \end{center}
\end{figure}

\subsection{The tension term and the eigenspectrum of circulant matrices}
\label{s:tension:circ}

In this section we define circulant matrices and prove a few simple but powerful properties that we will need later (see \citet{Davis79a} for general background). For periodic b.c., both our \D\ and \SS\ matrices are circulant, while for nonperiodic b.c.\ they will be approximately Toeplitz and the results below are expected to hold for large nets. The basic idea is that the eigenvalues of a circulant matrix can be computed from its first row via the inverse discrete Fourier transform and the eigenvectors are discrete plane waves. Thus, the spectrum is given directly by the stencil, and so we can characterise the tension term in the Fourier domain.

Notation: we will generally use the subindex $0$ for the first row or column of a matrix. Many of the operations involve complex values, so $\abs{\cdot}$, $\cdot^*$, $\cdot^T$ and $\cdot^H$ will mean the modulus, conjugate, transpose and Hermitian transpose, respectively. Although the results hold generally, when using the symbols $\varsigma$, \D\ and \SS\ we will implicitly assume that $\varsigma$ is the stencil that produces \D\ and that $\SS = \D^T \D$ is the tension term matrix. If a stencil is $\varsigma = (a,\ b,\ c)$ and the \D\ matrix is $M \times M$, the left-aligned stencil (padded with zeroes to the right) will be $(a,\ b,\ c,\ 0,\ 0,\dots,\ 0)_{1 \times M}$ and the centre-aligned stencil will be $(b,\ c,\ 0,\ 0,\dots,\ 0,\ a)_{1 \times M}$. We will call $\omega_m \bydef \exp(i 2 \pi m / M)$ an $M$th root of unity, for $m \in \{0,\dots,M-1\}$. The discrete Fourier transform (DFT) is defined in appendix~\ref{s:fourier}.

The proofs of the propositions are given in appendix~\ref{s:proofs}.

\begin{defn}[Circulant and Toeplitz matrices]
  Consider real, square, $M \times M$ matrices $\D =(d_{nm})$ and $\A = (a_{nm})$. Then \A\ is a Toeplitz matrix if $a_{nm} = a_{m-n}$, where $a_{-(M-1)},\dots,a_{-1},a_0,a_1,\dots,a_{M-1}$ are fixed real values; and \D\ is circulant if $d_{nm} = d_{(m-n) \bmod M}$, where $d_0,d_1,\dots,d_{M-1}$ are fixed real values. Thus, we can represent \A\ by its first column and row and \D\ by its first row (the rest being successive rotations of the first):
  \begin{equation*}
    \text{Toeplitz } \A =
      \begin{pmatrix}
        a_0        & a_1        & a_2        & \dots & a_{M-1} \\
        a_{-1}     & a_0        & a_1        & \dots & a_{M-2} \\
        a_{-2}     & a_{-1}     & a_0        & \dots & a_{M-3} \\
        \hdotsfor{5}                                           \\
        a_{-(M-1)} & a_{-(M-2)} & a_{-(M-3)} & \dots & a_0     \\
      \end{pmatrix}
    \quad
    \text{Circulant } \D =
      \begin{pmatrix}
        d_0     & d_1     & d_2 & \dots & d_{M-1} \\
        d_{M-1} & d_0     & d_1 & \dots & d_{M-2} \\
        d_{M-2} & d_{M-1} & d_0 & \dots & d_{M-3} \\
        \hdotsfor{5}                              \\
        d_1     & d_2     & d_3 & \dots & d_0     \\
      \end{pmatrix}.
  \end{equation*}
\end{defn}

We will need the following properties of circulant matrices.

\begin{prop}[Eigenspectrum of a circulant matrix \D]
  \label{p:eig-circ}
  The eigenvectors $\{\f_m\}^{M-1}_{m=0}$ of a circulant matrix \D\ are complex plane waves, i.e., $f_{mn} = e^{i 2 \pi \frac{nm}{M}}$ for $n = 0,\dots,M-1$, and are associated with complex eigenvalues $\lambda_m = \sum^{M-1}_{n=0}{d_n \omega^n_m}$, respectively.
\end{prop}
\begin{rmk}
  Using the identity $\sum^{M-1}_{k=0}{(e^{i 2 \pi \frac{m}{M}})^k} = M \delta_m$ (proof: sum of a geometric series), it is easy to see that $\{\f_m\}^{M-1}_{m=0}$ are orthogonal, thus linearly independent and so form a basis. Therefore, all circulant matrices of order $M$ are diagonalisable in this common basis, i.e., $\F^{-1} \D \F = \bLambda$ where $\F = (\f_0,\dots,\f_{M-1})$, $\bLambda = \diag{\lambda_0,\dots,\lambda_{M-1}}$ and $\F^{-1} = \frac{1}{M} \F^H = \frac{1}{M} \F^*$ and $\F = \F^T$. Thus, any circulant matrix \D\ can be written as $\D = \frac{1}{M} \F \bLambda \F^{*}$ in terms of its eigenvalues \bLambda. This also shows that if \A\ and \B\ are circulant then $\A\B$, $\A^T$ and $\A^{-1}$ (if it exists) are circulant too, and $\A\B = \B\A$ (circulant matrices commute). We will call $\F \bydef (\f_0,\dots,\f_{M-1})$ the \emph{Fourier matrix}.
\end{rmk}
\begin{rmk}
  The DFT of \y\ is $\F^H \y$. From proposition~\ref{p:eig-circ} we have $\blambda = \F^T \dd$ with $\blambda = (\lambda_0,\dots,\lambda_{M-1})^T$ and $\dd = (d_0,\dots,d_{M-1})^T$ and so we can compute the coefficients of the first row of a circulant matrix \D\ given its eigenvalues \blambda\ as $\dd = \frac{1}{M} \F^{*} \blambda$, or equivalently $d_m = \frac{1}{M} \sum^{M-1}_{n=0}{\lambda_n \omega^{*n}_m}$ for $m = 0,\dots,M-1$ (since \blambda\ is $M$ times the inverse DFT of \dd).
\end{rmk}
\begin{rmk}
  Since \D\ is a real matrix, complex eigenvalues come in conjugate pairs, i.e., for each $\lambda_m$ associated with eigenvector $\f_m$ we have $\lambda^*_m$ associated with $\f^*_m$. However, this does not mean there are $2M$ eigenvalues, because the identity $\omega_m = \omega^*_{M-m}$ implies $\lambda_m = \lambda^*_{M-m}$. If $M$ is odd, then we have $\frac{M-1}{2}$ pairs of distinct (in general) conjugate pairs plus $\lambda_0 = \sum_n{d_n}$ real. If $M$ is even, we have an additional $\lambda_{\frac{M}{2}} = \sum_n{d_n (-1)^n}$ real. Thus, there are at most $M$ distinct eigenvalues.
\end{rmk}
\begin{rmk}
  $\lambda_0$ is the sum of any row or column of \D\ and is associated with the \emph{constant} eigenvector $\f_0 = (1,\dots,1)^T$, while for even $M$, $\lambda_{\frac{M}{2}}$ is associated with the \emph{sawtooth} eigenvector $\f_{\frac{M}{2}} = (1,-1,1,-1,\dots,1,-1)^T$.
\end{rmk}

\begin{prop}[Eigenspectrum of a circulant matrix $\SS = \D^T\D$]
  \label{p:eig-circ2}
  If \D\ is real circulant then $\SS = \D^T\D$ is circulant too with real, nonnegative eigenvalues $\nu_m = \abs{\lambda_m}^2$ associated with eigenvectors $v_{mn} = \cos{\left( 2 \pi \frac{m}{M} n \right)}$ and $w_{mn} = \sin{\left( 2 \pi \frac{m}{M} n \right)}$.
\end{prop}
\begin{rmk}
  The previous proposition allows to obtain the eigenvalue spectrum of \SS\ from the coefficients of the stencil $\varsigma$. We can either construct the eigenvalues of \D, $\lambda_m = \sum^{M-1}_{n=0}{d_n \omega^n_m}$, from its first row $(d_0,d_1,\dots,d_{M-1})$ and those of \SS\ as $\nu_m = \abs{\lambda_m}^2$; or directly construct the eigenvalues of \SS\ as $\nu_m = \sum^{M-1}_{n=0}{s_n \cos{\left( 2 \pi \frac{m}{M} n \right)}}$ with $s_n$ given in eq.~\eqref{e:s-from-d}.
\end{rmk}
\begin{rmk}
  We will call $\v_m \bydef \frac{1}{2}(\f_m + \f^*_m) = \Re(\f_m)$, with $v_{mn} = \cos{\left( 2 \pi \frac{m}{M} n \right)}$, and $\w_m \bydef \frac{1}{2i}(\f_m - \f^*_m) = \Im(\f_m)$, with $w_{mn} = \sin{\left( 2 \pi \frac{m}{M} n \right)}$, the \emph{cosine} and \emph{sine} eigenvectors, respectively. Note that, while $\norm{\f_m}^2 = \f^H_m\f_m = M$, the same does not hold for the cosine and sine eigenvectors $\v_m$ and $\w_m$. For the latter, $\norm{\v_m}^2 = \norm{\w_m}^2 = \frac{M}{2}$ for $m > 0$ and $M$ for $m = 0$.
\end{rmk}
\begin{rmk}
  Since \SS\ is invariant to permutations of the rows of \D\ and each row of \D\ is a rotation of the first row, we are free to take the stencil origin anywhere in that the eigenvalues of \SS\ do not change (those of \D\ do, by a factor $e^{i 2 \pi \frac{m}{M} m'}$ for an origin shift of $m'$).
\end{rmk}
\begin{rmk}
  \SS\ has $\lfloor \frac{M}{2} \rfloor$ distinct eigenvalues at most, since $\nu_m = \nu_{M-m}$ for $m=1,\dots,M-1$: if $M$ is odd, then we have $\nu_1,\dots,\nu_{\frac{M-1}{2}}$, each with a 2D subspace spanned by the cosine and sine eigenvectors, plus $\nu_0 = \sum^{M-1}_{n=0}{s_n} = \big(\sum^{M-1}_{n=0}{d_n}\big)^2$ with a 1D subspace spanned by the constant eigenvector; and if $M$ is even, we have additionally $\nu_{\frac{M}{2}} = \sum^{M-1}_{n=0}{s_n (-1)^n} = \big(\sum^{M-1}_{n=0}{d_n (-1)^n}\big)^2$ with a 1D subspace spanned by the sawtooth eigenvector. This is the generic situation; for specific values of $d_n$ some of the eigenvalues can degenerate.
\end{rmk}
\begin{rmk}
  The relation $\nu_m = \abs{\lambda_m}^2$ implies that \D\ and \SS\ have the same null eigenvalues and the same nullspace.
\end{rmk}
\begin{rmk}
  \label{rm:tension-penalty}
  The penalty due to the tension term that a net incurs can now be computed by expressing the net as a superposition of plane waves, i.e., uniquely decomposing the net $\Y = \y^T$ in the eigenvector basis of \SS:
\begin{equation}
  \label{e:tension-penalty}
  \y = \sum^{M-1}_{m=0}{a_m \f_m} \text{ for } a_m \in \bbR \text{ with } a_m = a_{M-m} \Longrightarrow \norm{\D\Y^T}^2 = \trace{\Y^T \Y \SS} = \y^T \SS \y = M \sum^{M-1}_{m=0}{a^2_m \nu_m}.
\end{equation}
Thus, frequency $m$ contributes a penalty proportional to $a^2_m \nu_m$. The constant eigenvector of ones is the sampled version of a constant net. In the continuous case, when applying a differential operator of any order to a constant function we obtain the zero function and so no penalty in eq.~\eqref{e:tension-cont}, from eq.~\eqref{e:tension-penalty} with $\y = \f_0$. This holds in the discrete case if $\nu_0 = \sum^{M-1}_{n=0}{s_n} = \big(\sum^{M-1}_{n=0}{d_n}\big)^2 = 0$, i.e., if the stencil is zero-sum---a condition demanded by the truncation error discussion too. However, eigenvectors of the form $v_{mn} = n^p$ for fixed $p \in \bbZ^+$, corresponding to a monomial $t^p$, are not nullified by \D\ in the circulant case, since the rest of eigenvalues are nonzero in general. With nonperiodic boundary conditions \D\ is not circulant but quasi-Toeplitz and can nullify polynomials.
\end{rmk}
\begin{prop}[Fourier power spectrum of stencil $\varsigma$]
  \label{p:stencil-power}
  If \D\ is the $M \times M$ circulant matrix associated with the stencil $\varsigma$, the power spectrum of the stencil $\varsigma$ is equal to the eigenspectrum of the matrix $\SS = \D^T\D$.
\end{prop}

Later, we will find useful the quantity $\sum_m{\varsigma^2_m}$, which we call the \emph{squared modulus} of the stencil. The following proposition shows its relation with the power spectrum and the matrix \SS.

\begin{prop}[Squared modulus of stencil $\varsigma$]
  \label{p:stencil-sqmodulus}
  If \D\ is the $M \times M$ circulant matrix associated with the stencil $\varsigma$ and $\SS = \D^T\D$, then the total power of the stencil is $P = M \sum_m{\varsigma^2_m} = \trace{\SS}$.
\end{prop}
\begin{rmk}
  The squared modulus of the stencil appears repeated along the diagonal of \SS. For Toeplitz matrices $P = M \sum_m{\varsigma^2_m}$ holds but is in general (slightly) different from $\trace{\SS}$.
\end{rmk}
\begin{rmk}
  The tension term penalty for a delta net $y = (1,\ 0,\ 0,\dots,\ 0)$ equals the stencil squared modulus (for any dimensionality $D$ of the net), since the convolution of the stencil with the delta net equals the stencil. If a $D$\nobreakdash-dimensional tension term is constructed by passing a 1D stencil along each dimension (section~\ref{s:lc:power}), by the additivity of the power the penalty is $D$ times the stencil squared modulus.
\end{rmk}

\subsection{Sawtooth stencils}
\label{s:sawtooth}

The eigenvector $\f_m$, or the real cosine and sine eigenvectors, represent waves $\cos{mt}$ and $\sin{mt}$ for $t = 2 \pi \frac{n}{M} \in [0,2\pi)$ of discrete frequency $m$, which---unlike in the continuous case---is upper bounded by $\frac{M}{2}$. This highest frequency corresponds to a sawtooth, or comb, wave $(1,-1,1,-1,\dots,1,-1)^T$, which plays a significant role with certain stencils (described later). The following proposition gives a condition for a stencil to have zero power at the sawtooth frequency, and we will call sawtooth stencil a stencil satisfying this condition.

\begin{prop}[Sawtooth condition in 1D]
  \label{p:sawtooth1D}
  A differential stencil $\varsigma$ over a 1D net with $M$ centroids ($M$ even) has zero power at the sawtooth frequency if and only if the even coefficients sum zero and the odd coefficients sum zero.
\end{prop}
\begin{prop}[Sawtooth condition in 2D]
  \label{p:sawtooth2D}
  A differential stencil $\varsigma$ over a 2D net with $M \times N$ centroids ($M$, $N$ even) has zero power at the sawtooth frequency if and only if, imagining a checkerboard pattern, the coefficients at the black squares sum zero and the coefficients at the white squares sum zero.
\end{prop}
\begin{rmk}
  The sawtooth conditions can be analogously expressed in terms of the \SS\ matrix generated by the stencil. For example in 1D, if \D\ is the $M \times M$ circulant matrix associated with the stencil $\varsigma$ and $\SS = \D^T\D$, the condition is $\nu_{\frac{M}{2}} = \sum^{M-1}_{m=0}{s_m (-1)^m} = 0$, where $(s_0,\dots,s_{M-1})$ is the first row of \SS. This is a more general condition than that on the stencil, since different stencils can result in the same matrix \SS\ (see section~\ref{s:circ2sten}).
\end{rmk}
\begin{rmk}
  For nets with an odd number of centroids, the highest-frequency wave is not the sawtooth wave, but nearly so, and a stencil satisfying the above condition will have nearly zero power at it (since $\{\nu_m\}^M_{m=1}$ are samples of a continuous function of $m$, if some $\nu_m$ is zero or nearly so at some point, then it must be low in an environment of it). Also, in general it is not necessary that the sawtooth wave has nearly zero power for the net to develop sawteeth. For this, it is usually enough that the other frequencies have more or equal power.
\end{rmk}
For stencils satisfying the sawtooth condition, the highest frequency incurs no penalty (or a negligible one) in the tension term, just as the zero-frequency wave (the constant net) does. This does not correspond to the continuous case, where for a wave of frequency $n$ we have an average penalty of 
\begin{equation*}
  \frac{1}{2 \pi} \int^{2 \pi}_{0}{\abs{\frac{d^p}{du^p}\sin{mu}}^2 \, du} = \frac{1}{2} m^{2p}
\end{equation*}
which grows monotonically with $m$ and $p$ (for $m \ge 1$).

\subsection{Fourier space and stencil normalisation}
\label{s:tension:norm}

The representation of a stencil in the net domain is very variable: two stencils that may look completely different and have coefficients of wildly different magnitude may be representing the same derivative (with different truncation error). However, in the Fourier domain the character of different stencils is obvious: they are high- or band-pass filters. This is because their power spectra happen to have a simple aspect, typically being monotonically increasing or unimodal curves. Naturally, if the power spectrum was not simple, we would have to look at a different representation. From the point of view of the elastic net, then, the truncation error is of secondary importance---after all, the ``step size'' in the net will always remain quite large for 2D nets due to the computational cost involved.

However, in order to compare stencils with each other (either of the same or of different order), and to compare nets with different number of centroids or different length%
\footnote{The ``length'' refers to the space of centroid indices $m = 1,\dots$ considered continuous. For example, in 1D the length is $Mh$ where $h$ is the step size (see appendix~\ref{s:stencil-norm}).}
applied to the same training set, we need to normalise the stencils. This normalisation will result in multiplying the matrix \SS, or alternatively $\beta$, times a constant. This constant will depend on two things: the order and squared modulus of the stencil, and the step size (i.e., the net length divided by the number of centroids). The relevant calculations can be done for nets of any dimension and are given in appendix~\ref{s:stencil-norm}. In this paper we use the normalisation mainly in figures such as~\ref{f:stencil-families1Da}, which plots a family of stencils, and~\ref{f:simul2D:OD}, which plots the resulting maps.

\subsection{Stencil families}
\label{s:tension:families}

We are still left with the question of what stencil to choose, since to represent a derivative of order $p$ one can design stencils of arbitrarily small truncation error by making zero as many coefficients $\alpha_r$ as desired (except $\alpha_p$) in the Taylor expansion of eq.~\eqref{e:D:Taylor}, which in turn will require the stencil to have many nonzero coefficients. We can also derive new stencils of order $p$ as linear combinations of existing stencils of the same order $p$, which we briefly consider in section~\ref{s:lc}. Instead of approaching this general question, we will deal here with a specific, but useful, way of generating a family of stencils: by iterating a fixed stencil. We will analyse the families associated with the first-order \emph{forward-difference} and \emph{central-difference} stencils (the backward-difference one being equivalent to the forward-difference one by the shift invariance). See fig.~\ref{f:stencil-families}. 

First of all, we have the following properties. The convolution is an associative and commutative operator (in both the continuous and discrete cases), which is reflected in the associativity and commutativity of the respective circulant matrices. Repeated application of a first-order differential stencil results in higher-order stencils (recall that when writing a convolution of a stencil $\varsigma$ with a function we must write the reversed stencil $\overleftarrow{\varsigma}$):
\begin{equation*}
  f' \sim \overleftarrow{\varsigma} \ast f \Rightarrow f'' \sim \overleftarrow{\varsigma} \ast f' = \overleftarrow{\varsigma} \ast (\overleftarrow{\varsigma} \ast f) = (\overleftarrow{\varsigma} \ast \overleftarrow{\varsigma}) \ast f.
\end{equation*}
We call $\{\leftexp{(p)}{\varsigma}\}^{\infty}_{p=0}$ the \emph{family} of stencils associated with $\leftexp{(0)}{\varsigma} \bydef \delta$, $\leftexp{(1)}{\varsigma}\bydef \varsigma$ and $\leftexp{(p)}{\varsigma} = \overleftarrow{\varsigma} \ast \leftexp{(p-1)}{\smash{\overleftarrow{\varsigma}}} = \overleftarrow{\varsigma} \ast \dots \ast \overleftarrow{\varsigma} \ast \delta$ ($p$ convolutions), where $\delta$ is the discrete delta function ($\delta_m = 1$ if $m = 0$, zero otherwise). Naturally, one can also construct hybrid stencils by applying successively different first-order stencils.
\begin{prop}[Composition of stencils]
  \label{p:stencil-comp}
  Let $\varsigma$ be a stencil of derivative order $p$ and truncation error $\calO(h^{p'})$, and $\varrho$ a stencil of derivative order $q$ and truncation error $\calO(h^{q'})$. Then $\overleftarrow{\varsigma} \ast \overleftarrow{\varrho} = \overleftarrow{\varrho} \ast \overleftarrow{\varsigma}$ is a stencil of derivative order $p+q$ and truncation error $\calO(h^{\min(p',q')})$.
\end{prop}
\begin{cor}
  If $\varsigma$ is a stencil of derivative order $1$, then $\leftexp{(p)}{\varsigma}$ is of order $p$ with the same truncation error order.
\end{cor}
\begin{prop}
  \label{p:stencil-comp2}
  Let \D\ and \E\ be the matrices associated with the stencils $\varsigma$ and $\varrho$. Then $\D\E$ is associated with $\overleftarrow{\varsigma} \ast \overleftarrow{\varrho}$.
\end{prop}
\begin{cor}
  A family of stencils $\{\leftexp{(p)}{\varsigma}\}^{\infty}_{p=0}$ has associated matrices $\{\D^p\}^{\infty}_{p=0}$ with $\D^0 = \I$ and \D\ the matrix associated with $\varsigma = \leftexp{(1)}{\varsigma}$. The matrix \D\ has eigenvalues $\lambda_m = \sum_n{\varsigma_n \omega^n_m}$ (times an unimportant phase factor), the matrix $\leftexp{(p)}{\D} = \D^p$ has eigenvalues $\leftexp{(p)}{\lambda_m} = \lambda^p_m$ and the matrix $\leftexp{(p)}{\SS} = (\D^p)^T \D^p$ has eigenvalues $\leftexp{(p)}{\nu_m} = \abs{\smash{\leftexp{(p)}{\lambda_m}}}^2 = \abs{\lambda_m}^{2p}$ associated with the cosine and sine eigenvectors.
\end{cor}
\begin{prop}[Sawtooth dominance]
  \label{p:sawtooth-dom}
  Let $\varsigma$ be a stencil that has zero power at the sawtooth frequency. Then $\overleftarrow{\varrho} \ast \overleftarrow{\varsigma} = \overleftarrow{\varsigma} \ast \overleftarrow{\varrho}$ also have zero power at the sawtooth frequency for any stencil $\varrho$.
\end{prop}
This means that if a stencil is the result of the convolution of several stencils, then if any one of these is a sawtooth stencil, the convolution will also be a sawtooth stencil.

\begin{figure}
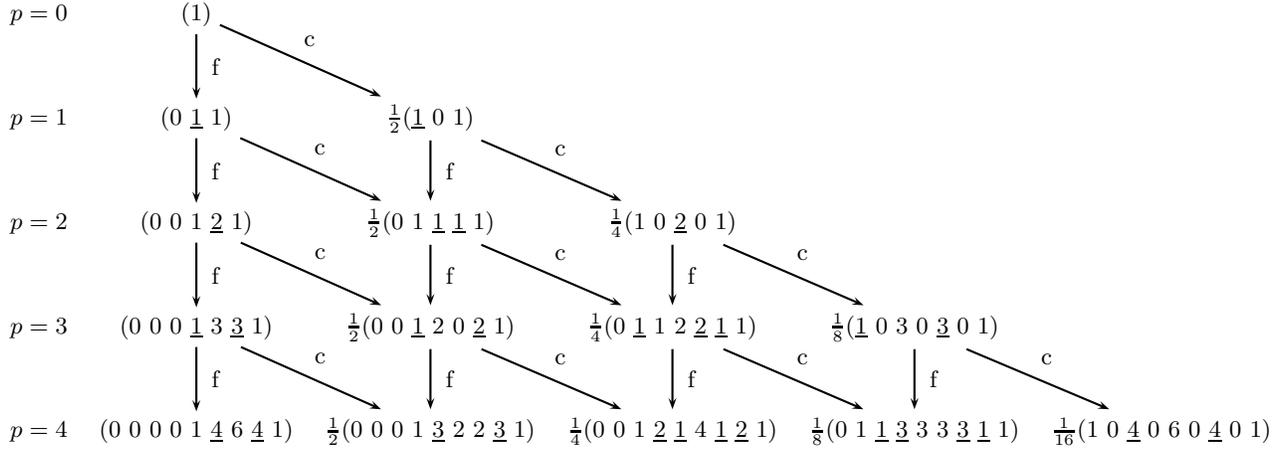

  \small
  \begin{center}
    \begin{tabular}{@{}lccccc@{}}
      $p = 0$ & \rnode{n:00}{$(1)$} \\[1cm]
      $p = 1$ & \rnode{n:10}{$(0\ \underline{1}\ 1)$} & \rnode{n:11}{$\frac{1}{2}(\underline{1}\ 0\ 1)$} \\[1cm]
      $p = 2$ & \rnode{n:20}{$(0\ 0\ 1\ \underline{2}\ 1)$} & \rnode{n:21}{$\frac{1}{2}(0\ 1\ \underline{1}\ \underline{1}\ 1)$} & \rnode{n:22}{$\frac{1}{4}(1\ 0\ \underline{2}\ 0\ 1)$} \\[1cm]
      $p = 3$ & \rnode{n:30}{$(0\ 0\ 0\ \underline{1}\ 3\ \underline{3}\ 1)$} & \rnode{n:31}{$\frac{1}{2}(0\ 0\ \underline{1}\ 2\ 0\ \underline{2}\ 1)$} & \rnode{n:32}{$\frac{1}{4}(0\ \underline{1}\ 1\ 2\ \underline{2}\ \underline{1}\ 1)$} & \rnode{n:33}{$\frac{1}{8}(\underline{1}\ 0\ 3\ 0\ \underline{3}\ 0\ 1)$} \\[1cm]
      $p = 4$ & \rnode{n:40}{$(0\ 0\ 0\ 0\ 1\ \underline{4}\ 6\ \underline{4}\ 1)$} & \rnode{n:41}{$\frac{1}{2}(0\ 0\ 0\ 1\ \underline{3}\ 2\ 2\ \underline{3}\ 1)$} & \rnode{n:42}{$\frac{1}{4}(0\ 0\ 1\ \underline{2}\ \underline{1}\ 4\ \underline{1}\ \underline{2}\ 1)$} & \rnode{n:43}{$\frac{1}{8}(0\ 1\ \underline{1}\ \underline{3}\ 3\ 3\ \underline{3}\ \underline{1}\ 1)$} & \rnode{n:44}{$\frac{1}{16}(1\ 0\ \underline{4}\ 0\ 6\ 0\ \underline{4}\ 0\ 1)$}
    \end{tabular}
    \psset{nodesep=3pt}
    \ncline{->}{n:00}{n:10}\Aput{f}
    \ncline{->}{n:00}{n:11}\Aput{c}
    \ncline{->}{n:10}{n:20}\Aput{f}
    \ncline{->}{n:10}{n:21}\Aput{c}
    \ncline{->}{n:11}{n:21}\Aput{f}
    \ncline{->}{n:11}{n:22}\Aput{c}
    \ncline{->}{n:20}{n:30}\Aput{f}
    \ncline{->}{n:20}{n:31}\Aput{c}
    \ncline{->}{n:21}{n:31}\Aput{f}
    \ncline{->}{n:21}{n:32}\Aput{c}
    \ncline{->}{n:22}{n:32}\Aput{f}
    \ncline{->}{n:22}{n:33}\Aput{c}
    \ncline{->}{n:30}{n:40}\Aput{f}
    \ncline{->}{n:30}{n:41}\Aput{c}
    \ncline{->}{n:31}{n:41}\Aput{f}
    \ncline{->}{n:31}{n:42}\Aput{c}
    \ncline{->}{n:32}{n:42}\Aput{f}
    \ncline{->}{n:32}{n:43}\Aput{c}
    \ncline{->}{n:33}{n:43}\Aput{f}
    \ncline{->}{n:33}{n:44}\Aput{c}
    \caption{Stencils $\varsigma = (\dots,\varsigma_{-1},\varsigma_0,\varsigma_1,\dots)$ obtained by repeated application of the forward-difference $\varsigma_{\text{f}} = (0,\ -1,\ 1)$ and the central-difference $\varsigma_{\text{c}} = \left(-\frac{1}{2},\ 0,\ \frac{1}{2}\right)$ up to order $p = 4$. The forward-difference family appears on the left vertical lines; note that the leading zeroes are irrelevant for the creation of the \SS\ matrix and so are usually omitted elsewhere in the paper. The backward-difference family (not drawn) results from $\varsigma_{\text{b}} = (-1,\ 1,\ 0)$ and produces the same stencils but with trailing zeroes, and thus results in the same \SS\ as the forward-difference family (although the truncation error is not exactly the same). The central-difference family appears on the top diagonal lines. Hybrid stencils, obtained by applying $\varsigma_{\text{f}}$ and $\varsigma_{\text{c}}$, appear in the middle. All stencils in this figure are sawtooth except for the forward-difference family. Since the convolution is commutative, $\protect\overleftarrow{\varsigma} \ast \protect\overleftarrow{\varrho} = \protect\overleftarrow{\varrho} \ast \protect\overleftarrow{\varsigma}$, and so are the \D\ matrices, $\D_{\varsigma}\D_{\varrho} = \D_{\varrho}\D_{\varsigma}$, the diagram is commutative.}
    \label{f:stencil-families}
  \end{center}
\end{figure}

\subsubsection{Forward-difference family: the continuous-case correlate}
\label{s:fwddiff-family}

This is defined by the first-order forward-difference stencil $\varsigma = (0,\ -1,\ 1)$, so \D\ has eigenvalues $\lambda_m = -1 + e^{i 2 \pi \frac{m}{M}}$ and \SS\ has eigenvalues $\nu_m = 2 \left( 1 - \cos{\left( 2 \pi \frac{m}{M} \right)} \right) = 4 \sin^2{\left( \pi \frac{m}{M} \right)}$. Thus the $p$th\nobreakdash-order derivative stencil has eigenvalues $\leftexp{(p)}{\nu_m} = \left( 2 \sin{\left( \pi \frac{m}{M} \right)} \right)^{2p} \in [0,2^{2p}]$ $\forall m$. Note that $\sum_m{\varsigma_m} = 0$ but $\sum_m{\varsigma_m (-1)^m} \neq 0$ and so it is not a sawtooth stencil (as also indicated by $\nu_{\frac{M}{2}} = 4 \neq 0$).
The case $p = 2$ corresponds to the small vibrations of a string with periodic b.c., the sinewaves $\f_0,\dots,\f_{M-1}$ being the modes of vibration.

The following propositions show that this family forms the Pascal triangle when piled up without the $p$ leading zeroes (see fig.~\ref{f:stencil-families}), and that the stencil squared modulus (which appears along the diagonal of $\leftexp{(p)}{\SS}$) equals $\binom{2p}{p}$.
\begin{prop}
  \label{p:fwd-Pascal}
  For $m = -p,\dots,p$: $\leftexp{(p)}{\varsigma_m} = (-1)^{p+m} \binom{p}{m}$ for $m \ge 0$, zero otherwise.
\end{prop}
\begin{prop}
  \label{p:fwd-sqmodulus}
  $\sum_m{\leftexp{(p)}{\varsigma^2_m}} = \binom{2p}{p}$ and $\trace{\smash{\leftexp{(p)}{\SS}}} = M \binom{2p}{p}$.
\end{prop}

Fig.~\ref{f:stencil-families1Da} shows that the forward-difference family forms a progression with $p$ similar to that of the continuous case (fig.~\ref{f:1Dcderiv}) where the curves slope up more slowly for larger $p$. Note that even though the nullspace is strictly that of the constant waves, since the only null eigenvalue is $\nu_0$, as $p$ increases there are more and more near-zero eigenvalues for the low frequencies. In other words, the \emph{effective} nullspace increases with $p$, consistently with the nonperiodic case, where the \emph{actual} nullspace increases with $p$. Thus, in this family low frequencies are practically not penalised for high $p$.

The combination of the power curves of figures~\ref{f:stencil-families1Da} and~\ref{f:stencil-families2Da} with a qualitative argument based on a \emph{drifting cutoff} allows a partial explanation of the behaviour of the GEN during the deterministic annealing algorithm (though it does not allow us to compute the actual preferred frequency as a function of the stencil and $\beta$). The predictions are demonstrated by our simulations in fig.~\ref{f:simul1D-periodic} for 1D nets and fig.~\ref{f:simul2D:OD} for 2D nets. The idea is as follows. The fitness term would like to have access to all frequencies so that the net matches the data points, while the tension term penalises frequencies proportionally to the power spectrum (remark~\ref{rm:tension-penalty}). If we assume that an optimal net can have at most a certain tension, then it will have at most a certain power $p_{k^*}$ (where $k^*$ is the largest frequency with power less or equal than that power). This power will increase as $\sigma$ decreases, since the relative influence of the tension term decreases. We call the power $p_{k^*}(\sigma)$ the drifting cutoff, since frequencies whose power exceeds $p_{k^*}(\sigma)$ cannot be present in the net at scale $\sigma$.

The drifting cutoff predicts that the nets that arise are expected to show higher and higher frequency as $p$ increases. For large $\sigma$ the tension term dominates, so that even a small tension penalty is not allowed, and the net collapses to a point at the centre of mass of the training set. This, which corresponds to the fact that the energy has a single minimum, was proven by \citet{Durbin_89a} for the original elastic net and is easy to extend to our general case. As $\sigma$ is decreased, the influence of the tension term gradually decreases in favour of the fitness term; the energy experiences a series of bifurcations and develops more and more minima which correspond to the net stretching more and more in various ways to approach the training set points. Now imagine a cutoff power $p_{k^*}$ in fig.~\ref{f:stencil-families1Da} (the horizontal dashed line) that corresponds to a maximum allowed tension penalty at a given $\sigma$. For large $\sigma$, the cutoff is at $p_{k^*} = 0$ and the only frequency possible for any stencil order is $k = 0$, i.e., a constant net (all centroids at the centre of mass). As $\sigma$ decreases, the cutoff power increases (as in the figure), so that for a small $p_{k^*}$ the corresponding cutoff frequency $k^*$ is given by the power curves; thus, $k^*$ increases with $p$, i.e., the early behaviour of the nets shows higher frequency for high $p$ (narrower stripes in cortical maps). Note that the increase of the cutoff proceeds in jumps (corresponding to the bifurcations of the energy) rather than uniformly. As $\sigma$ is further decreased and the cutoff power further increases, higher frequencies appear in the net, but they are constrained to develop on the already existing low-frequency net, and so result in local quirks and stretchings towards the training set points, superimposed on a lower-frequency structure.

It follows that if one starts the training at a low value of $\sigma$, so that the cutoff corresponds to a high frequency $k^*$, the emerging net loses its smooth structure with little difference between orders $p$, as we have confirmed in simulations. The same occurs if annealing too fast. This can also be understood without recourse to the power curves by noting that at low $\sigma$ the fitness term dominates and the centroids are drawn towards the training set points with little regard to the tension.

From fig.~\ref{f:stencil-families1Da} we see the cutoff frequencies at low $\sigma$ cluster for high $p$, i.e., $p = 1$ and $2$ differ much more from each other than $p = 3$ and $4$. This is also confirmed by the simulations and has been quantitatively evaluated in cortical maps by \citet{CarreirGoodhil03b}. Specifically, the original elastic net ($p = 1$) often comes out as qualitatively different from nets with $p > 1$.

The drifting cutoff argument also explains the effect of $\beta$: at any $\sigma$ or $p$, increasing $\beta$ strengthens the tension term penalty and so lowers the cutoff power $p_{k^*}$, which results in lower cutoff frequencies (wider stripes in cortical maps, see fig.~\ref{f:simul2D:OD}).

In 2D nets, the drifting cutoff argument explains not only the behaviour of the preferred frequency of the plane waves in the net but also the preferred direction. Note in fig.~\ref{f:stencil-families2Da} how for $p = 1$ the contour lines near the $\bk = (0,0)$ frequency are approximately circular. Thus, for a given power cutoff the largest frequency $\norm{\bk}$ is attained approximately equally in any direction, and the resulting stripes in cortical map simulations show no preferred direction (fig.~\ref{f:simul2D:OD}, $p = 1$). But for $p > 1$ (again in fig.~\ref{f:stencil-families2Da}) the contour lines near $\bk = (0,0)$ become approximately square. Thus, the largest frequency occurs along the line $k_1 = k_2$ (or $k_1 = -k_2$), and the resulting stripes run preferentially along the diagonals. Other 2D stencils (e.g.\ $\nabla^2_9$ in fig.~\ref{f:stencil-families2Db}) are more isotropic around $\bk = \0$ and do not result in preferred stripe directions.

\newlength{\MACPlengthA}
\setlength{\MACPlengthA}{.24\textheight}
\begin{figure}[th!]
  \psfrag{k}[t]{$k$}
  \psfrag{P(k)}[r][l][1][-90]{$p_k$}
  \begin{center}
    \begin{tabular}{@{}c@{\hspace{0.5cm}}c@{\hspace{1.5cm}}c@{}}
      \rotatebox{90}{\makebox[\MACPlengthA][c]{Forward-difference family}} &
      \psfrag{1}{$p = 1$}
      \psfrag{2}{$p = 2$}
      \psfrag{3}{$p = 3$}
      \psfrag{4}{$p = 4$}
      \includegraphics[height=\MACPlengthA]{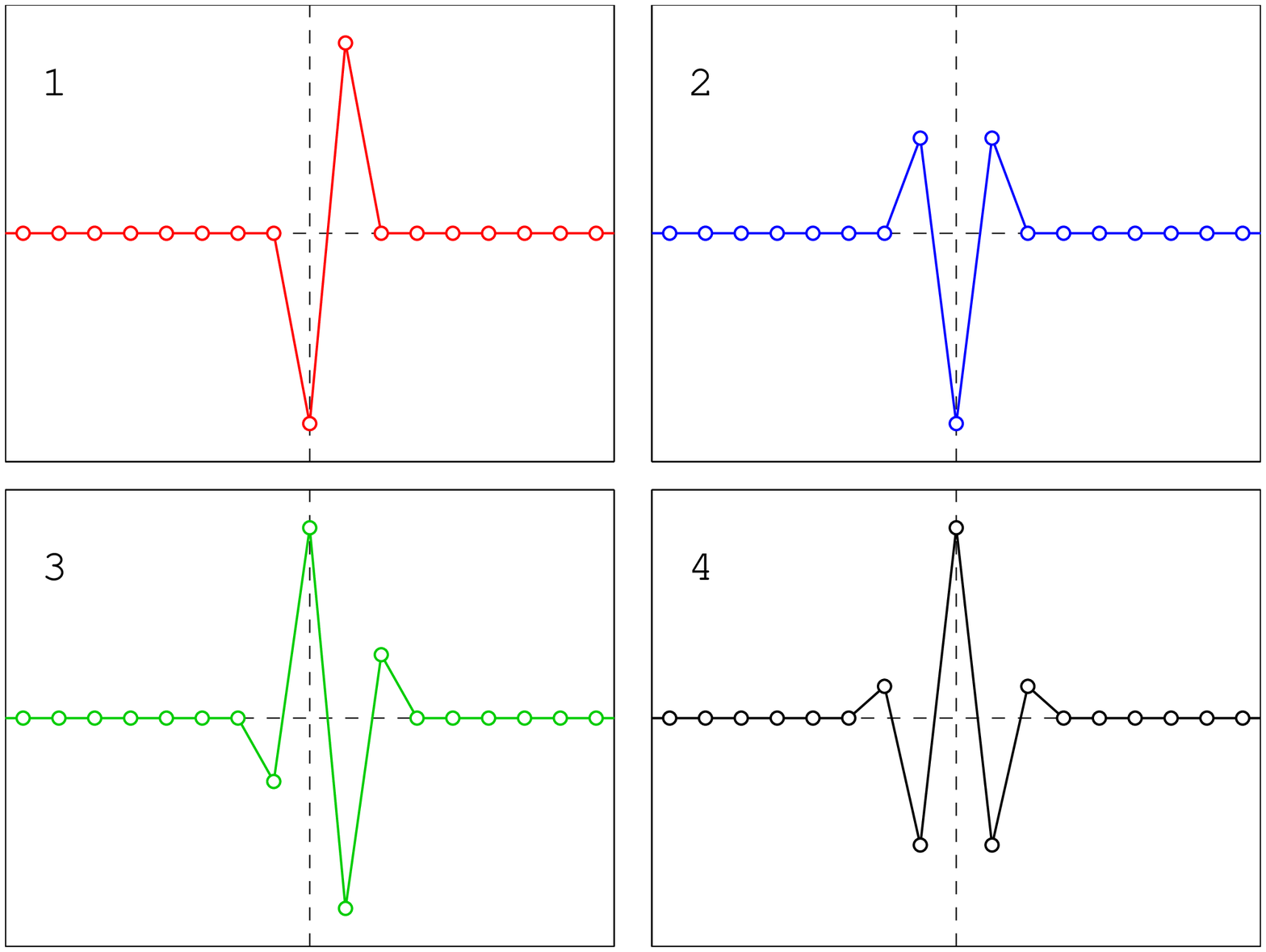} &
      \psfrag{1}{$1$}
      \psfrag{2}{$2$}
      \psfrag{3}{$3$}
      \psfrag{4}{$4$}
      \includegraphics[height=\MACPlengthA]{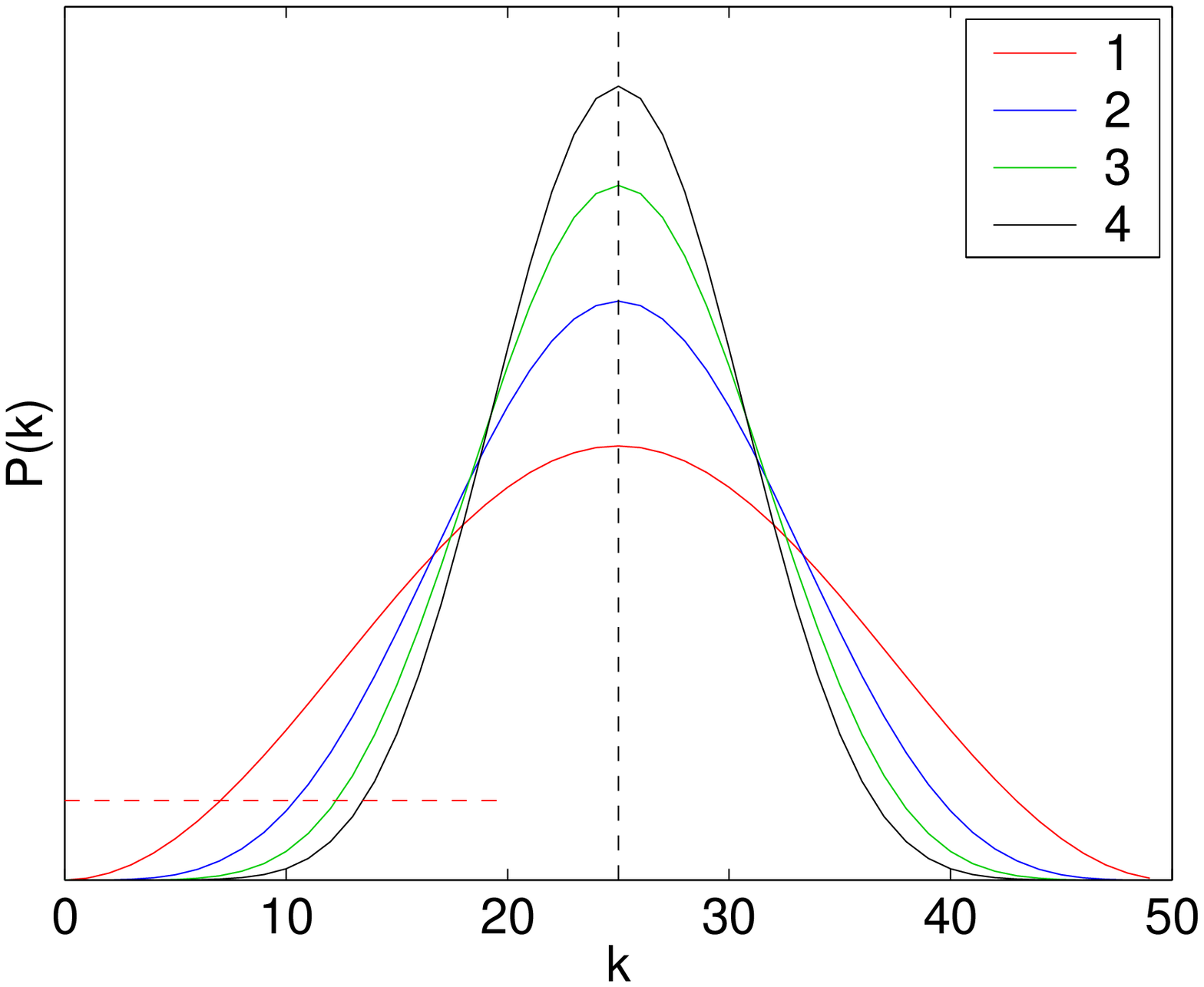} \\[2ex]
      \rotatebox{90}{\makebox[\MACPlengthA][c]{Central-difference family}} &
      \psfrag{1}{$p = 1$}
      \psfrag{2}{$p = 2$}
      \psfrag{3}{$p = 3$}
      \psfrag{4}{$p = 4$}
      \includegraphics[height=\MACPlengthA]{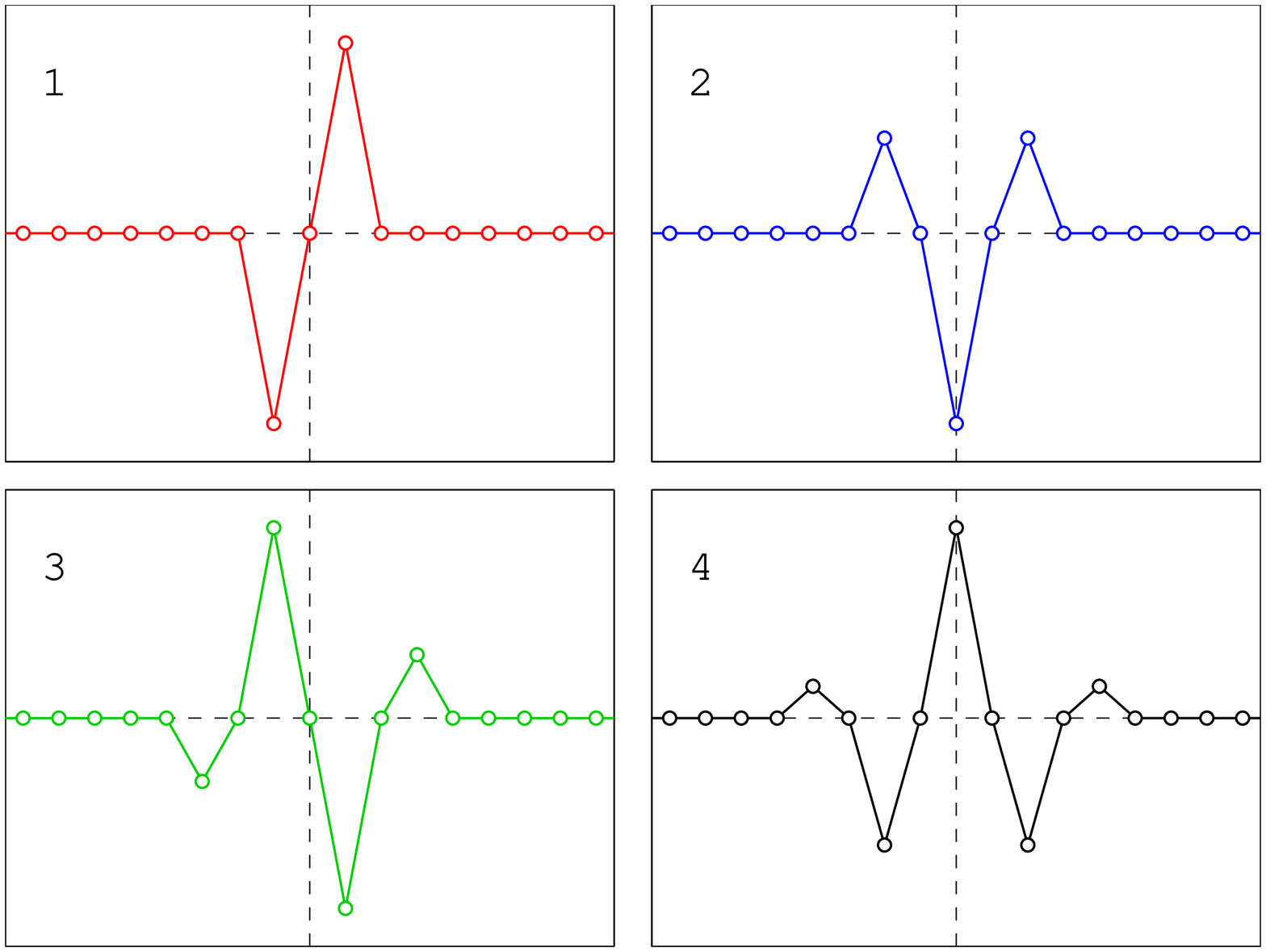} &
      \psfrag{1}{$1$}
      \psfrag{2}{$2$}
      \psfrag{3}{$3$}
      \psfrag{4}{$4$}
      \includegraphics[height=\MACPlengthA]{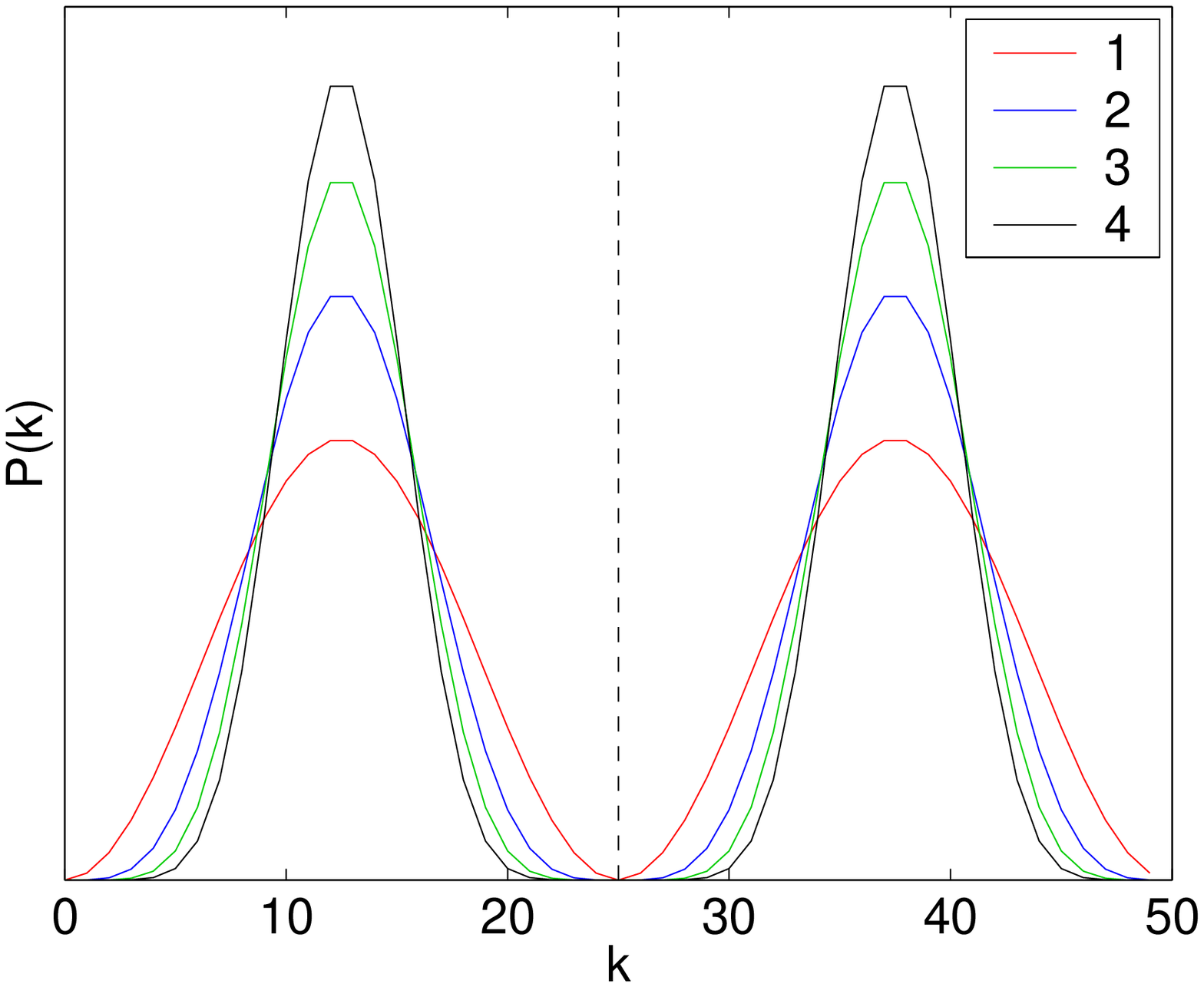}
    \end{tabular}
    \caption{1D stencils for the forward- and central-difference families, for a net with $M = 50$ centroids. The left panels show the stencils (in the net domain), normalised in maximal amplitude. The right panel shows the power spectrum (in the Fourier domain), normalised to integrate to $1$ (see appendix~\ref{s:stencil-norm}); only its left half, $k \in [0,\frac{M}{2}]$, is nonredundant. A relatively small $M$ is used to stress the fact that the power is discrete, although underlying it is a continuous curve. The horizontal dashed line represents a cutoff power corresponding to a certain value of the annealing parameter $\sigma$ (see section~\ref{s:fwddiff-family}).}
    \label{f:stencil-families1Da}
  \end{center}
\end{figure}

\subsubsection{Central-difference family: nets with sawtooth waves}
\label{s:cendiff-family}

This is defined by the first-order central-difference stencil $\varsigma = \left(-\frac{1}{2},\ 0,\ \frac{1}{2}\right)$, so \D\ has eigenvalues $\lambda_m = \frac{1}{2} \left( -1 + e^{i 4 \pi \frac{m}{M}} \right)$ and \SS\ has eigenvalues $\nu_m = \frac{1}{2} \left(1 - \cos{\left( 4 \pi \frac{m}{M} \right)} \right) = \sin^2{\left( 2 \pi \frac{m}{M} \right)}$. Thus the $p$th\nobreakdash-order derivative stencil has eigenvalues $\leftexp{(p)}{\nu_m} = \sin^{2p}{\left( 2 \pi \frac{m}{M} \right)} \in [0,1]$ $\forall m$. Note that $\sum_m{\varsigma_m} = \sum_m{\varsigma_m (-1)^m} = 0$ and so it is a sawtooth stencil (as also indicated by $\nu_{\frac{M}{2}} = 0$); thus, all stencils in this family are sawtooth.

The following proposition shows that the stencil of order $p$ of this family can be obtained from the forward-difference stencil of order $p$ by intercalating zero every two components and dividing by $2^p$ (see fig.~\ref{f:stencil-families}).
\begin{prop}
  \label{p:cen-Pascal}
  For $m = 0,\dots,2p$: $\leftexp{(p)}{\varsigma_m} = \frac{1}{2^p} (-1)^{p+n} \binom{p}{n}$ for $m = 2n$ even, zero otherwise.
\end{prop}
\begin{prop}
  \label{p:cen-sqmodulus}
  $\sum_m{\leftexp{(p)}{\varsigma^2_m}} = \frac{1}{2^{2p}} \binom{2p}{p}$ and $\trace{\smash{\leftexp{(p)}{\SS}}} = M \frac{1}{2^{2p}} \binom{2p}{p}$.
\end{prop}

This family also has a progression with decreasing slopes at low frequencies (see fig.~\ref{f:stencil-families1Da}), but every one of its stencils is a sawtooth stencil. Thus, both the low and high frequencies are practically not penalised. Given that the fitness term will favour high frequencies, because this generally allows to match training set points better, the elastic nets resulting from this family of stencils very often contain sawtooth patterns (for low enough $\sigma$). Such sawtooth patterns may take all the net or part of it, and can appear superimposed on a lower-frequency wave for some values of $\sigma$. One can also understand why this happens by noting that the tension term decouples into two terms, one for the even centroids and the other for the odd centroids (the zero coefficients of the stencil alternate with nonzero ones). However, in general it may not be obvious from the structure of the stencil whether it is sawtooth, as in $\frac{1}{2} (-1,\ 2,\ 0,\ -2,\ 1)$ in fig.~\ref{f:sawtooth}; of course, the sawtooth conditions of section~\ref{s:sawtooth} can always be used.

Naturally, 2D stencils obtained by combining a 1D central-difference stencil along the horizontal and vertical directions are also sawtooth. As a different example of sawtooth stencil in 2D, consider $\nabla^2_{\times} \bydef \frac{1}{2} \left(\begin{smallmatrix} 1 &   & 1 \\   & -4 &   \\ 1 &   & 1 \end{smallmatrix}\right)$, which is a well-known finite-difference approximation to the 2D Laplacian with quadratic truncation error. This stencil verifies the sawtooth condition and so develops sawtooth in 2D, which appear as grating- or checkerboard-like patterns in the OD and OR maps. One can also understand why this happens in an analogous way to the 1D case, as follows. The linear combinations of consecutive points in the matrix \D\ do not have any element in common; this can be visualised by shifting the stencil either horizontally or vertically and seeing that no nonzero coefficients overlap. Consequently, the tension term becomes the sum of two uncoupled terms, one for the ``black squares'' and the other for the ``white'' ones (imagining again a checkerboard).

The fact that the central-difference stencil will typically result in a net with sawteeth suggests that it should be avoided in applications where continuity of representation is important, such as cortical maps.  However note that, while the whole net develops sawteeth, the two uncoupled subnets show individually a smooth structure (see e.g.\ fig.~\ref{f:sawtooth}A,B). This suggests we can use the central-difference stencil to solve a multiple TSP, where the training set ``cities'' must be visited by a given number of salesmen (see fig.~\ref{f:sawtooth-TSP} and section~\ref{e:discussion}).

The same type of techniques can be applied to any matrix \D, not necessarily a differential operator (although, as discussed in section~\ref{s:D:types-op}, non-differential operators are not desirable because they bias the centroids towards the origin of coordinates). For example, for $\varsigma = (1)$, which corresponds to $\D = \SS = \I$, the eigenvalues of \SS\ are $\nu_m = 1$ (DFT of a delta). Thus, all frequencies are equally penalized, in agreement with the fact that the prior $p(\Y)$ factorizes and all centroids are independent. For $\varsigma = \frac{1}{3} (1,\ 4,\ 1)$ (Simpson's integration rule) we get $\nu_m = \smash{\left( \frac{2}{3} \left( 2 + \cos{\left( 2 \pi \frac{m}{M} \right)} \right) \right)}^2$, which decreases from a maximum at frequency $0$ to a minimum at frequency $\frac{M}{2}$, the opposite to the forward difference. Thus, even though $\nu_{\frac{M}{2}} > 0$, the sawtooth frequency is the least penalized and so the net develops sawteeth---consistent with the fact that this is an integral, or smoothing, stencil.

\begin{FPfigure}
    \psfrag{k}[t]{$k$}
    \psfrag{P(k)}[r][l][1][-90]{$p_k$}
    \psfrag{1A}{\small 1A}
    \psfrag{1B}{\small 1B}
    \psfrag{1C}{\small 1C}
    \psfrag{1D}{\small 1D}
    \psfrag{1E}{\small 1E}
    \psfrag{1F}{\small 1F}
    \psfrag{1G}{\small 1G}
    \psfrag{1H}{\small 1H}
    \psfrag{1I}{\small 1I}
    \psfrag{1J}{\small 1J}
    \psfrag{1K}{\small 1K}
    \psfrag{2A}{\small 2A}
    \psfrag{2B}{\small 2B}
    \psfrag{2C}{\small 2C}
    \psfrag{2D}{\small 2D}
    \psfrag{2E}{\small 2E}
    \psfrag{2F}{\small 2F}
    \psfrag{2G}{\small 2G}
    \psfrag{2H}{\small 2H}
    \psfrag{2I}{\small 2I}
    \psfrag{2J}{\small 2J}
    \psfrag{2K}{\small 2K}
    \psfrag{3A}{\small 3A}
    \psfrag{3B}{\small 3B}
    \psfrag{3C}{\small 3C}
    \psfrag{3D}{\small 3D}
    \psfrag{3E}{\small 3E}
    \psfrag{3F}{\small 3F}
    \psfrag{3G}{\small 3G}
    \psfrag{3H}{\small 3H}
    \psfrag{3I}{\small 3I}
    \psfrag{3J}{\small 3J}
    \psfrag{3K}{\small 3K}
    \psfrag{4A}{\small 4A}
    \psfrag{4B}{\small 4B}
    \psfrag{4C}{\small 4C}
    \psfrag{4D}{\small 4D}
    \psfrag{4E}{\small 4E}
    \psfrag{4F}{\small 4F}
    \psfrag{4G}{\small 4G}
    \psfrag{4H}{\small 4H}
    \psfrag{4I}{\small 4I}
    \psfrag{4J}{\small 4J}
    \psfrag{4K}{\small 4K}
    \begin{tabular}{@{}c@{\hspace{0.5cm}}c@{\hspace{1.5cm}}c@{}}
      \rotatebox{90}{\makebox[\MACPlengthA][c]{First-order differences}} &
      \includegraphics[height=\MACPlengthA]{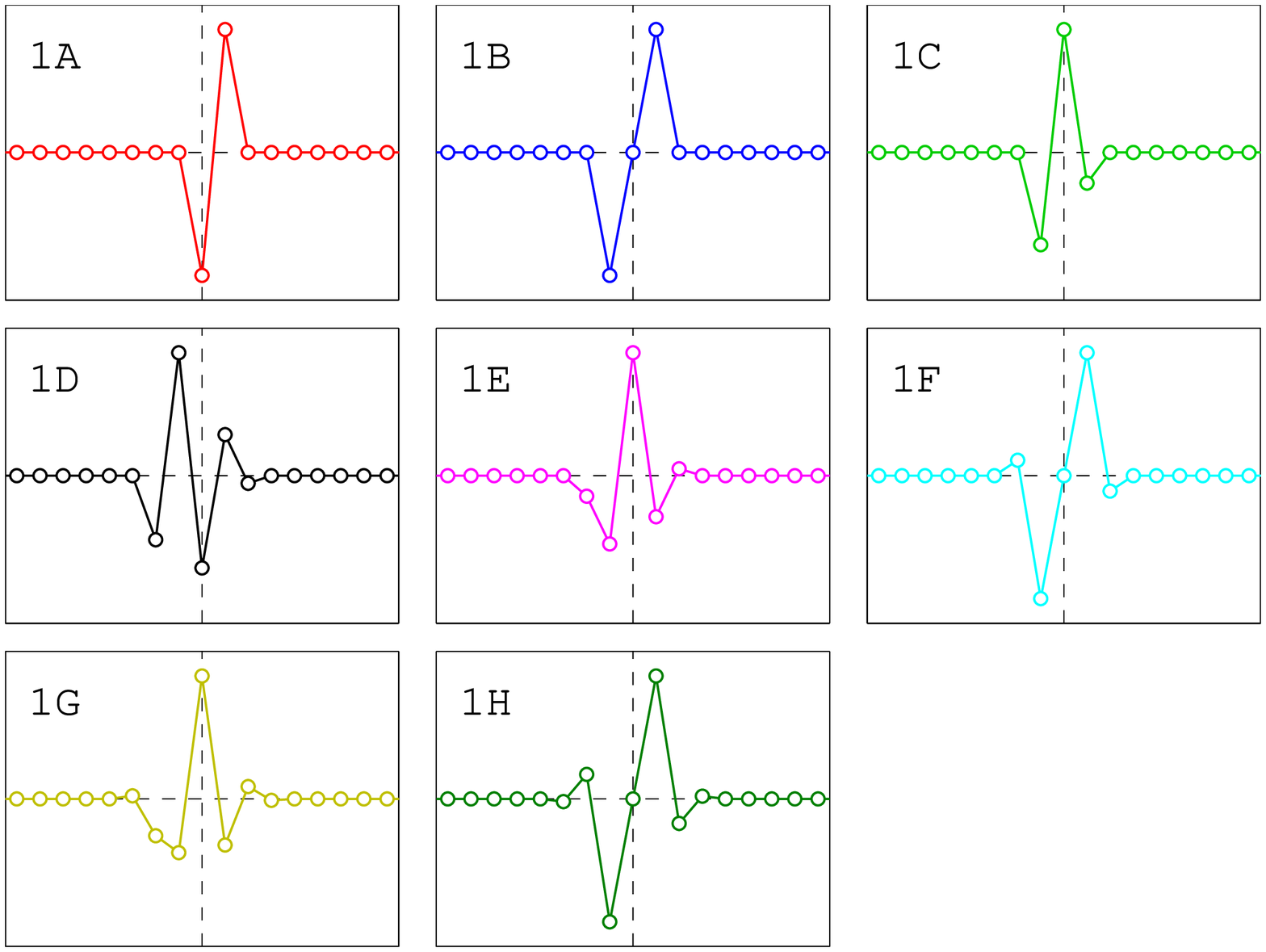} &
      \includegraphics[height=\MACPlengthA]{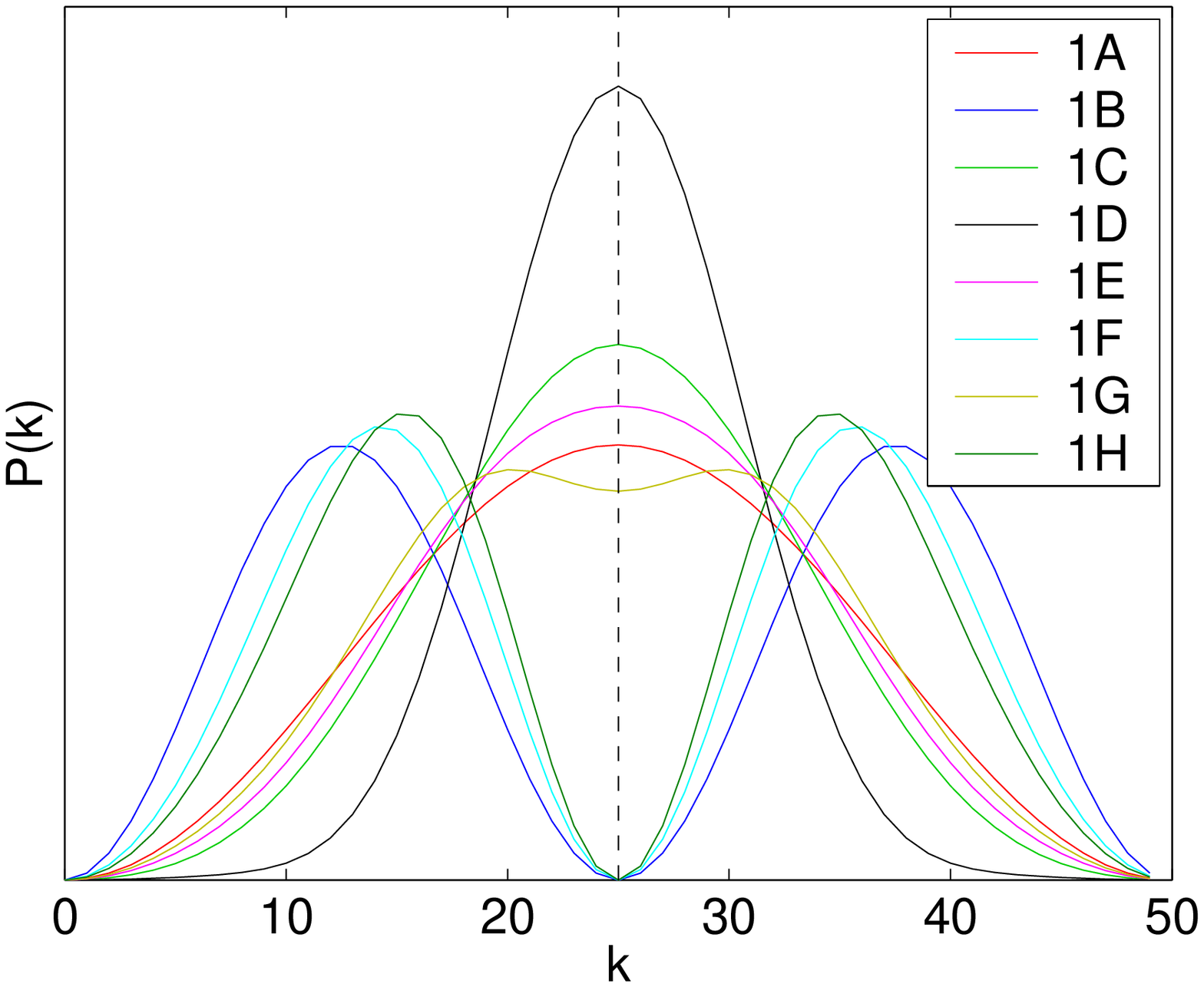} \\[2ex]
      \rotatebox{90}{\makebox[\MACPlengthA][c]{Second-order differences}} &
      \includegraphics[height=\MACPlengthA]{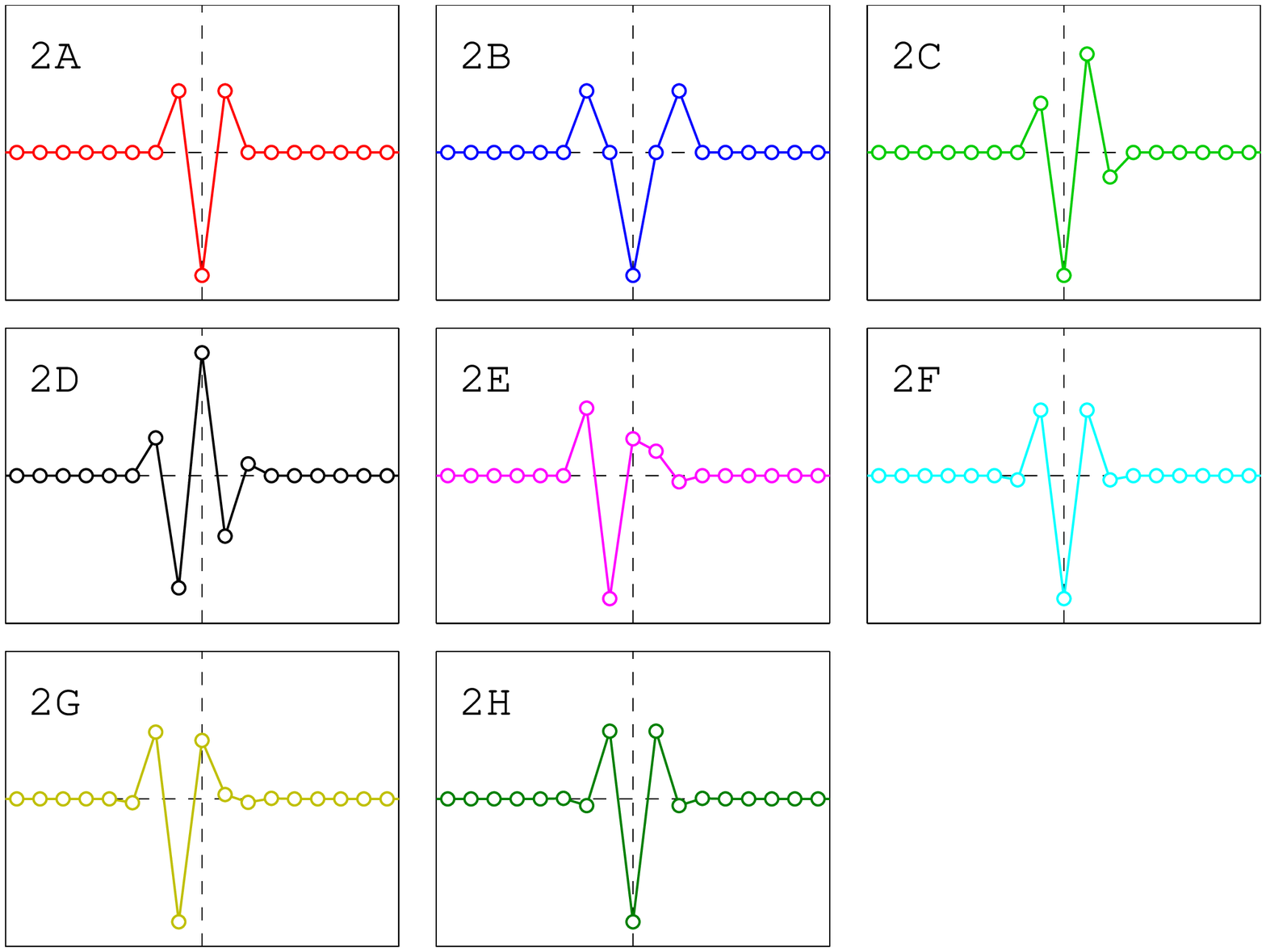} &
      \includegraphics[height=\MACPlengthA]{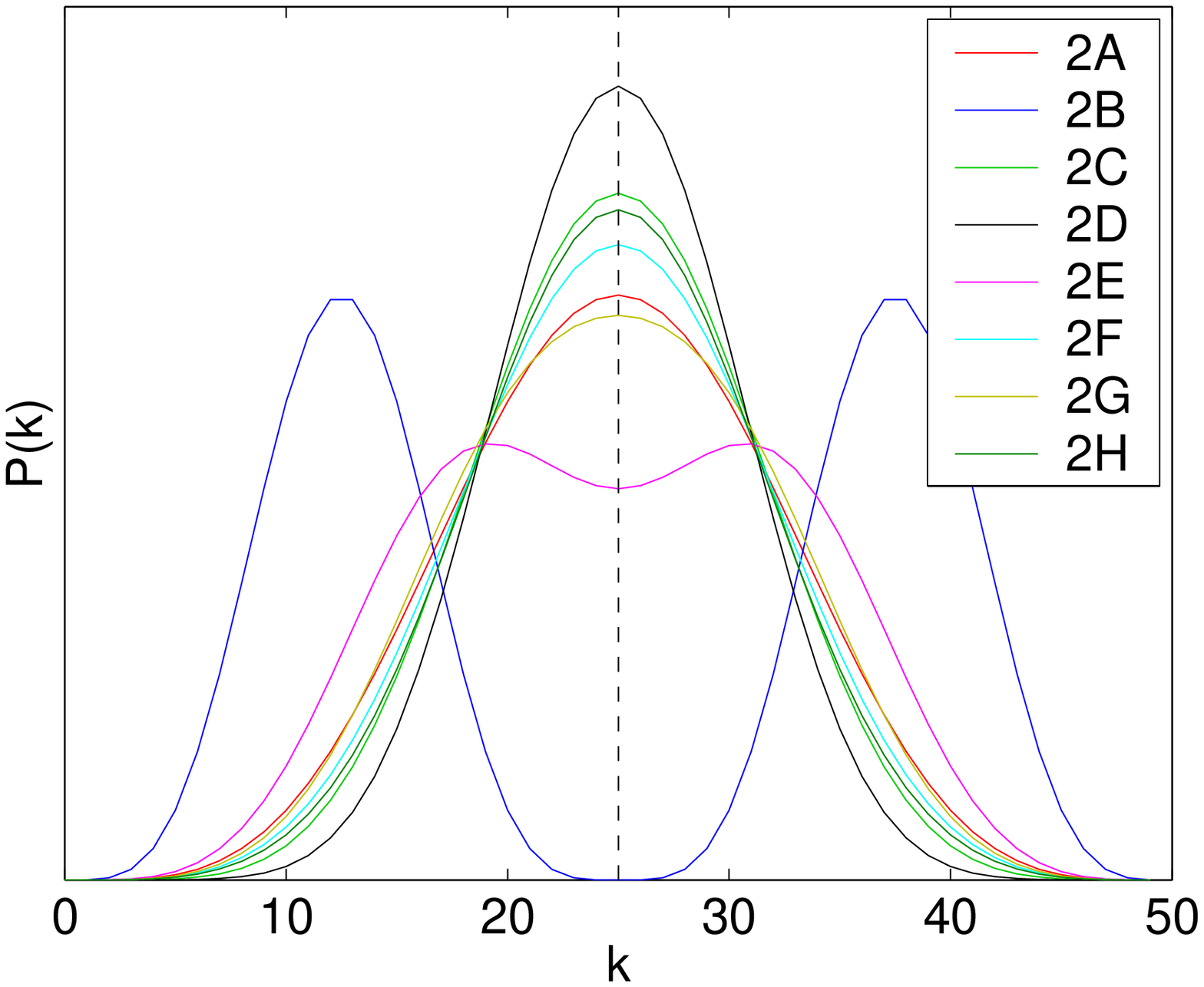} \\[2ex]
      \rotatebox{90}{\makebox[\MACPlengthA][c]{Third-order differences}} &
      \includegraphics[height=\MACPlengthA]{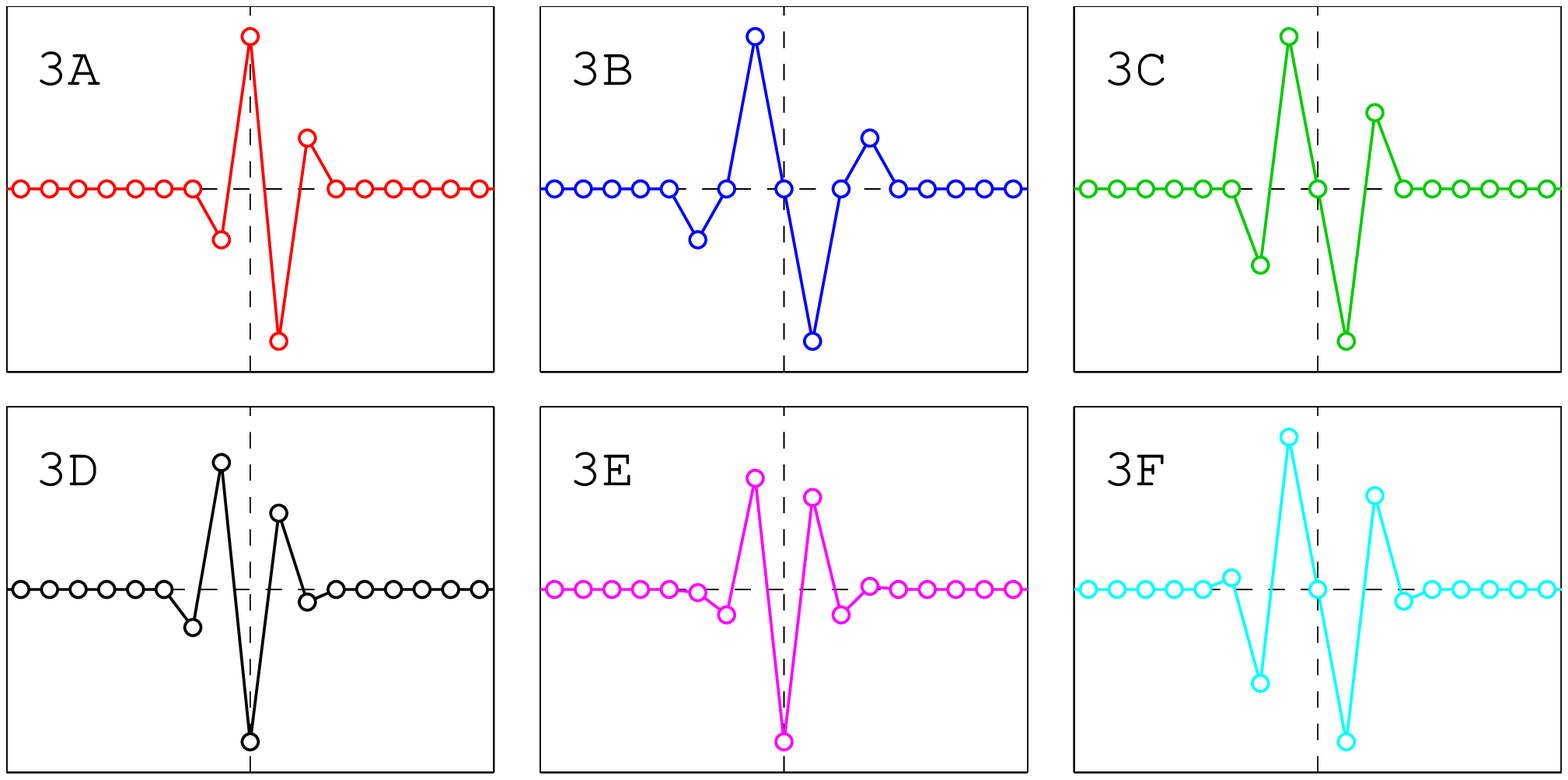} &
      \includegraphics[height=\MACPlengthA]{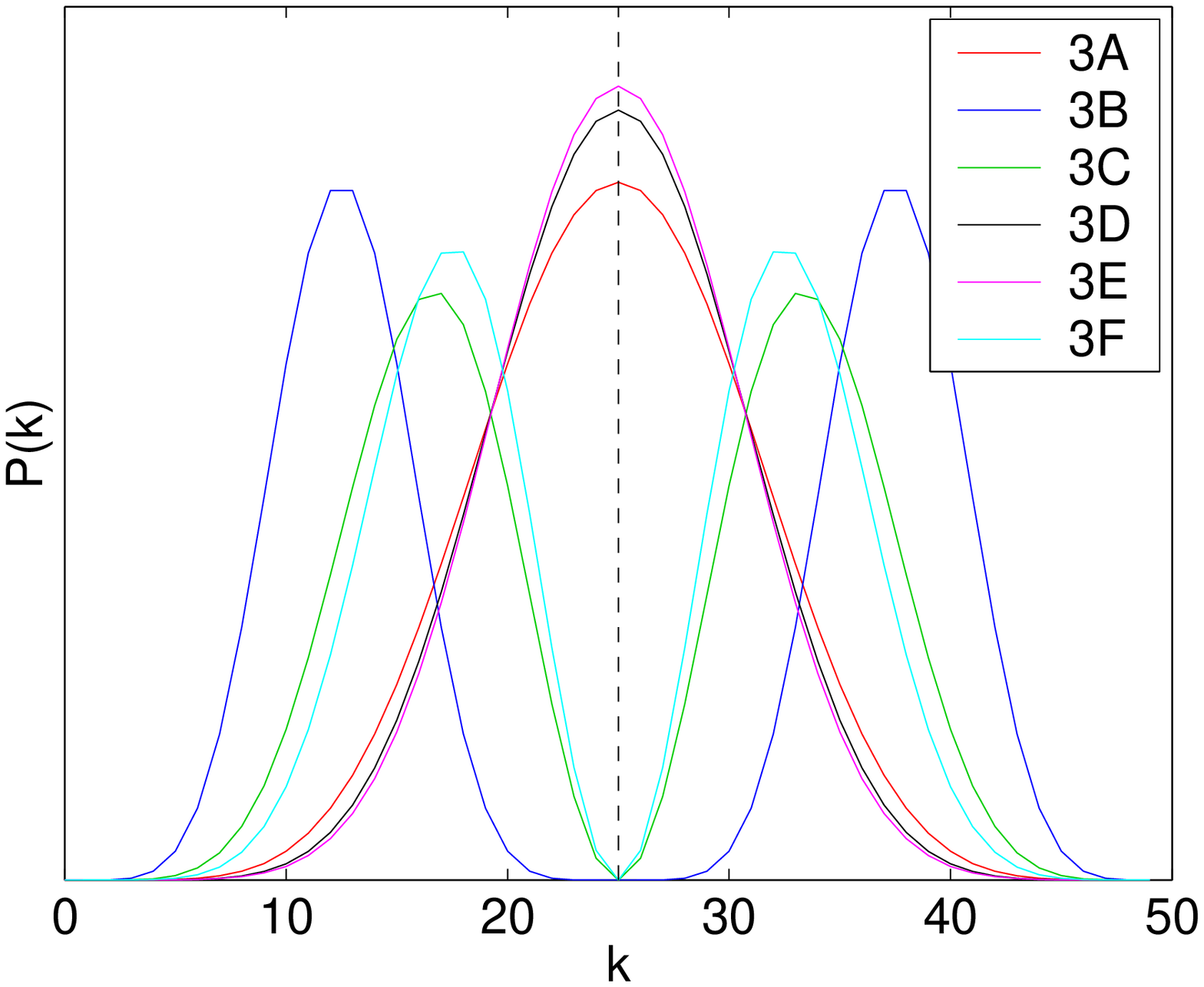} \\[2ex]
      \rotatebox{90}{\makebox[\MACPlengthA][c]{Fourth-order differences}} &
      \includegraphics[height=\MACPlengthA]{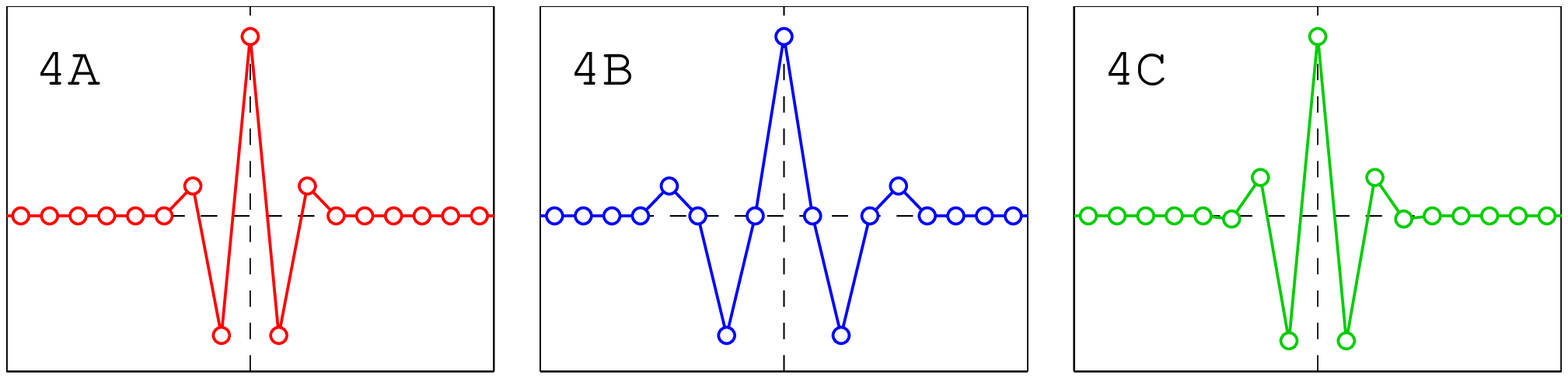} &
      \includegraphics[height=\MACPlengthA]{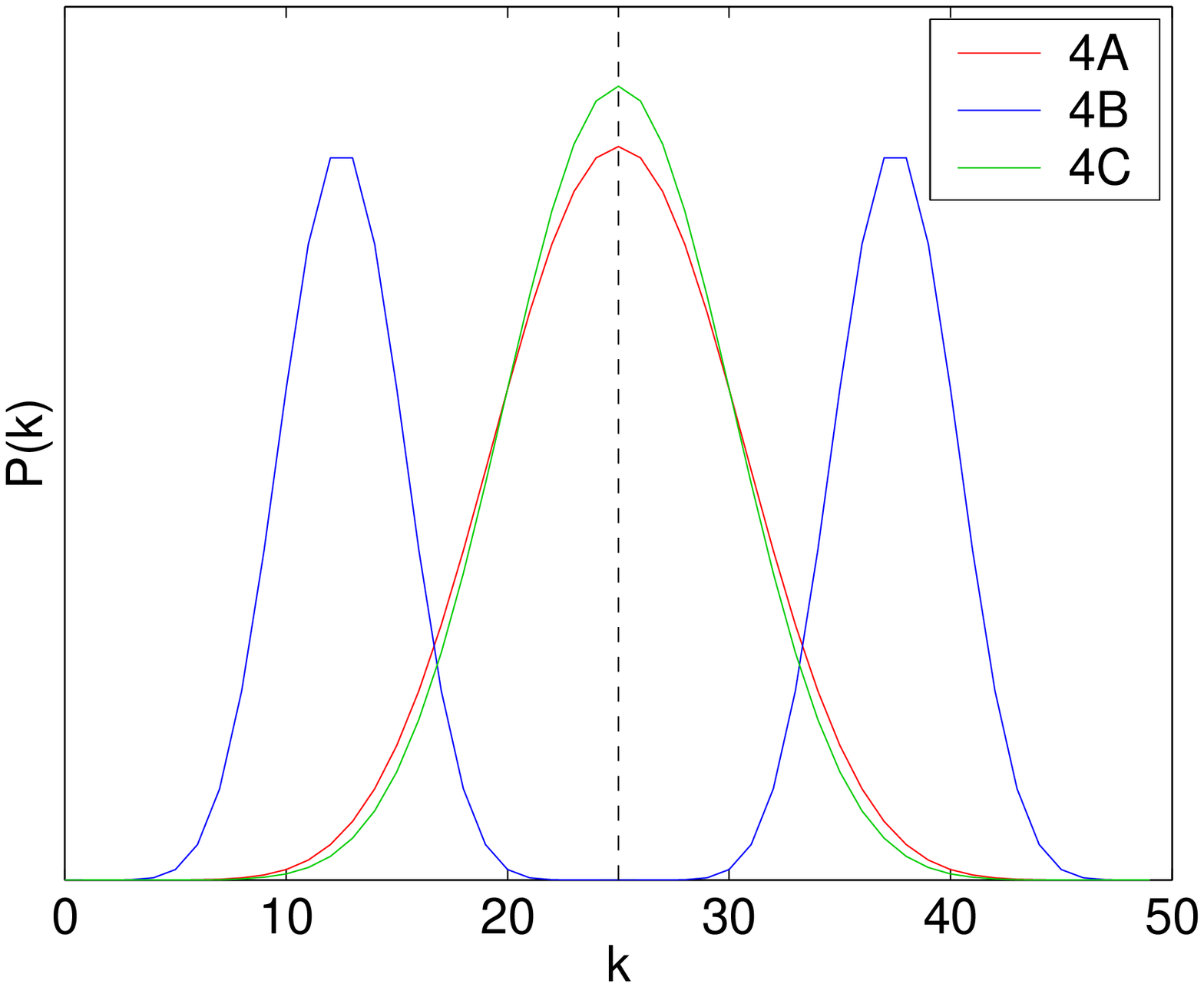}
    \end{tabular}
    \caption{A gallery of 1D stencils of orders $p = 1$ to $4$, as in fig.~\ref{f:stencil-families1Da}. The keys 1A, etc.\ correspond to those in table~\ref{t:D:diff:1D}. While in the net domain the stencils may look very different, the power curves look similar, typically either peaking at $p_{\frac{M}{2}}$ (like the forward-difference stencils) or dipping at $p_{\frac{M}{2}}$ (like the central-difference ones), the latter being sawtooth stencils. Also note that, while for a fixed $p$ the non-sawtooth power curves differ little from each other, there are occasional outliers (e.g.\ 1D).}
    \label{f:stencil-families1Db}
\end{FPfigure}

\begin{figure}
  \psfrag{k1}[t]{}
  \psfrag{k2}[t]{}
  \begin{center}
    \begin{tabular}{@{}c@{\hspace{.02\textwidth}}c@{}c@{}c@{}c@{}}
      & $p = 1$ & $p = 2$ & $p = 3$ & $p = 4$ \\
      \rotatebox{90}{\makebox[.24\textwidth][c]{Forward-difference family}} &
      \includegraphics[height=.24\textwidth]{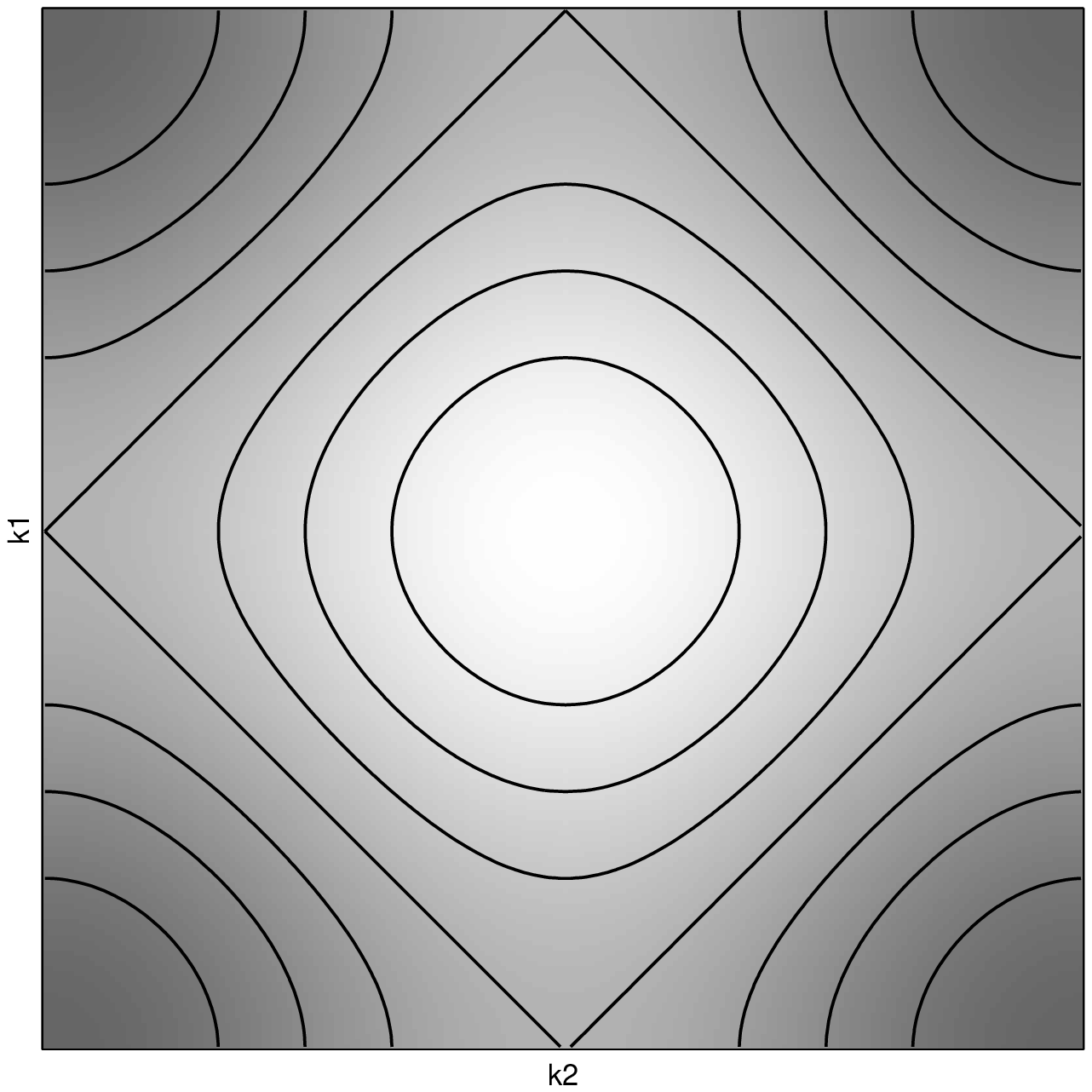} &
      \includegraphics[height=.24\textwidth]{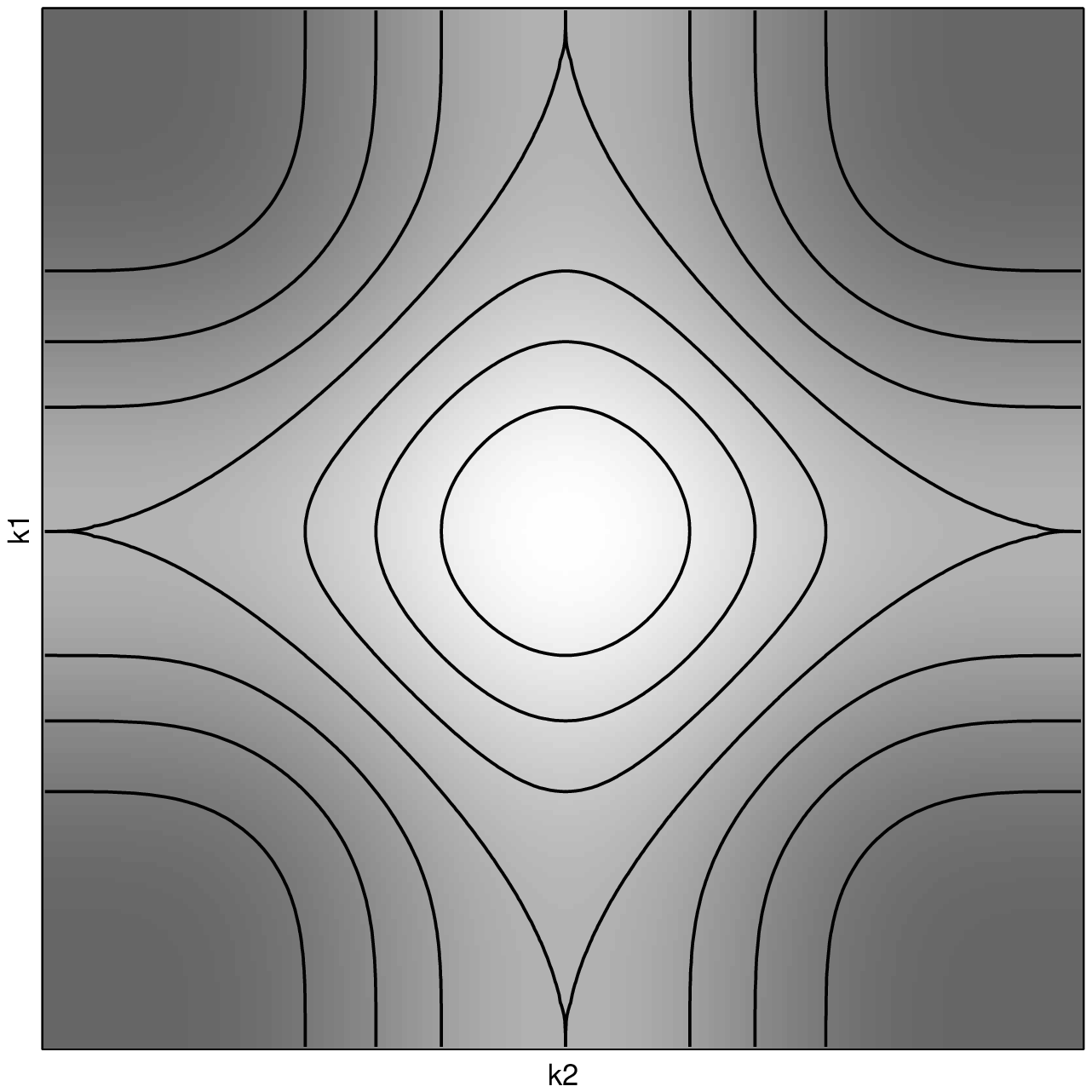} &
      \includegraphics[height=.24\textwidth]{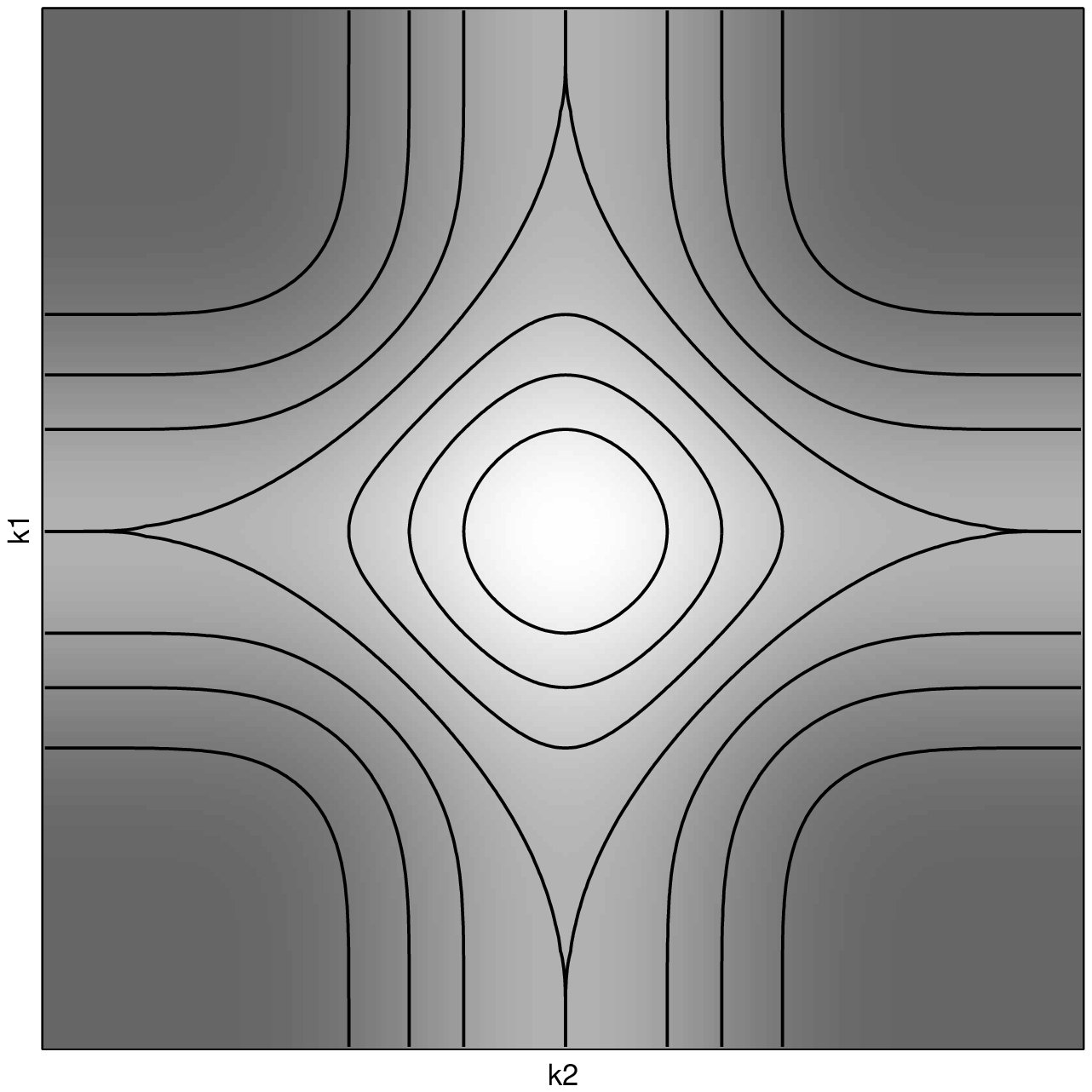} &
      \includegraphics[height=.24\textwidth]{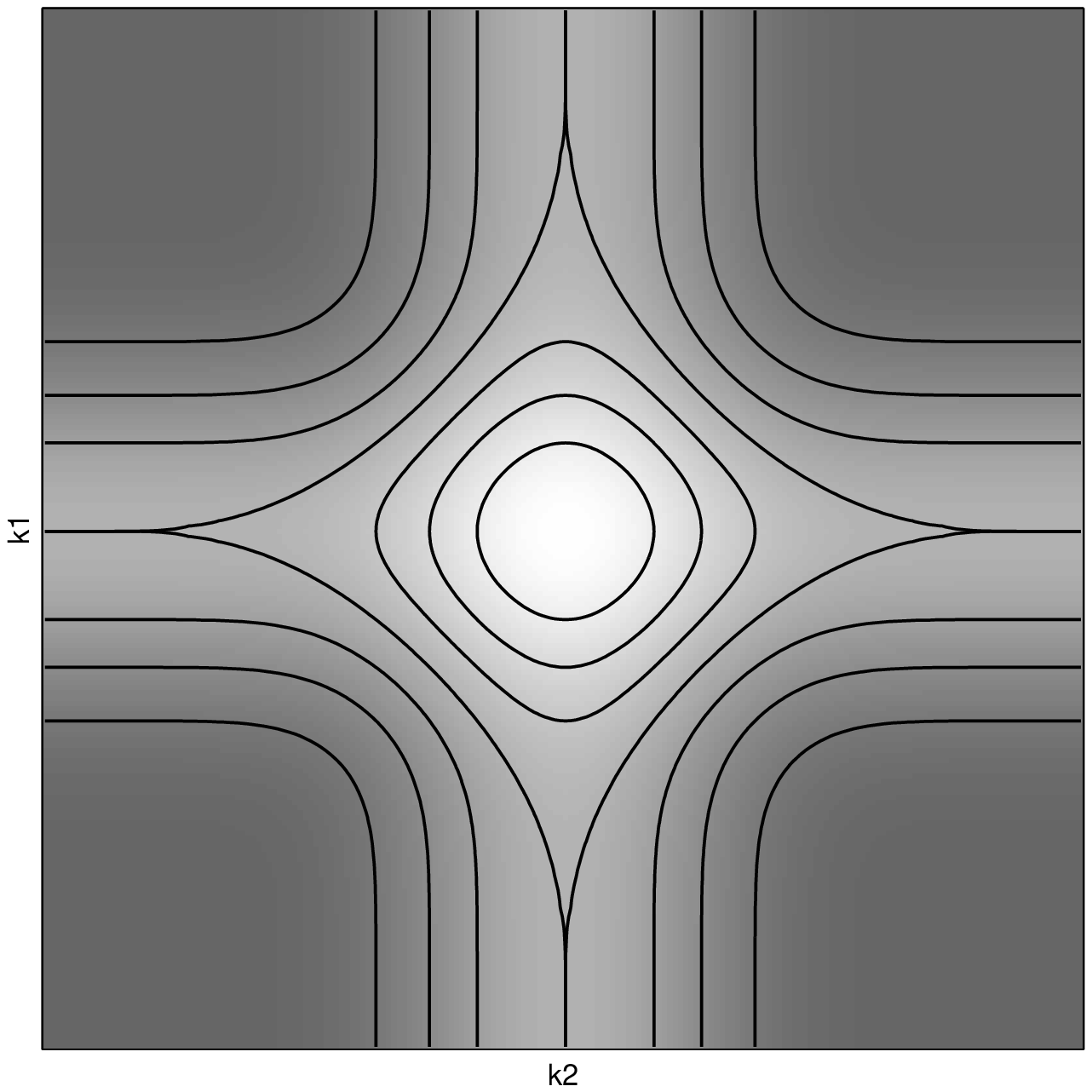} \\
      \rotatebox{90}{\makebox[.24\textwidth][c]{Central-difference family}} &
      \includegraphics[height=.24\textwidth]{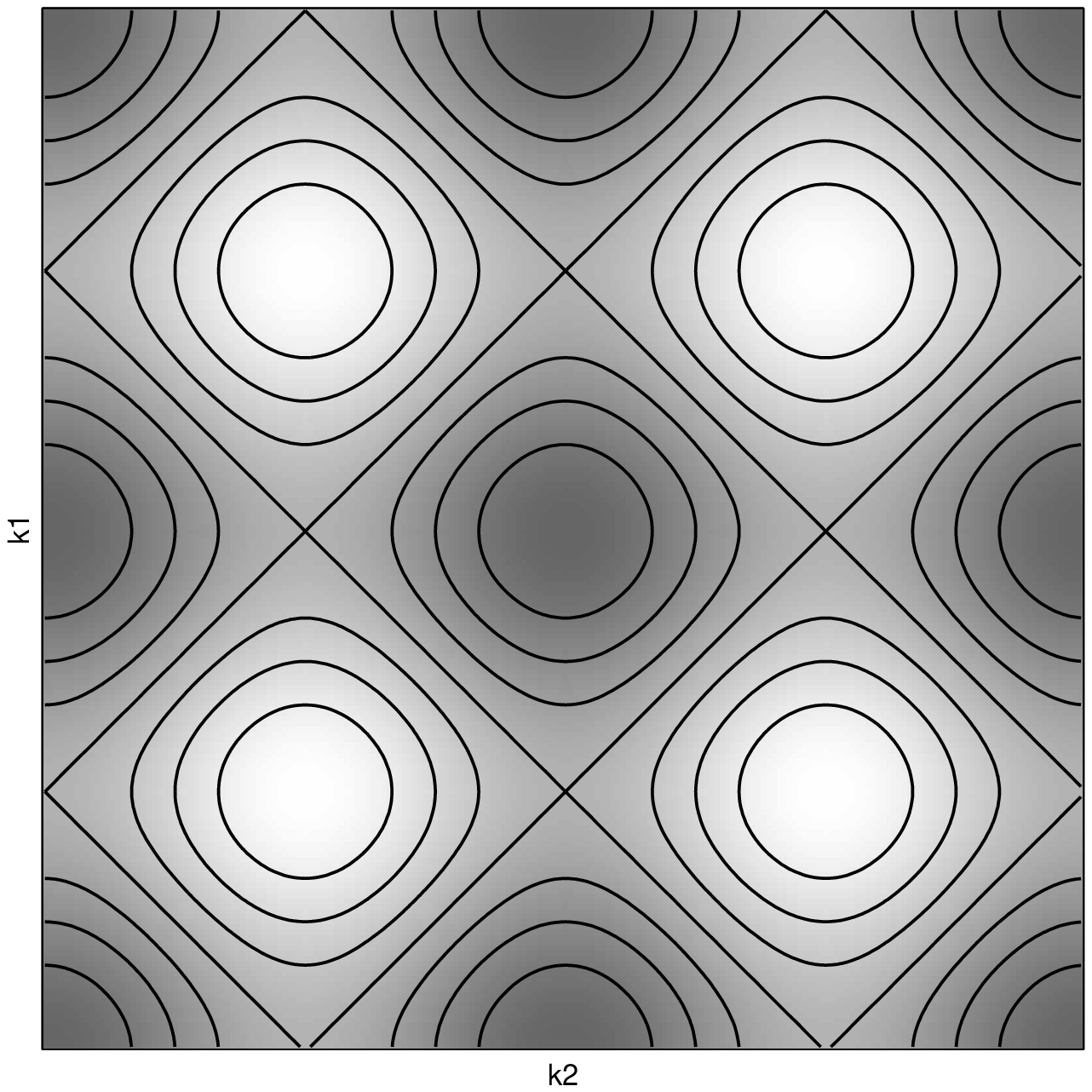} &
      \includegraphics[height=.24\textwidth]{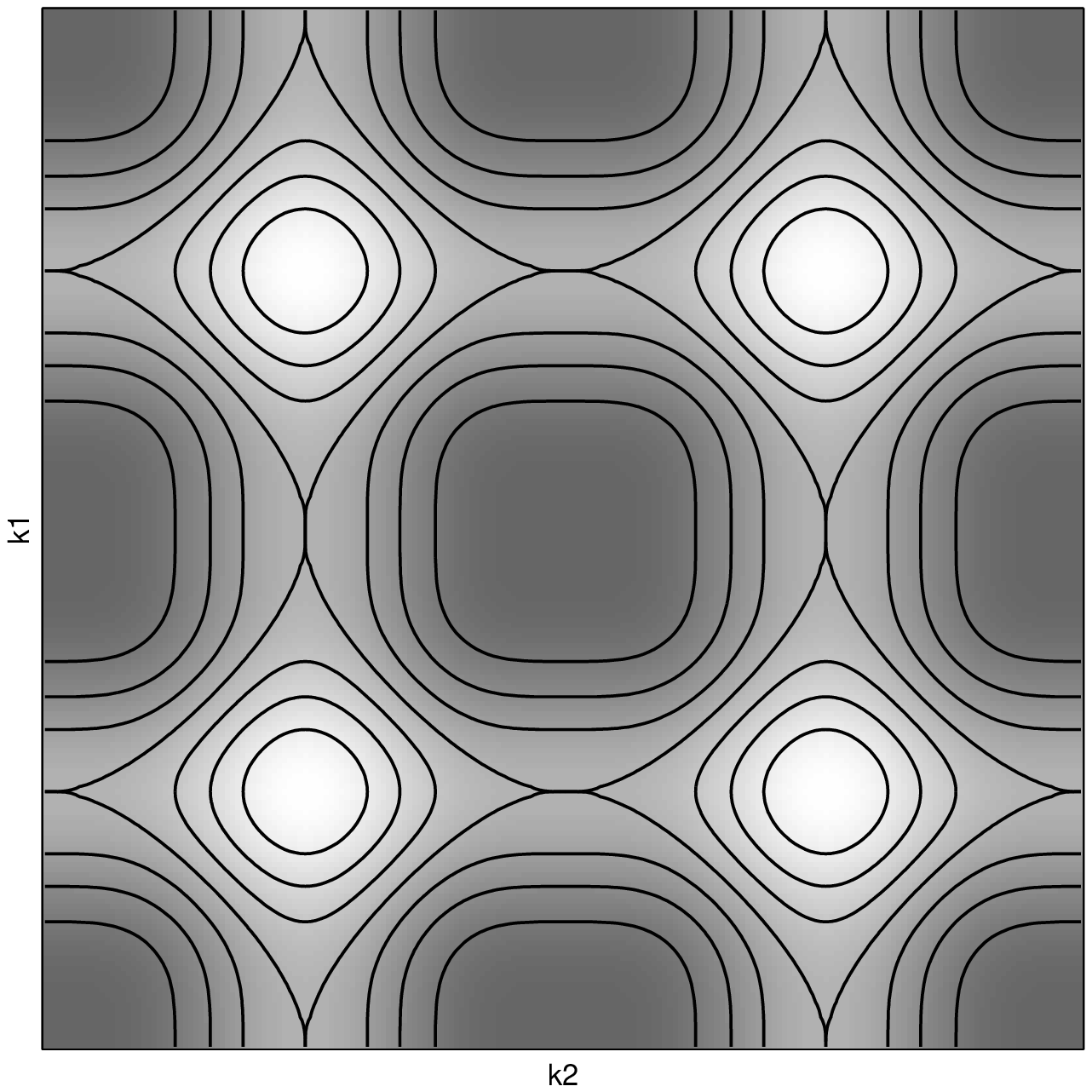} &
      \includegraphics[height=.24\textwidth]{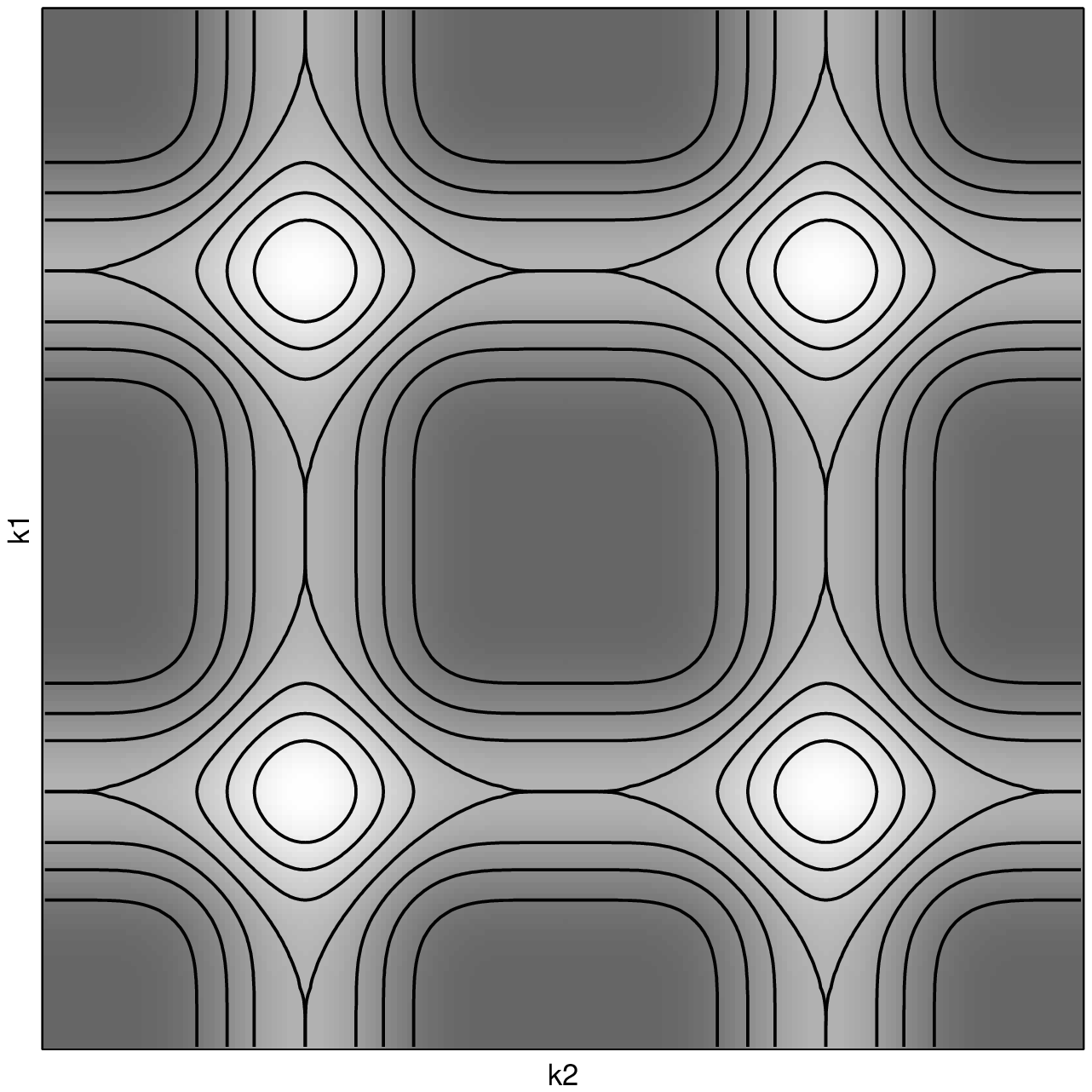} &
      \includegraphics[height=.24\textwidth]{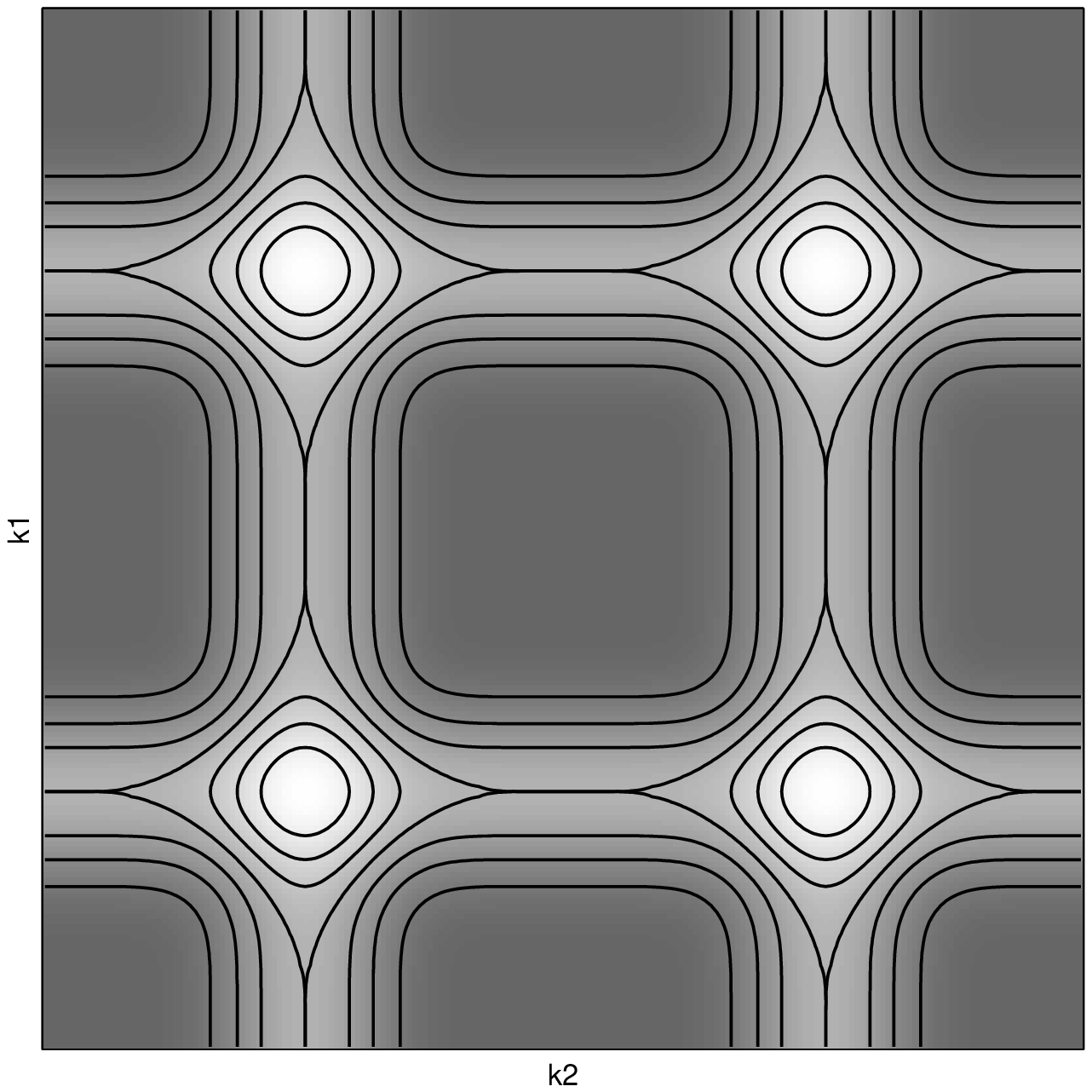}
    \end{tabular}
    \caption{2D stencils for the forward- and central-difference families (resulting from passing the respective 1D stencil horizontally and vertically). The graphs show the power spectrum in the space $(k_1,k_2) \in [0,M]^2$, with $(k_1,k_2) = (0,0)$ at the lower left corner, for a square net with $M \times M = 50 \times 50$ centroids; white means higher power. The nonredundant power spectrum is in its first quadrant $\left[0,\frac{M}{2}\right]^2$.}
    \label{f:stencil-families2Da}
  \end{center}
\end{figure}

\begin{figure}
  \psfrag{k1}[t]{}
  \psfrag{k2}[t]{}
  \begin{center}
    \begin{tabular}{@{}c@{\hspace{.02\textwidth}}c@{}c@{}c@{}c@{}}
      & $\nabla^2_{+}$ & $\nabla^2_{\times}$ & $\nabla^2_9$ & $\nabla^2_{\text{a}}$ \\
      \rotatebox{90}{\makebox[.24\textwidth][c]{Laplacian ($p = 2$)}} &
      \includegraphics[height=.24\textwidth]{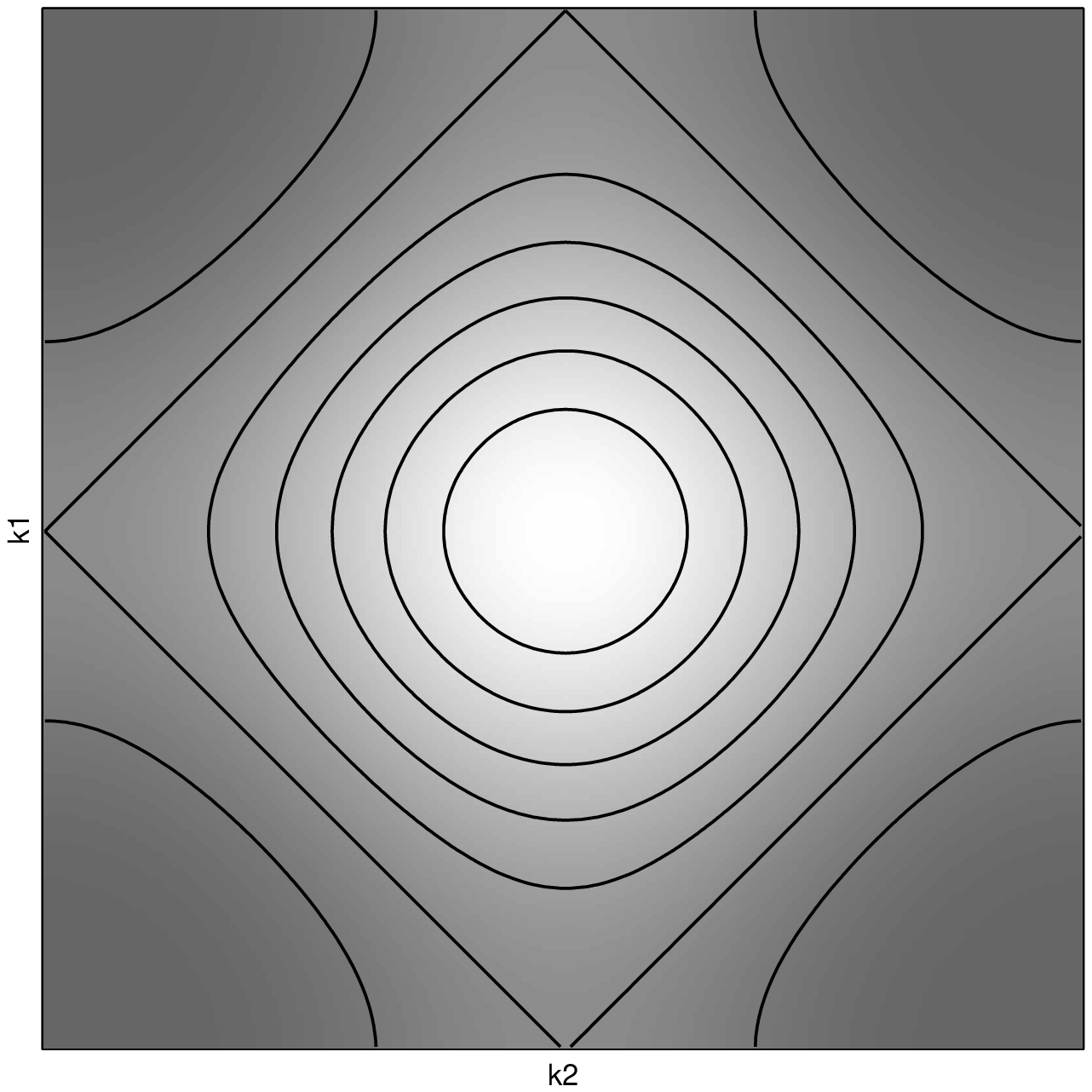} &
      \includegraphics[height=.24\textwidth]{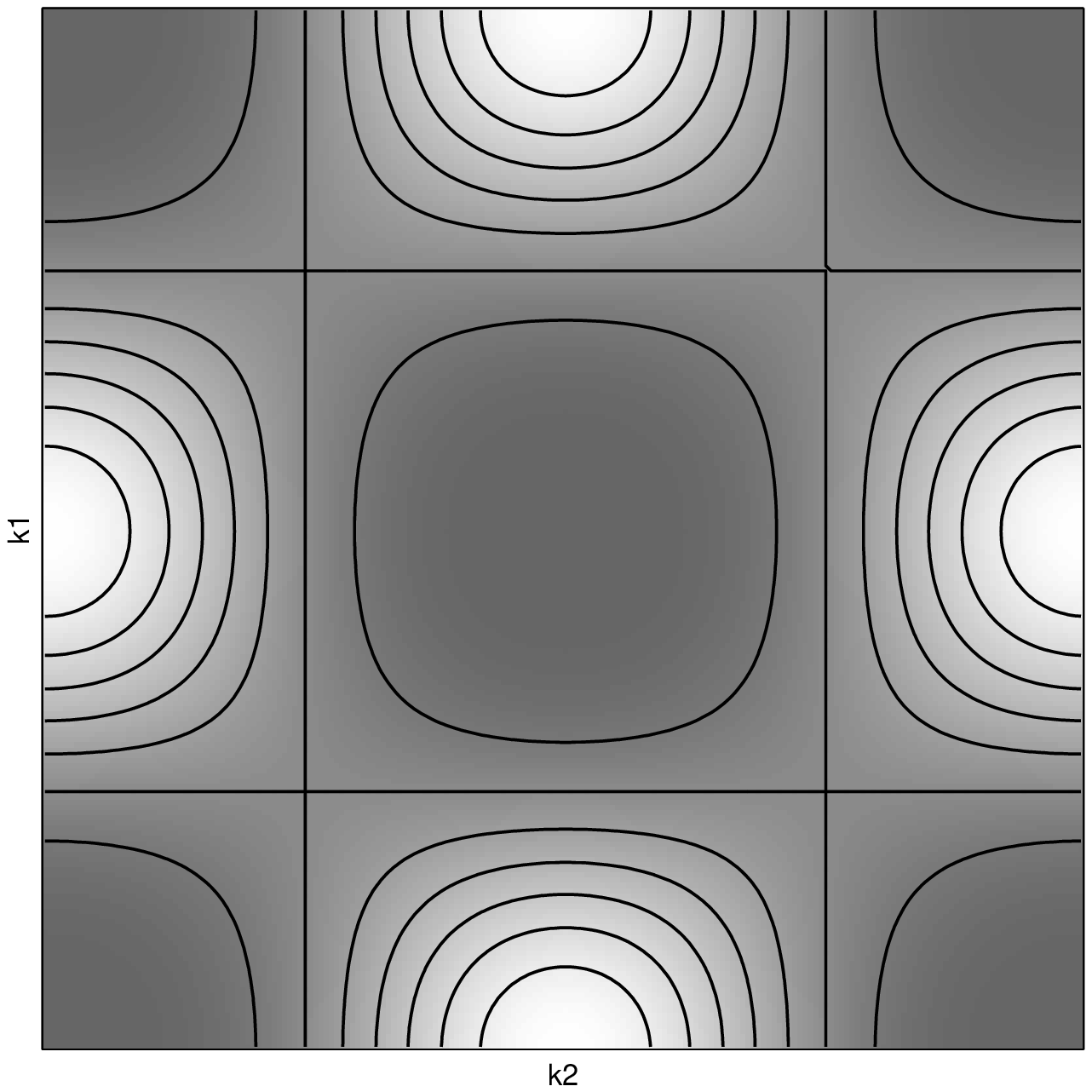} &
      \includegraphics[height=.24\textwidth]{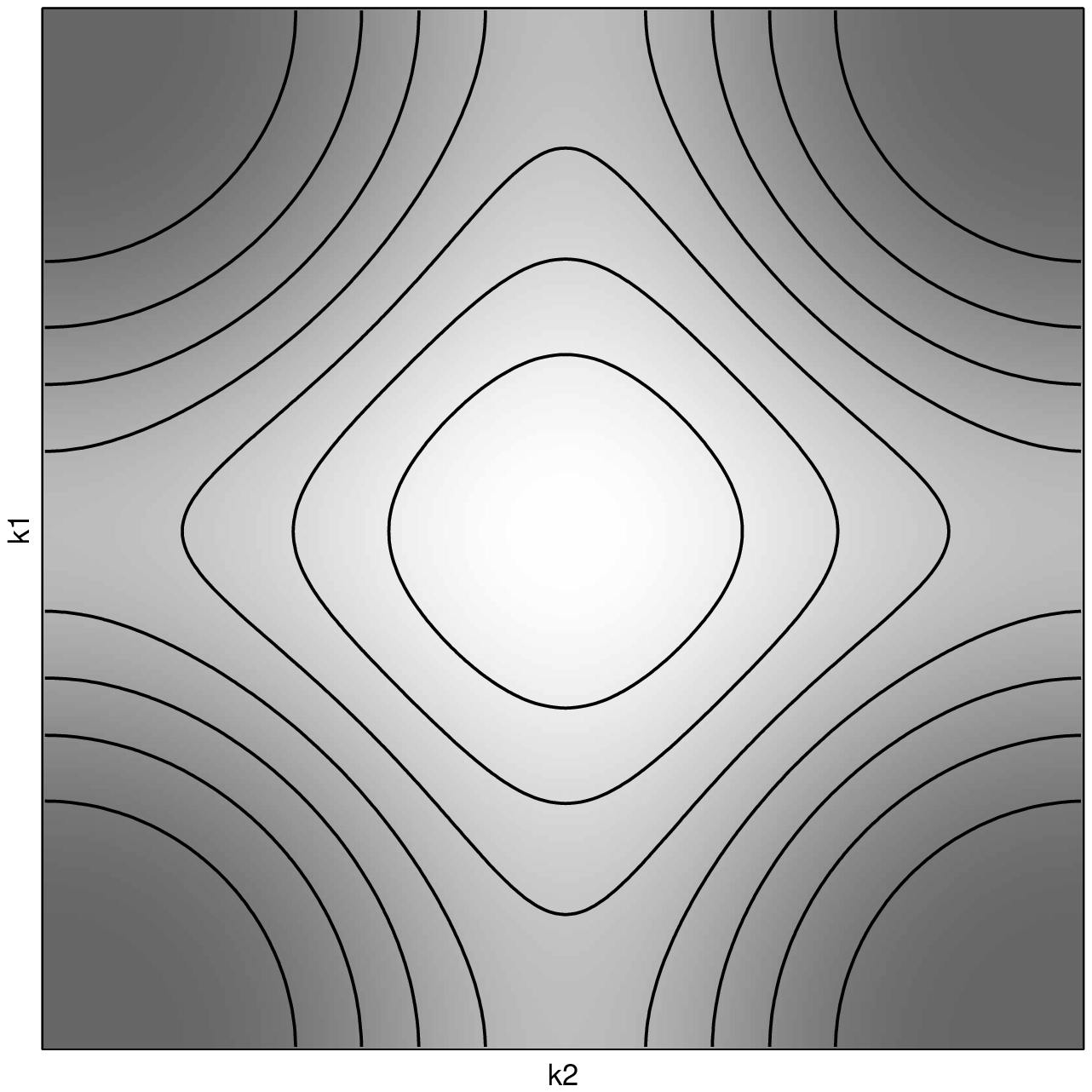} &
      \includegraphics[height=.24\textwidth]{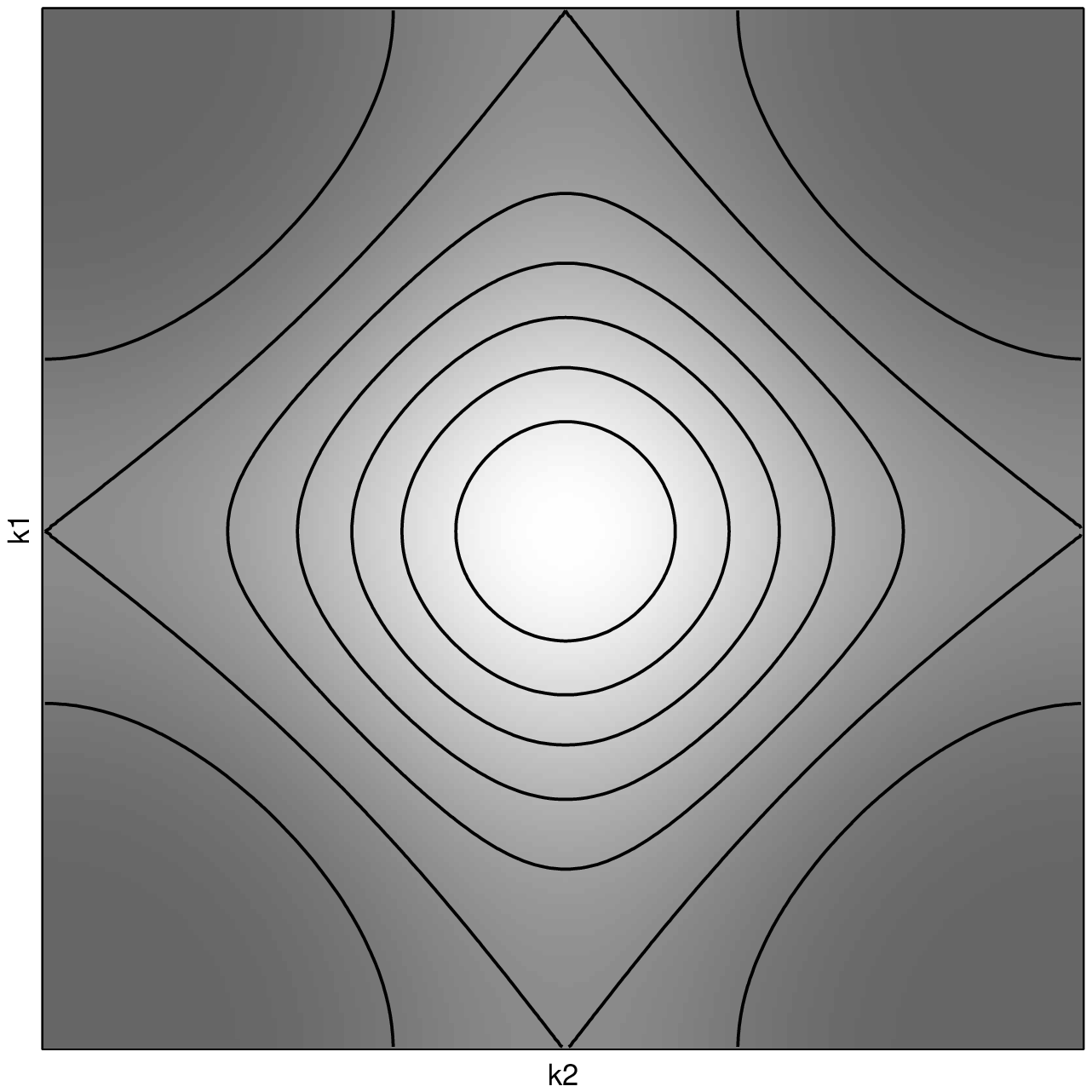} \\
      & $\nabla^4_{\text{a}}$ & $\nabla^4_{\text{b}}$ \\
      \rotatebox{90}{\makebox[.24\textwidth][c]{Biharmonic ($p = 4$)}} &
      \includegraphics[height=.24\textwidth]{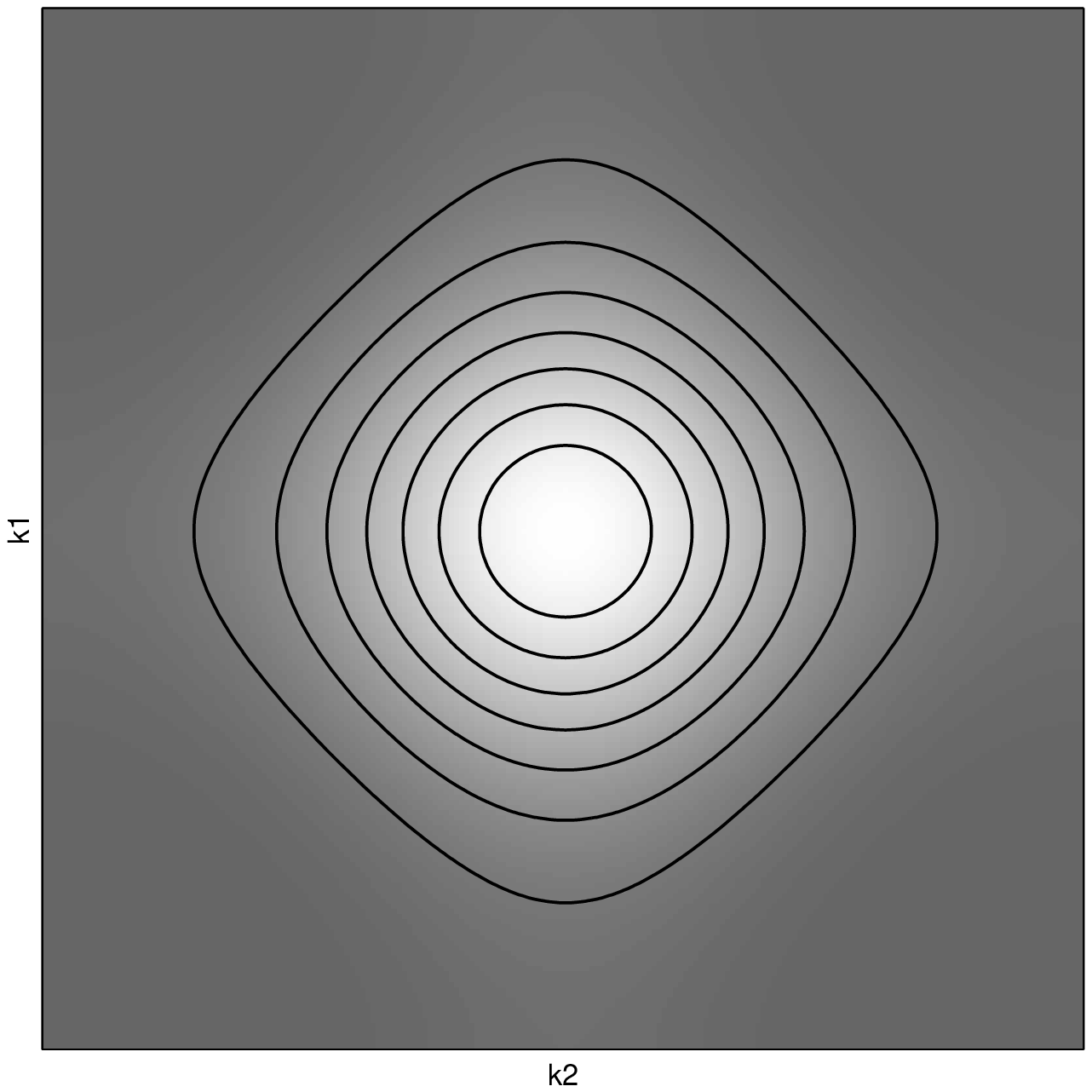} &
      \includegraphics[height=.24\textwidth]{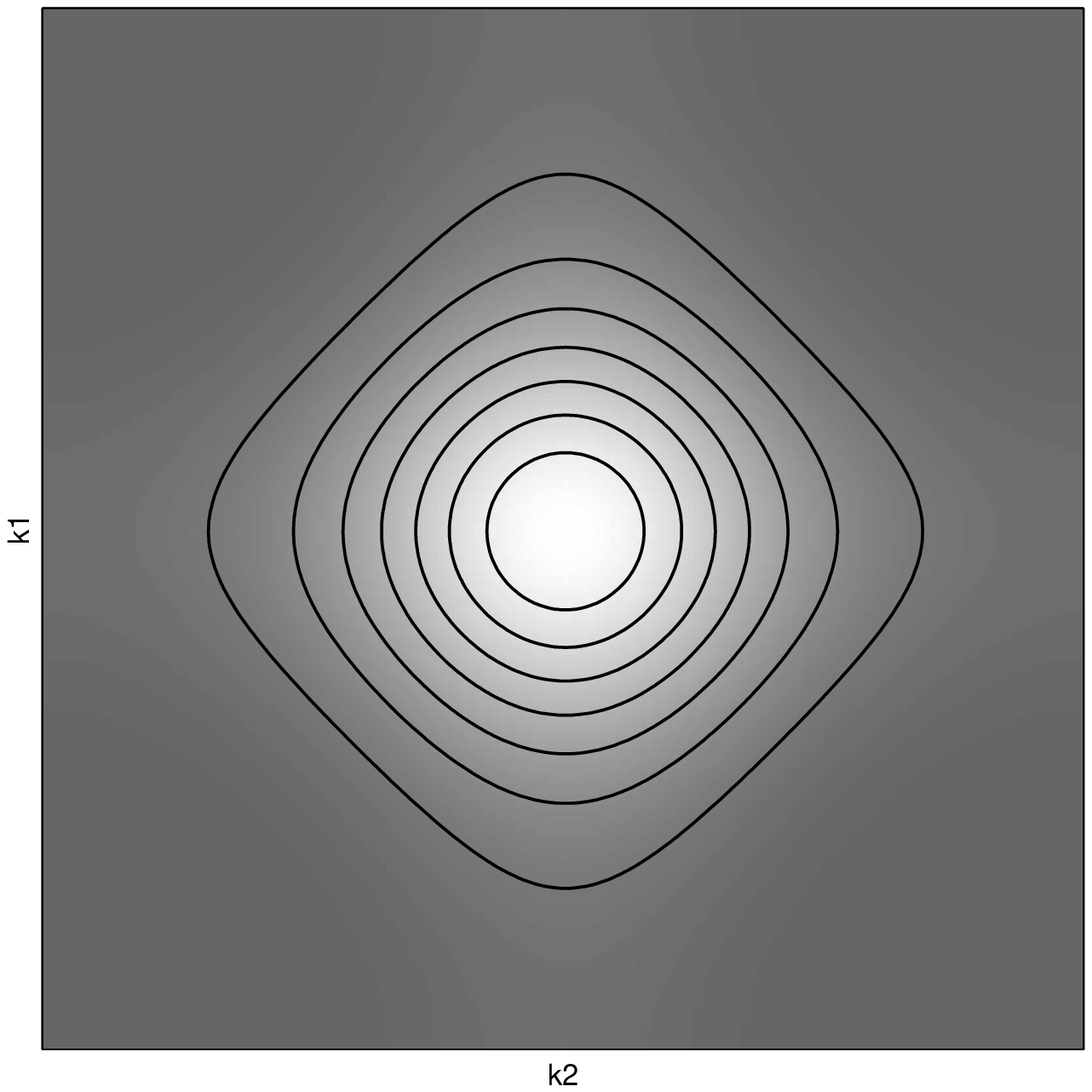}
    \end{tabular}
    \caption{A gallery of 2D stencils of orders $p = 2$ and $4$ from table~\ref{t:D:diff:2D}, as in fig.~\ref{f:stencil-families2Da}. \emph{Top row}: Laplacian stencils $\nabla^2_{+}$, $\nabla^2_{\times}$, $\nabla^2_9$ and $\nabla^2_{\text{a}}$. \emph{Bottom row}: biharmonic stencils $\nabla^4_{\text{a}}$, $\nabla^4_{\text{b}}$.}
    \label{f:stencil-families2Db}
  \end{center}
\end{figure}


\begin{figure}
  \begin{center}
    \psfrag{A}[rB][rB]{A}
    \psfrag{B}[rB][rB]{B}
    \psfrag{C}[rB][rB]{C}
    \psfrag{D}[rB][rB]{D}
    \includegraphics[width=\textwidth]{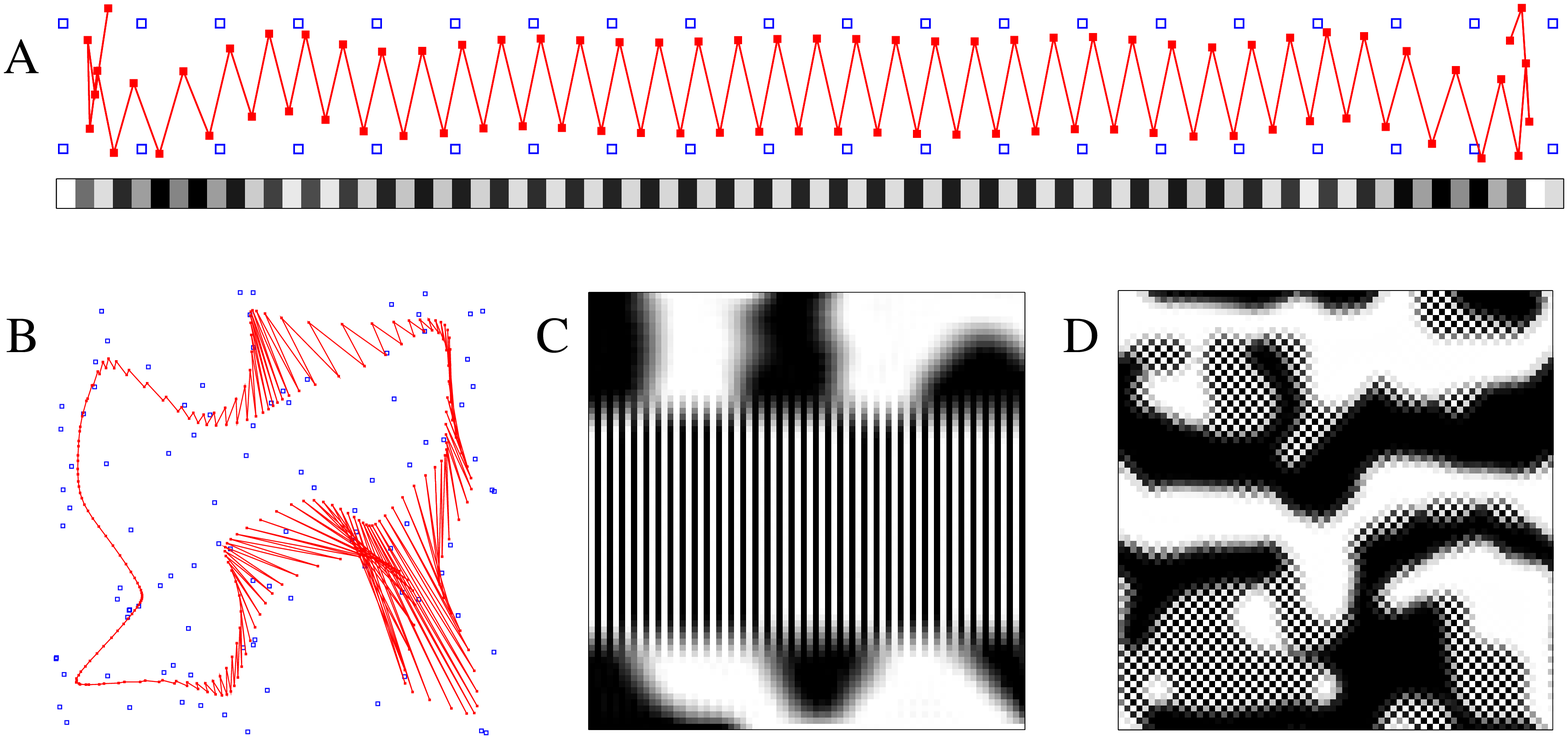}
    \caption{Typical nets resulting from sawtooth stencils. \emph{A}: sawtooth arising in a 1D nonperiodic net for ocular dominance and retinotopy with the third-order stencil $\frac{1}{2} (-1,\ 2,\ 0,\ -2,\ 1)$ (3C in table~\ref{t:D:diff:1D}), which results from applying twice the forward-difference $(0,\ -1,\ 1)$ and once the central-difference $\frac{1}{2} (-1,\ 0,\ 1)$. \emph{B}: sawtooth arising in a 1D periodic net for a TSP, with the central-difference stencil $\frac{1}{2} (-1,\ 0,\ 1)$. \emph{C}: sawtooth arising in 2D nets for ocular dominance with the third-order stencil $\frac{1}{2} (-1,\ 2,\ 0,\ -2,\ 1)$. This stencil does not penalise the sawtooth frequency either horizontally, vertically or diagonally (see fig.~\ref{f:stencil-families2Da}). \emph{D}: sawtooth arising in 2D nets for ocular dominance with the Laplacian stencil $\nabla^2_{\times}$. This stencil does not penalise the sawtooth frequency diagonally, but it does horizontally and vertically (see fig.~\ref{f:stencil-families2Db}), so a grid pattern as in \emph{C} is not possible, but a checkerboard pattern is. Compare with the maps resulting from the forward-difference family in fig.~\ref{f:simul2D:OD}.}
    \label{f:sawtooth}
  \end{center}
\end{figure}

\begin{figure}
  \begin{center}
    \begin{tabular}{@{}c@{\hspace{.1\textwidth}}c@{}}
      \includegraphics[width=.3\textwidth]{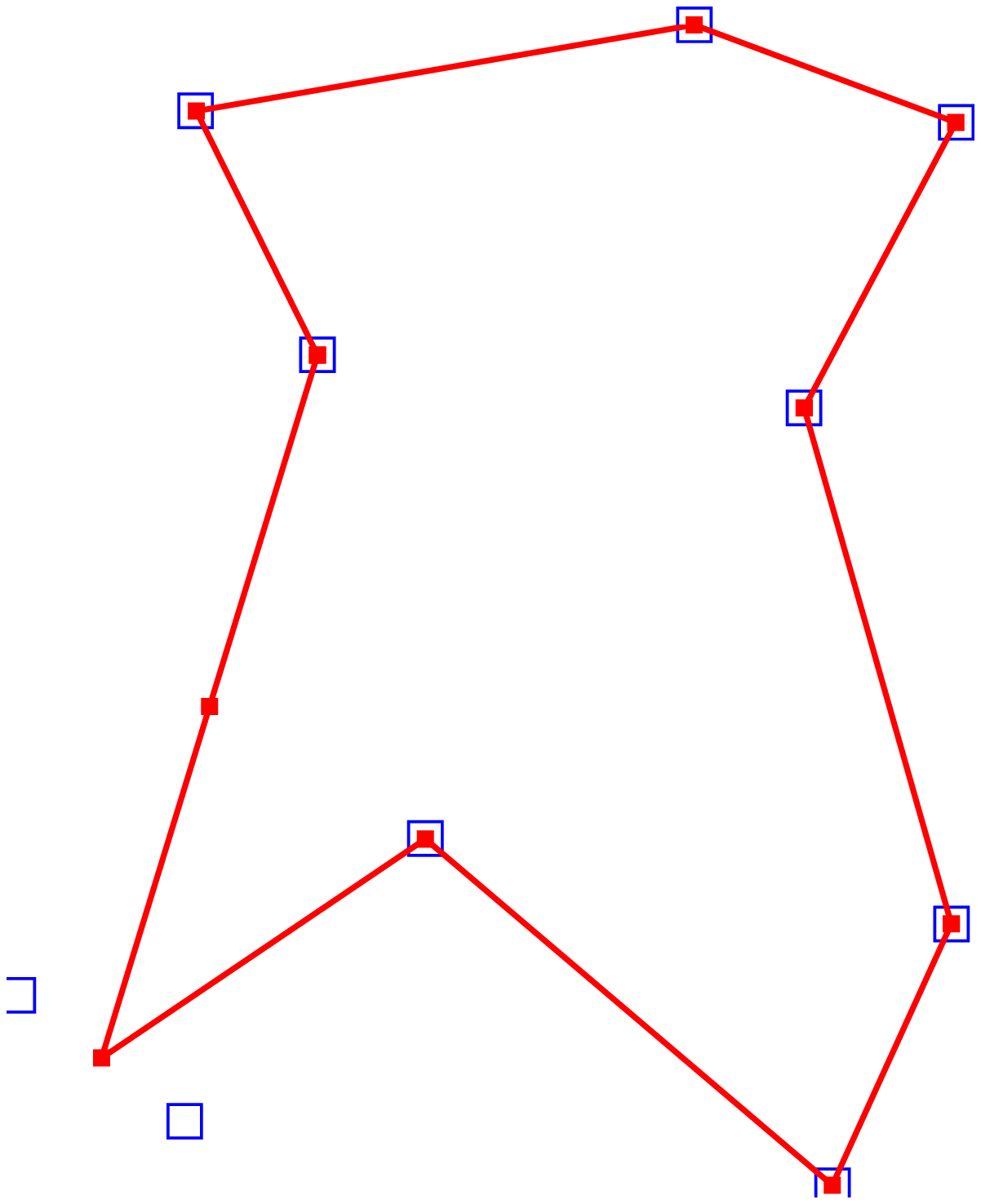} &
      \includegraphics[width=.3\textwidth]{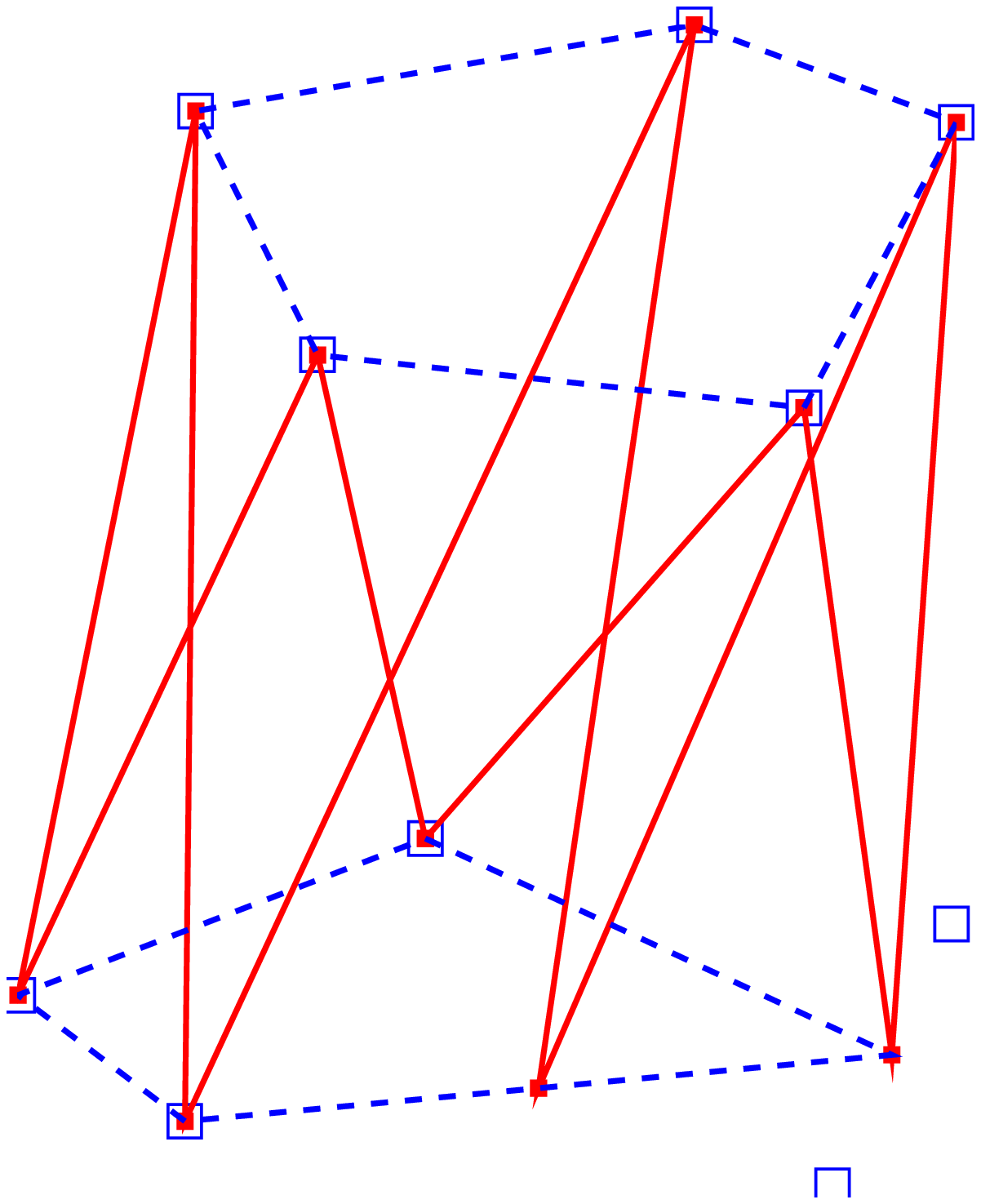}
    \end{tabular}
    \caption{A possible application of sawtooth stencils. \emph{Left}: TSP with original elastic net, stencil $\varsigma = (0,\ -1,\ 1)$ (forward difference). \emph{Right}: the same city set but using a sawtooth stencil $\varsigma = \frac{1}{2}(-1,\ 0,\ 1)$ (central difference). The tours defined by the even and odd centroids (represented by the dotted lines) can be used to obtain a good solution to a multiple TSP. In both cases the b.c.\ are periodic and the figures correspond to a small but nonzero value of the annealing parameter $\sigma$. Compare with fig.~\ref{f:sawtooth}B, which shows (for a different city set) the central-difference stencil at a larger $\sigma$ (earlier development stage).}
    \label{f:sawtooth-TSP}
  \end{center}
\end{figure}

\subsection{The inverse problem and ``evenised'' stencils}
\label{s:circ2sten}

In section~\ref{s:D:inv} we saw that many different matrices \D\ could result in the same $\SS = \D^T \D$, thus having the same power spectrum and the same tension term value for any net. However, many of these do not correspond to a stencil, i.e., their rows are not cyclic permutations of each other, or their elements are complex numbers. In this section we (1) show that two nontrivially different stencils (i.e., discounting shifts and sign reversals) can result in the same \SS\ and (2) show how to obtain a stencil that produces a given positive (semi)definite circulant matrix \SS\ (the inverse problem). As a corollary of (1) for the forward-difference family we rewrite the sum-of-square-distances term of the original elastic net as a Mexican-hat stencil and show that higher-order stencils add more oscillations to this Mexican hat.

We start with the inverse problem, in which we want to obtain a stencil that results in a given power spectrum $\{p_m\}^{M-1}_{m=0}$. In general, there are infinitely many (complex) stencils that would have that power because the discrete Fourier transform at every frequency must only agree in modulus and can have any phase (once we fix the DFT, the stencil is of course unique). That is, taking the eigenvalues $\lambda_m = \sqrt{p_m} e^{i \theta_m}$ for any phase $\theta_m \in [0,2\pi)$ we satisfy $\abs{\lambda_m}^2 = p_m$, but then the corresponding stencil is complex for most choices of the phases. However, in the particular case where the power spectrum is even, i.e., $p_m = p_{M-m}$ for $m = 1,\dots,M-1$, there are several real solutions, thanks to the fact that both the DFT and the inverse DFT of a real, even collection of numbers is also real and even (as can easily be seen). If we take the eigenvalues of \D\ as $\lambda_m = \pm \sqrt{p_m}$ (any sign is valid as long as $\lambda_m = \lambda_{M-m}$) then the first row of \D\ (i.e., the stencil) is the DFT of $\{\lambda_m\}^{M-1}_{m=0}$ divided by $M$, $d_m = \frac{1}{M} \sum^{M-1}_{n=0}{\lambda_n e^{-i 2 \pi \frac{m}{M} n}}$, which is real and even. Note that if $\{\lambda_m\}^{M-1}_{m=0}$ are real but not even then $\{d_m\}^{M-1}_{m=0}$ may be complex.

The inverse problem can be stated equivalently in terms of circulant matrices as follows: given a circulant, symmetric positive (semi)definite matrix \SS\ whose eigenvalues are $\{p_m\}^{M-1}_{m=0}$, decompose $\SS = \D^T\D$ with \D\ real circulant. To solve the problem, we first prove that the eigenvalues of \SS\ are a real, even collection, and then obtain \D\ from them.
\begin{prop}
  \label{p:sym-circ}
  Let \SS\ be a real circulant matrix with eigenvalues $\{\nu_0,\dots,\nu_{M-1}\}$ associated with the Fourier eigenvectors $\f_0,\dots,\f_{M-1}$. Then $\nu_m = \nu_{M-m}$, $m = 1,\dots,M-1 \Leftrightarrow \SS$ is symmetric.
\end{prop}
\begin{rmk}
  Likewise, if \SS\ is real circulant symmetric, then its first row satisfies $s_m = s_{M-m}$, $m = 1,\dots,M-1$:
  \begin{equation*}
    s_m \bydef s_{0m} = \frac{1}{M} \sum^{M-1}_{k=0}{\nu_k \Re\left(e^{-i 2 \pi \frac{m}{M} k}\right)} = \frac{1}{M} \sum^{M-1}_{k=0}{\nu_k \Re\left(e^{-i 2 \pi \left(\frac{M-m}{M}\right) k}\right)} = s_{M-m}.
  \end{equation*}
\end{rmk}
\begin{rmk}
  Any real circulant matrix $\SS = \frac{1}{M} \F \N \F^{*}$ (symmetric or not) can be decomposed as $\SS = \A^2$ with $\A = \frac{1}{M} \F \N^{\frac{1}{2}} \F^{*}$ (where $\nu^{1/2}_n$ can have an arbitrary sign), but \A\ need not be real
. Proposition~\ref{p:sym-circ} proves that a \emph{symmetric} real circulant matrix \SS\ can be decomposed as $\SS = \A^T\A$ with symmetric real circulant $\A = \A^T = \frac{1}{M} \F \N^{\frac{1}{2}} \F^{*}$ (where the roots $\sqrt{\nu_m}$ must be chosen with a sign that preserves $\sqrt{\nu_m} = \sqrt{\nu_{M-m}}$). Thus, we can associate one (or more) real stencils with the first row of \SS, and then \SS\ is a penalty matrix for that stencil, i.e., $\y^T \SS \y = \norm{\A \y}^2 = \norm{\overleftarrow{\varsigma} \ast \y}^2$.
\end{rmk}
\begin{rmk}
  One cannot use the spectral decomposition of \SS\ as $\SS = \U \N \U^T$ with \U\ real orthogonal (made up of the sine and cosine eigenvectors) and \N\ real diagonal because then $\N^{\frac{1}{2}} \U^T$ is not circulant in general.
\end{rmk}
In summary, we have proven that, given a power spectrum $\{p_m\}^{M-1}_{m=0}$ verifying $p_m = p_{M-m}$ for $m = 1,\dots,M-1$ (or a real symmetric circulant positive (semi)definite penalty matrix \SS), we can always obtain an even, real stencil $\varsigma$ for it (i.e., verifying $\varsigma_m = \varsigma_{M-m}$). This is done by simply taking the eigenvalues of the \D\ matrix as $\lambda_m = +\sqrt{p_m}$ and obtaining the stencil as the inverse DFT of them. This also gives a method to \emph{evenise} any stencil that does not verify $\varsigma_m = \varsigma_{M-m}$ (such as $\varsigma = (0,\ -1,\ 1)$). We apply it to the forward- and central-difference families in the next section.

Note that the evenisation procedure also implies that we can write any \SS\ matrix that results from a stencil as $\SS = \D^2$ (since when using evenised stencils, \D\ is symmetric). Thus, the $p$th\nobreakdash-order matrix is simply $\leftexp{(p)}{\SS} = \D^{2p} = \frac{1}{M} \F \bLambda^{2p} \F^*$ if $\D = \frac{1}{M} \F \bLambda \F^*$.

\subsubsection{Forward-difference family}

From section~\ref{s:fwddiff-family} we have that the eigenvalues of $\leftexp{(p)}{\SS}$ are $\leftexp{(p)}{\nu_m} = \left( 2 \sin{\left( \pi \frac{m}{M} \right)} \right)^{2p}$, $m = 0,\dots,M-1$. Since $\sin{\left( \pi \frac{m}{M} \right)} = \sin{\left( \pi \left(\frac{M-m}{M}\right) \right)} > 0$ for $m = 1,\dots,M-1$, we can take $\sqrt{\nu_m} = \abs{2 \sin{\left( \pi \frac{m}{M} \right)}}^p = \left( 2 \sin{\left( \pi \frac{m}{M} \right)} \right)^p$. Thus the even, real stencil is
\begin{equation}
  \label{e:evenised-fwddiff}
  \leftexp{(p)}{d_m} = \frac{1}{M} \sum^{M-1}_{k=0}{\leftexp{(p)}{\nu}^{\frac{1}{2}}_k \cos{\left( 2 \pi \frac{m}{M} k \right)}} = \frac{2^p}{M} \sum^{M-1}_{k=0}{\sin^p{\left( \pi \frac{k}{M} \right)} \cos{\left( 2 \pi \frac{m}{M} k \right)}} \qquad m = 0,\dots,M-1
\end{equation}
verifying $\leftexp{(p)}{d_m} = \leftexp{(p)}{d_{M-m}}$ for $m = 1,\dots,M-1$.

We do not have a formula for this sum in general, but the following proposition characterises the stencils asymptotically on the number of centroids $M$ (the limit stencils are good approximations for $M \gtrsim 10$ already).
\begin{prop}
  \label{p:evenised-fwddiff-inf}
  For the stencil defined by~\eqref{e:evenised-fwddiff}, we have
  \begin{equation*}
    \leftexp{(p)}{d^{\infty}_m} \bydef \lim_{M \rightarrow \infty}{\leftexp{(p)}{d_m}} =
    \begin{cases}
      p \text{ even:} &
      \begin{cases}
        (-1)^m \, \binom{p}{\frac{p}{2} - m}, & m \le \frac{p}{2} \\
        0, & m > \frac{p}{2}
      \end{cases} \\[3ex]
      p \text{ odd:} &
      \displaystyle
      \frac{(-1)^{\frac{p+1}{2}} \, 2^{\frac{3p+1}{2}} \left( \frac{p-1}{2} \right)! \, p!!}{\pi \, \prod^{\frac{p-1}{2}}_{k=0}{( 4m^2 - (2k+1)^2)}},\quad m = 0,1,2,\dots,\infty
    \end{cases}
  \end{equation*}
  where we define $(2n)!! = 2 \cdot 4 \cdot 6 \cdots (2n)$, $(2n+1)!! = 1 \cdot 3 \cdot 5 \cdots (2n+1)$, $0!! = (-1)!! = 1$.
\end{prop}
We also give the exact expression for the case $p = 1$. It is easy to confirm that~\eqref{e:evenised-fwddiff1} tends to the corresponding expression in proposition~\ref{p:evenised-fwddiff-inf} for $M \rightarrow \infty$.
\begin{prop}
  \label{p:evenised-fwddiff}
  For the stencil defined by~\eqref{e:evenised-fwddiff}, we have for the case $p = 1$:
\begin{equation}
  \label{e:evenised-fwddiff1}
  \leftexp{(1)}{d_m} = \frac{2}{M} \sum^{M-1}_{k=0}{\sin{\left( \pi \frac{k}{M} \right)} \cos{\left( 2 \pi \frac{m}{M} k \right)}} = \frac{2}{M} \frac{\sin{\frac{\pi}{M}}}{\cos{\frac{2 \pi m}{M}} - \cos{\frac{\pi}{M}}} \qquad m = 0,\dots,M-1.
\end{equation}
\end{prop}
Consequently, for even $p$ the limit stencil coincides with the original forward-difference stencil (with an immaterial sign change) and is therefore sparse (or compact-support) since only a finite number of elements is nonzero. In fact, the stencils must coincide not only in the limit but for finite $M$, since they are already even, but we do not have an expression for $\sum^{M-1}_{k=0}{\sin^{2p}{\left( \pi \frac{k}{M} \right)} \cos{\left( 2 \pi \frac{m}{M} k \right)}}$ to prove it. For odd $p$ the evenised stencils are not sparse, but decay as $\calO(m^{p+1})$ with a finite number of sign changes, or oscillations, for small $m$, as given by the following proposition.
\begin{prop}
  \label{p:evenised-fwddiff-signs}
  For $p = 2n+1$ odd, $\leftexp{(2n+1)}{d^{\infty}_m}$ has $n+1$ sign changes in $m = 0,1,2,\dots,n$ and constant sign for $m > n$, the sign being positive for $n$ odd and negative for $n$ even. That is:
  \begin{equation*}
    \sgn{\leftexp{(2n+1)}{d^{\infty}_m}} =
    \begin{cases}
      (-1)^m, & m \le n \\
      (-1)^{n+1}, & m > n.
    \end{cases}
  \end{equation*}
\end{prop}
The sign alternation has then the pattern shown in table~\ref{t:evenised-fwddiff-inf-odd}. In particular, $\leftexp{(1)}{d^{\infty}_m}$ (the original elastic net) is positive only at $m = 0$, the central element, and is thus a discrete Mexican hat (or centre-surround kernel), with higher-order stencils being oscillatory versions of Mexican hats (i.e., with more rings). Stencils with $p = 1,5,\dots,4n+1$ for $n = 0,1,2,\dots$ have long-range inhibitory connections while stencils with $p = 3,7,\dots,4n+3$ have long-range excitatory connections. Fig.~\ref{f:evenised-fwddiff} plots the evenised forward-difference stencils for the first few orders $p$ (for $M = 50$); note how for even $p$ the stencils are sparse and coincide with those of fig.~\ref{f:stencil-families}.

From proposition~\ref{p:evenised-fwddiff-inf} for even $p = 2n$ we have that
\begin{equation*}
  \sgn{\leftexp{(2n)}{d^{\infty}_m}} = 
  \begin{cases}
    (-1)^m, & m \le n \\
    0, & m > n
  \end{cases}
\end{equation*}
just as for odd $p$. The sign alternation has the pattern shown in table~\ref{t:evenised-fwddiff-inf-even}. The stencils are, as before, oscillatory Mexican hats but with compact support, i.e., strictly short-range connections. The case $p = 2$ is similar to $p = 1$ in having only the central element as excitatory connection.

The analysis of the stencils in 2D becomes more complicated, with the \SS\ matrix being block-circulant, though the results should carry over in an analogous way.

\begin{table}
  \begin{center}
    \begin{tabular}{c|c|ccccccccc|l}
$p = 2n+1$ & $n$ & sign for $m =$ & $0$ & $1$ & $2$ & $3$ & $4$ & $5$ & $6$ & $\dots$ & \multicolumn{1}{c}{$\pi 2^{-p} \leftexp{(p)}{d^{\infty}_m}$} \\
\hline
$1$ & $0$ & & $+$ & $-$ & $-$ & $-$ & $-$ & $-$ & $-$ & & $-2 / (4m^2-1)$ \\
$3$ & $1$ & & $+$ & $-$ & $+$ & $+$ & $+$ & $+$ & $+$ & & $12 / (4m^2-1) (4m^2-9)$ \\
$5$ & $2$ & & $+$ & $-$ & $+$ & $-$ & $-$ & $-$ & $-$ & & $-240 / (4m^2-1) (4m^2-9) (4m^2-25)$ \\
$7$ & $3$ & & $+$ & $-$ & $+$ & $-$ & $+$ & $+$ & $+$ & & etc.
    \end{tabular}
    \caption{Pattern of sign alternation of the 1D ``evenised'' stencils of the forward-difference family, for odd order $p$.}
    \label{t:evenised-fwddiff-inf-odd}
  \end{center}
\end{table}

\begin{table}
  \begin{center}
    \begin{tabular}{c|c|ccccccccc|l}
$p = 2n$ & $n$ & sign for $m =$ & $0$ & $1$ & $2$ & $3$ & $4$ & $5$ & $6$ & $\dots$ & \multicolumn{1}{c}{$\leftexp{(p)}{d^{\infty}_m}$} \\
\hline
$2$ & $1$ & & $+$ & $-$ & $0$ & $0$ & $0$ & $0$ & $0$ & & $(2,\ -1,\ 0,\ \dots$ \\
$4$ & $2$ & & $+$ & $-$ & $+$ & $0$ & $0$ & $0$ & $0$ & & $(6,\ -4,\ 1,\ 0,\ \dots$ \\
$6$ & $3$ & & $+$ & $-$ & $+$ & $-$ & $0$ & $0$ & $0$ & & $(20,\ -15,\ 6,\ -1,\ 0,\ \dots$ \\
$8$ & $4$ & & $+$ & $-$ & $+$ & $-$ & $+$ & $0$ & $0$ & & $(70,\ -56,\ 28,\ -8,\ 1,\ 0,\ \dots$
    \end{tabular}
    \caption{Pattern of sign alternation of the 1D ``evenised'' stencils of the forward-difference family, for even order $p$.}
    \label{t:evenised-fwddiff-inf-even}
  \end{center}
\end{table}

\begin{figure}
  \begin{center}
    \psfrag{1}{$p = 1$}
    \psfrag{2}{$p = 2$}
    \psfrag{3}{$p = 3$}
    \psfrag{4}{$p = 4$}
    \begin{tabular}{@{}c@{\hspace{1.5cm}}c@{}}
      Forward-difference family & Central-difference family \\
      \includegraphics[height=\MACPlengthA]{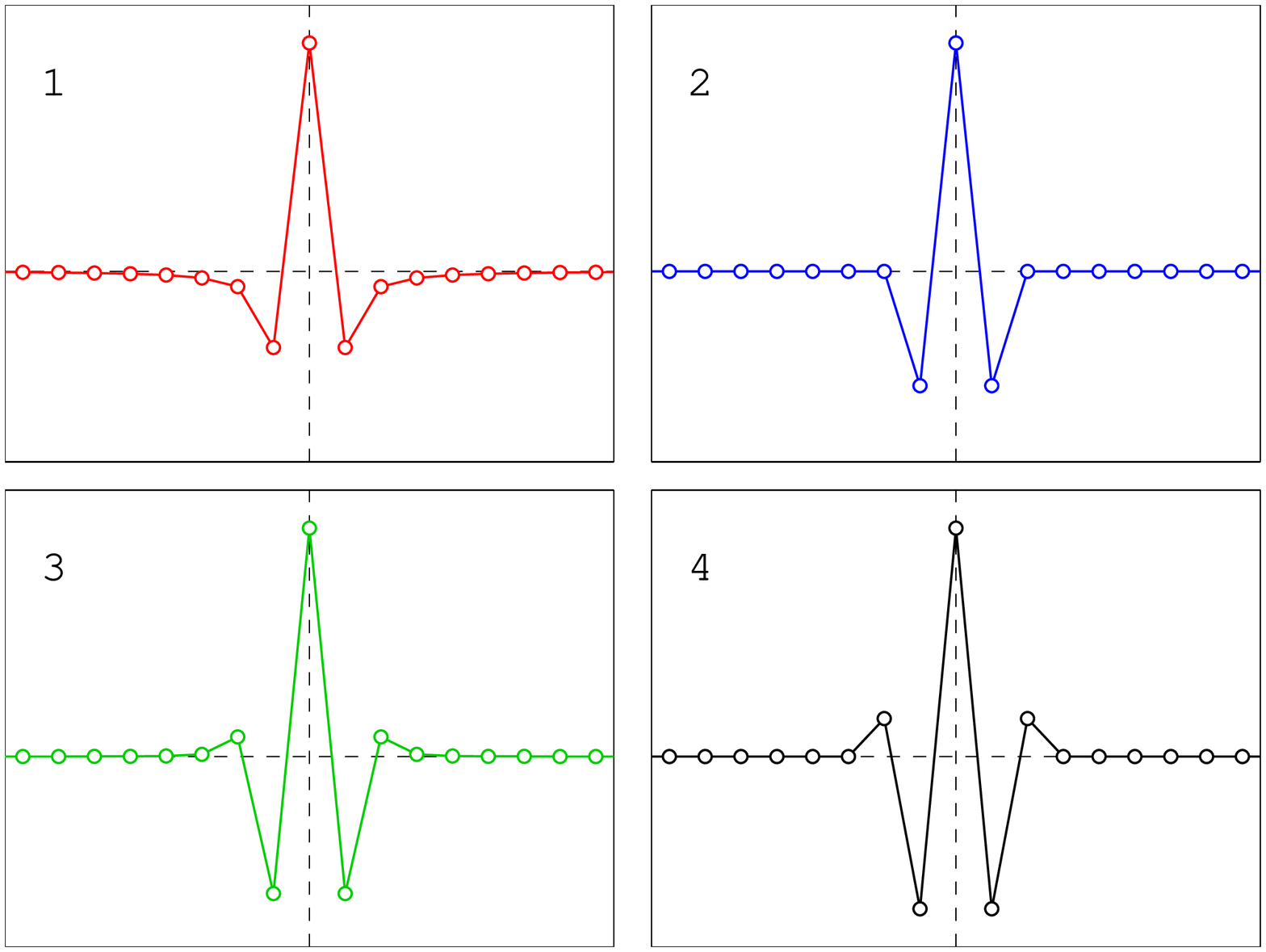} &
      \includegraphics[height=\MACPlengthA]{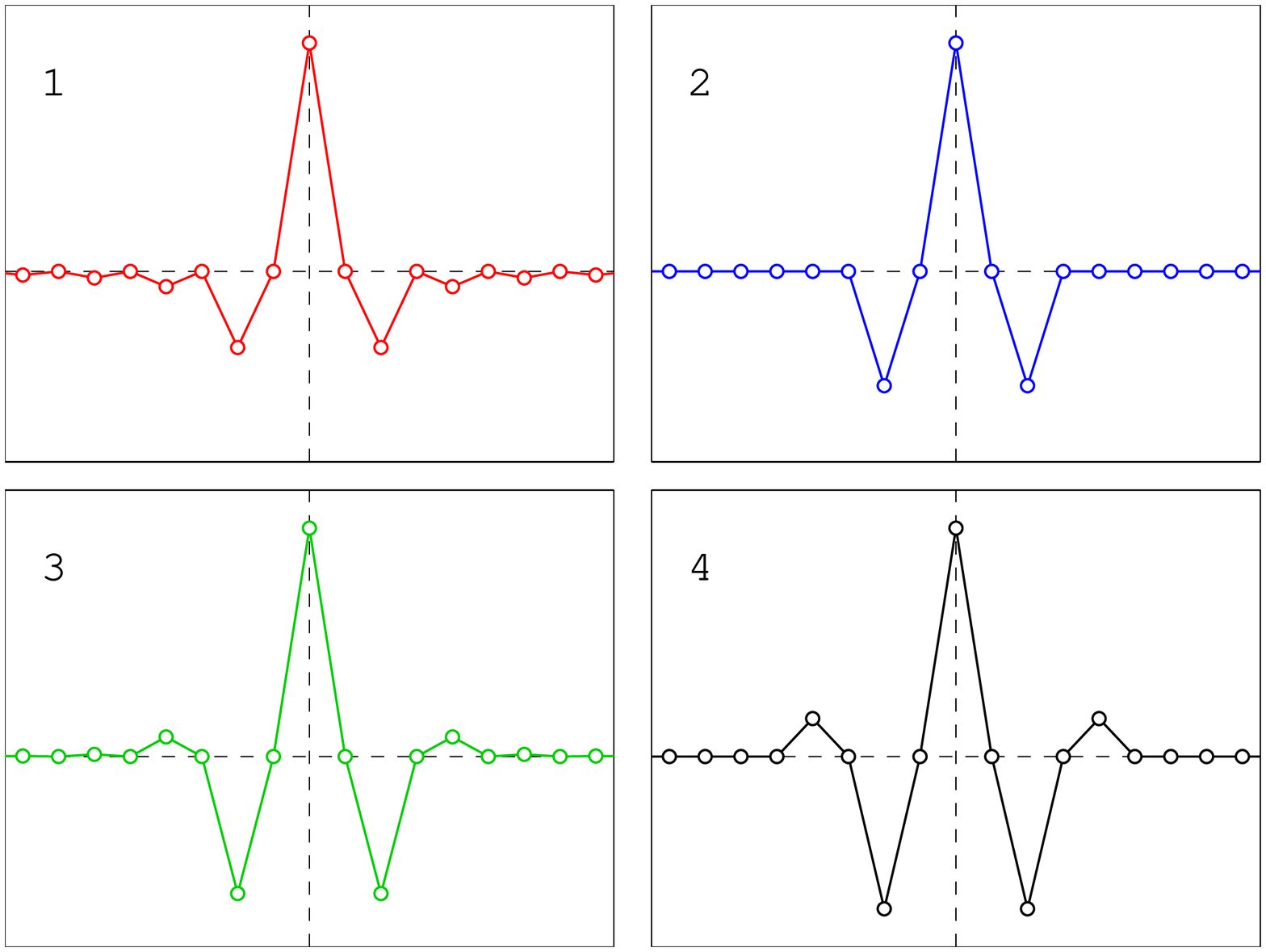}
    \end{tabular}
    \caption{1D ``evenised'' stencils of the forward-  and central-difference families, normalised in maximal amplitude. Compare with the net-domain plots of fig.~\ref{f:stencil-families1Da} (the power spectra are the same by definition). For even $p$ they coincide with the original ones in fig.~\ref{f:stencil-families1Da}, except for a change of sign.}
    \label{f:evenised-fwddiff}
  \end{center}
\end{figure}

\subsubsection{Central-difference family}

From section~\ref{s:cendiff-family} we have that the eigenvalues of $\leftexp{(p)}{\SS}$ are $\leftexp{(p)}{\nu_m} = \sin^{2p}{\left( 2 \pi \frac{m}{M} \right)}$, $m = 0,\dots,M-1$. Thus the even, real stencil is
\begin{equation}
  \label{e:evenised-cendiff}
  \leftexp{(p)}{d_m} = \frac{1}{M} \sum^{M-1}_{k=0}{\abs{ \sin^p{\left( 2 \pi \frac{k}{M} \right)} } \cos{\left( 2 \pi \frac{m}{M} k \right)}} \qquad m = 0,\dots,M-1
\end{equation}
and the following proposition characterises it asymptotically.
\begin{prop}
  \label{p:evenised-cendiff-inf}
  For the stencil defined by~\eqref{e:evenised-cendiff}, we have
  \begin{equation*}
    \leftexp{(p)}{d^{\infty}_m} \bydef \lim_{M \rightarrow \infty}{\leftexp{(p)}{d_m}} =
    \begin{cases}
      p \text{ even:} &
      \begin{cases}
        (-1)^{\frac{m}{2}} \, 2^{-p} \, \binom{p}{\frac{p}{2} - \frac{m}{2}}, & m \text{ even and } m \le p \\
        0, & m \text{ odd or } m > p
      \end{cases} \\[3ex]
      p \text{ odd:} &
      \begin{cases}
        \displaystyle
        \frac{(-1)^{\frac{p+1}{2}} \, 2^{\frac{p+1}{2}} \left( \frac{p-1}{2} \right)! \, p!!}{\pi \, \prod^{\frac{p-1}{2}}_{k=0}{( m^2 - (2k+1)^2)}}, & m \text{ even } = 0,2,4,\dots,\infty \\
        0, & m \text{ odd}.
      \end{cases}
    \end{cases}
  \end{equation*}
\end{prop}
Again, for even $p$ the limit stencil coincides with the original central-difference stencil (with an immaterial sign change) and is therefore sparse, while for odd $p$ it is not sparse, decaying as $\calO(m^{p+1})$ with a finite number of oscillations. The pattern of sign alternation is like that of the forward-difference one but over the even indices $m$ (since the odd ones are zero). It is possible to obtain the stencil exactly for $p = 1$ as in eq.~\eqref{e:evenised-fwddiff1} but we shall not dwell in precise statements.

\subsection{Linear combinations of stencils}
\label{s:lc}

Another way to create new stencils is to combine linearly stencils of the same or different order. We have two options: either to combine in the net domain or in the power domain. We consider both cases for stencils of the same order; for stencils of different order, we also need to include the step size $h$ (since it is powered to the order).

\subsubsection{Net domain}

Consider without loss of generality a linear combination (l.c.) stencil $\varsigma \bydef a \varrho + b \varpi$ where both $\varrho$ and $\varpi$ are of order $p$. From the truncation error expression~\eqref{e:D:trunc}, in order for $\varsigma$ to be a differential stencil of order $p$ we must have $a + b = 1 = \alpha_p$ (the $\alpha_p$ coefficient for $\varsigma$). If both $a$ and $b$ were nonnegative we could further restrict them to be in $[0,1]$ and so we would have a convex l.c.\ $\varsigma = \gamma \varrho + (1 - \gamma) \varpi$, $\gamma \in [0,1]$; but we accept negative values for either $a$ or $b$. Also note that if $a + b = 0$ then $\varsigma$ would be of order higher than $p$.

Some examples using first-order stencils:
\begin{center}
  \renewcommand{\arraystretch}{1.2}
  \begin{tabular}{ll}
    $2 (-\frac{1}{2},\ 0,\ \frac{1}{2}) - (0,\ -1,\ 1) = (-1,\ 1,\ 0)$ & ($2 \times \text{central} - 1 \times \text{forward} = \text{backward}$) \\
    $2 \left( (0,\ -1,\ 1) - (-\frac{1}{2},\ 0,\ \frac{1}{2}) \right) = (1,\ -2,\ 1)$ & ($2 \times \left( \text{forward} - \text{central} \right) = \text{forward, with}$ $p = 2$, since $a + b = 0$) \\
    $\gamma (0,\ -1,\ 1) + (1 - \gamma) (-\frac{1}{2},\ 0,\ \frac{1}{2})$ & (smooth transition between forward and central)
  \end{tabular}
\end{center}
We see that the l.c.\ with a sawtooth stencil need no result in a sawtooth stencil.

By the linearity of the convolution and the Fourier transform, combining in the net domain is the same as combining in the Fourier domain and so we have:
\begin{equation*}
  \D_{\varsigma} = a \D_{\varrho} + b \D_{\varpi} \qquad \lambda_{\varsigma,m} = a \lambda_{\varrho,m} + b \lambda_{\varpi,m} \qquad \SS_{\varsigma} = \D^T_{\varsigma} \D_{\varsigma} \qquad \nu_{\varsigma,m} = \abs{a \lambda_{\varrho,m} + b \lambda_{\varpi,m}}^2.
\end{equation*}
This is very different from combining the power spectra, as we show below. Note that a l.c.\ in the net domain can be designed from scratch, since it is just a particular way of zeroing $\alpha_r$'s coefficients.

Studying l.c.'s of stencils provides another way to characterise the space of stencils. For example, the differential stencils having at most $3$ consecutive non-zero coefficients $(a,\ b,\ c)$ form a 2D vector subspace (since $a+b+c = 0$) spanned by the forward- and central-difference stencils. We will not pursue this approach further here.

Note that using a l.c.\ of stencils of different order, such as $\sum^P_{p=1}{a_p \, \leftexp{(p)}{\varsigma}}$, is the discrete correlate of penalising functions $f$ that do not satisfy the differential equation $\sum^P_{p=1}{a_p \frac{d^p f}{dx^p}} = 0$.

\subsubsection{Power domain}
\label{s:lc:power}

The l.c.\ should now be convex in general because the power cannot be negative. Then, for $\abs{\hat{\varsigma}_{k}}^2 = \gamma \abs{\hat{\varrho}_{k}}^2 + (1 - \gamma) \abs{\hat{\varpi}_{k}}^2$ for each $k = 0,\dots,M-1$ we have $\nu_{\varsigma,k} = \gamma \nu_{\varrho,k} + (1 - \gamma) \nu_{\varpi,k}$ and we can express the matrix $\D_{\varsigma}$ as resulting from two matrices, $\D_{\varsigma} = \left( \begin{smallmatrix} \sqrt{\gamma} \, \D_{\varrho} \\ \sqrt{1 - \gamma} \, \D_{\varpi}\end{smallmatrix} \right)$, so that $\SS_{\varsigma} = \D^T_{\varsigma} \D_{\varsigma} = \gamma \SS_{\varrho} + (1 - \gamma) \SS_{\varpi}$. Basically, we are adding more rows to \D, which end up being further squared l.c.'s to add to the tension term.

This is exactly what we do to create a 2D stencil out of a 1D one: we first pass the 1D stencil horizontally to produce a matrix $\D_{\text{h}}$; and then pass it vertically to produce $\D_{\text{v}}$. The resulting 2D stencil has a matrix $\SS = \D^T_{\text{h}} \D_{\text{h}} + \D^T_{\text{v}} \D_{\text{v}}$ (see figure~\ref{f:stencil-families2Da} for the power spectrum and figures~\ref{f:simul2D:OD}--\ref{f:simul2D:ODOR} for simulations from cortical maps resulting from them). For example, for the first-order forward-difference stencil, we obtain terms of the form $\norm{\y_{m+1,n} - \y_{mn}}^2 + \norm{\y_{m,n+1} - \y_{mn}}^2$, whose derivative in the gradient descent algorithm is of the form $4 \y_{mn} - (\y_{m+1,n} + \y_{m-1,n} + \y_{m,n+1} + \y_{m,n-1})$ and has been used in the past in cortical map models. Note that the term $\norm{\y_{m+1,n} - \y_{mn}}^2 + \norm{\y_{m,n+1} - \y_{mn}}^2$ cannot be expressed as the square of a single l.c.\ $\norm{A \y_{mn} + B \y_{m+1,n} + C \y_{m,n+1}}^2$ with real coefficients $A$, $B$, $C$ (though it can with complex coefficients, e.g.\ $A = -\sqrt{2}$, $B = (1+i)/\sqrt{2} = C^*$). Although we could find a (nonsparse) evenised 2D stencil for \SS, in practice it is simpler to use $\D_{2M \times M} = \left( \begin{smallmatrix} \sqrt{\gamma} \D_{\text{h}} \\ \sqrt{1 - \gamma} \D_{\text{v}}\end{smallmatrix} \right)$.

For 2D stencils it is also possible to design the stencil directly in the net domain by forcing the relevant terms of the 2D Taylor expansion to vanish, e.g.\ as for the Laplacian stencils $\nabla^2_{+}$ or $\nabla^2_{\times}$. This also allows one to introduce other constraints, such as isotropy in some sense---deriving a 2D stencil from horizontal and vertical passes of a 1D one may result in nets that heavily emphasise those directions. Compare the contour lines around $\bk = (0,0)$ of the power spectrum of $(1,\ -2,\ 1)$ in fig.~\ref{f:stencil-families2Da} (which align with the coordinate axes) and those of any of $\nabla^2_{+}$, $\nabla^2_{9}$ or $\nabla^2_{\text{a}}$ in fig.~\ref{f:stencil-families2Db} (which are roughly circular). However, the horizontal-vertical method results in simple b.c.\ (basically the same as in 1D), while stencils such as $\nabla^2_{+}$ require a more careful definition of b.c. For example, for $\nabla^2_{+}$ one cannot simply zero all rows of \D\ for which the stencil overspills, because the corner points of the net would end up in no l.c.\ (their associated columns in \D\ would be zeroes) and thus be unconstrained.

In the power domain, the l.c.\ $\SS_{\varsigma} = \gamma \SS_{\varrho} + (1 - \gamma) \SS_{\varpi}$ of a sawtooth stencil $\varrho$ with a nonsawtooth one $\varpi$ smoothly transitions as $\gamma$ takes values from $0$ to $1$. An interesting question we leave for future research is for what $\gamma$ the l.c.\ starts showing sawtooth behaviour.

\subsection{Extensions}
\label{s:ext}

We have given a general methodology to characterise the spectral properties of families of stencils with periodic b.c.\ and studied in detail two important families. The space of stencils is, however, much larger. Here we discuss additional ways of defining families and nonperiodic b.c.

\subsubsection{The nonperiodic case}

Nonperiodic b.c.\ can be defined in different ways, perhaps using ideas from partial differential equations (such as Dirichlet or Neumann b.c.; \citealp{Press_92a,GodunovRyaben87a}). They are particularly necessary when modelling domains that are not rectangular or have holes, but also when modelling adjacent domains with different properties (e.g.\ to investigate the lateral connections at the border between V1 and V2). In the simulations of section~\ref{s:cmap} we use a simple type of nonperiodic b.c.\ by analogy with the original elastic net: if (a nonzero coefficient of) a stencil falls out of the net when applied at a given centroid, we discard it by making zero the corresponding row of \D. If the stencil size is small compared to the net dimensions, only a few rows are zeroed.

With nonperiodic b.c.\ the \D\ and \SS\ matrices are not circulant anymore, but almost Toeplitz (see e.g.\ eq.~\eqref{e:origEN:D}). The analysis becomes harder because in general there are no formulas for the eigenvalues (though there are asymptotic results for the distribution of the eigenvalues, such as the Szeg{\H{o}} formula; \citealp{GrenanSzego58a}). The eigenvectors are not plane waves anymore; in fact, for matrices such as those of eq.~\eqref{e:origEN:D} (for $a = 0$) the nullspace of \D\ and \SS\ becomes larger and can contain discrete polynomials of degree less than $p$, as in the continuous case. However, as $M \rightarrow \infty$, we woould expect our analysis to hold for nonperiodic b.c.

Many of our results can be considered as an extension of spectral graph theory \citep{Chung97a} to higher-order derivatives in regular, periodic grids. However, one could think more generally of defining the tension term not on a regular grid in dimension $D$, but on arbitrary graphs. Higher-order derivatives defined on graphs may have applications in related areas, such as spectral clustering, of recent interest (e.g.~\citealp{ShiMalik00a}).

Also note that, even for the forward-difference family, the discrete case is not exactly equivalent to the continuous one. For example, while the second-derivative annihilates any linear function, in the discrete case the centroids must not only be collinear but also equispaced for the net to have zero penalty.

\subsubsection{Stencil design in power space}
\label{s:ext:design-power}

So far we have designed stencils in the net space, either from scratch (truncation error equations) or by repeated application of known stencils (families). Instead, we can design the stencil in the Fourier domain by specifying desirable power spectra (satisfying $p_0 = 0$ for the stencil to be differential and $p_m = p_{M-m}$ for it to be real) and then invert it to obtain an even stencil in the net domain. In particular, we can use intermediate power curves that would correspond to noninteger derivative orders $p$ (fig.~\ref{f:stencil-families1Da}). A desirable requirement would probably be that the power is a nondecreasing curve, since otherwise some high frequencies would be less penalised than the lower ones. One problem with this approach is that inverting a power spectrum to obtain an stencil results in general in nonsparse stencils (with many, or typically $M$, nonzero coefficients). This means that the matrix \SS\ is not sparse anymore and the computational cost of the learning algorithm becomes very large in both memory and time. One possibility would be to find the stencil that most closely approximates a given power spectrum subject to having a few nonzero coefficients (for a related problem see \citealp{ChuPlemmon03a}).

Note that these noninteger-order stencils cannot be obtained as linear combinations of stencils in the net space; the power spectra must be averaged instead. In other words, as we saw in section~\ref{s:lc}, the linear superposition of stencils in the power space is very different from that in the net space (analogously to the interference of incoherent and coherent light, respectively). For example, the average of the forward and backward differences in net space is the central difference, while in the power space it is the forward difference.

This also suggests a way of defining noninteger derivative orders $p$ in the continuous case by taking the inverse Fourier transform of $(i 2 \pi k)^p$ for $p \in \bbR^+$.

\subsubsection{Stencils with a scale parameter}

The differential stencils we have used are strictly local in that they do not depend on a scale parameter. For example, the forward difference involves only two neighbouring points, irrespectively of the net size $M$. One way (inspired by scale-space theory; \citealp{Lindeb94a}) of introducing a scale parameter in the stencil, and thus define a different class of stencils, is as follows. Consider a sequence of continuous functions indexed by a scale parameter $s \in \bbR^+$ that converges to the delta function, such as Gaussians $g_s(\uu) = \exp{(-\norm{\uu}^2/2 s)}$. Now, instead of taking the derivative of a function $f$, we take the derivative of $f$ convolved (smoothed) with $g_s$ (or, since the derivative and convolution operators commute, smooth the derivative of $f$). We have:
\begin{equation*}
  g_s \ast \frac{d^p f}{du^p} = \frac{d^p}{du^p}(g_s \ast f) = \frac{d^p g_s}{du^p} \ast f.
\end{equation*}
Thus, we can define a stencil by discretising $g^{(p)}_s = \frac{d^p g_s}{du^p}$. For the Gaussian case, $g^{(p)}_s$ has a similar shape as the forward-difference family and is given by the Hermite polynomial of order $p$ times the Gaussian kernel $g_s$.

\section{Application to cortical map modelling}
\label{s:cmap}

In this section we demonstrate by simulation the behaviour of generalised elastic nets, corroborating the previous analysis (section~\ref{s:fwddiff-family} in particular) and also pointing out characteristics that need theoretical explanation. All the simulations were obtained for the forward-difference family (orders $p = 1$ to $4$) using the Cholesky factorisation method. We consider the problem of cortical maps in primary visual cortex, to which the elastic net was originally applied by \citet{DurbinMitchis90a} and \citet{GoodhilWillsh90a}. In the low-dimensional version of this problem, the training set of ``cities'' are values of stimuli to which primary visual cortex cells respond (we consider the orientation of an edge OR, its position in the visual field VF and its eye of origin OD) and the centroids are the preferred values (receptive fields) of the cells, arranged as a 2D rectangular grid. The stimulus space is densely populated with training points, arranged for computational convenience as a grid%
\footnote{This grid acts as a scaffolding for the stimulus space and is computationally convenient for batch learning. It is also possible to use online training and generate a random data point uniformly in the stimulus space (e.g.\ \citealp{WolfGeisel98a}). This results in very similar maps, though for fixed $\sigma$ the online method takes very long to stabilise. Note that the annealing schedule may need to be adjusted if different values of $\beta$, the training set or the stencil are used.}
whose respective dimensions determine the development of the net. The typical development sequence with annealing of the $\sigma$ parameter from high to low values is the same as for the original elastic net, as follows. At large $\sigma$ the energy $E$ has a single minimum at the centre of mass of the training set. At a value $\sigma_0$ the net topographically expands along the principal components of the training set, chosen to be the VF axes; $\sigma_0$ depends on the training set geometry and the \SS\ matrix, and can be calculated as in \citet{Durbin_89a} as the value of $\sigma$ where the Hessian matrix of $E$ stops being positive definite (the energy surface bifurcates), noting that the Hessian of the tension term is here $\frac{\beta}{2} \I_D \otimes (\SS + \SS^T)$ (in Kronecker product notation). At a critical value $\sigma_{\text{c}} < \sigma_0$ the net expands along the OD and/or OR axis (whichever has the largest training set variance) and develops stripes. Fig.~\ref{f:simul1D-nonperiodic} shows this sequence for stencil orders $1$ and $2$, for a simplified problem where the net is 1D and the stimulus space consists of the 1D VF position and the OD. Note how, unlike the 1st-order net, the 2nd-order net avoids sharp corners, which have locally high curvature, and goes slightly out of the convex hull of the training set. These two characteristics always occur with order $p > 1$ and are most apparent for intermediate values of $\sigma$; when $\sigma \rightarrow 0$, the fitness term overwhelms the tension term and the net develops kinks anyway in an effort to interpolate the training set (high frequencies are less and less penalised as $\sigma$ decreases). Fig.~\ref{f:simul1D-periodic} shows the decrease of the stripe width with the stencil order for $p = 1$ to $4$ (with constant $\beta$), this time using periodic b.c.\ for the net (not for the training set).

In our simulations for $p = 1$ (at least in dimension $D \le 3$, which can be visualised) the net centroids (and therefore the links between them) always remain inside the convex hull of the training set. This can be intuitively understood by noting the following. (1) If a centroid lies outside the convex hull, the fitness term value can be increased by bringing it towards the convex hull (since the Gaussian is isotropic and monotonically decreasing from its centre). (2) The tension term encourages centroids to concentrate. Thus, maxima of $E$ should have no centroids outside the convex hull. However, for higher-order tension terms, step (2) is not true anymore, since stretched nets can have zero penalty and the centroids can lie outside the convex hull (the tendency increasing with the order $p$). For example, in fig.~\ref{f:simul1D-nonperiodic} the rounded ends of the net that exceed the convex hull result in a lower tension than flattening the ends and producing two sharp corners with a heavy cost in the second derivative. In cortical map models one would rather not have centroids outside the convex hull, since they result in stimulus preferences outside their valid range (for VF, OD and OR modulus). If annealing slowly enough and if stopping the training shortly after all maps have arisen, the centroids are inside the convex hull except for a few that are slightly outside; further, one can define the ranges to exceed the training set range.
Another possibility would be to constrain the learning algorithm to keep the net inside the convex hull. However, apart from being a difficult task, this would probably result in rather different nets.

We investigate now the effect of $\beta$ and the stencil order $p$ on the ocular dominance and orientation maps with a 2D net. As is well known from earlier work on dimension reduction models such as the elastic net and self-organising map applied to cortical maps \citep{WolfGeisel98a,Scherf_99a}, the space of $\beta$ and $\sigma$ can be decomposed in regions as in a phase diagram. In each region or phase the resulting cortical maps of OD and OR have qualitatively different properties. Here we have an additional variable%
\footnote{The stencil order $p$ as used here is discrete, but could be considered continuous if interpolating the power spectra as described in section~\ref{s:ext:design-power}.}, 
the stencil order $p$, and we discuss the resulting phase diagram obtained by simulation (assuming the forward-difference family with normalisation to unit power). For a given $p$, the OD and OR maps arise at a given critical value $\sigma_{\text{c}}$ of $\sigma$ (where $\sigma_{\text{c}}$ decreases with $\beta$), and do not arise at all if $\beta$ is too large (though for $p > 2$ we have not obtained unsegregated maps for the range of $\beta$ considered). As shown in fig.~\ref{f:phase} (right), $\sigma_{\text{c}}$ increases with $p$. Figures~\ref{f:simul2D:OD}--\ref{f:simul2D:ODOR} show the maps on a slice of the space $(\beta,p,\sigma)$ for a value of $\sigma$ small enough that the OD and OR maps have arisen. Fig.~\ref{f:phase} (left) shows a schematic phase diagram over $(\beta,p)$. Region $1$ (for small $\beta$) contains salt-and-pepper maps, without continuity among neighbouring centroids; intuitively this reflects the fact that if $\beta$ is very small then the smoothing effect of the tension term disappears. Region $2$ contains unsegregated maps, where the net stretches (to some extent) along the VF variables but not along the OD and OR variables, even as $\sigma$ tends to $0$. Region $3$ contains columnar maps with local geometric characteristics (such as stripe width, pinwheel density, crossing angles or gradient matching) that depend on $\beta$ and $p$. The stripe width increases with $\beta$ and decreases with $p$, as predicted in section~\ref{s:fwddiff-family}. For fixed $\beta$, \citet{CarreirGoodhil03b} have determined the following effects dependent on $p$: (1) the distribution of crossing angles between the contours of the OD and OR maps is biased towards orthogonality for $p = 1$ but quickly flattens to become uniform as $p$ increases; (2) the OR pinwheels tend to lie away from OD borders for $p = 1$ but gradually occupy OD borders (like beads on a string) as $p$ increases; and (3) the contours of OD and OR align preferentially along the diagonals as $p$ increases (due to the anisotropy of the stencil explained in section~\ref{s:fwddiff-family}). The maps that have been characterised experimentally so far are best matched in region 3 for $p = 1$.

In the $\sigma$ variable there is a narrow interval $I_{\text{L}}$ just below the critical $\sigma_{\text{c}}$ in which the following phenomenon occurs: if training at a fixed $\sigma \in I_{\text{L}}$, the stripe widths of OD and OR increase and the number of pinwheels decrease (where pinwheels of opposite sign approach and annihilate) as a function of the iteration number. \emph{Pinwheel annihilation} has been noted in the literature of cortical map models \citep{WolfGeisel98a}. If training for a very long time, it is possible to get rid of nearly all OD borders and OR pinwheels. The reduction of the number of OD borders and OR pinwheels are really the same phenomenon, which we call \emph{loop elimination} (see fig.~\ref{f:loop-elim}), and seems particularly strong for $p > 1$. It corresponds to the fact that the energy function $E$ presents a deep minimum at the net without loops that is only reachable for $\sigma \in I_{\text{L}}$; for smaller $\sigma$, $E$ develops another minimum to which the algorithm is always attracted, corresponding to the striped net. Thus, a fast annealing will result in striped maps (which are observed biologically) while a very slow one will result in maps without stripes or pinwheels (which have not been observed biologically). The phase diagram for the slow annealing case changes as follows. Region $1$ (for small $\beta$) contains a constant map, practically identical for all values of $p$ and $\beta$, which is severely distorted and contains stripes of different widths. Region $2$ contains unsegregated maps as before. Region $3$ contains columnar maps with loop elimination, i.e., wide stripes and few pinwheels (depending on the slowness of the annealing).

Another intriguing issue is the fact that the global structure of the maps can sometimes be very similar for different values of $p$ in part of region $3$ of the $(\beta,p)$ space. For example, the maps in fig.~\ref{f:simul2D:OD} for $(\beta,p) \in \{(10^1,1),\ (10^3,2),\ (10^5,3)\}$; and for the diagonal from $(10^0,1)$ to $(10^3,4)$. This effect can be partially explained by the drifting cutoff argument of section~\ref{s:fwddiff-family} as follows. For the forward-difference family the power is $p_k = \bigl( \smash{2 \sin{\left( \pi \frac{k}{M} \right)}} \bigr)^{2p}$ and in 2D we normalise by dividing it by twice the stencil squared modulus, $2 \binom{2p}{p}$ (see appendix~\ref{s:stencil-norm}), so $p_k \approx \frac{1}{2} \sqrt{\pi p} \sin^{2p}{\left( \frac{\pi k}{M} \right)}$. Assuming the emerging net \Y\ has frequency $k^*$, the tension term value will be $\frac{\beta}{2} \trace{\Y^T\Y\SS} \propto \frac{\beta}{2} p_{k^*} = \frac{1}{4} \beta \sqrt{\pi p} \sin^{2p}{ \left( \frac{\pi \smash{k^*}}{M} \right) }$, and any combination of $\beta$ and $p$ for which this tension equals a given constant (dependent on $k^*$) should result in roughly the same map. So we would expect that similar maps occur along the family of curves
\begin{equation*}
  \log{\beta} = -2 p \log{\sin{ \left( \frac{\pi k^*}{M} \right) }} -\frac{1}{2} \log{p} + \text{constant}
\end{equation*}
where all the logarithms are in base $10$. For $k^* << M$ (wide stripes) these are straight lines with positive slope $2 \log{(M/\pi k^*)}$. This expression indeed matches well the slopes for the cases mentioned above, and also justifies the use of a logarithmic scale for $\beta$ (otherwise evident from fig.~\ref{f:simul2D:OD}). Note that the slope is at most $2 \log{(M/\pi)}$ and occurs for a single-striped map ($k^* = 1$). The region for $\log{\beta} > 2 p \log{(M/\pi)}$ should correspond to maps that are either single-striped or unsegregated. In the region of thin stripes (for $k^*$ close to $M/2$) the curves become horizontal and logarithmic in $p$.

However, even though some OD maps may be similar for different $p$, detailed examination shows that the geometric relations between the OD and OR maps still differ as mentioned above (e.g.\ the pinwheels colocate with OD borders as $p$ increases). A rigorous theoretical explanation for the structure of the 3D phase map $(\beta,p,\sigma)$ both in the continuous and discrete cases is left for future research.

The GEN could be further extended in two useful ways for cortical map modelling. (1) Instead of using deterministic annealing, we could learn both \Y\ and $\sigma$ (e.g.\ with an EM algorithm) and interpret the resulting $\sigma$ as a receptive field size. However, a prior on $\sigma$ that penalises high values will be necessary, because otherwise the result is typically a net with large $\sigma$ that stretches along the principal component of the training set but not along orthogonal directions with lower variance. In effect, the freedom in the selection of the annealing schedule becomes the freedom in the selection of the initial net and the prior on $\sigma$ (biologically interpretable as a constraint on the receptive field size). (2) We could consider a stencil $\varsigma_m$ or tension strength $\beta_m$ that depends on the centroid index. For example, the stripe width of OD is smaller at the fovea than at the periphery of the visual field and this might be due to a dependence of the lateral connection pattern on the cortical location. Both things can be readily incorporated in the model without modifying the learning algorithm, by e.g.\ defining manually the rows of the matrix \D, or defining a tension term $\frac{1}{2} \trace{\Y^T \Y \SS}$ with $\SS = \D^T \B \D$, where $\B = \diag{\beta_1,\dots,\beta_M}$. However, it would be more parsimonious to express explicitly the dependence on $m$, e.g.\ by defining $\beta_m = a m + \beta_0$ or $\beta_m = \beta_0 e^{am}$ for some parameters $\beta_0$ and $a$ (which could also be learned from the data).

\begin{figure}
  \begin{center}
    \psfrag{p1}[B][B]{$p = 1$, $\leftexp{(1)}{\varsigma} = (0,\ -1,\ 1)$}
    \psfrag{p2}[B][B]{$p = 2$, $\leftexp{(2)}{\varsigma} = (1,\ -2,\ 1)$}
    \psfrag{small_s}[l][Bl]{\footnotesize small $\sigma$}
    \psfrag{large_s}[l][Bl]{\footnotesize large $\sigma$}
    \includegraphics[width=\textwidth]{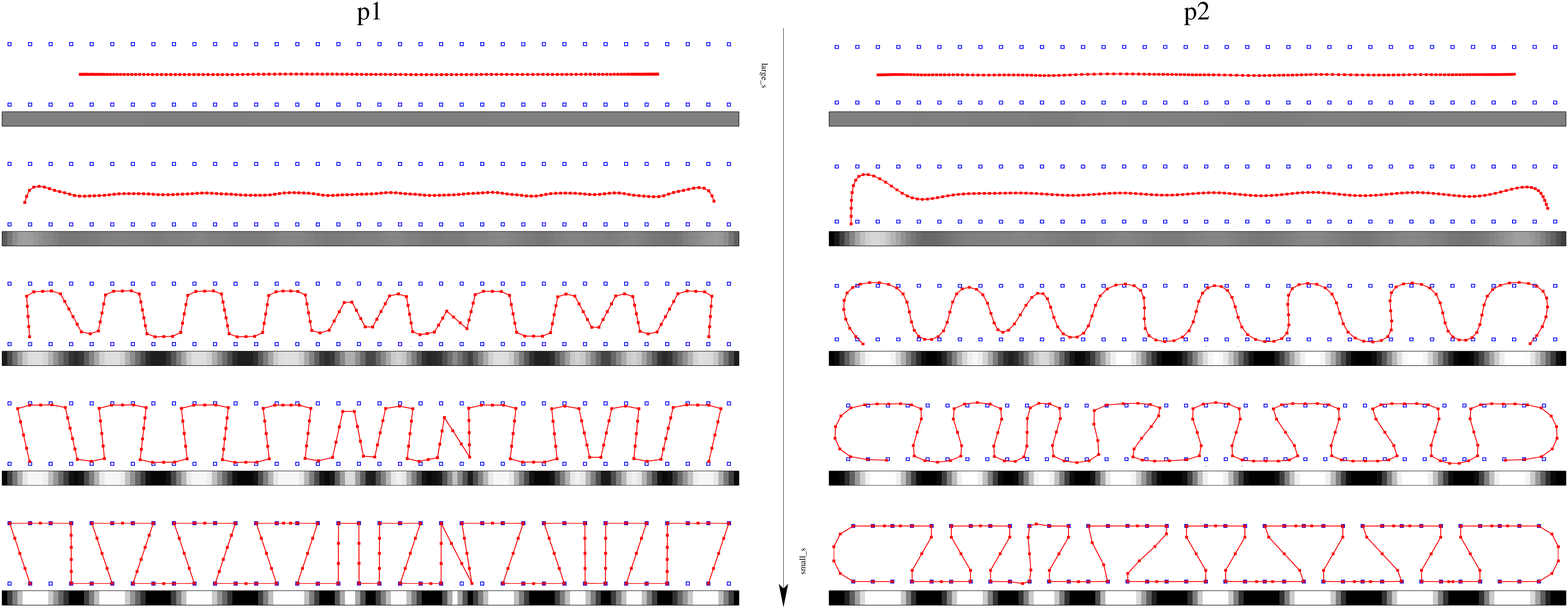}
    \caption{Typical sequence of map development for a 1D generalized elastic net in a 2D stimulus space of VF (X axis) and OD (Y axis) for the forward-difference family (stencils of order $1$ and $2$), with nonperiodic b.c.\ and annealing. The thin greyscale bar below the plots represents the OD value (black for one eye, white for the other) along the centroids $m = 1,\dots,M$ as a 1D ocular dominance map. The net is initially unselective to ocularity, after which an initial wave appears, and the final state is a set of stripes. Note the difference in curvature between both stencils, e.g.\ the absence of corners in the second-order stencil (for intermediate $\sigma$).}
    \label{f:simul1D-nonperiodic}
  \end{center}
\end{figure}

\begin{figure}
  \begin{center}
    \psfrag{p1}{$p = 1$, $\leftexp{(1)}{\varsigma} = (0,\ -1,\ 1)$}
    \psfrag{p2}{$p = 2$, $\leftexp{(2)}{\varsigma} = (1,\ -2,\ 1)$}
    \psfrag{p3}{$p = 3$, $\leftexp{(3)}{\varsigma} = (0,\ -1,\ 3,\ -3,\ 1)$}
    \psfrag{p4}{$p = 4$, $\leftexp{(4)}{\varsigma} = (1,\ -4,\ 6,\ -4,\ 1)$}
    \includegraphics[width=\textwidth]{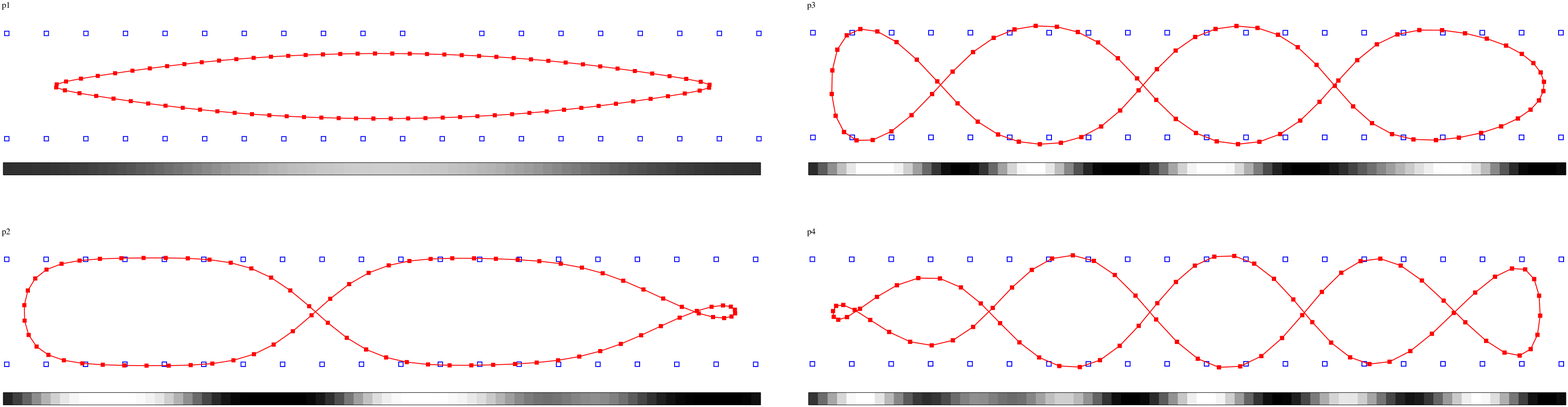}
    \caption{Selected subset of cortical map simulations with a generalized elastic net with the forward-difference family, for the same problem type as in fig.~\ref{f:simul1D-nonperiodic}, with periodic b.c.\ (for the net, but not the training set) and $\beta = 10^{2}$. Except for $p$, all other parameters and the initial conditions are the same. The nets are shown at a value of the annealing parameter $\sigma$ shortly after the net expands along the OD axis. Note how the frequency increases with $p$.}
    \label{f:simul1D-periodic}
  \end{center}
\end{figure}

\begin{figure}
  \begin{center}
      \psfrag{t4}[l][Bl]{\scriptsize \caja{c}{c}{$\tau =$ \\ $4$}}
      \psfrag{t20}[l][Bl]{\scriptsize \caja{c}{c}{$\tau =$ \\ $20$}}
      \psfrag{t200}[r][Br]{\scriptsize \caja{c}{c}{ $\tau =$ \\ $200$}}
      \psfrag{t250}[l][Bl]{\scriptsize \caja{c}{c}{$\tau =$ \\ $250$}}
      \psfrag{t300}[r][Br]{\scriptsize \caja{c}{c}{$\tau =$ \\ $300$}}
      \psfrag{t800}[r][Br]{\scriptsize \caja{c}{c}{$\tau =$ \\ $800$}}
      \psfrag{t900}[r][Br]{\scriptsize \caja{c}{c}{$\tau =$ \\ $900$}}
      \psfrag{t920}[l][Bl]{\scriptsize \caja{c}{c}{$\tau =$ \\ $920$}}
      \psfrag{t1020}[l][Bl]{\scriptsize \caja{c}{c}{$\tau =$ \\ $1020$}}
      \psfrag{t1080}[l][Bl]{\scriptsize \caja{c}{c}{$\tau =$ \\ $1080$}}
      \psfrag{t1450}[l][Bl]{\scriptsize \caja{c}{c}{$\tau =$ \\ $1450$}}
      \psfrag{t1460}[l][Bl]{\scriptsize \caja{c}{c}{$\tau =$ \\ $1460$}}
      \psfrag{t3000}[l][Bl]{\scriptsize \caja{c}{c}{$\tau =$ \\ $3000$}}
      \psfrag{t3010}[l][Bl]{\scriptsize \caja{c}{c}{$\tau =$ \\ $3010$}}
      \psfrag{t3400}[r][Br]{\scriptsize \caja{c}{c}{$\tau =$ \\ $3400$}}
      \psfrag{t3500}[r][Br]{\scriptsize \caja{c}{c}{$\tau =$ \\ $3500$}}
      \psfrag{t4100}[r][Br]{\scriptsize \caja{c}{c}{$\tau =$ \\ $4100$}}
      \psfrag{t5000}[l][Bl]{\scriptsize \caja{c}{c}{$\tau =$ \\ $5000$}}
      \includegraphics[width=\textwidth]{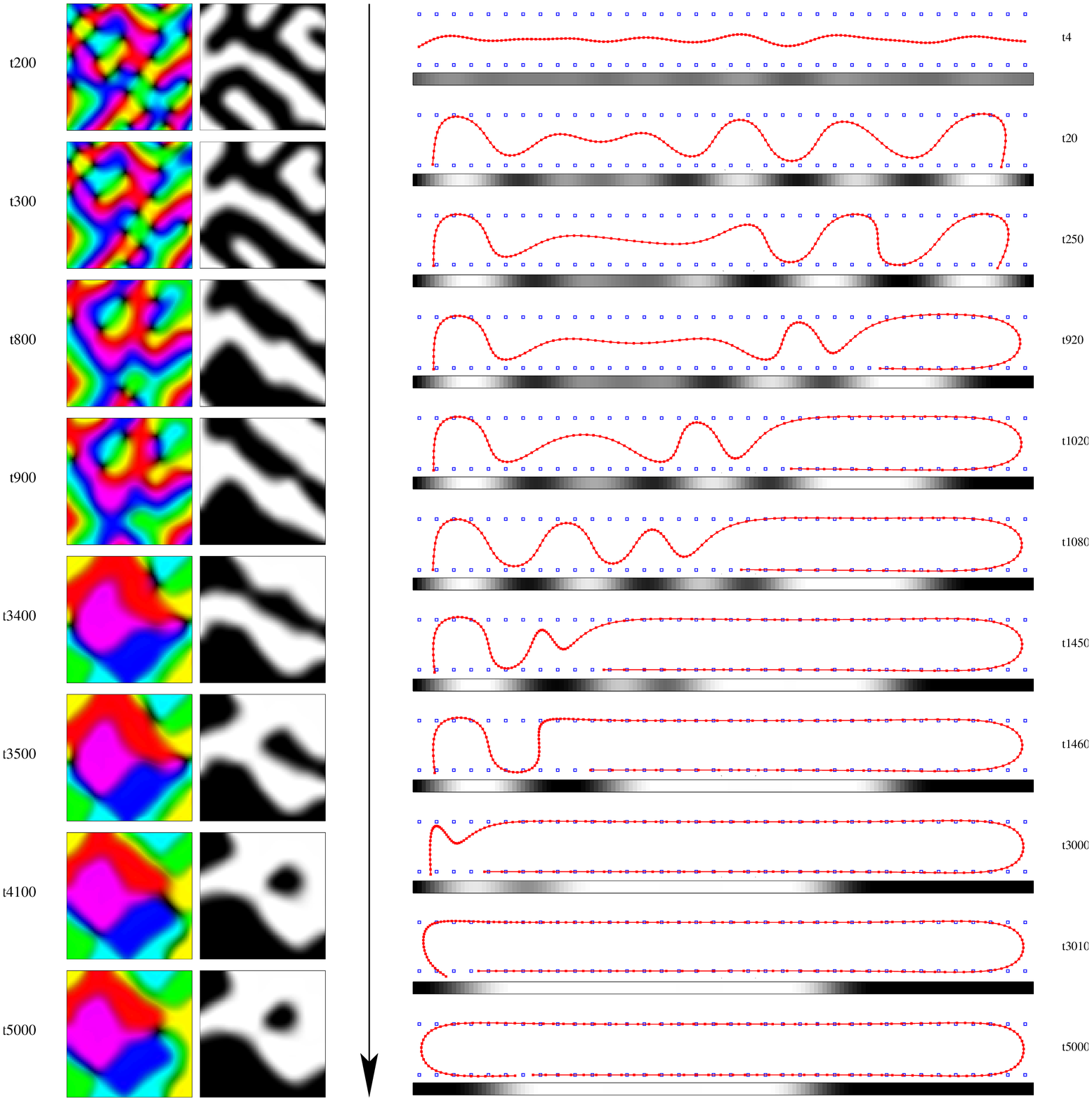}
    \caption{The phenomenon of loop elimination without annealing (for fixed $\sigma$), for $p = 2$ (forward difference). $\tau$ indicates the number of iterations (solutions of the system~\eqref{e:ann:zerograd} with the Cholesky factorisation method). \emph{Left}: pinwheel annihilation and stripe widening for a 2D cortical map model ($\sigma = 0.05$). The training set is a grid in the rectangle $[0,1] \times [0,1] \times [-0.14,0.14] \times [-\frac{\pi}{2},\frac{\pi}{2}] \times \{0.2\}$ with $N = 10 \times 10 \times 2 \times 6 \times 1$ points and the net has $M = 72 \times 72$ centroids and $\beta = 100$ with nonperiodic b.c. \emph{Right}: loop elimination for a 1D cortical map model ($\sigma = 0.016$). For this particular case loop elimination happened only for $\sigma \in [0.016,0.0245]$. For other $\sigma$ the net converged to a striped configuration. The training set is a grid in the rectangle $[0,1] \times [-0.0422,0.0422]$ with $N = 36 \times 2$ points and the net has $M = 144$ centroids and $\beta = 5\,000$ with nonperiodic b.c. In both cases, the net that appears initially is striped; then, as the training continues over a long term, loops are eliminated in sudden but widely separated events.}
    \label{f:loop-elim}
  \end{center}
\end{figure}

\begin{figure}
  \begin{center}
    \begin{tabular}{@{}c@{\hspace{1.5cm}}c@{}}
      \psfrag{p}[][B]{$p$}
      \psfrag{b}[][B]{$\log{\beta}$~}
      \psfrag{1}[][B]{\Large $1$}
      \psfrag{2}[][B]{\Large $2$}
      \psfrag{3}[][B]{\Large $3$}
      \psfrag{k1}[][]{$k^*_1$}
      \psfrag{k2}[][]{$k^*_2$}
      \psfrag{k3}[][]{$k^*_3$}
      \includegraphics[height=.3\textheight]{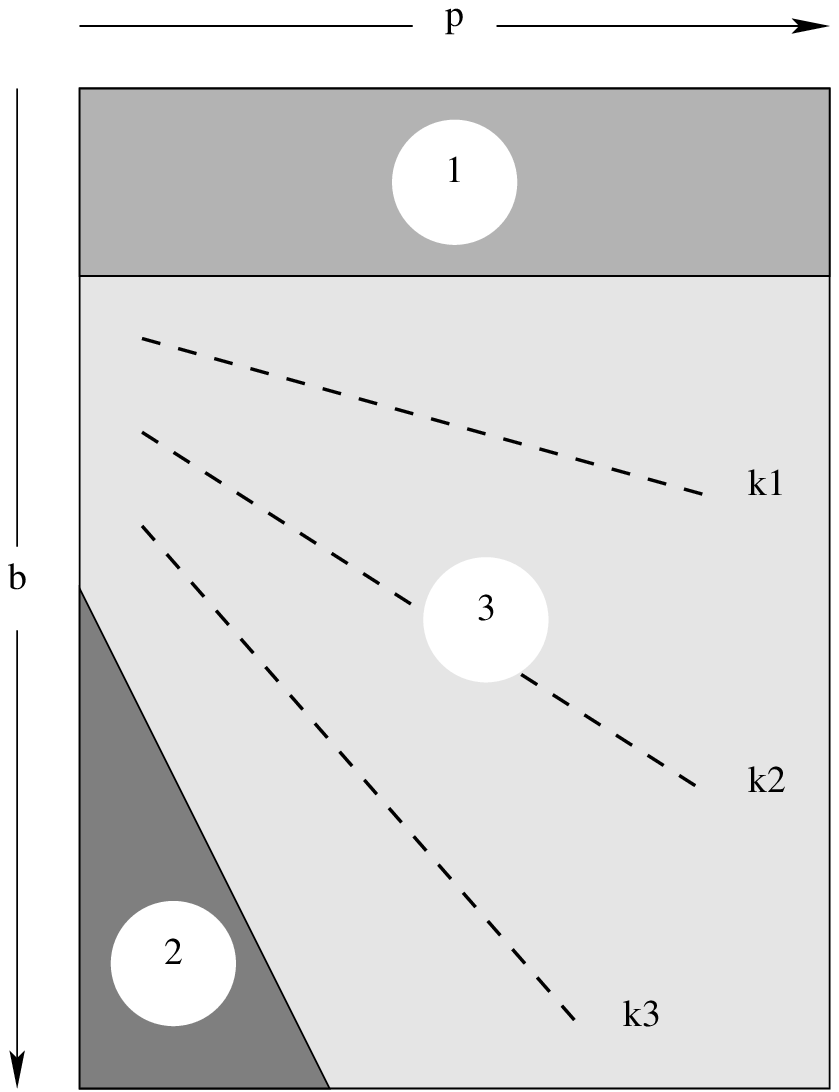} &
      \footnotesize
      \psfrag{p1}[lb][l]{~~~$p = 1$}
      \psfrag{p2}[lb][l]{~~~$p = 2$}
      \psfrag{p3}[lt][lB]{~~~$p = 3$}
      \psfrag{p4}[lb][lt]{~~~$p = 4$}
      \psfrag{b1}[][]{\raisebox{-3ex}{$10^{-1}$}}
      \psfrag{b2}[][]{\raisebox{-3ex}{$10^{0}$}}
      \psfrag{b3}[][]{\raisebox{-3ex}{$10^{1}$}}
      \psfrag{b4}[][]{\raisebox{-3ex}{$10^{2}$}}
      \psfrag{b5}[][]{\raisebox{-3ex}{$10^{3}$}}
      \psfrag{b6}[][]{\raisebox{-3ex}{$10^{4}$}}
      \psfrag{b7}[][]{\raisebox{-3ex}{$10^{5}$}}
      \psfrag{b8}[][]{\raisebox{-3ex}{$10^{6}$}}
      \psfrag{0}[r][r]{$0$}
      \psfrag{0.05}[r][r]{$0.05$}
      \psfrag{0.1}[r][r]{$0.1$}
      \psfrag{beta}[t][b]{\normalsize $\beta$}
      \psfrag{Kc}[][][1][-90]{\normalsize \raisebox{1cm}{$\sigma_{\text{c}}$}}
      \includegraphics[height=.3\textheight]{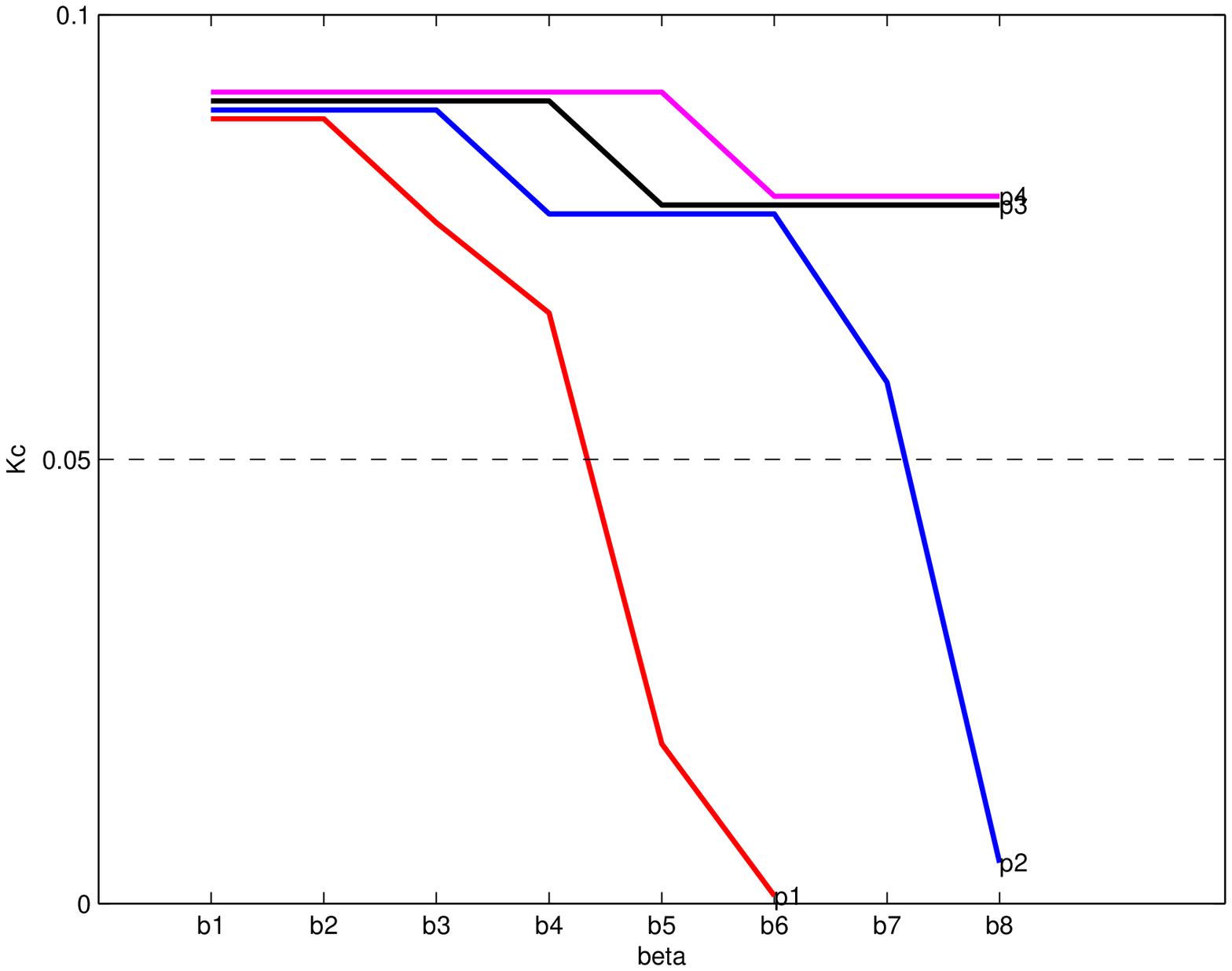}
    \end{tabular}
    \caption{\emph{Left}: schematic phase diagram in the space of $(\beta,p)$ for a value of $\sigma$ small enough that the OD and OR maps have arisen ($\sigma = 0.05$ in the actual simulations). For fast annealing the regions mean: salt-and-pepper maps ($1$), striped map ($2$) and unsegregated map ($3$); compare with fig.~\ref{f:simul2D:OD}. For slow annealing, region $1$ contains a map independent of $p$ and $\beta$ (but not necessarily salt-and-pepper) and region $3$ contains maps with wide stripes and few pinwheels. The boundaries between regions are drawn approximately. The dotted lines in region $3$ represent curves along which the stripe width is constant, for $\frac{M}{2} > k^*_1 > k^*_2 > k^*_3 > 1$. \emph{Right}: critical value $\sigma_{\text{c}}$ for which the OD and OR maps start to arise, as a function of $\beta$ for each $p$. For each $p$, values of $\sigma$ above the curve represent unsegregated maps (region $3$). The horizontal line at $\sigma = 0.05$ marks the point when maps are plotted in figures~\ref{f:simul2D:OD}--\ref{f:simul2D:ODOR}. The curves for $p = 2$ to $4$ have been raised slightly to distinguish them.}
    \label{f:phase}
  \end{center}
\end{figure}

\begin{FPfigure}
    \begin{tabular}{@{}c@{\hspace{0.5cm}}c@{\hspace{0.5cm}}c@{\hspace{0.5cm}}c@{\hspace{0.5cm}}c@{}}
      & $p = 1$ & $p = 2$ & $p = 3$ & $p = 4$ \\
      \rotatebox{90}{\makebox[.1200\textheight][c]{$\beta = 10^{-1}$}} &
      \includegraphics[height=.1200\textheight]{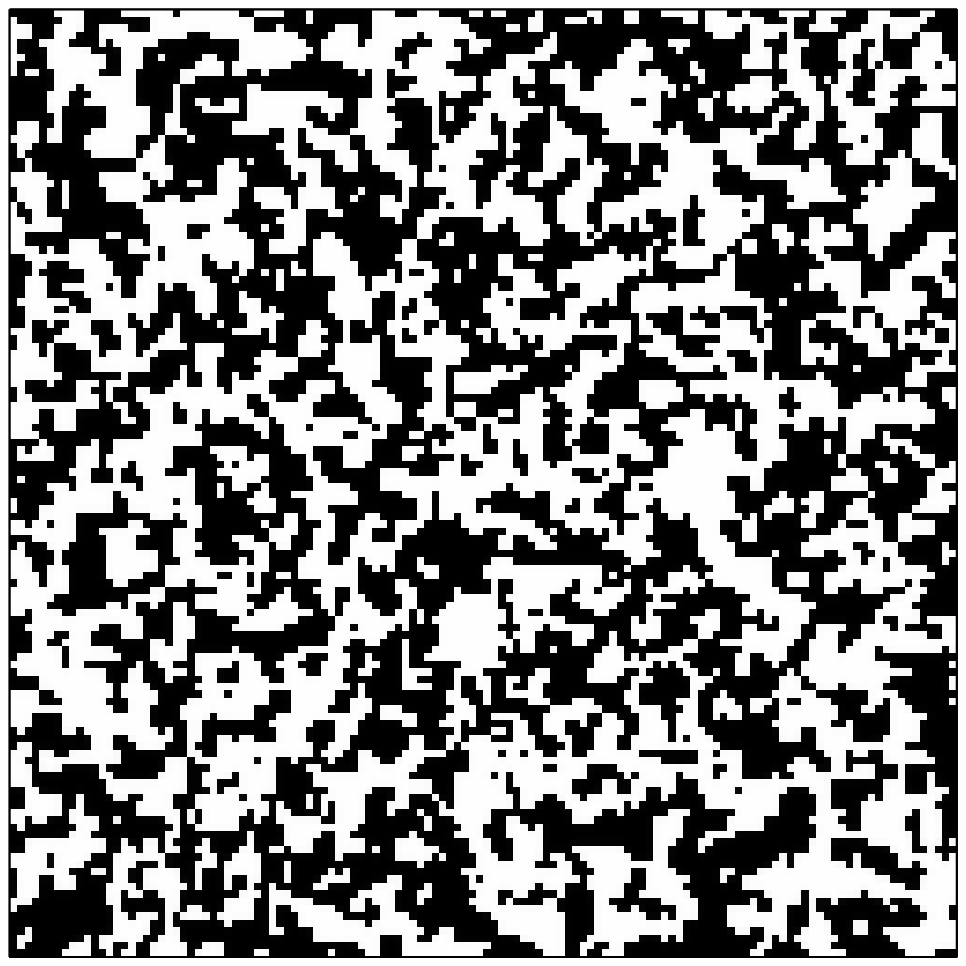} &
      \includegraphics[height=.1200\textheight]{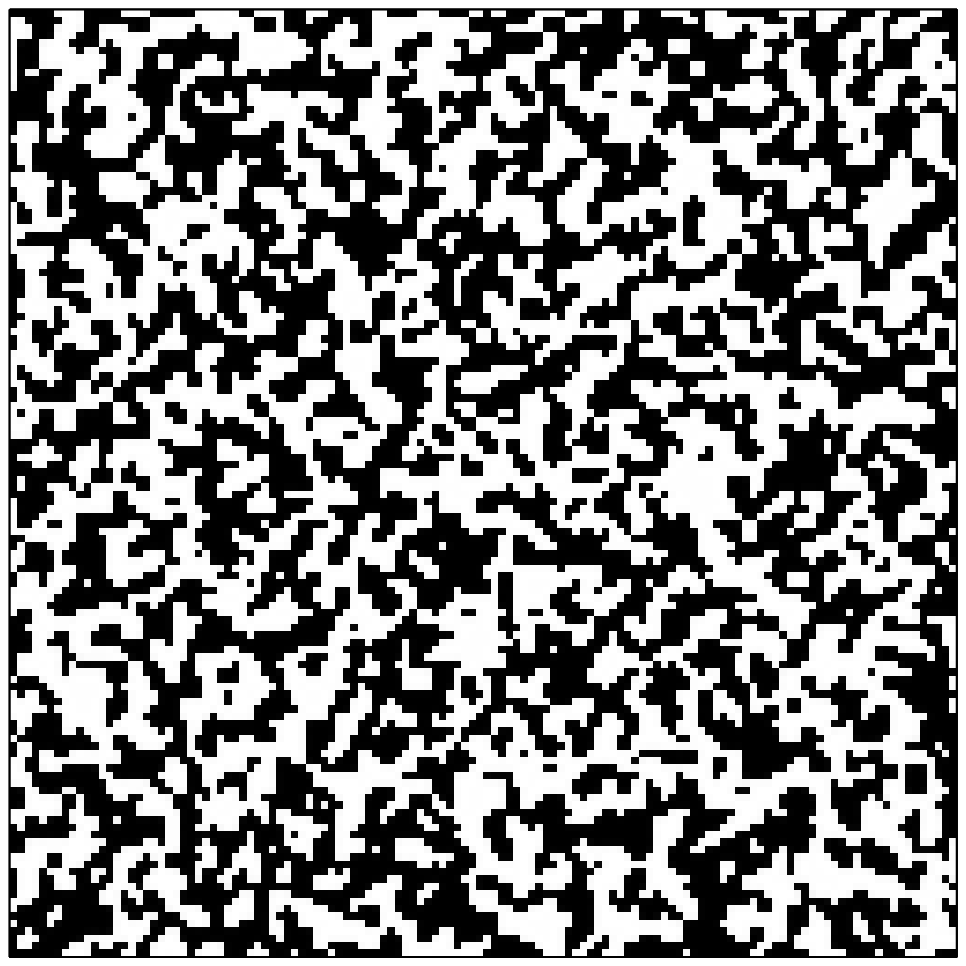} &
      \includegraphics[height=.1200\textheight]{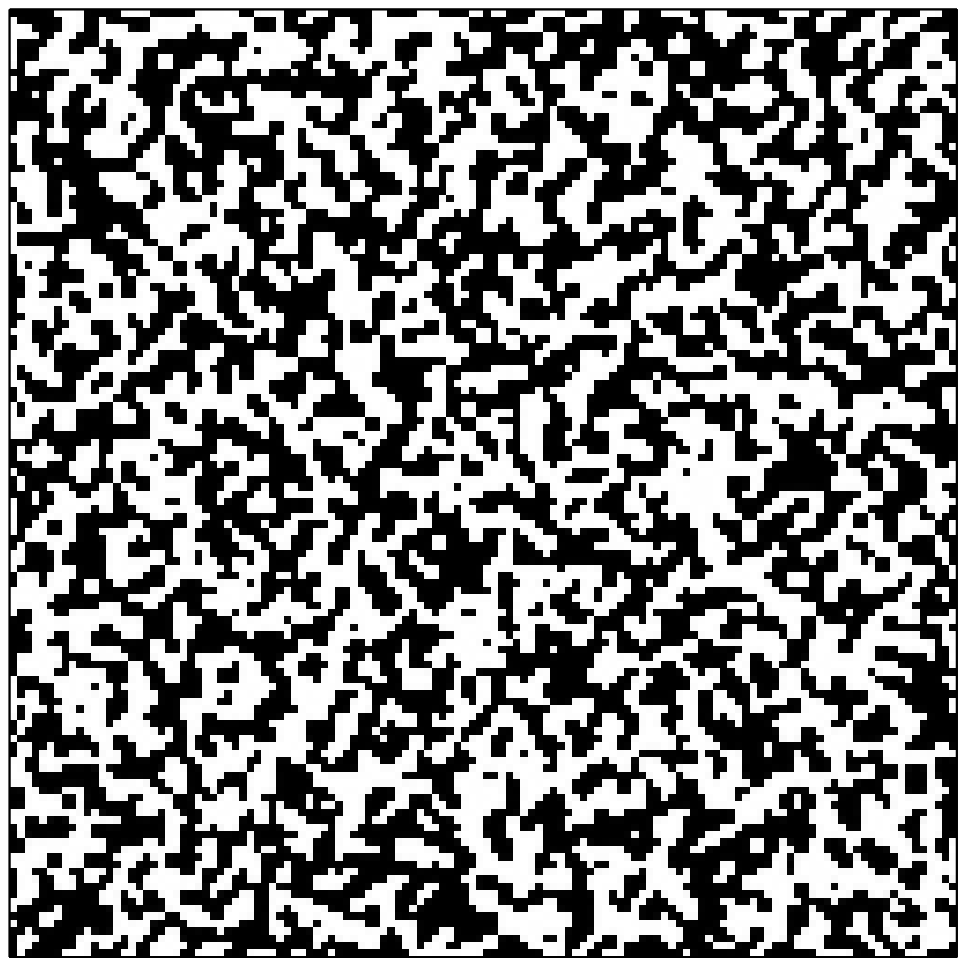} &
      \includegraphics[height=.1200\textheight]{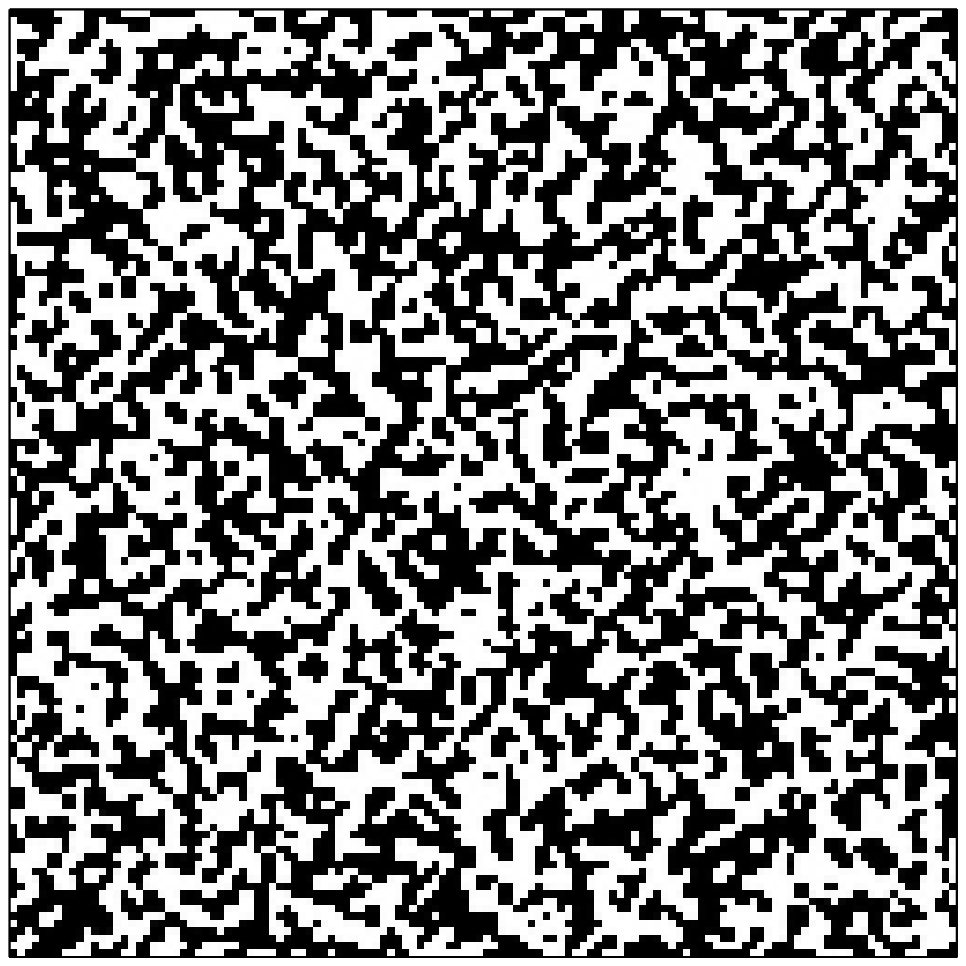} \\[-1ex]
      \rotatebox{90}{\makebox[.1200\textheight][c]{$\beta = 10^{0}$}} &
      \includegraphics[height=.1200\textheight]{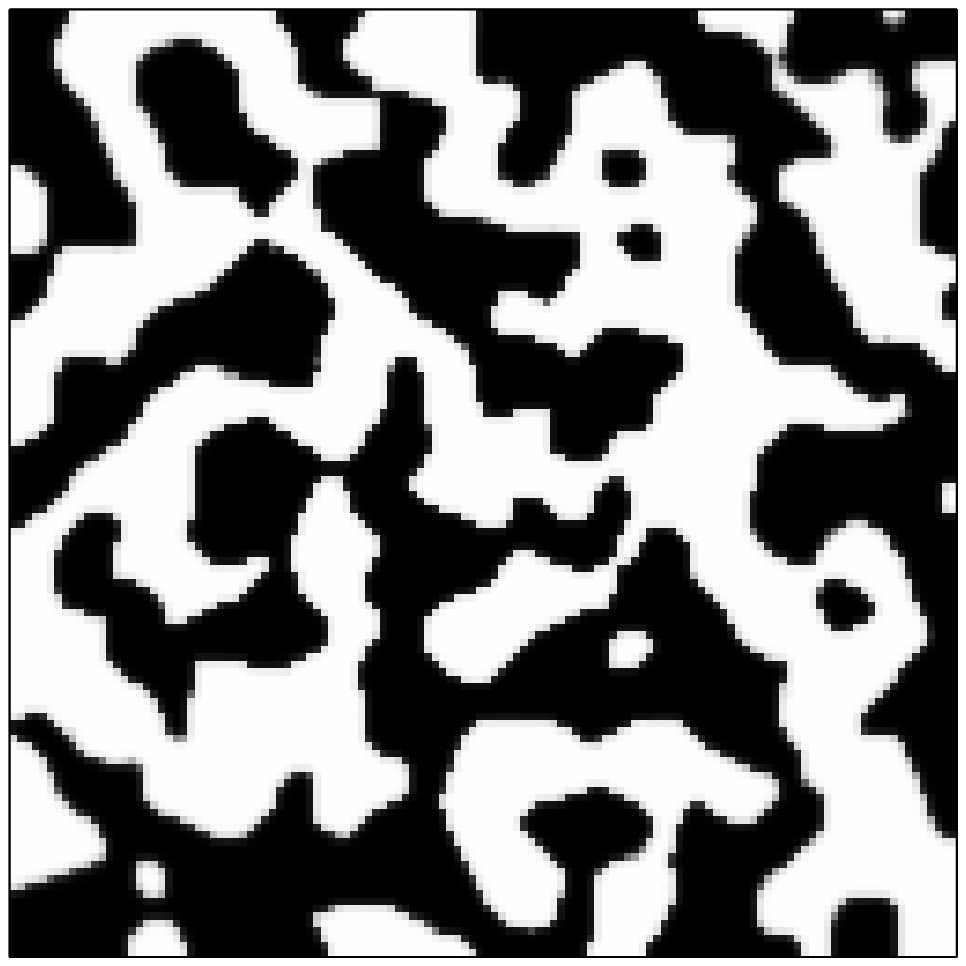} &
      \includegraphics[height=.1200\textheight]{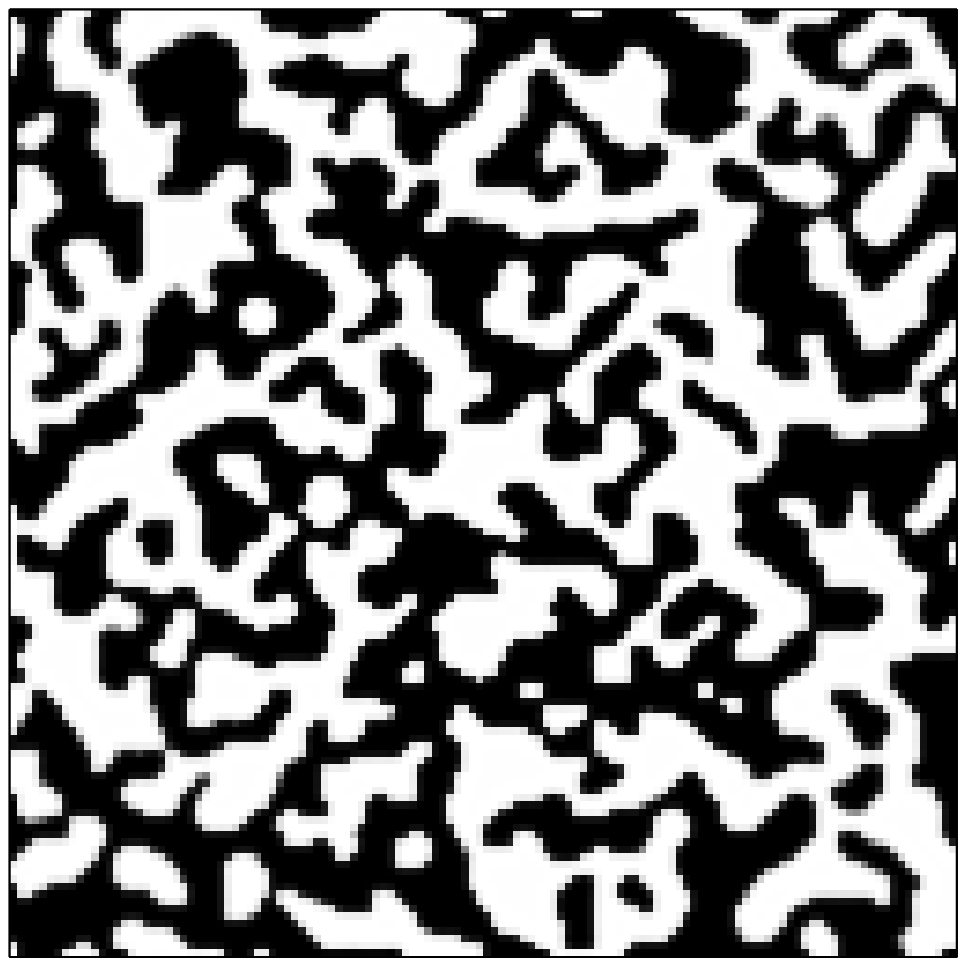} &
      \includegraphics[height=.1200\textheight]{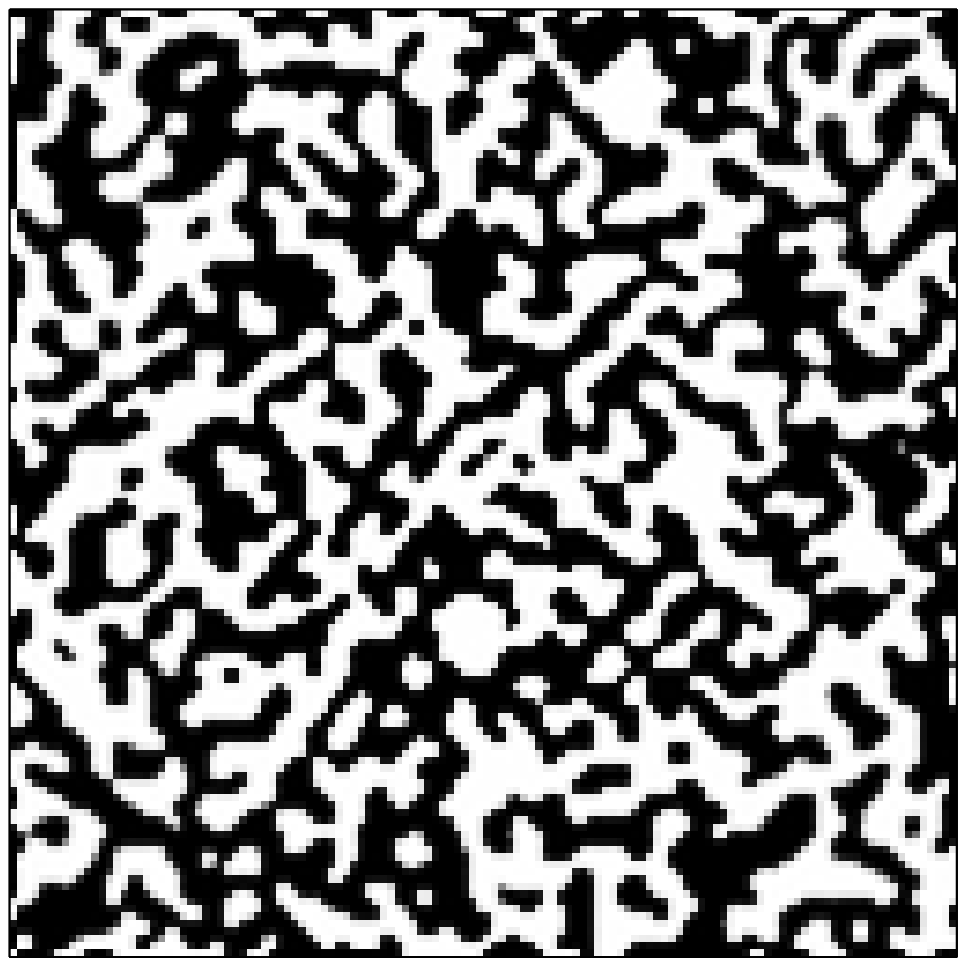} &
      \includegraphics[height=.1200\textheight]{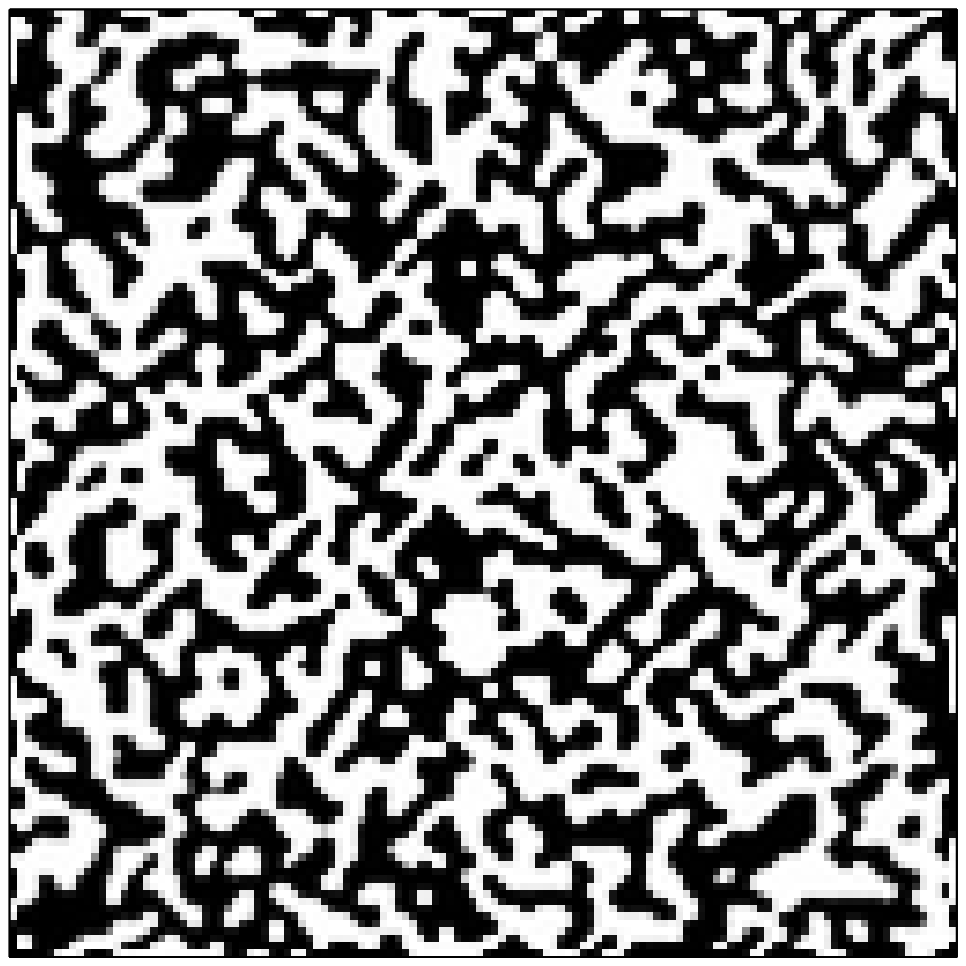} \\[-1ex]
      \rotatebox{90}{\makebox[.1200\textheight][c]{$\beta = 10^{1}$}} &
      \includegraphics[height=.1200\textheight]{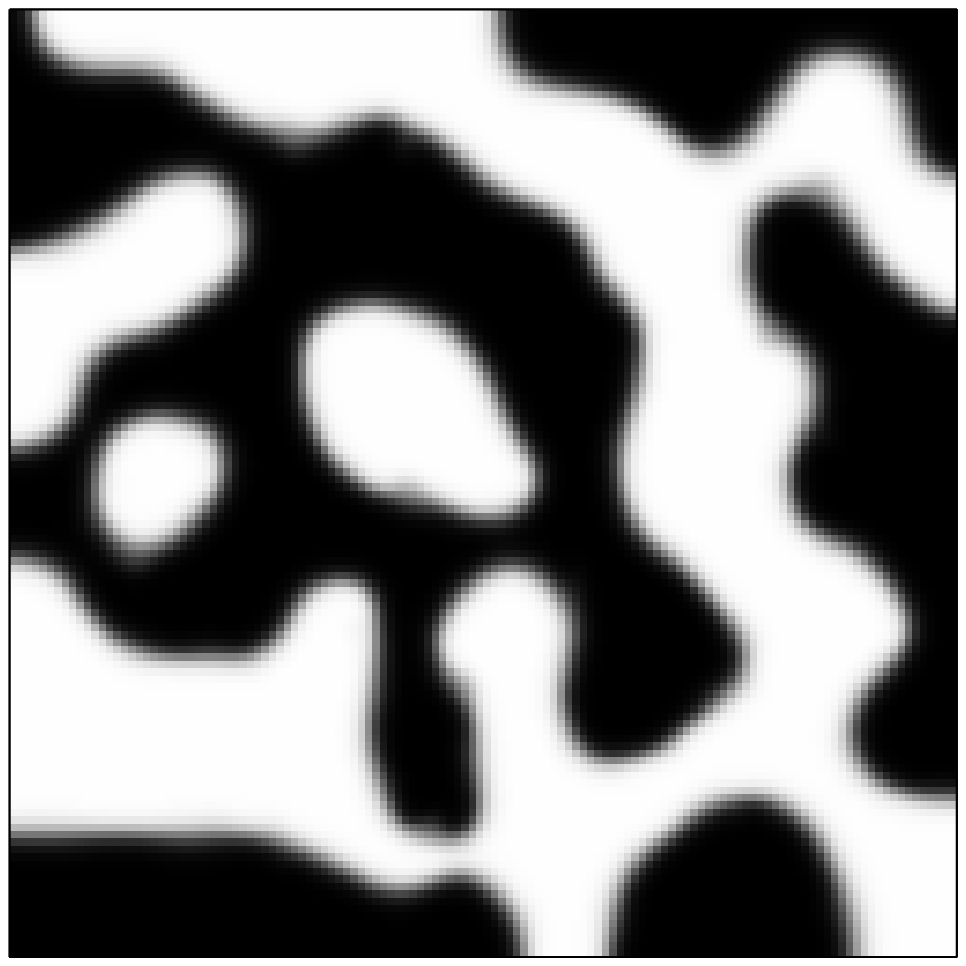} &
      \includegraphics[height=.1200\textheight]{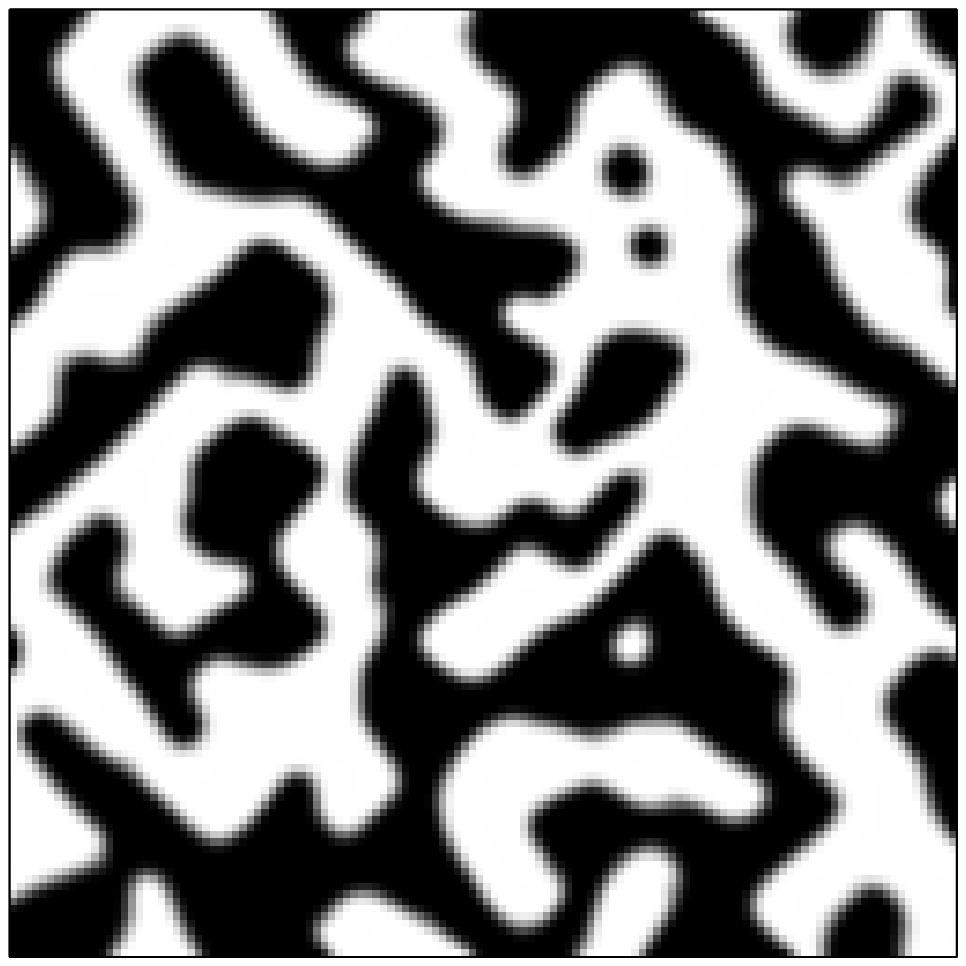} &
      \includegraphics[height=.1200\textheight]{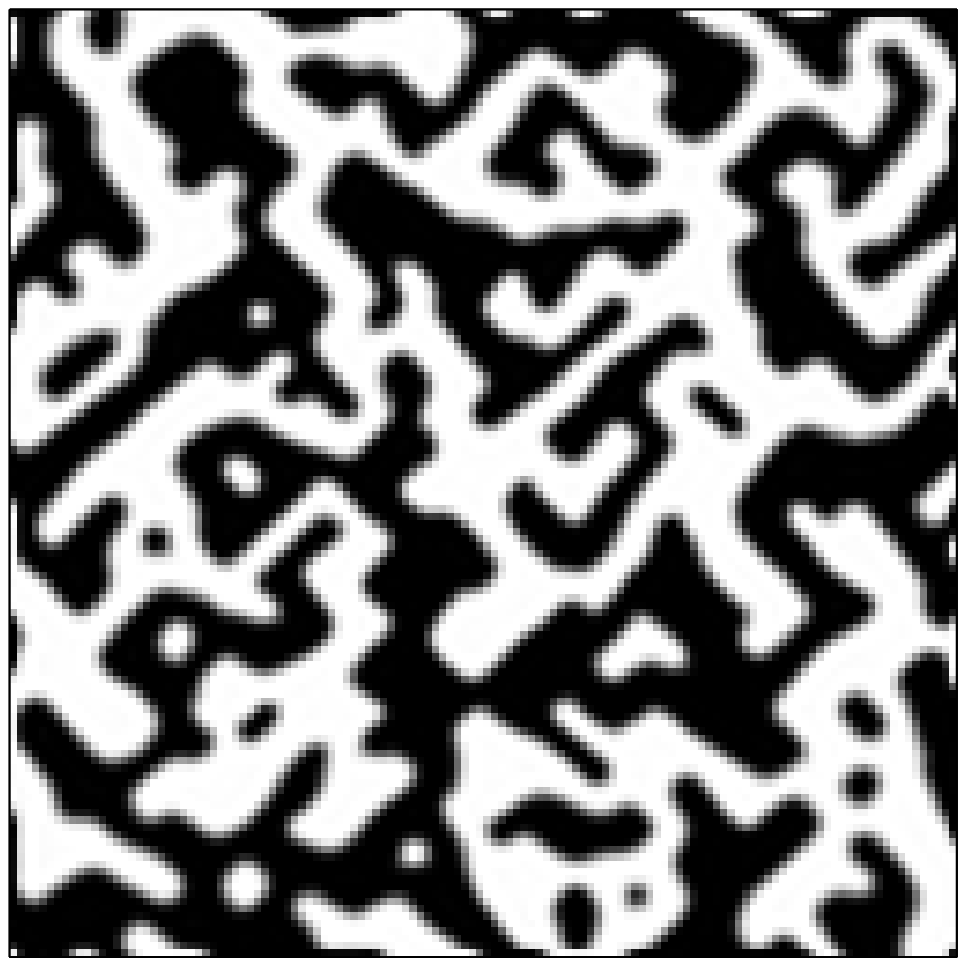} &
      \includegraphics[height=.1200\textheight]{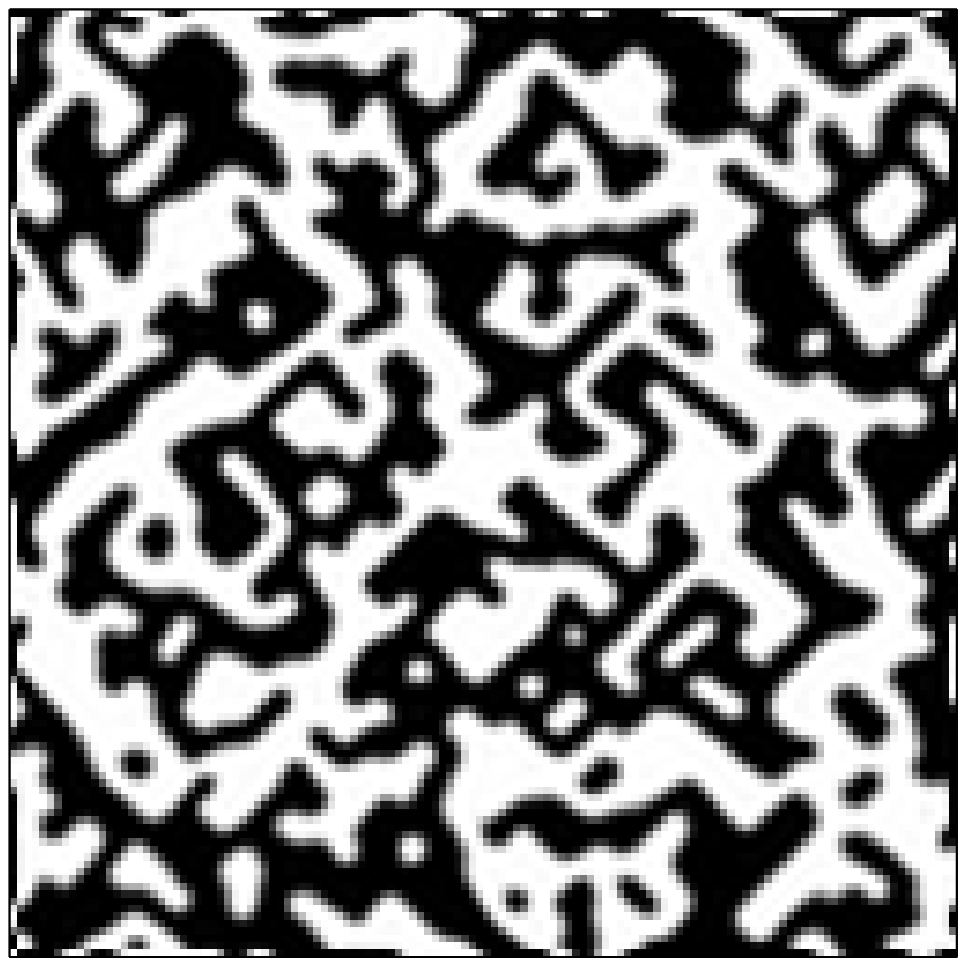} \\[-1ex]
      \rotatebox{90}{\makebox[.1200\textheight][c]{$\beta = 10^{2}$}} &
      \includegraphics[height=.1200\textheight]{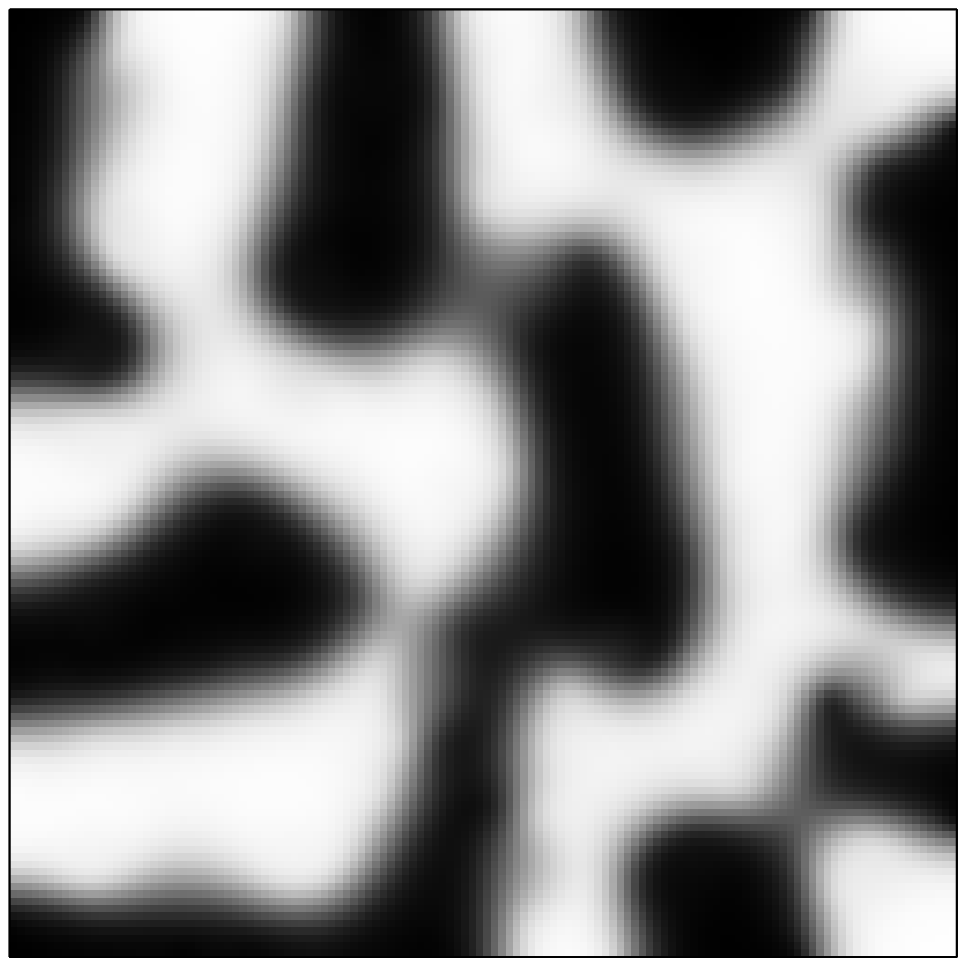} &
      \includegraphics[height=.1200\textheight]{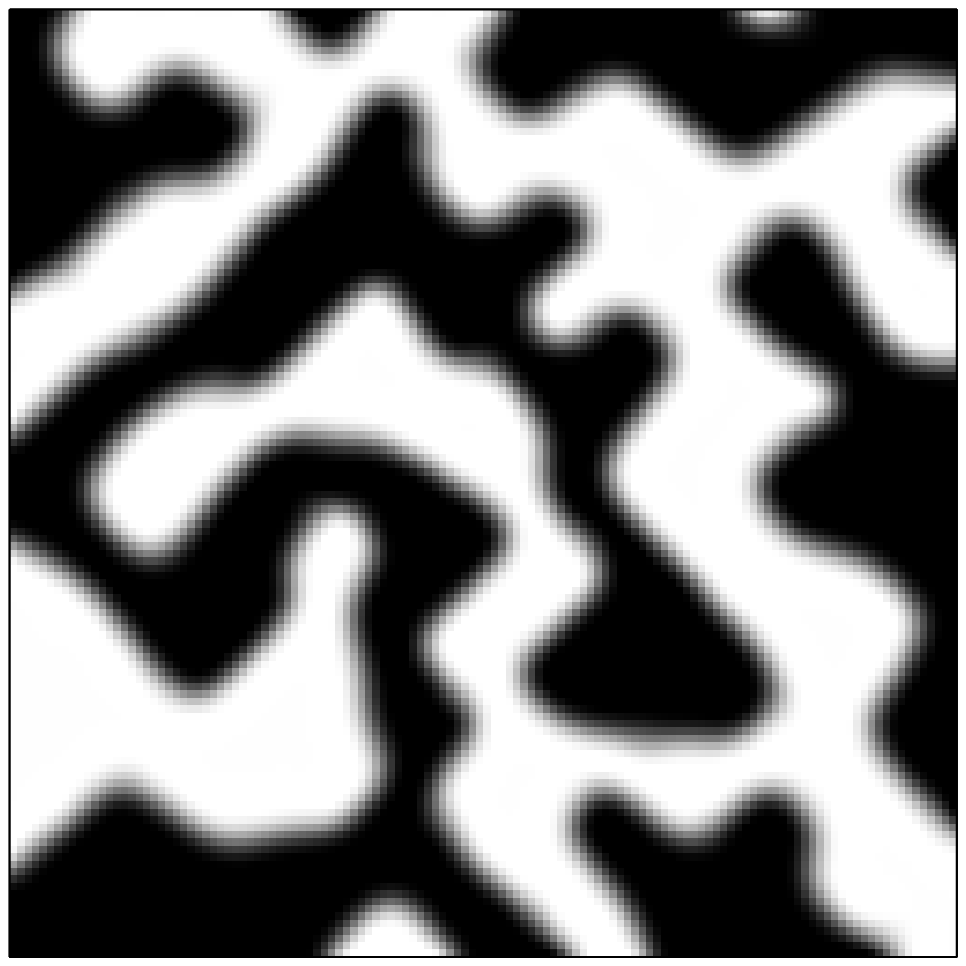} &
      \includegraphics[height=.1200\textheight]{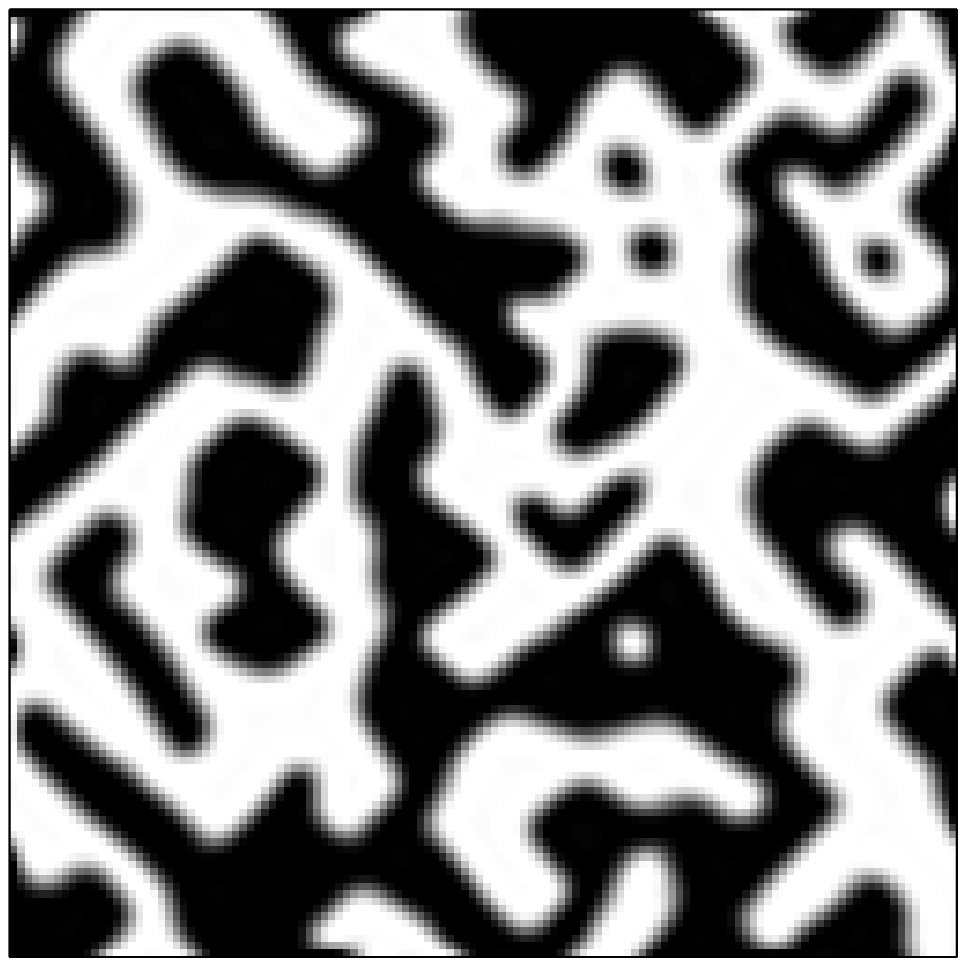} &
      \includegraphics[height=.1200\textheight]{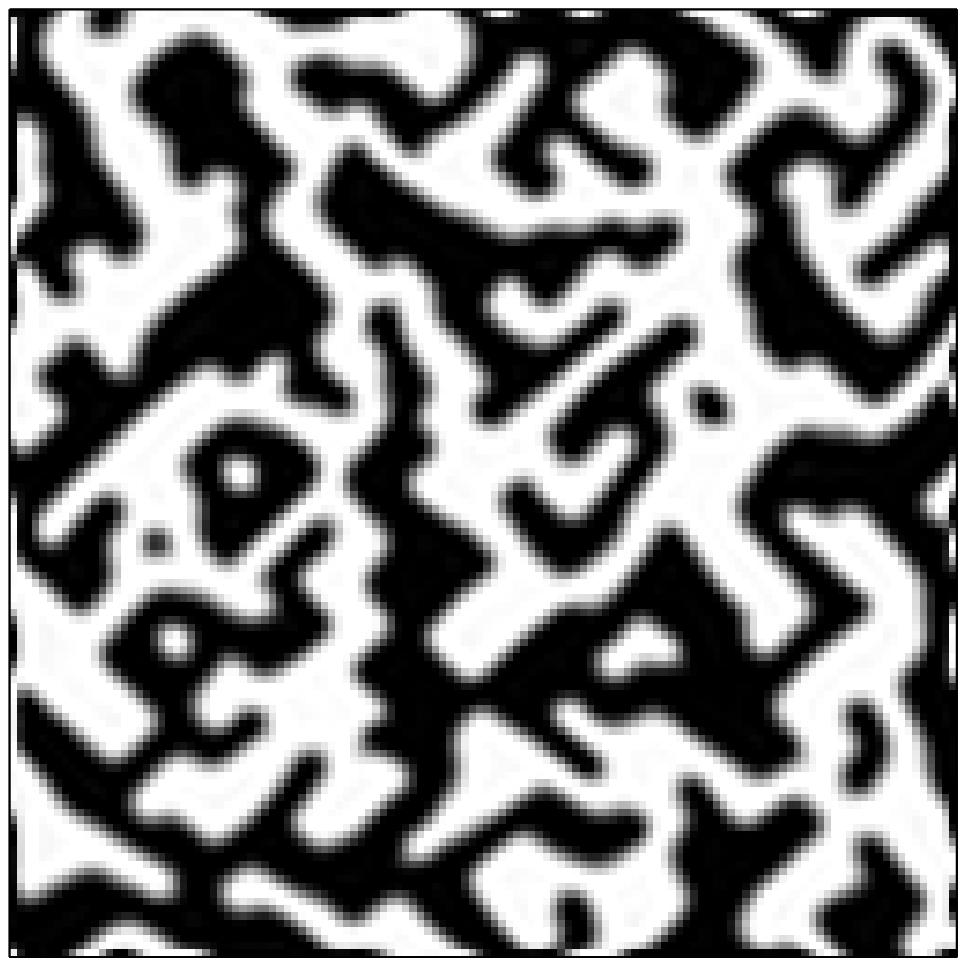} \\[-1ex]
      \rotatebox{90}{\makebox[.1200\textheight][c]{$\beta = 10^{3}$}} &
      \includegraphics[height=.1200\textheight]{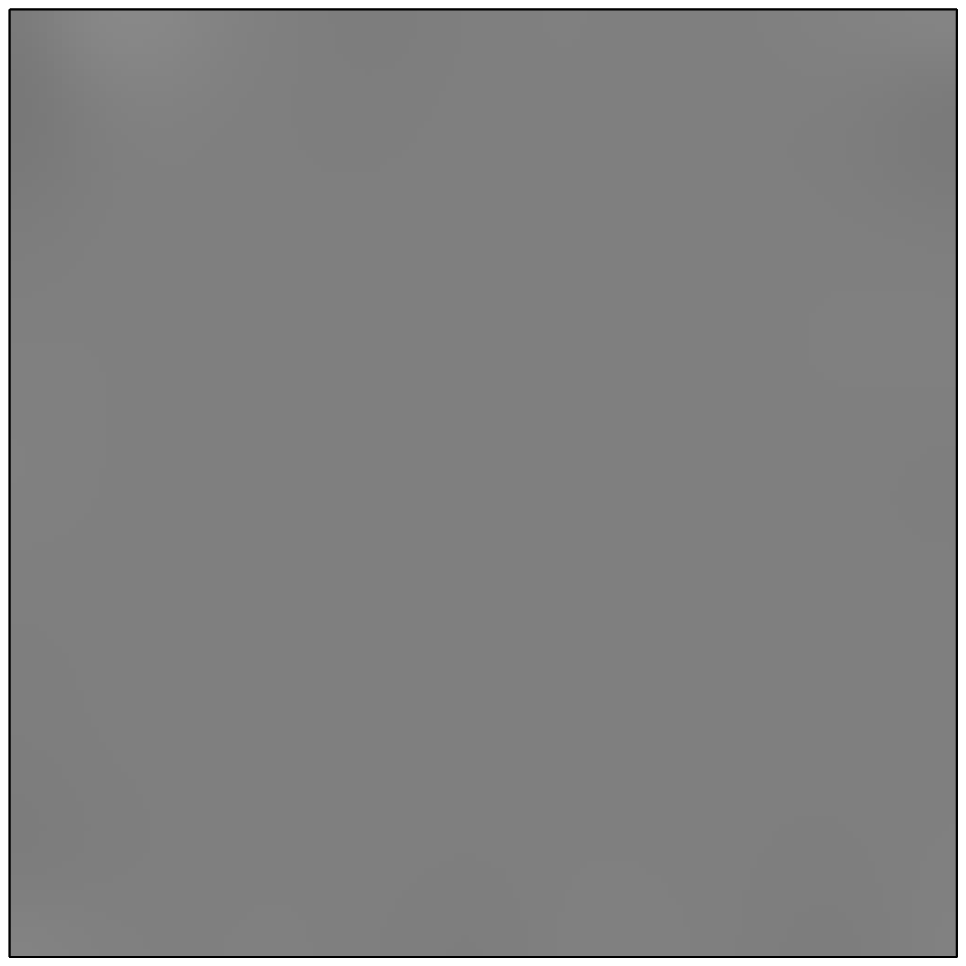} &
      \includegraphics[height=.1200\textheight]{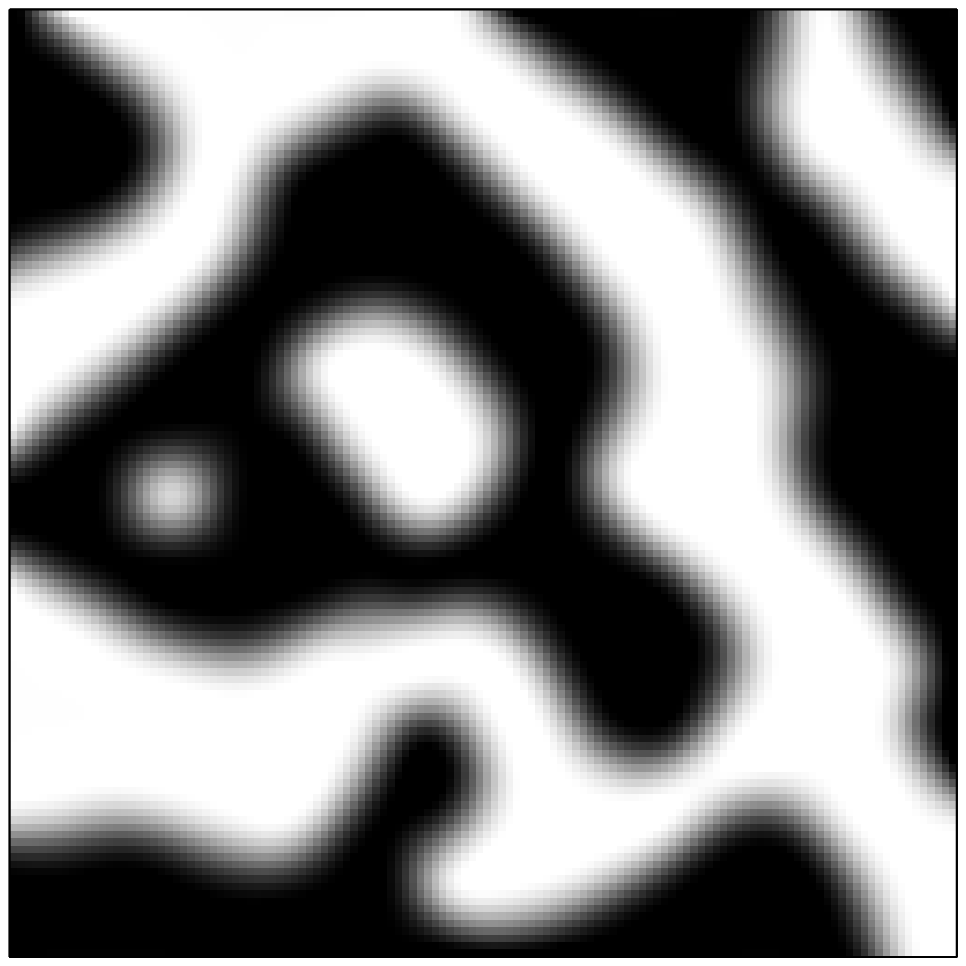} &
      \includegraphics[height=.1200\textheight]{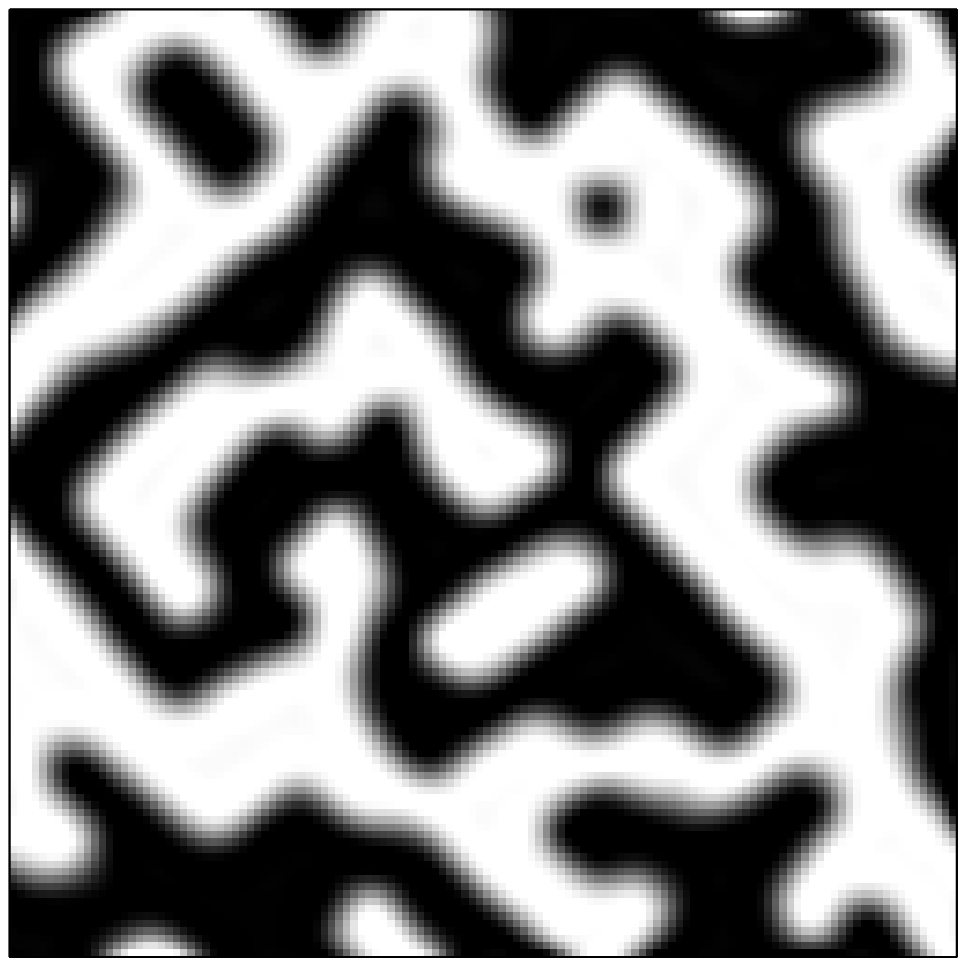} &
      \includegraphics[height=.1200\textheight]{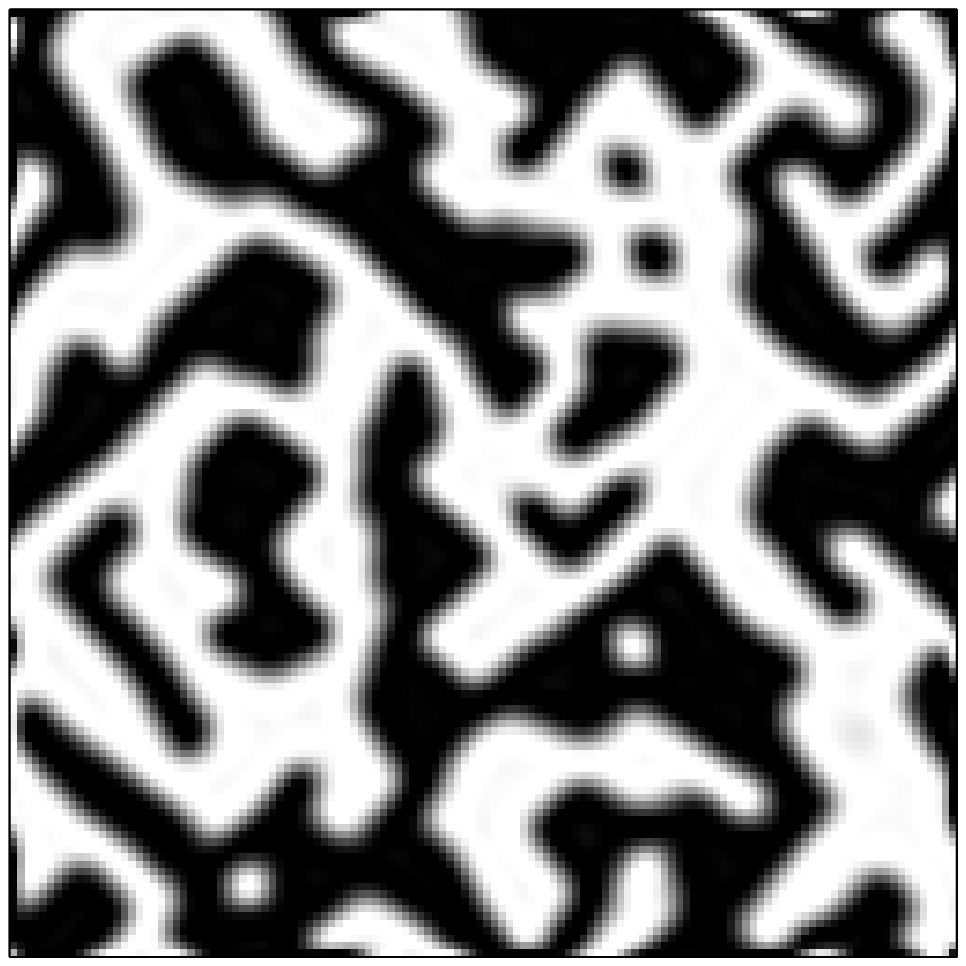} \\[-1ex]
      \rotatebox{90}{\makebox[.1200\textheight][c]{$\beta = 10^{4}$}} &
      \includegraphics[height=.1200\textheight]{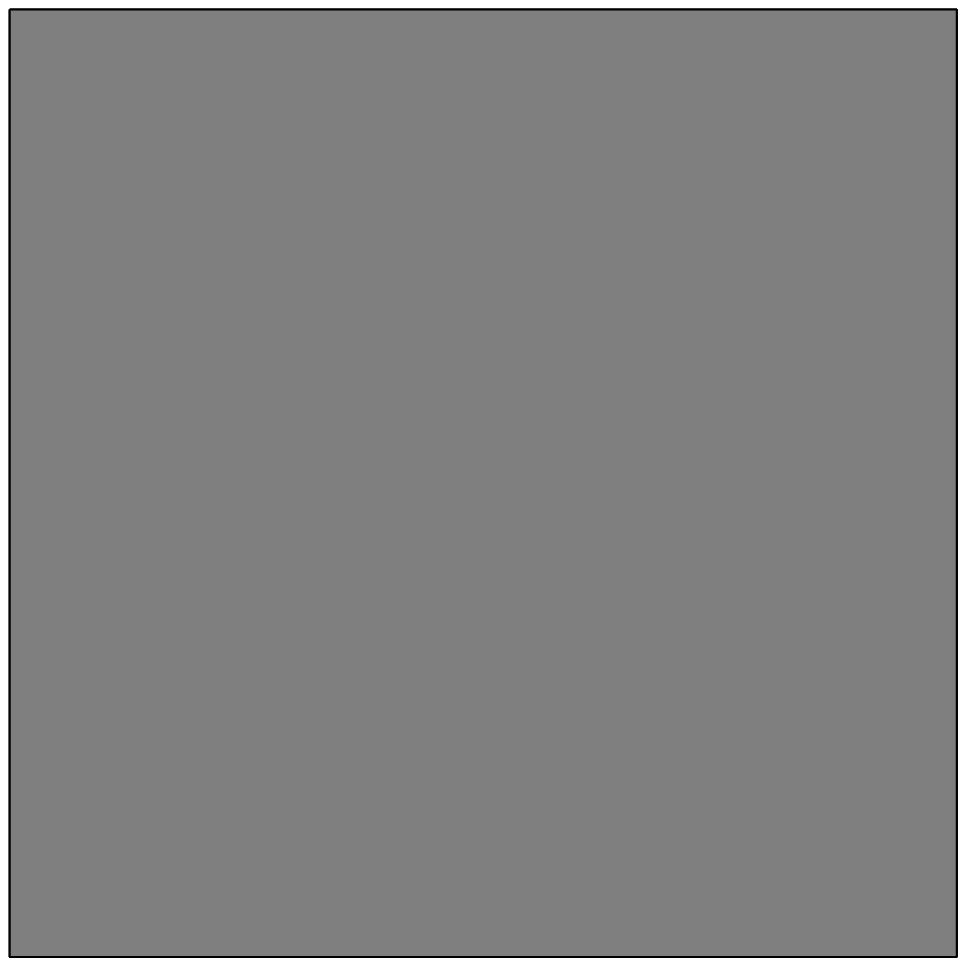} &
      \includegraphics[height=.1200\textheight]{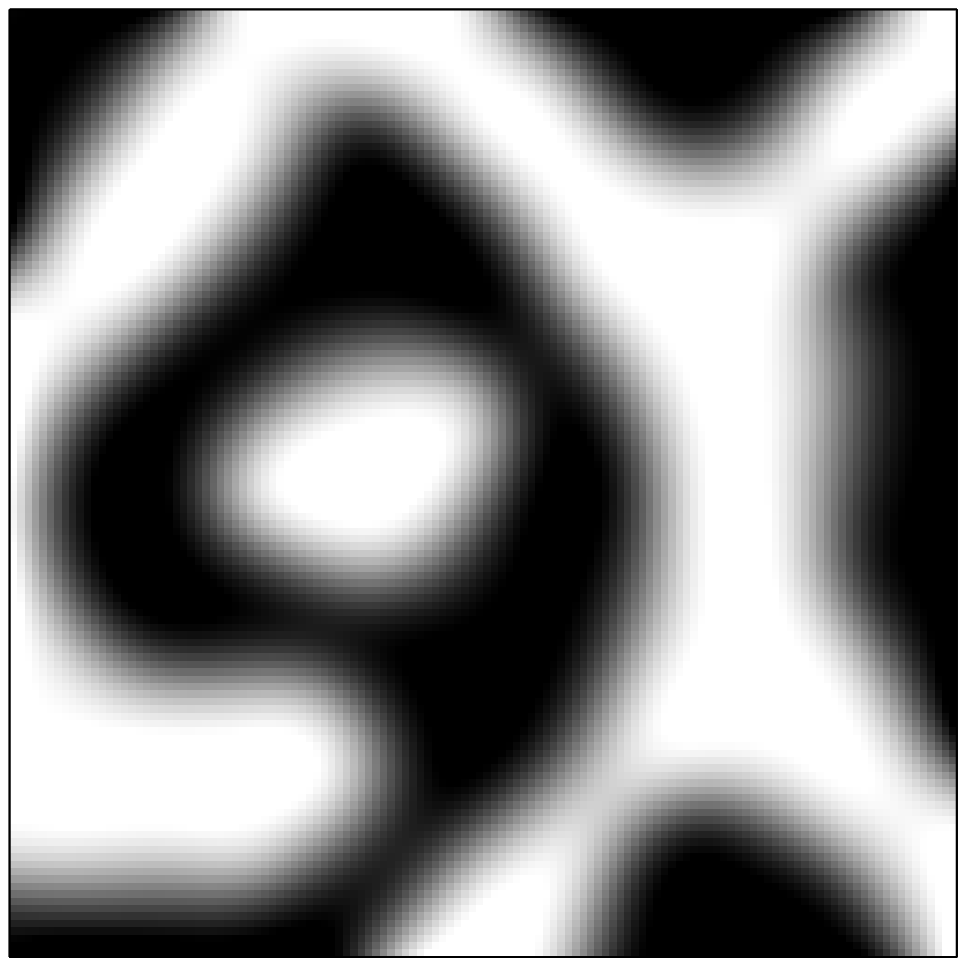} &
      \includegraphics[height=.1200\textheight]{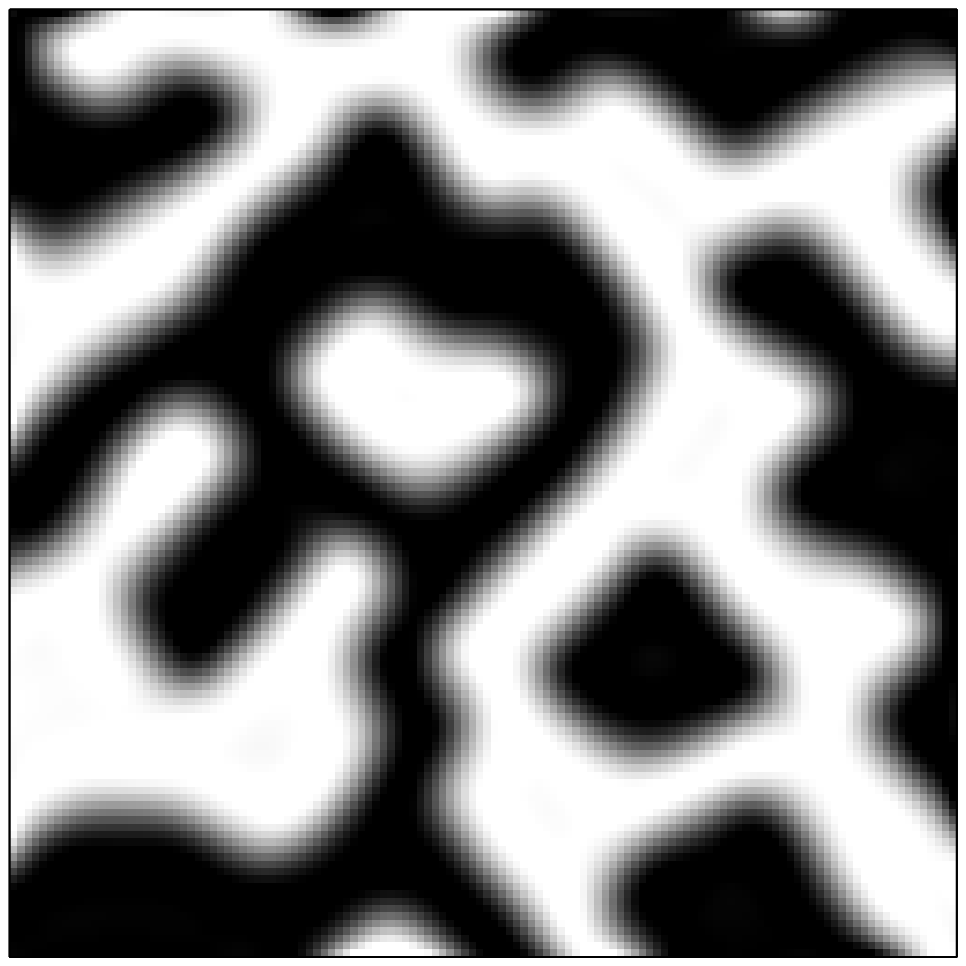} &
      \includegraphics[height=.1200\textheight]{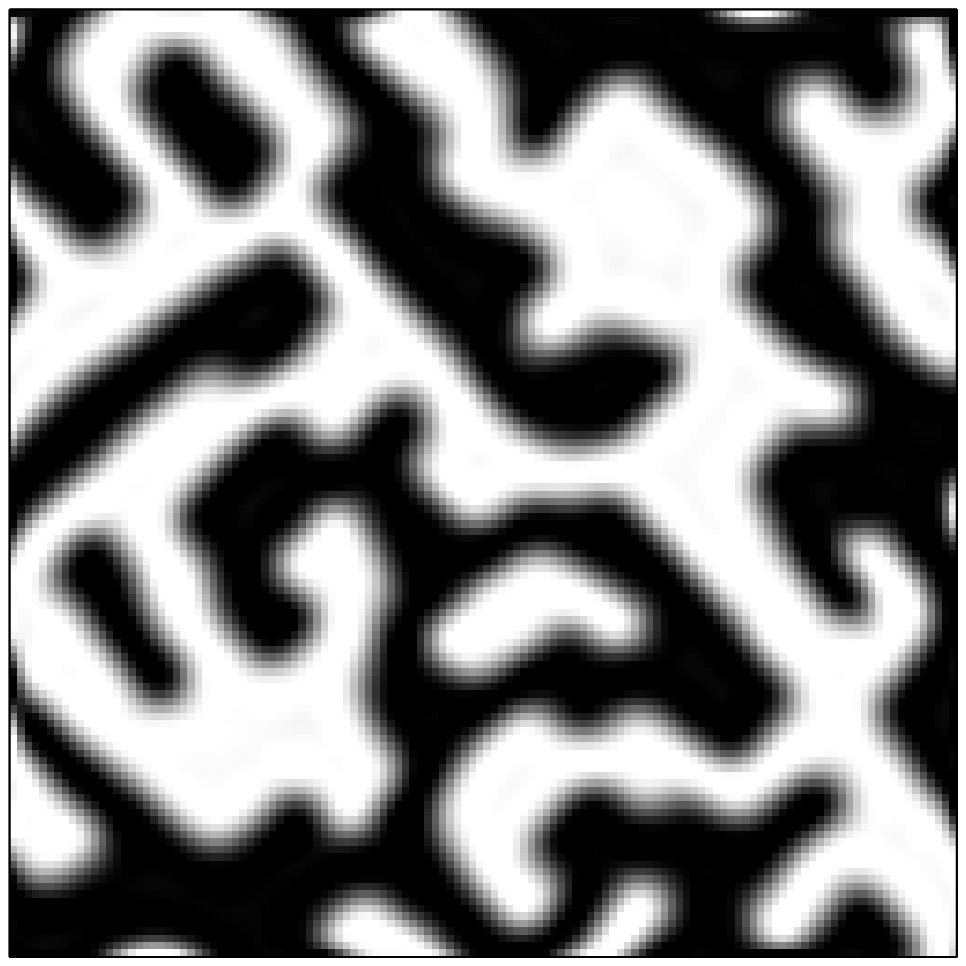} \\[-1ex]
      \rotatebox{90}{\makebox[.1200\textheight][c]{$\beta = 10^{5}$}} &
      \includegraphics[height=.1200\textheight]{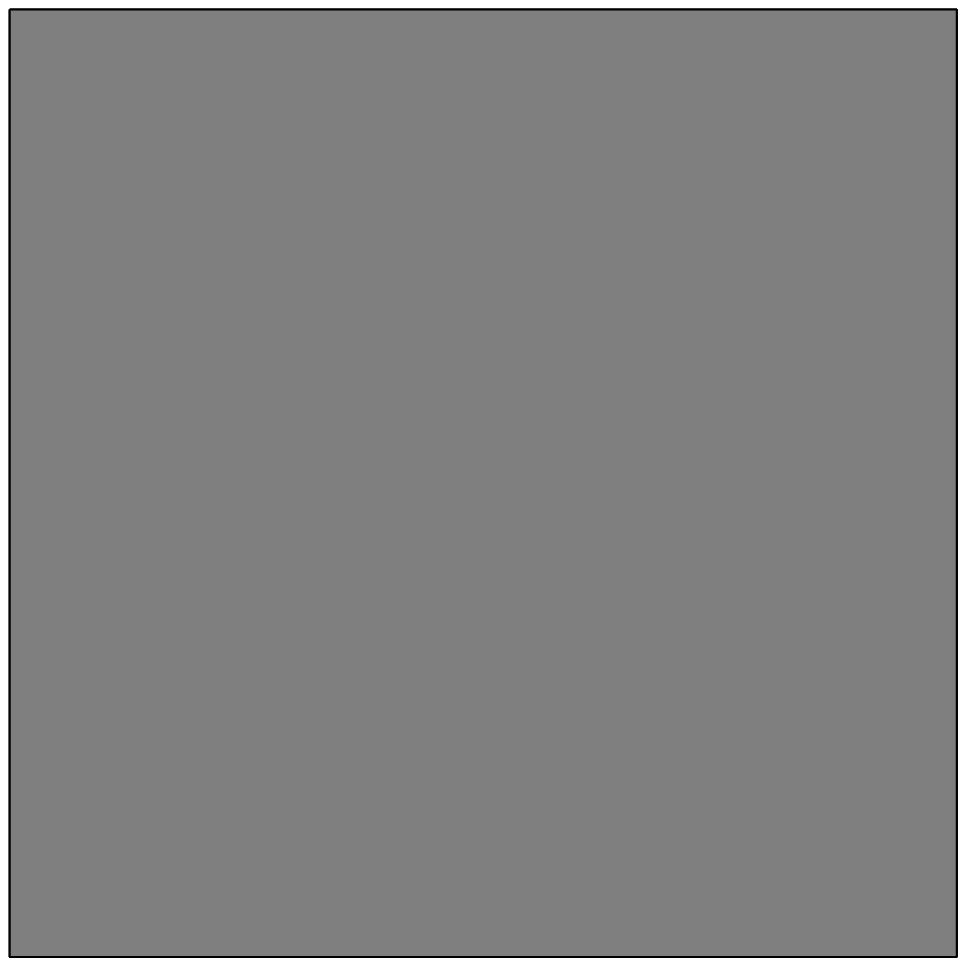} &
      \includegraphics[height=.1200\textheight]{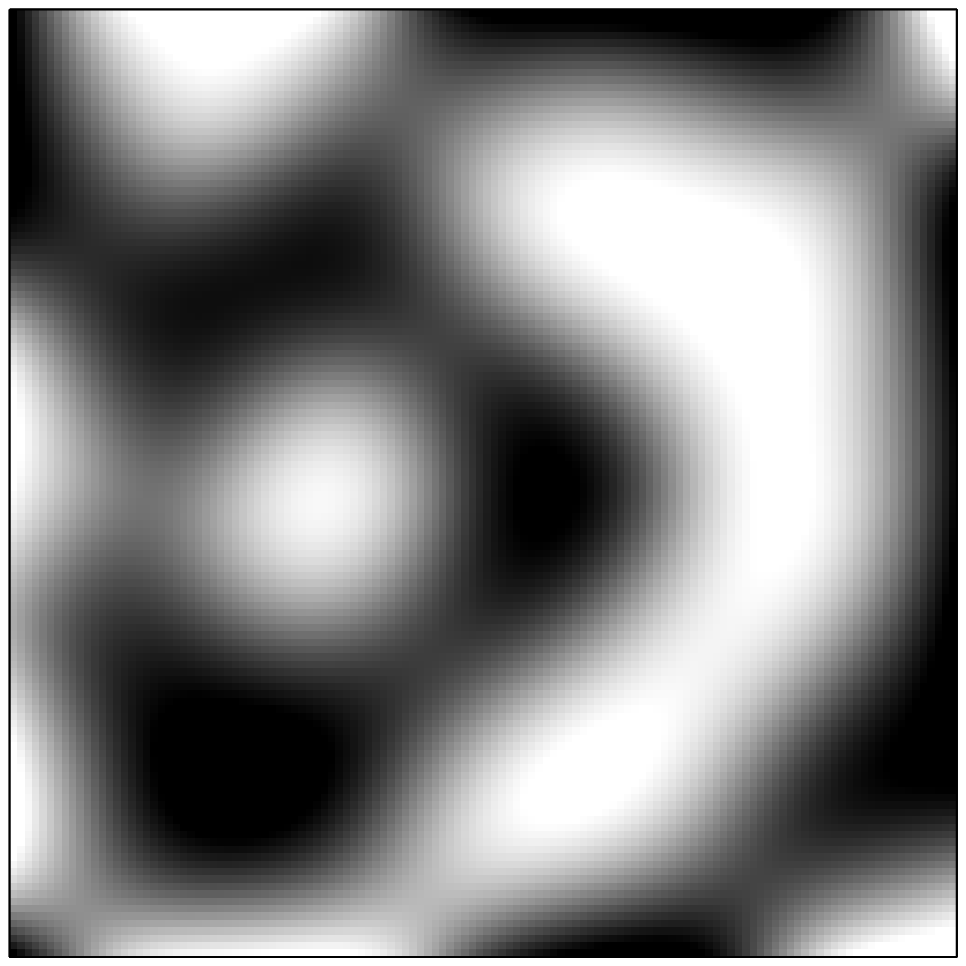} &
      \includegraphics[height=.1200\textheight]{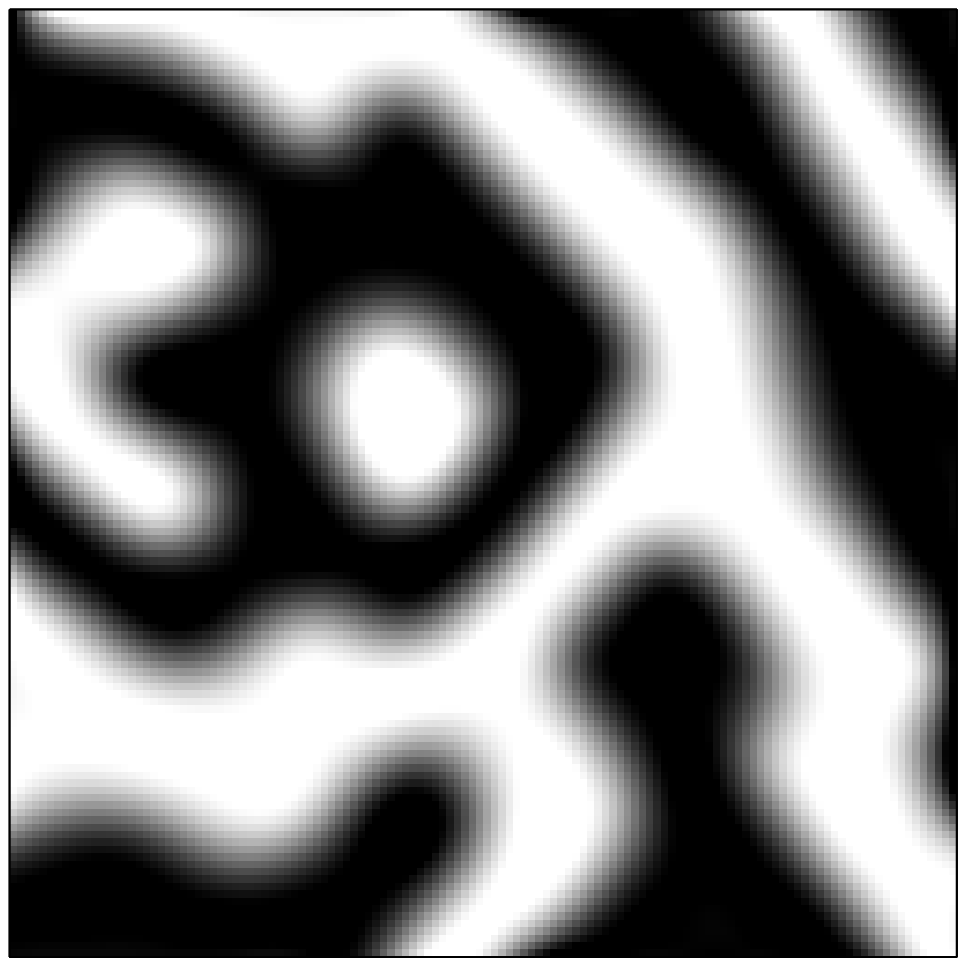} &
      \includegraphics[height=.1200\textheight]{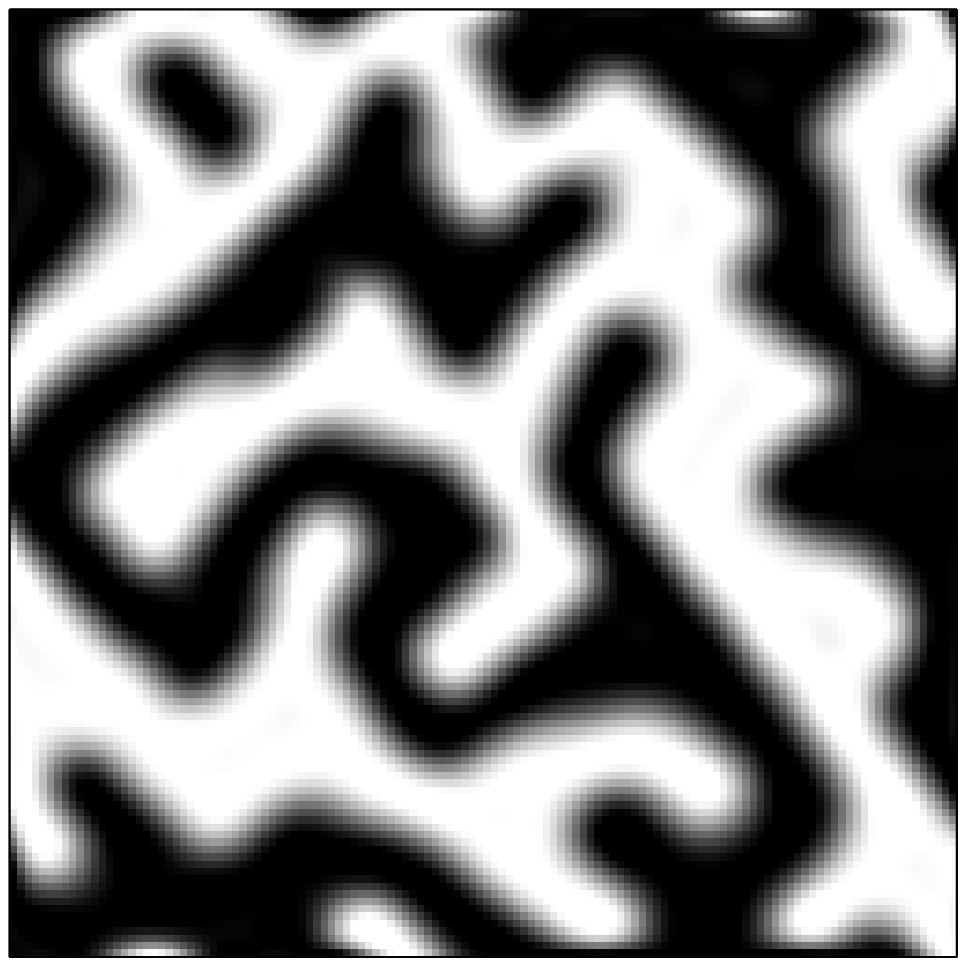} \\[-1ex]
      \rotatebox{90}{\makebox[.1200\textheight][c]{$\beta = 10^{6}$}} &
      \includegraphics[height=.1200\textheight]{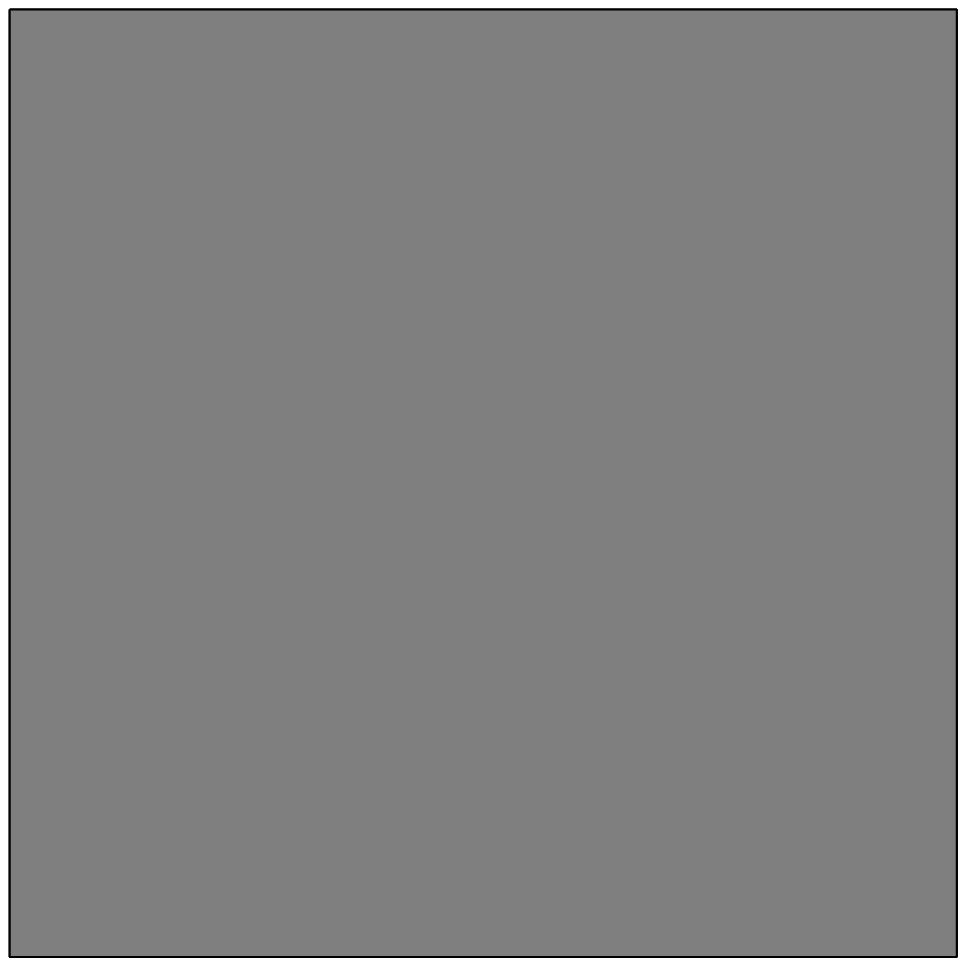} &
      \includegraphics[height=.1200\textheight]{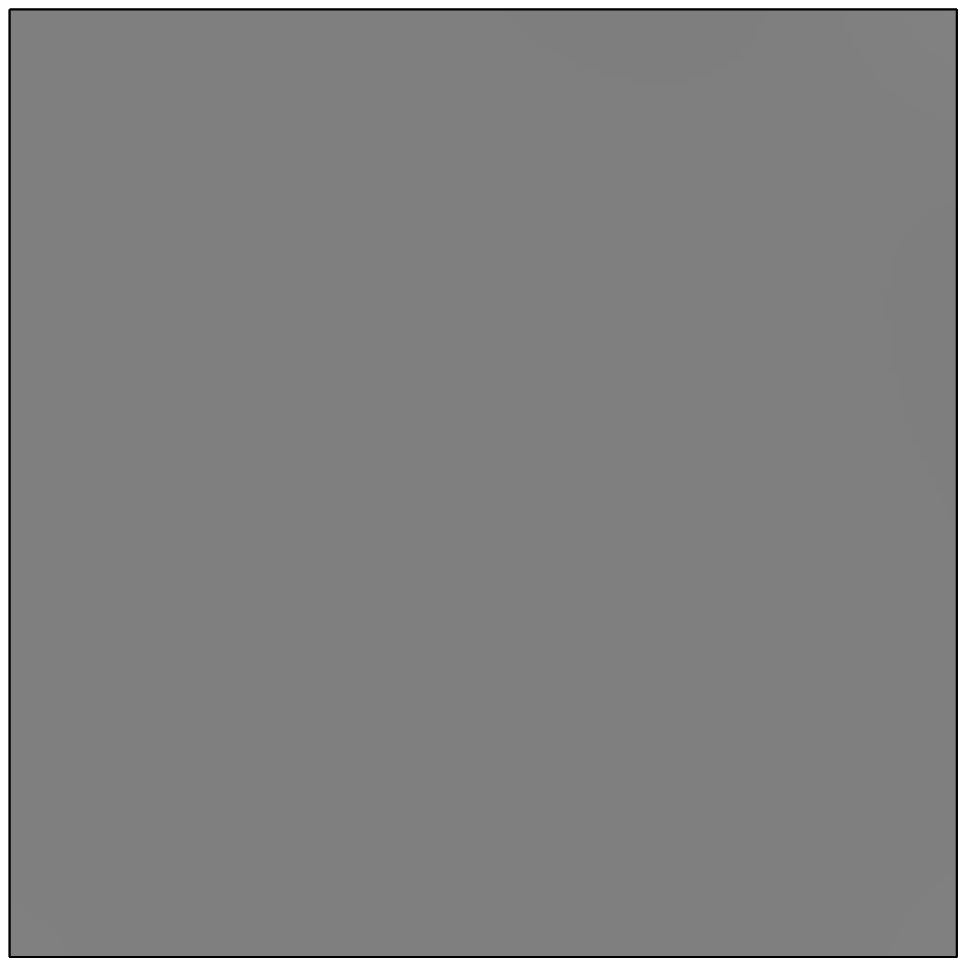} &
      \includegraphics[height=.1200\textheight]{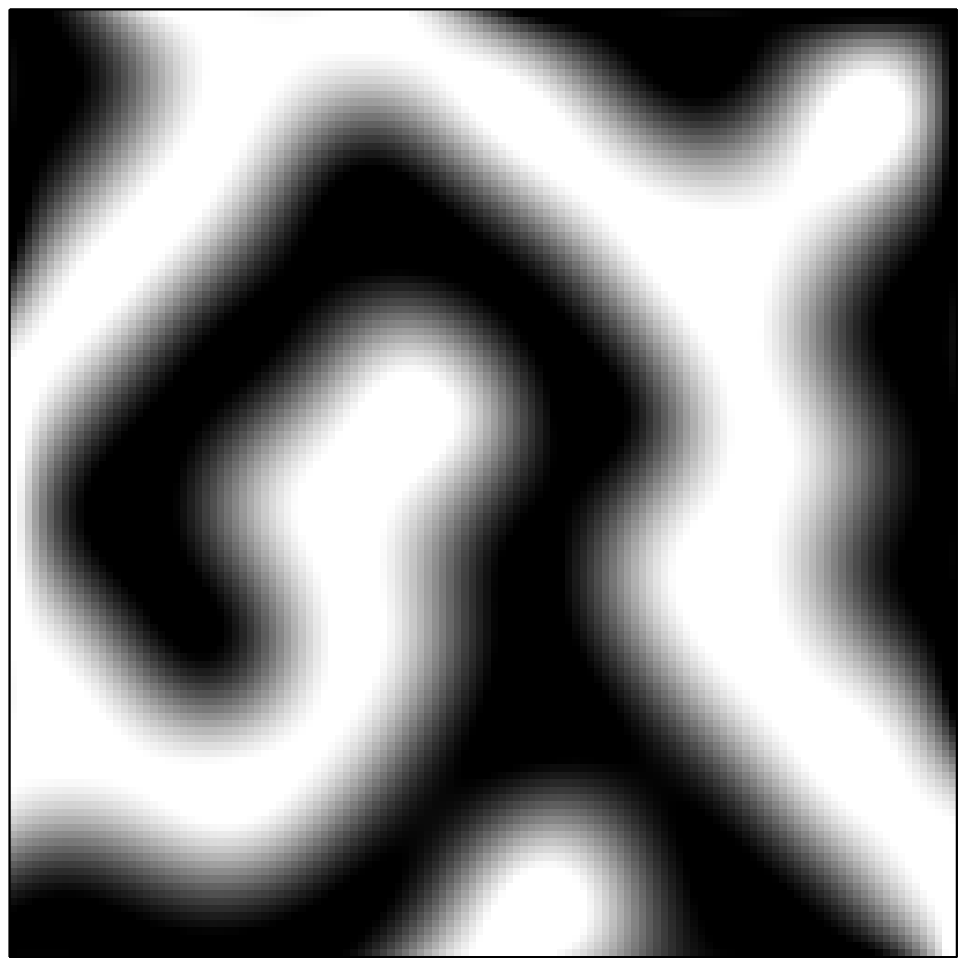} &
      \includegraphics[height=.1200\textheight]{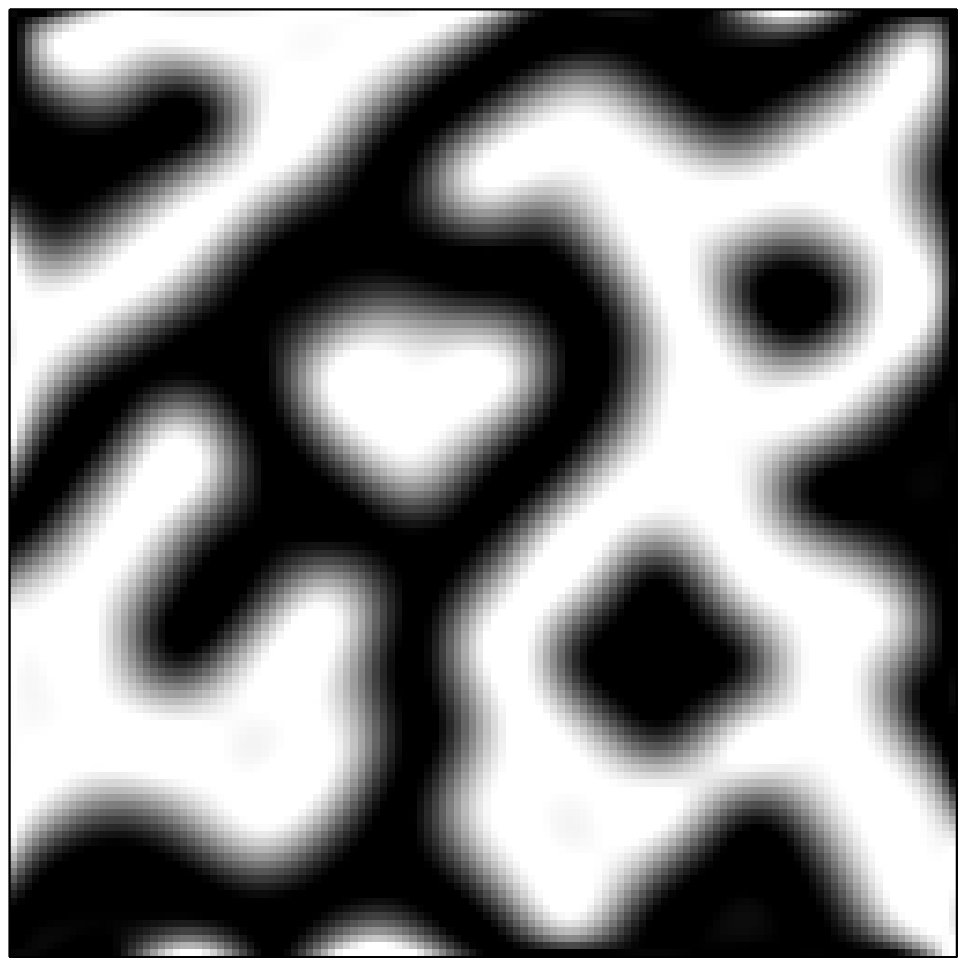} \\[-1ex]
    \end{tabular}
    \caption{Simulations for a 2D generalised elastic net as a model for maps of visual field VF, ocular dominance OD and orientation OR. This panel shows the OD maps obtained for a range of values of $\beta$ and stencil order $p$ (forward-difference family with normalisation to unit power), at a value $\sigma = 0.05$ small enough that the OD and OR maps have arisen, with all other parameters and the initial conditions being the same. The stripes are narrower for low $\beta$ and high $p$. Simulation details: square net with $M \times M$ centroids ($M = 128$); nonperiodic b.c.; annealing schedule: from $\sigma = 0.2$ down to $\sigma = 0.05$ in a geometric progression with $40$ iterations. The training set consisted of a regular grid of $(10,10,2,12)$ values in $[0,1] \times [0,1] \times [-l,l] \times [-\frac{\pi}{2},\frac{\pi}{2}]$ for the 2D visual field, OD ($l = 0.07$) and OR, respectively; the periodic OR variable is implemented as a ring of radius $r = 0.2$ in 2D Cartesian space (see \citealp{CarreirGoodhil03b} for details). The relative stripe widths of the OD and OR maps depend on the correlation parameters of OD $l$ and OR $r$ and has been studied elsewhere \citep{GoodhilCimpon00a}.}
    \label{f:simul2D:OD}
\end{FPfigure}
            
\begin{figure}
  \begin{center}
    \begin{tabular}{@{}c@{\hspace{0.5cm}}c@{\hspace{0.5cm}}c@{\hspace{0.5cm}}c@{\hspace{0.5cm}}c@{}}
      & $p = 1$ & $p = 2$ & $p = 3$ & $p = 4$ \\
      \rotatebox{90}{\makebox[.1200\textheight][c]{$\beta = 10^{-1}$}} &
      \includegraphics[height=.1200\textheight]{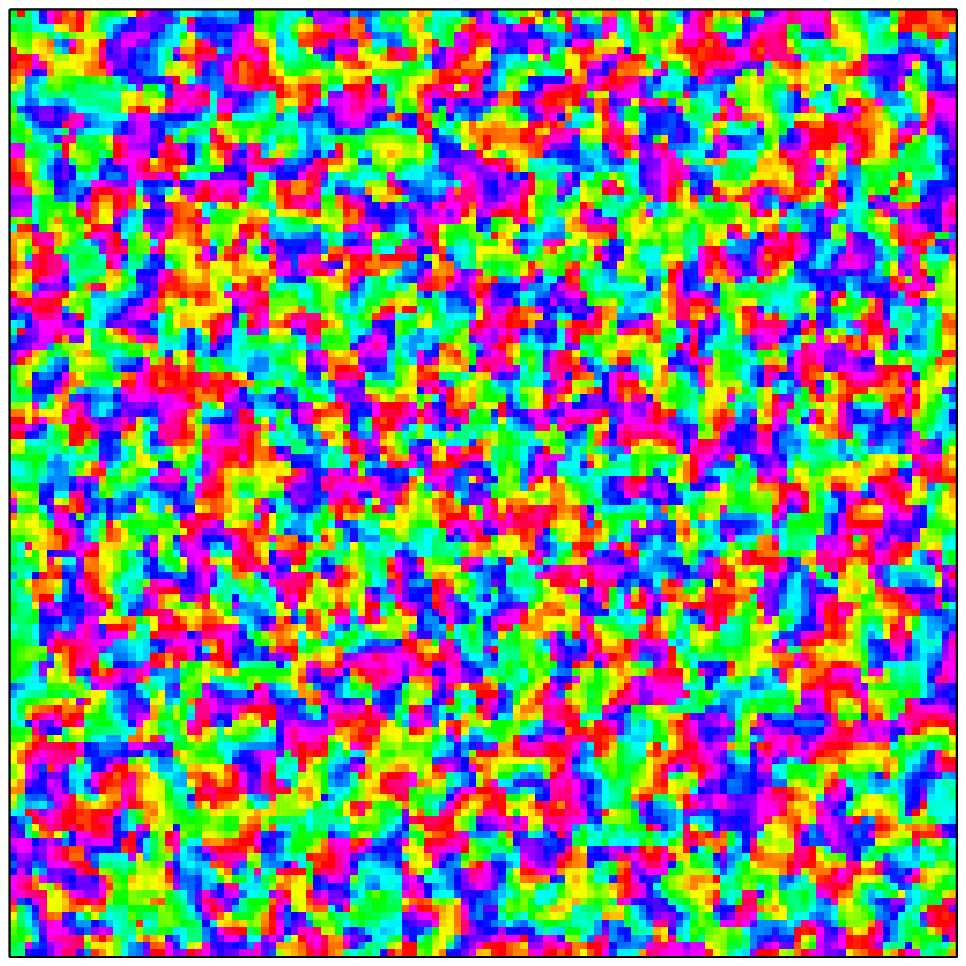} &
      \includegraphics[height=.1200\textheight]{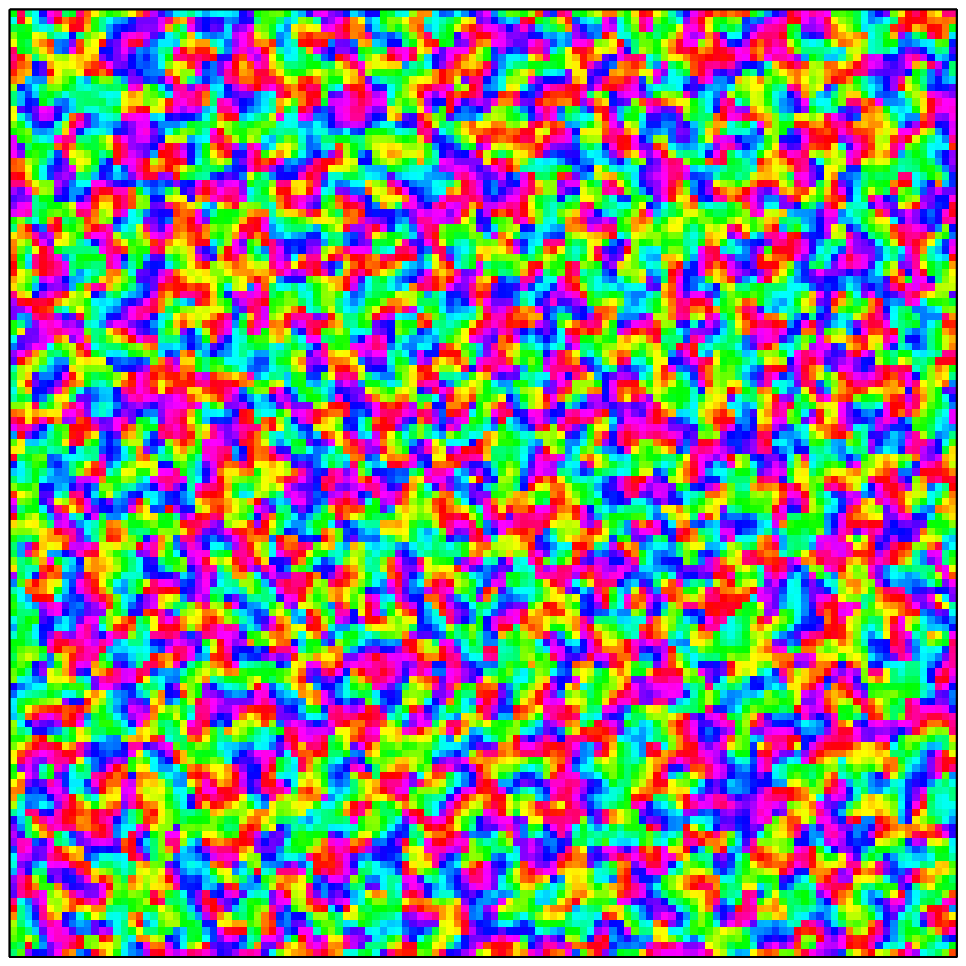} &
      \includegraphics[height=.1200\textheight]{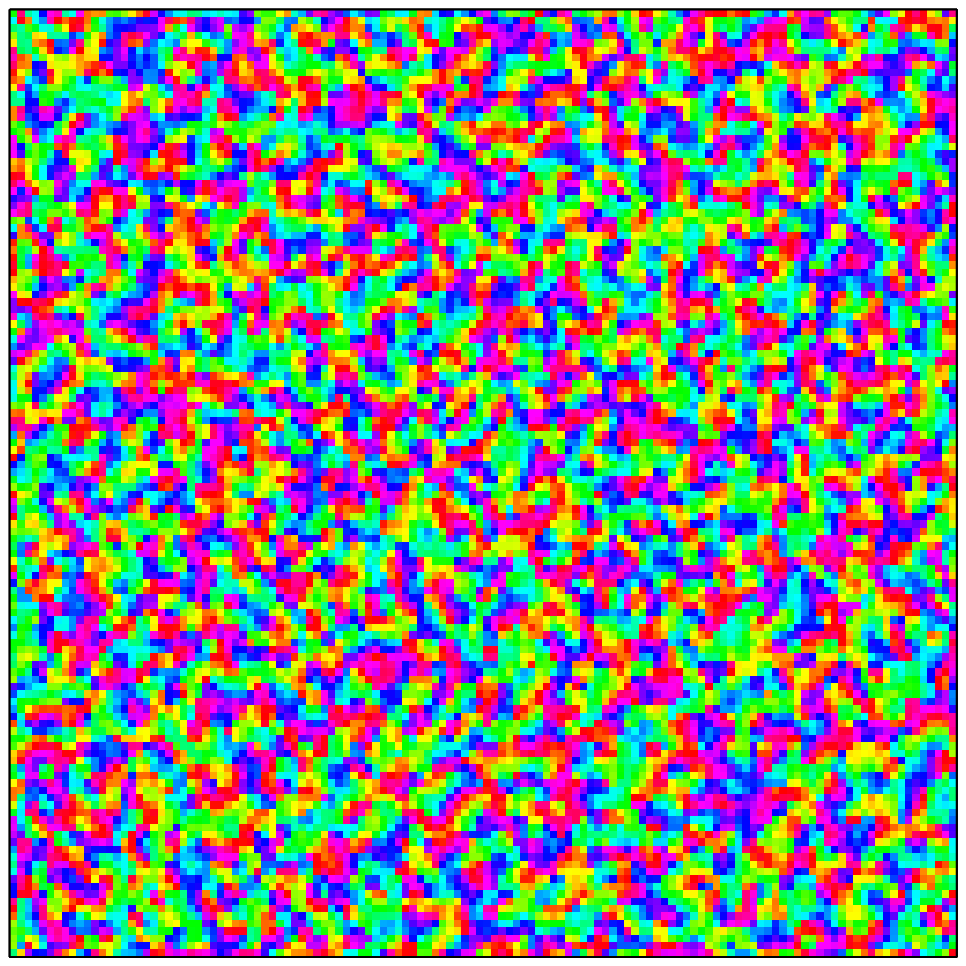} &
      \includegraphics[height=.1200\textheight]{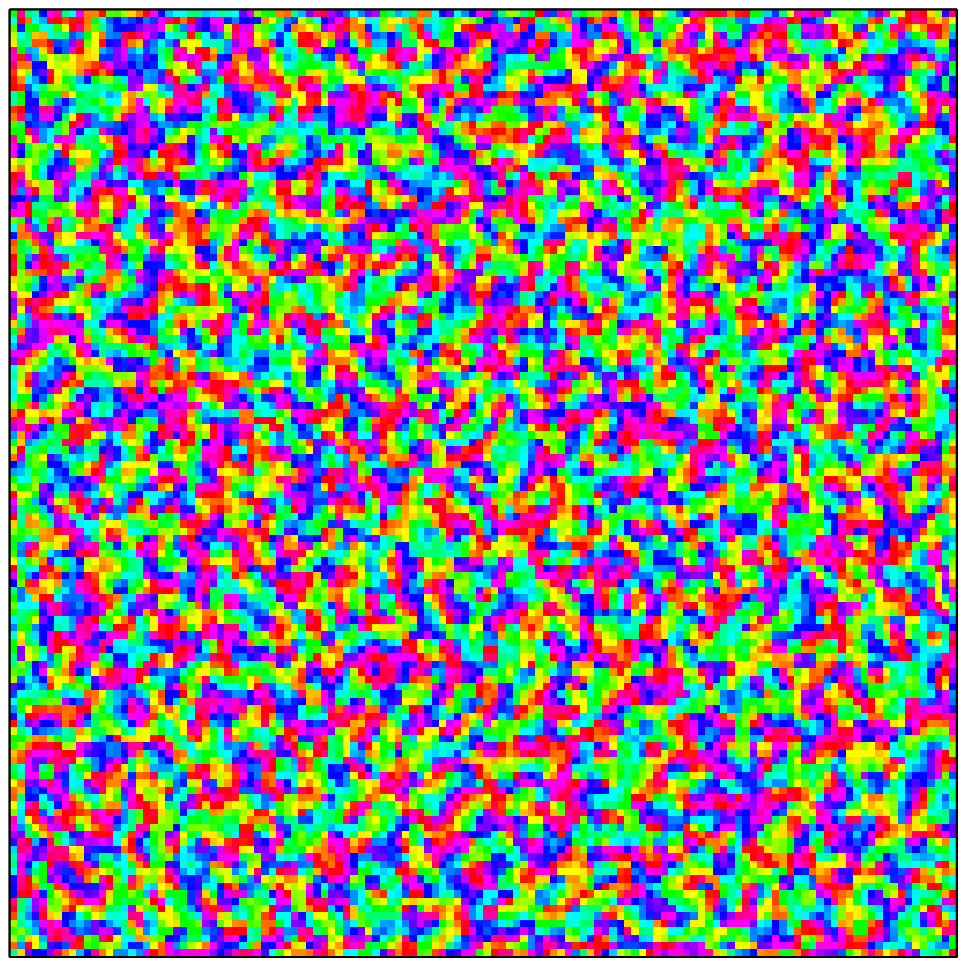} \\[-1ex]
      \rotatebox{90}{\makebox[.1200\textheight][c]{$\beta = 10^{0}$}} &
      \includegraphics[height=.1200\textheight]{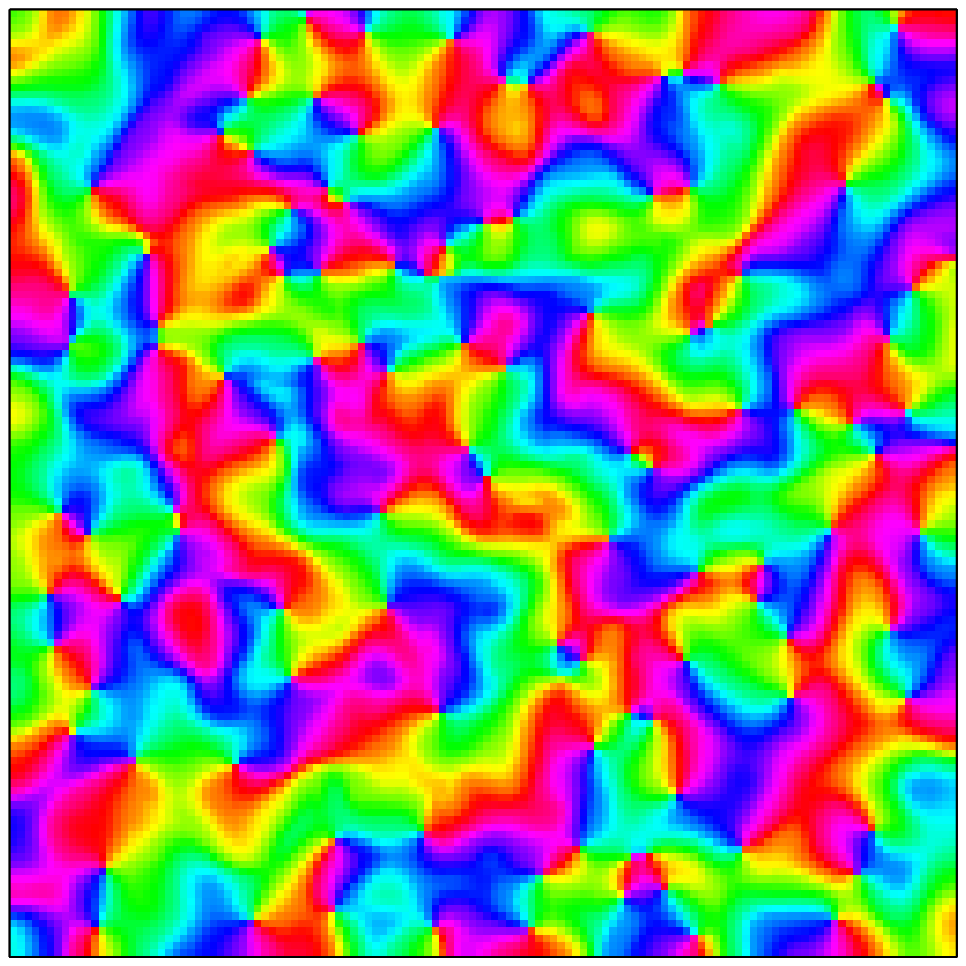} &
      \includegraphics[height=.1200\textheight]{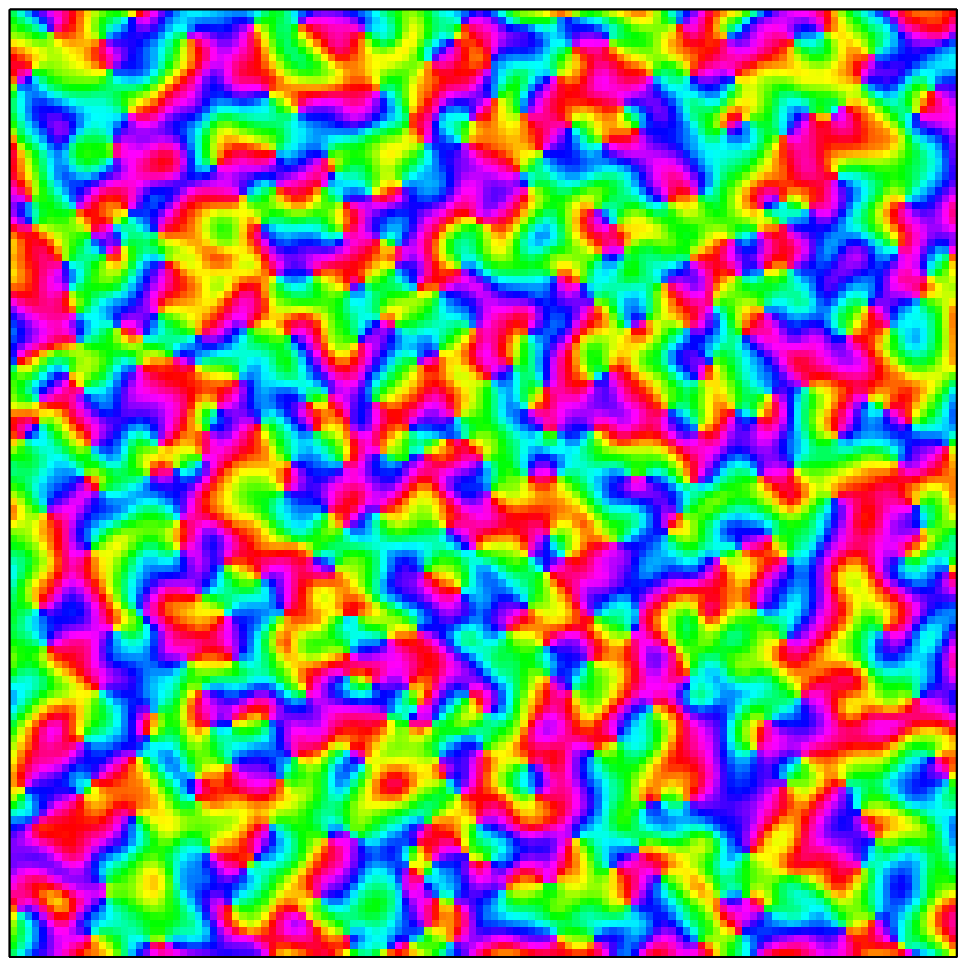} &
      \includegraphics[height=.1200\textheight]{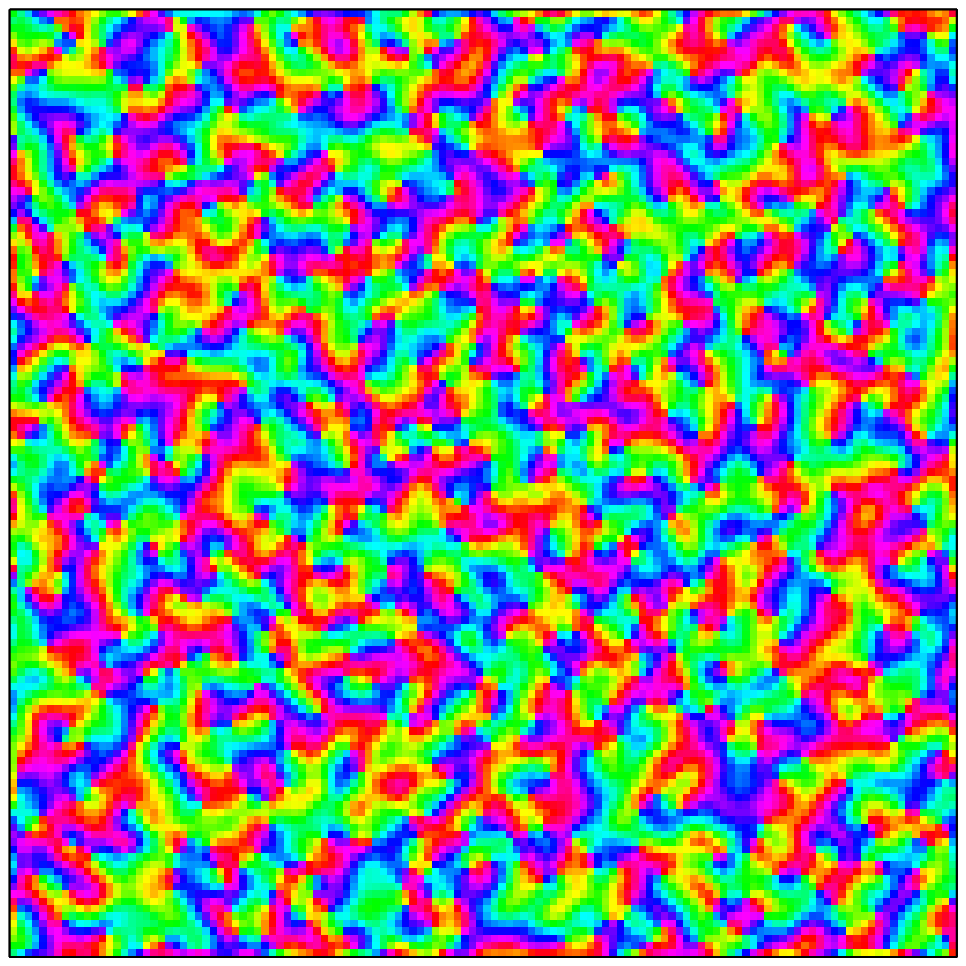} &
      \includegraphics[height=.1200\textheight]{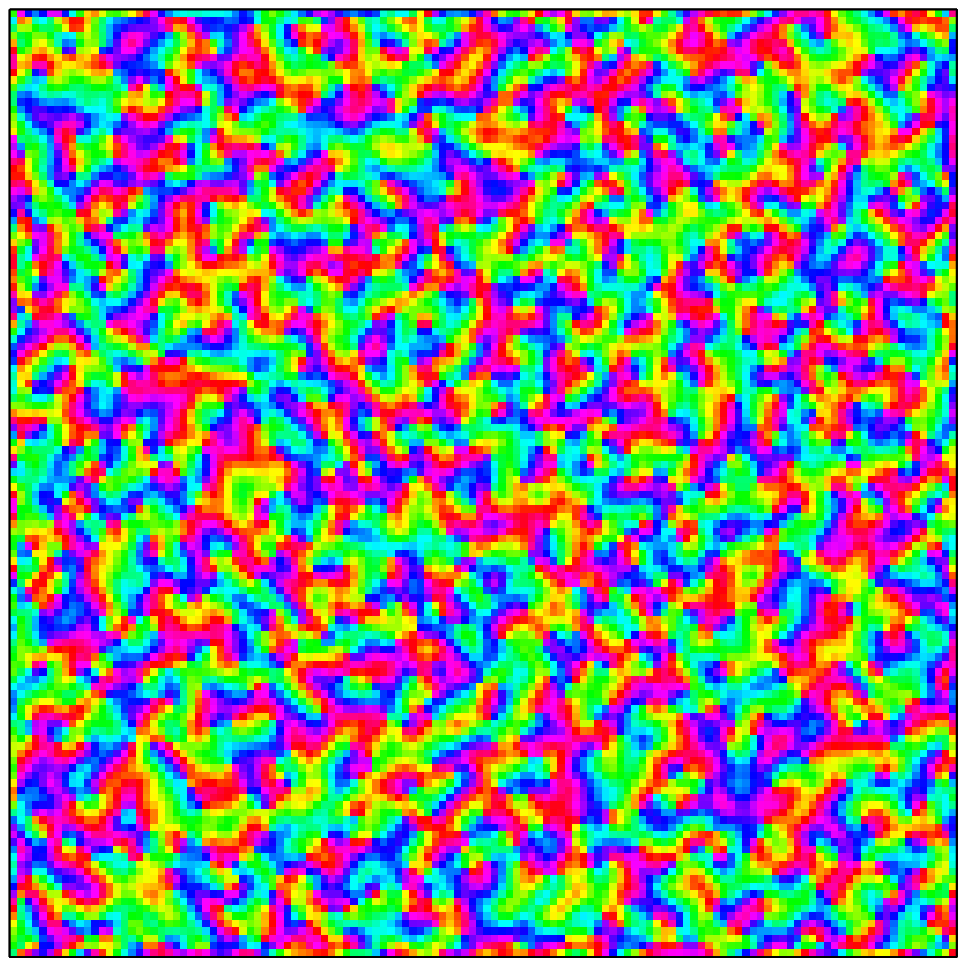} \\[-1ex]
      \rotatebox{90}{\makebox[.1200\textheight][c]{$\beta = 10^{1}$}} &
      \includegraphics[height=.1200\textheight]{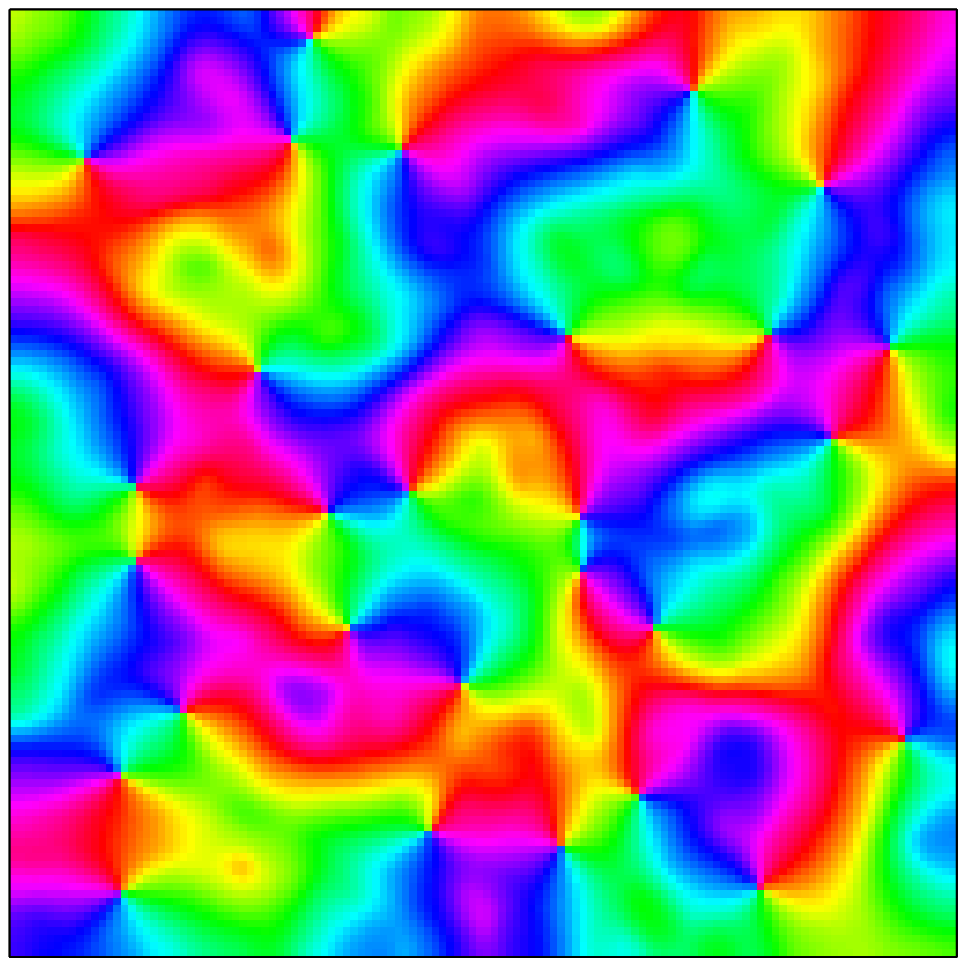} &
      \includegraphics[height=.1200\textheight]{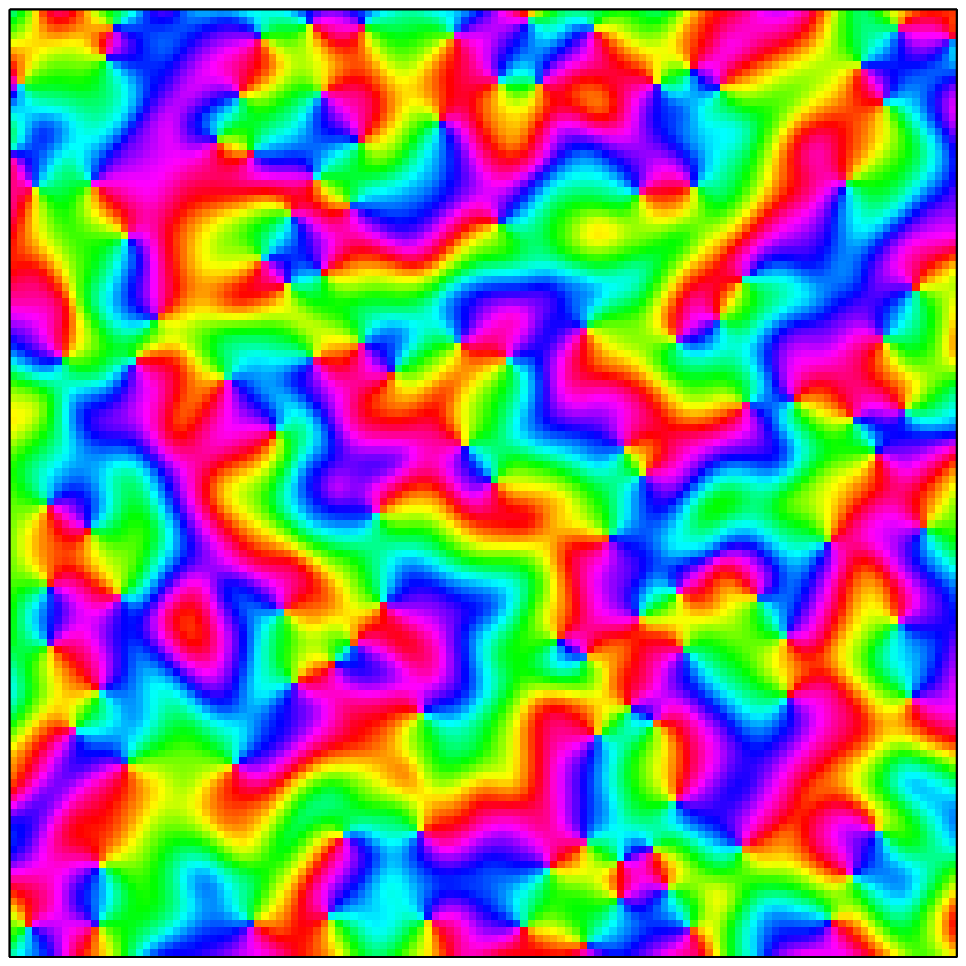} &
      \includegraphics[height=.1200\textheight]{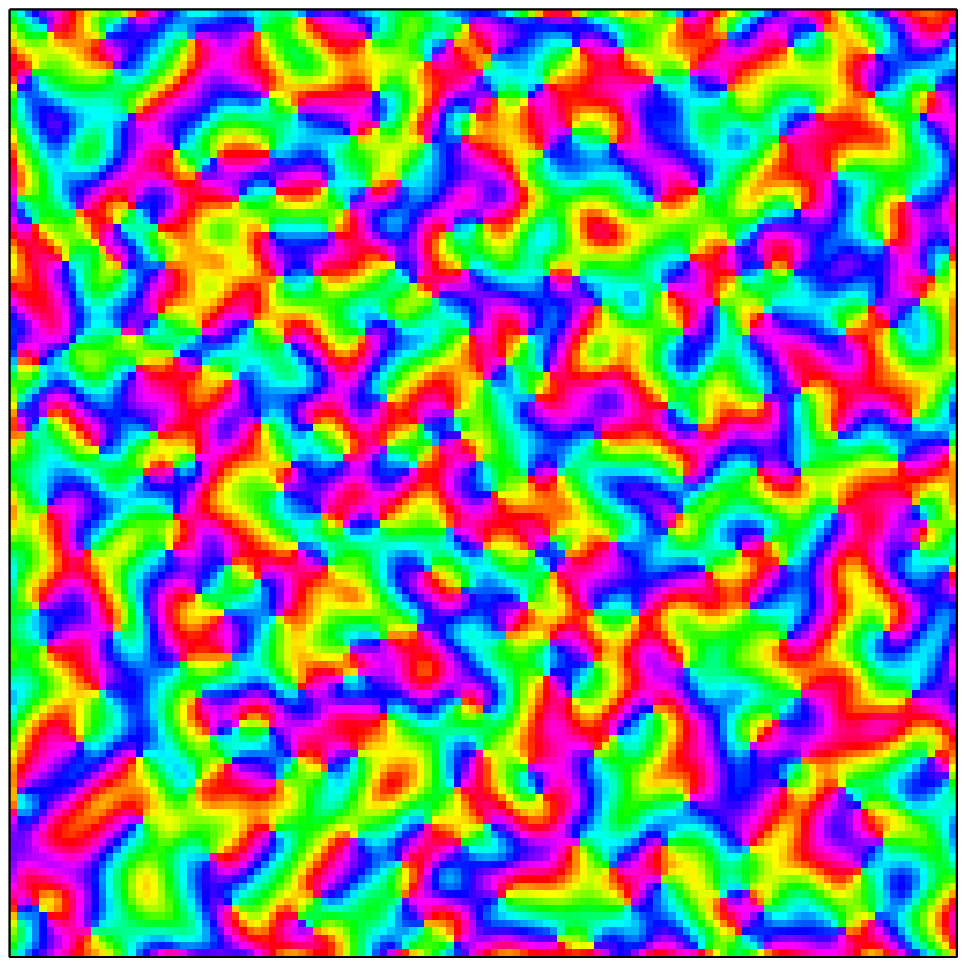} &
      \includegraphics[height=.1200\textheight]{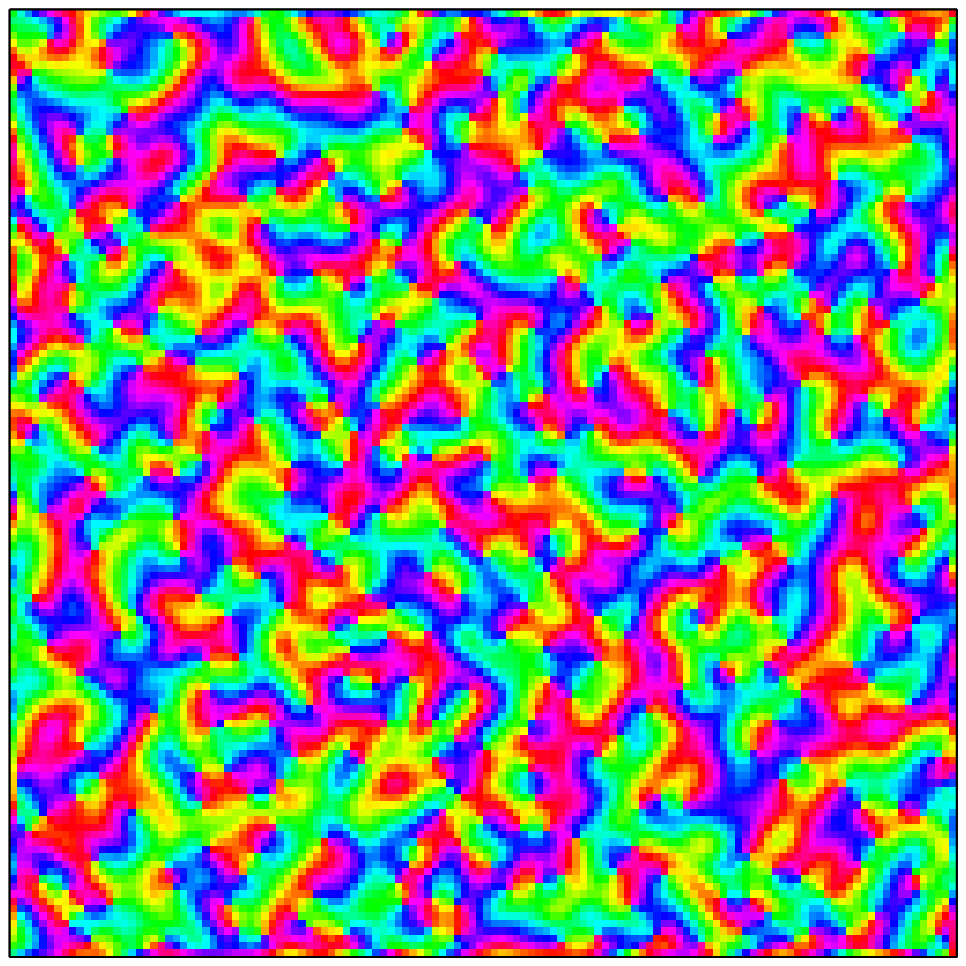} \\[-1ex]
      \rotatebox{90}{\makebox[.1200\textheight][c]{$\beta = 10^{2}$}} &
      \includegraphics[height=.1200\textheight]{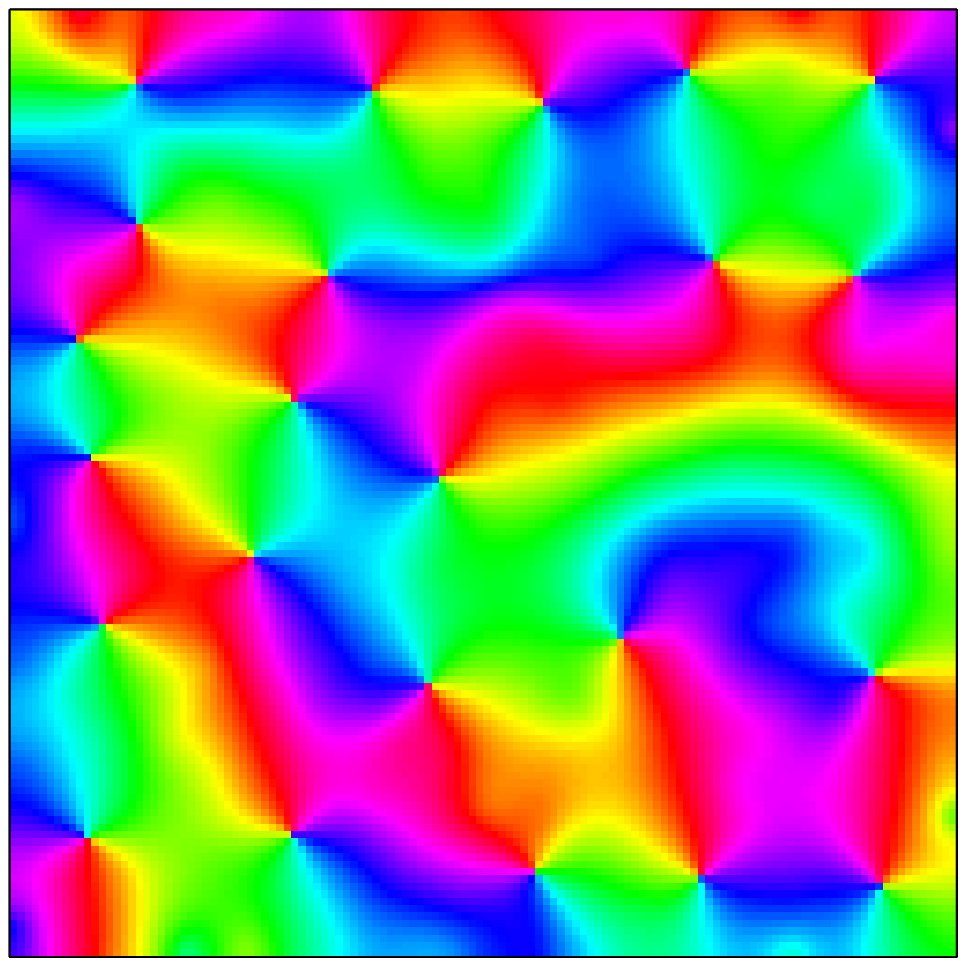} &
      \includegraphics[height=.1200\textheight]{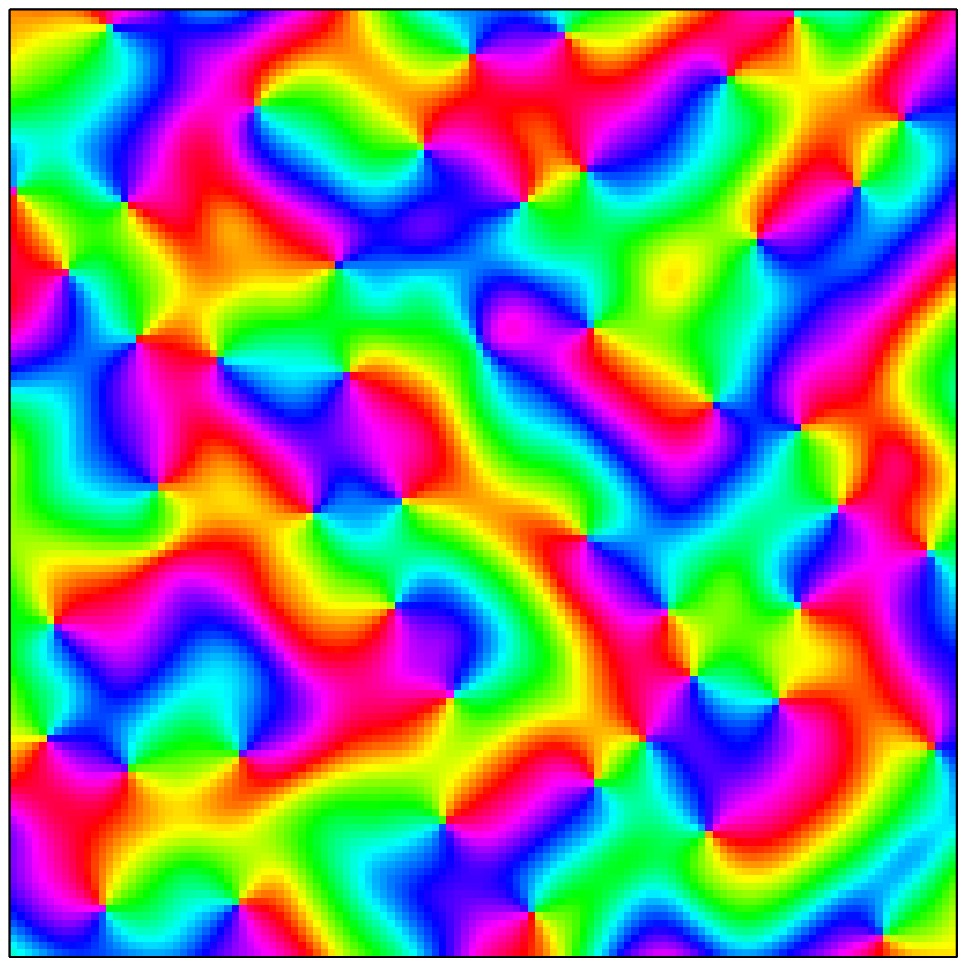} &
      \includegraphics[height=.1200\textheight]{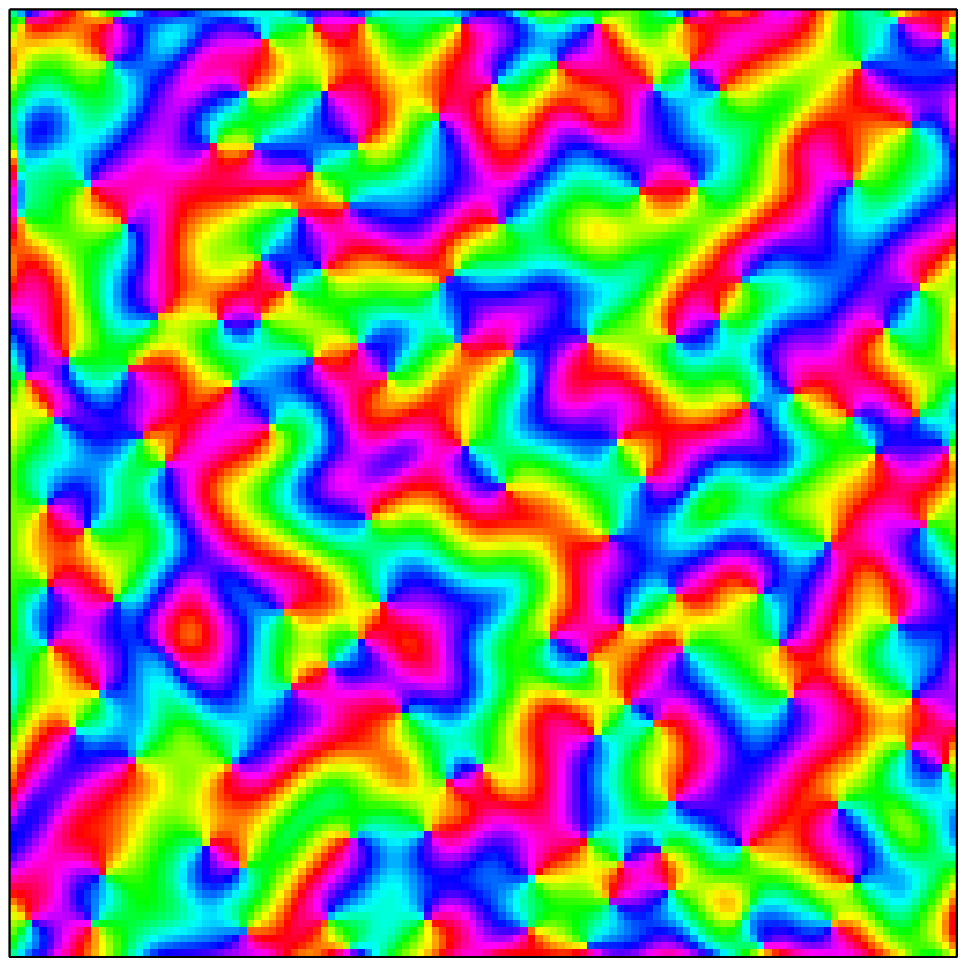} &
      \includegraphics[height=.1200\textheight]{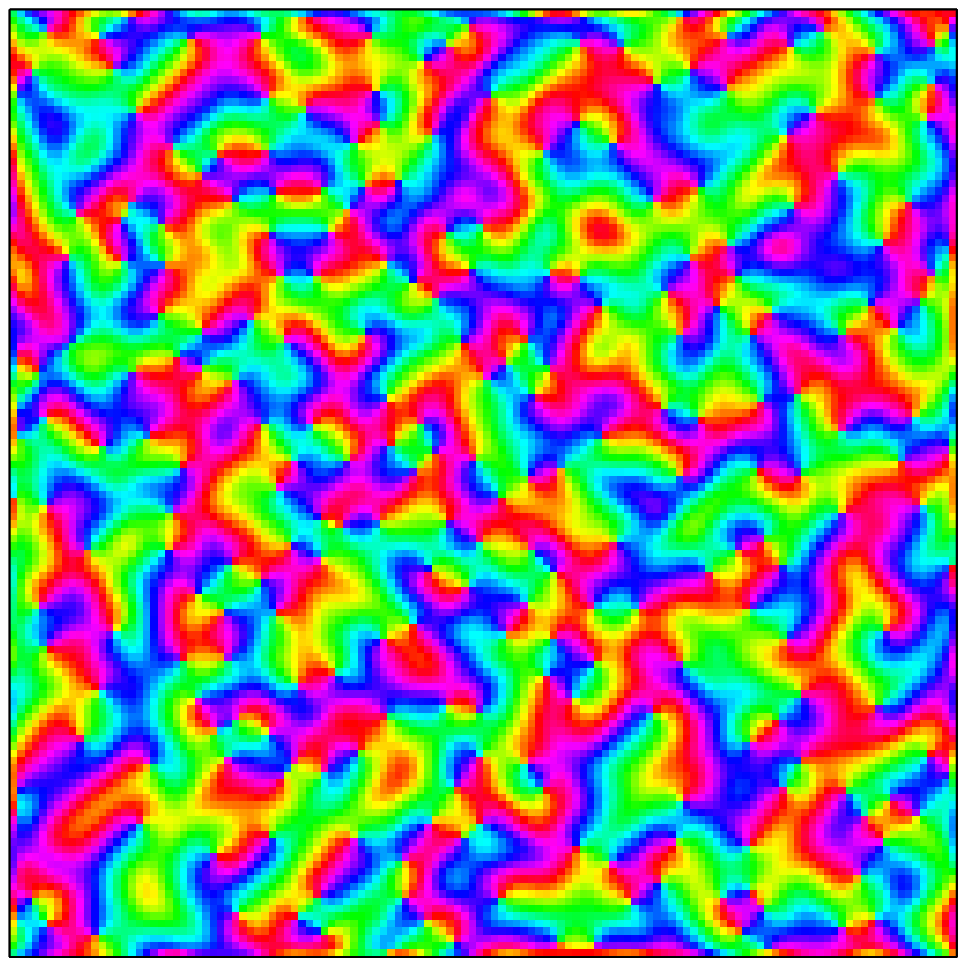} \\[-1ex]
      \rotatebox{90}{\makebox[.1200\textheight][c]{$\beta = 10^{3}$}} &
      \includegraphics[height=.1200\textheight]{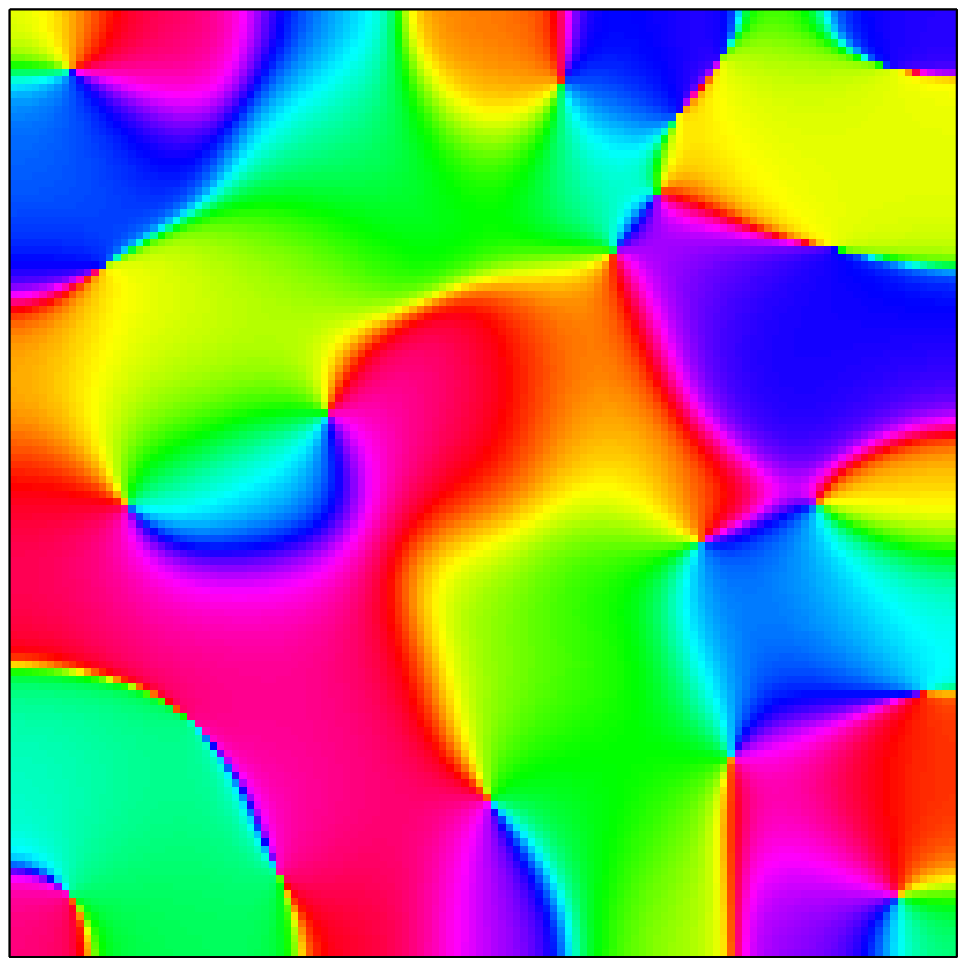} &
      \includegraphics[height=.1200\textheight]{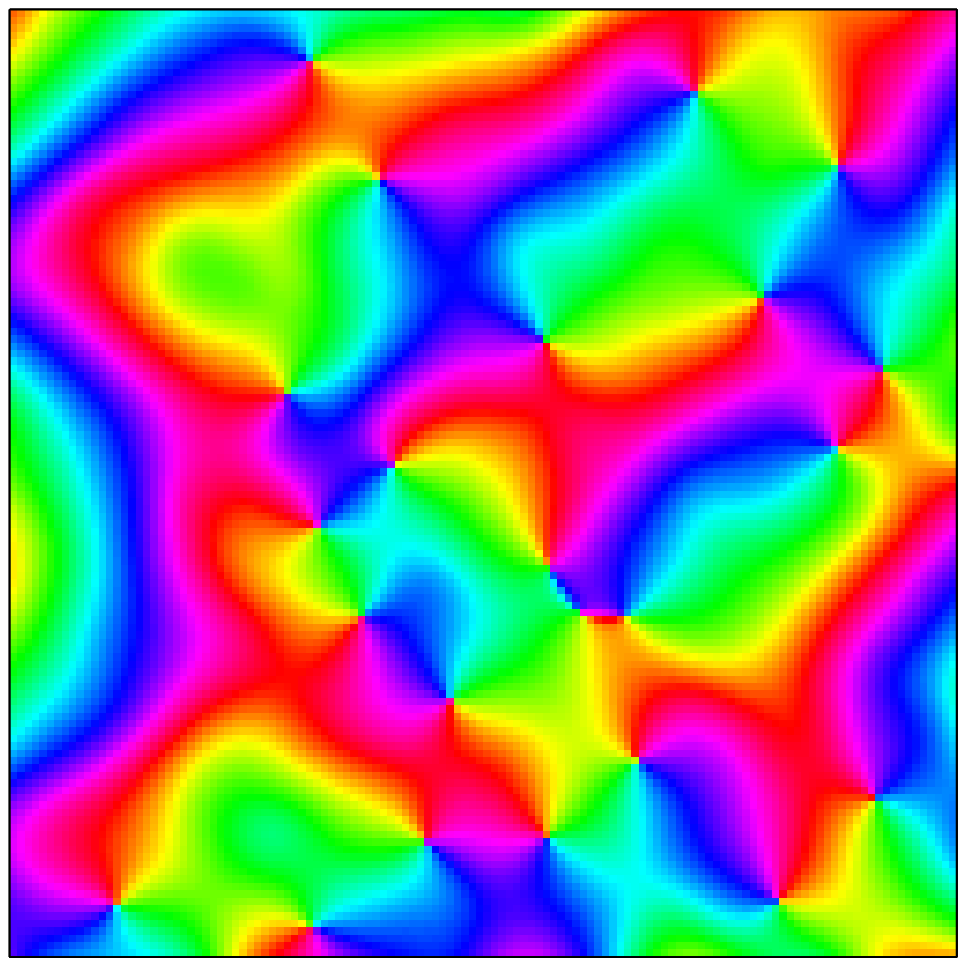} &
      \includegraphics[height=.1200\textheight]{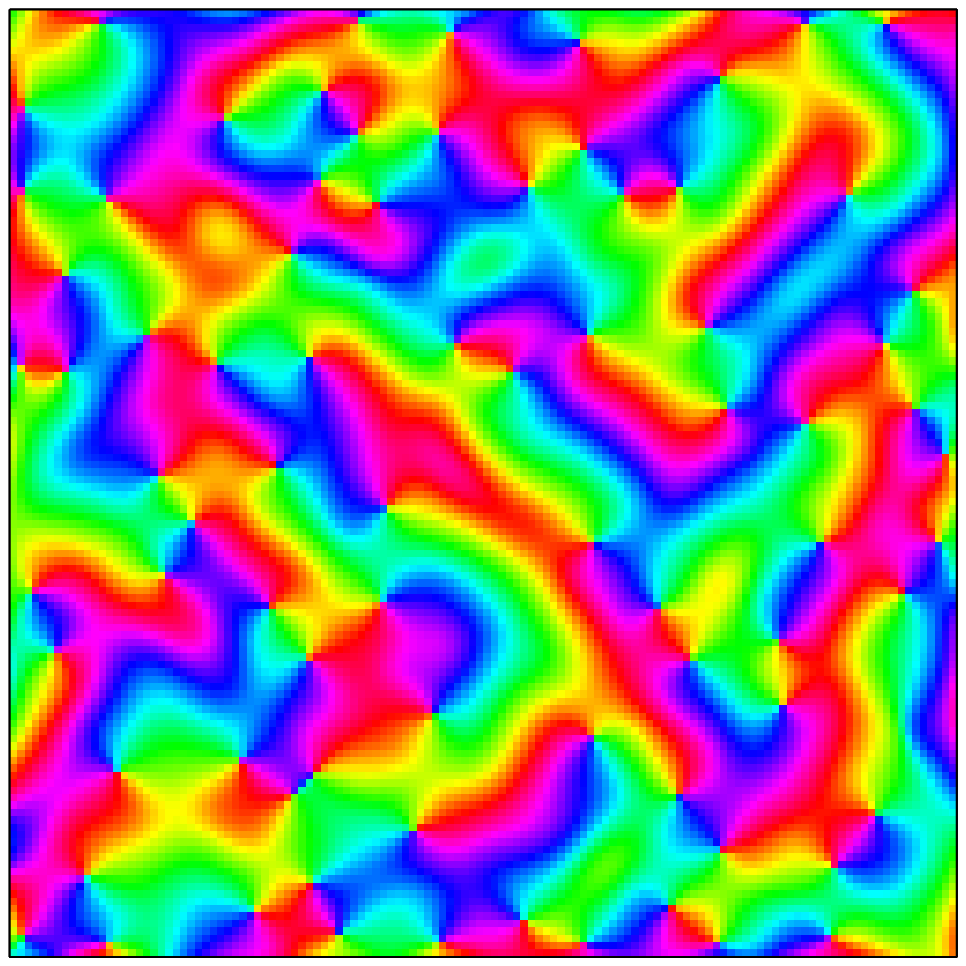} &
      \includegraphics[height=.1200\textheight]{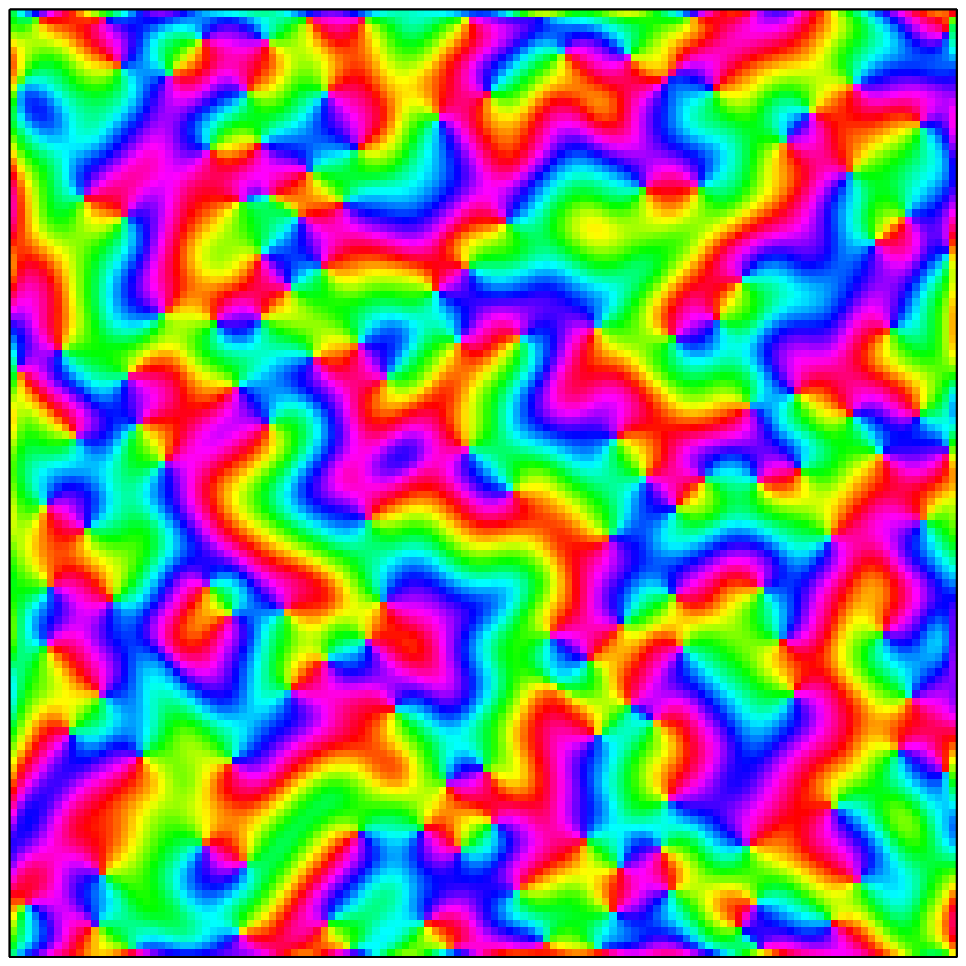} \\[-1ex]
      \rotatebox{90}{\makebox[.1200\textheight][c]{$\beta = 10^{4}$}} &
      \includegraphics[height=.1200\textheight]{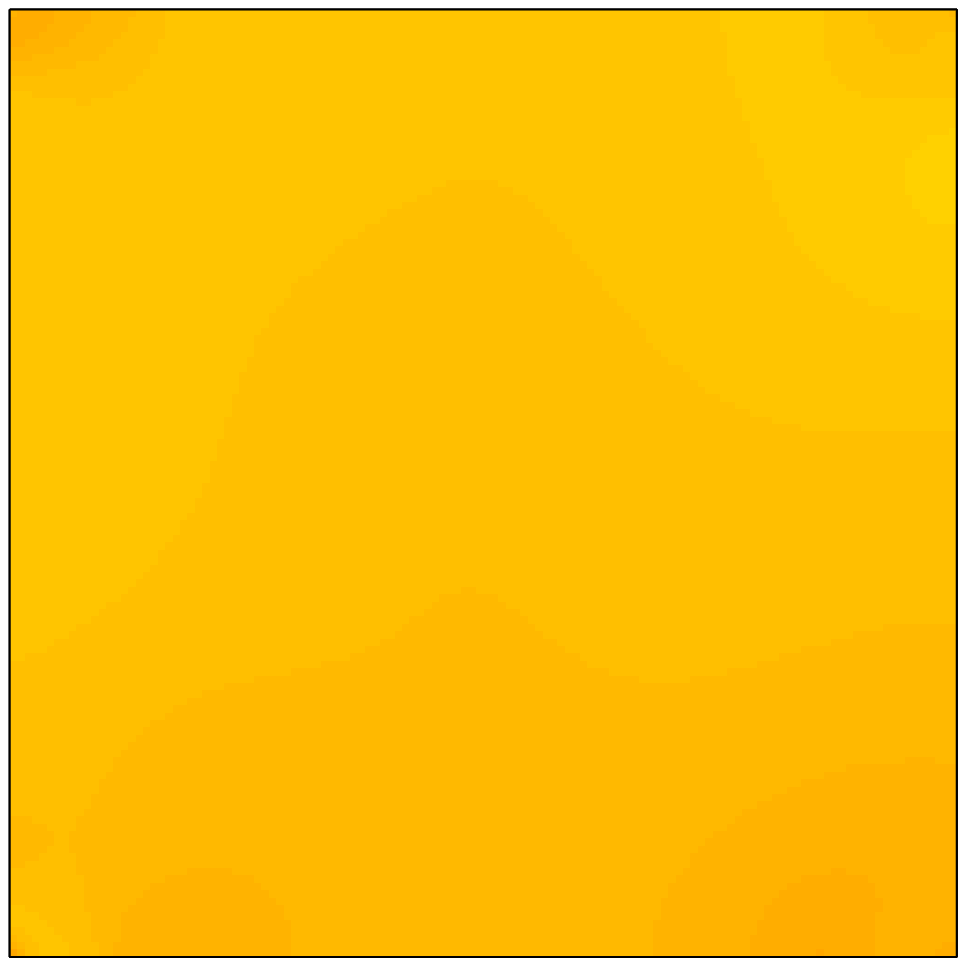} &
      \includegraphics[height=.1200\textheight]{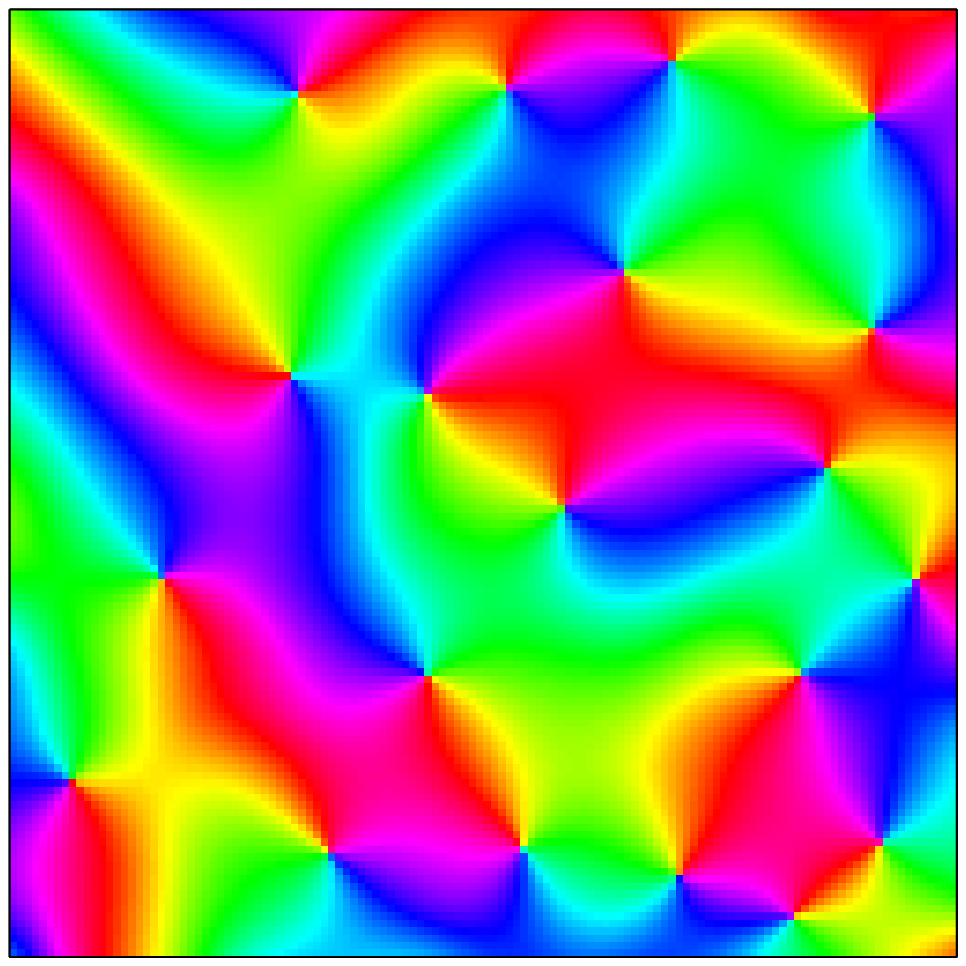} &
      \includegraphics[height=.1200\textheight]{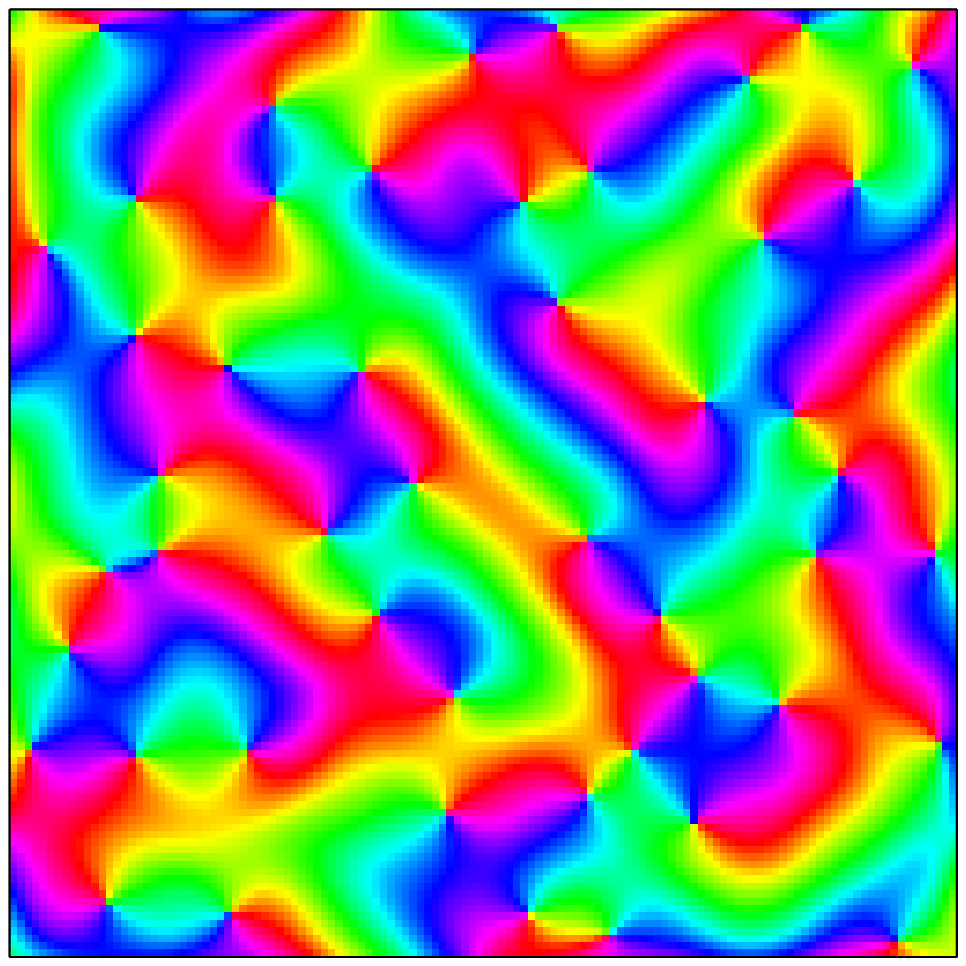} &
      \includegraphics[height=.1200\textheight]{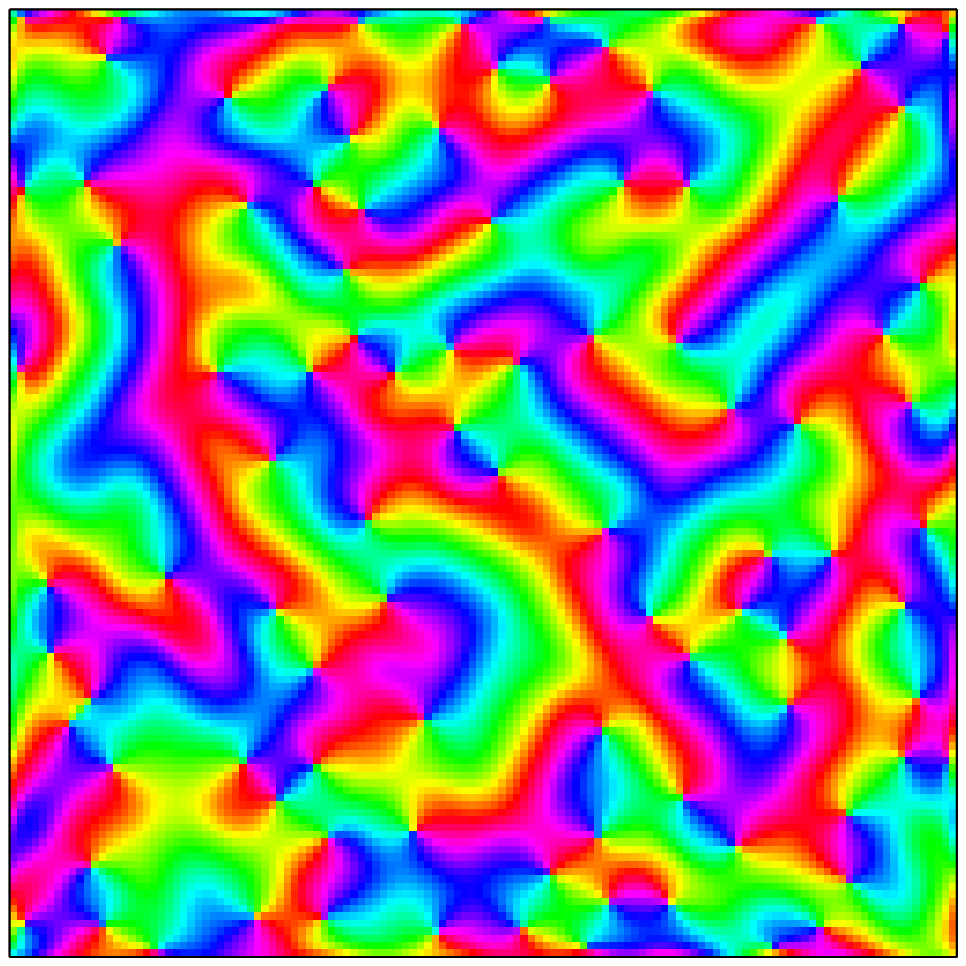} \\[-1ex]
      \rotatebox{90}{\makebox[.1200\textheight][c]{$\beta = 10^{5}$}} &
      \includegraphics[height=.1200\textheight]{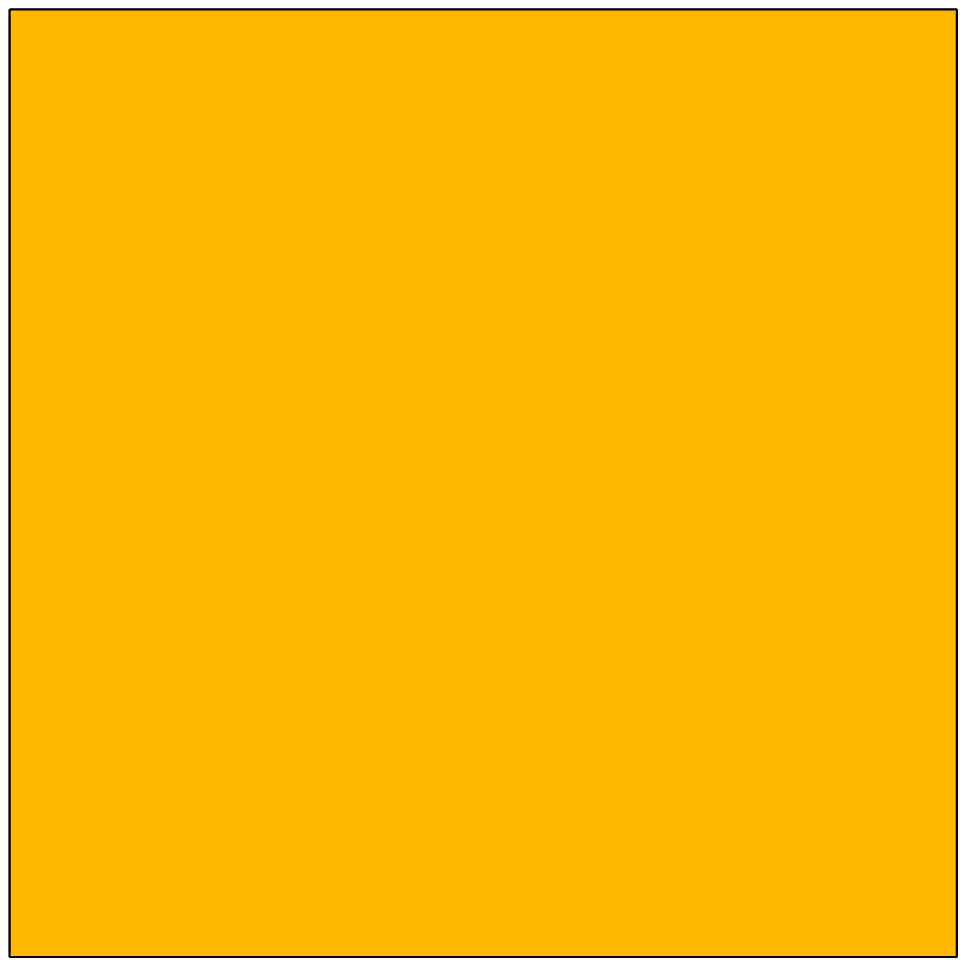} &
      \includegraphics[height=.1200\textheight]{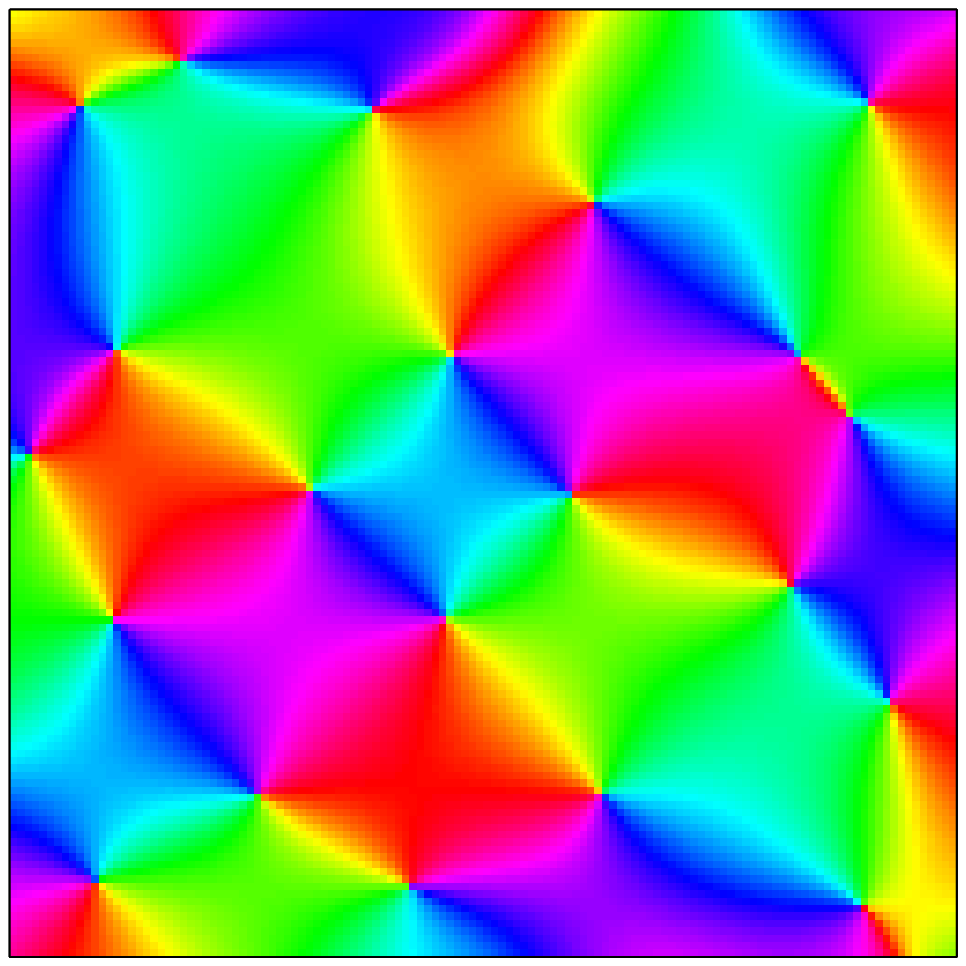} &
      \includegraphics[height=.1200\textheight]{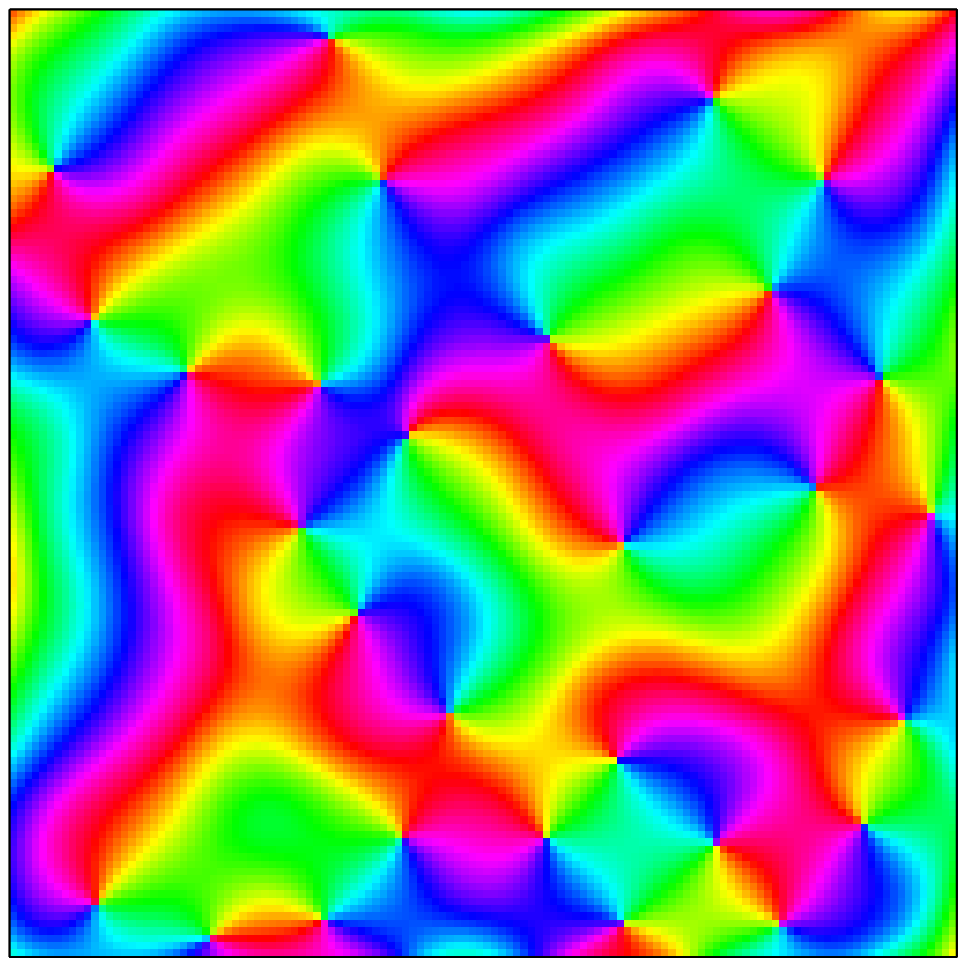} &
      \includegraphics[height=.1200\textheight]{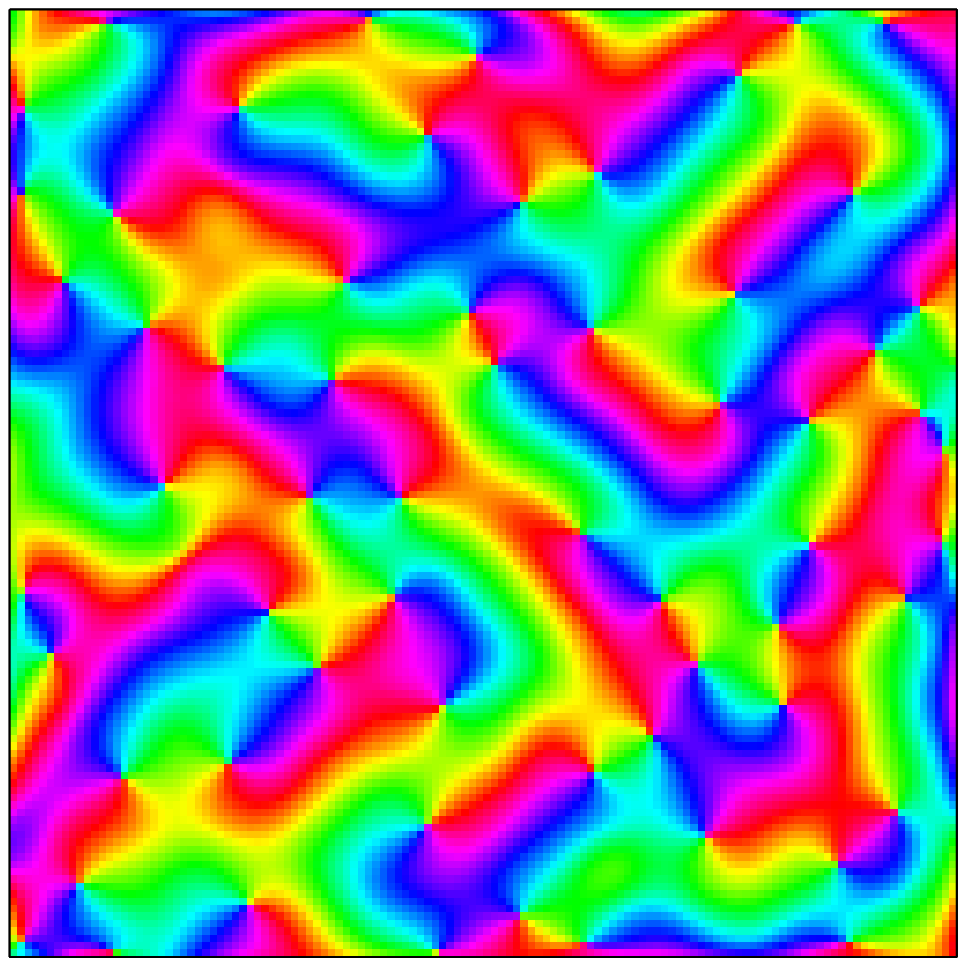} \\[-1ex]
      \rotatebox{90}{\makebox[.1200\textheight][c]{$\beta = 10^{6}$}} &
      \includegraphics[height=.1200\textheight]{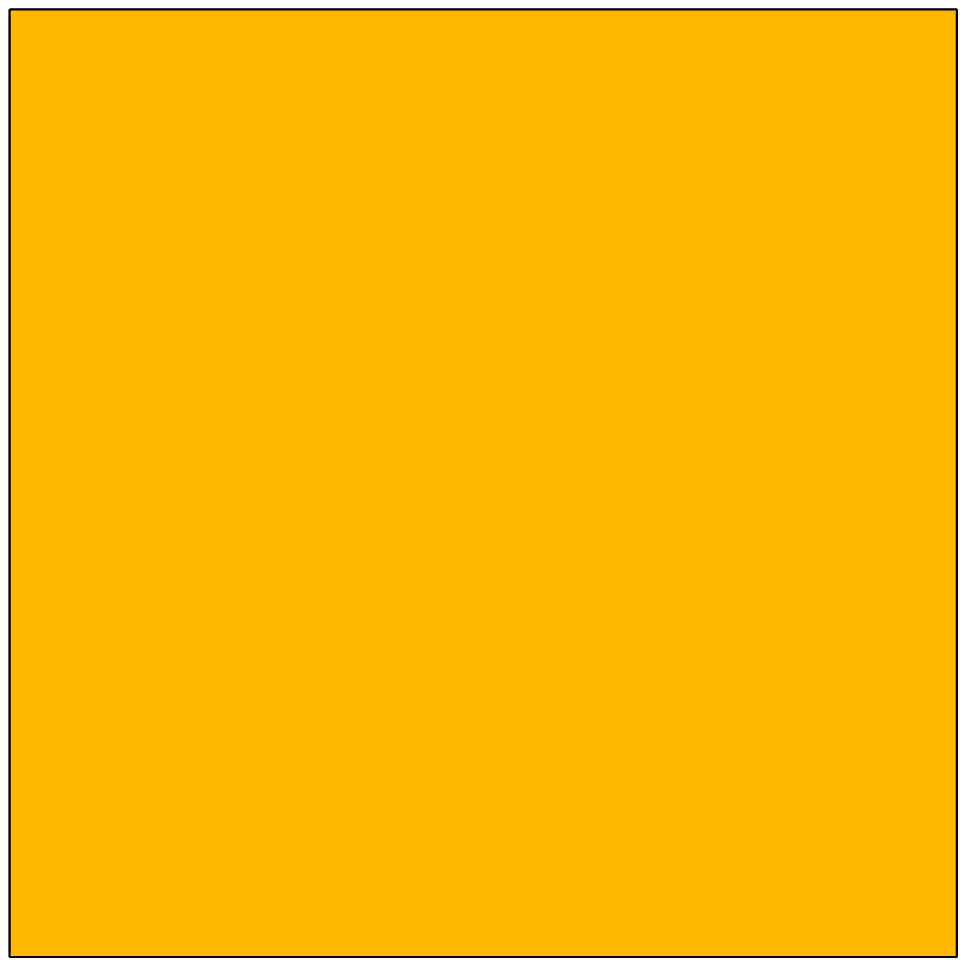} &
      \includegraphics[height=.1200\textheight]{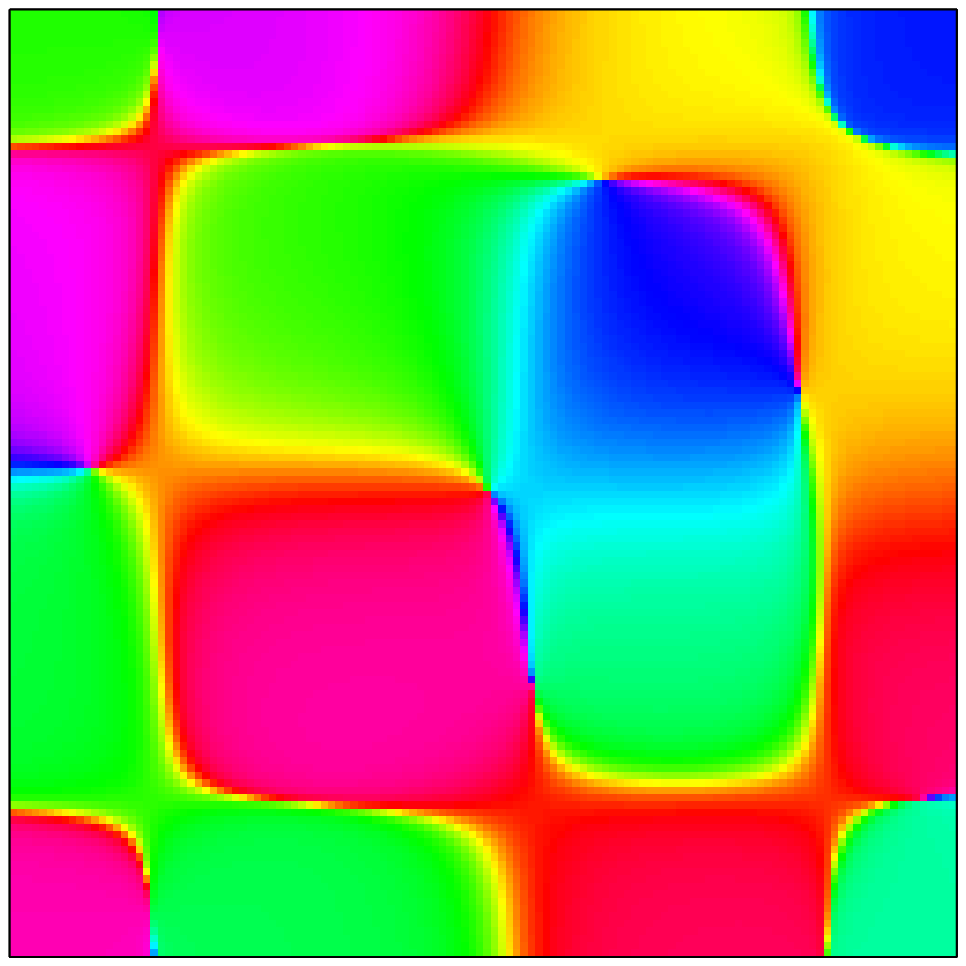} &
      \includegraphics[height=.1200\textheight]{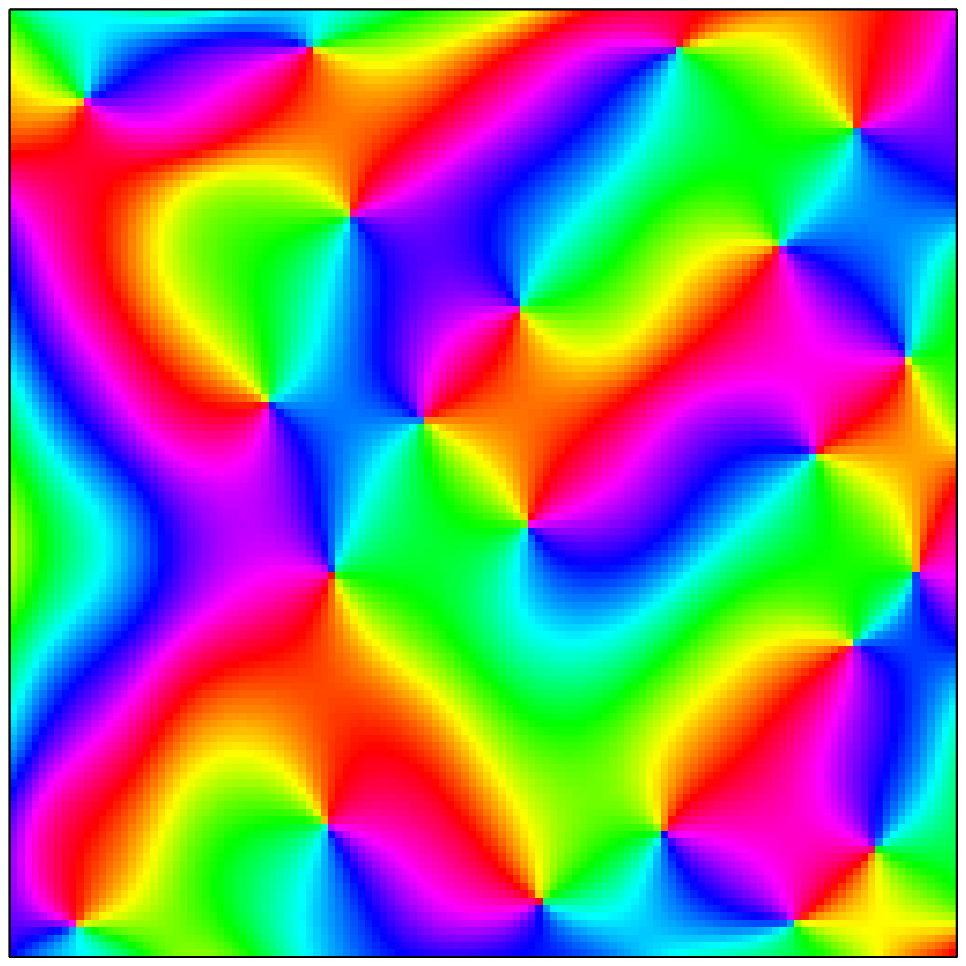} &
      \includegraphics[height=.1200\textheight]{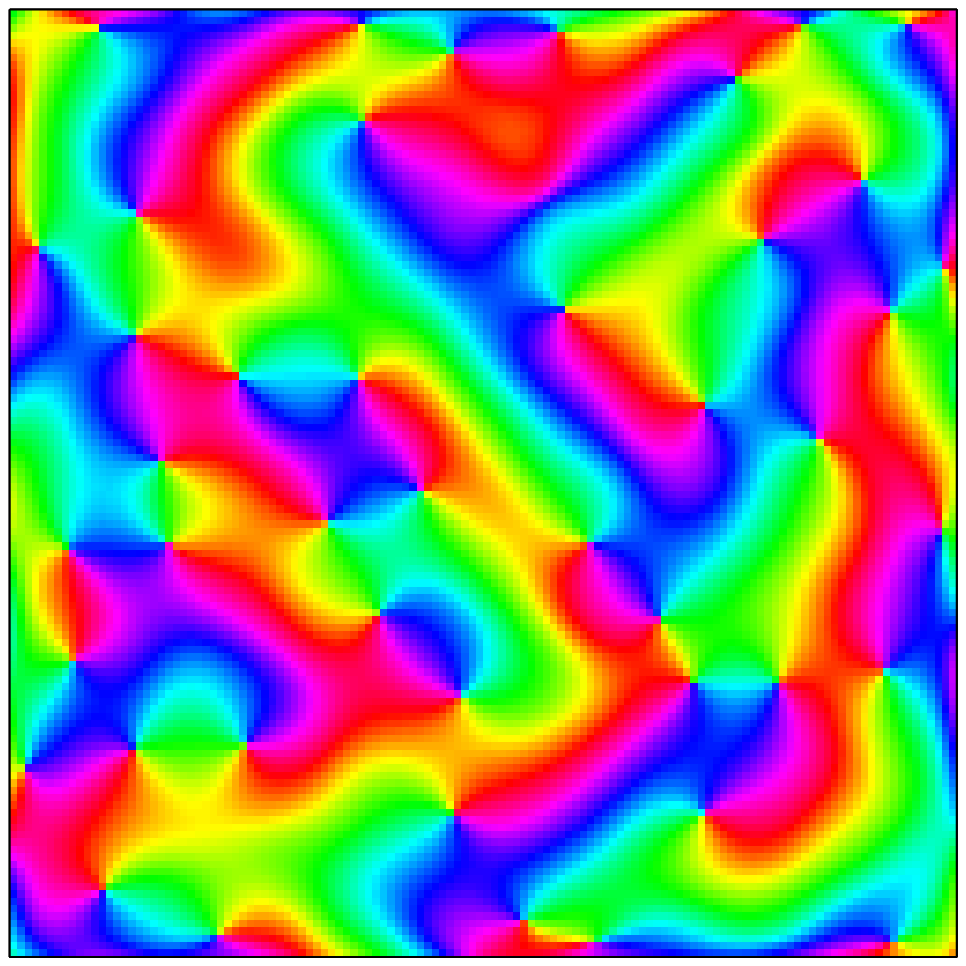} \\[-1ex]
    \end{tabular}
    \caption{As fig.~\ref{f:simul2D:OD} but for orientation maps.}
    \label{f:simul2D:OR}
  \end{center}
\end{figure}
            
\begin{figure}
  \begin{center}
    \begin{tabular}{@{}c@{\hspace{0.5cm}}c@{\hspace{0.5cm}}c@{\hspace{0.5cm}}c@{\hspace{0.5cm}}c@{}}
      & $p = 1$ & $p = 2$ & $p = 3$ & $p = 4$ \\
      \rotatebox{90}{\makebox[.1200\textheight][c]{$\beta = 10^{-1}$}} &
      \includegraphics[height=.1200\textheight]{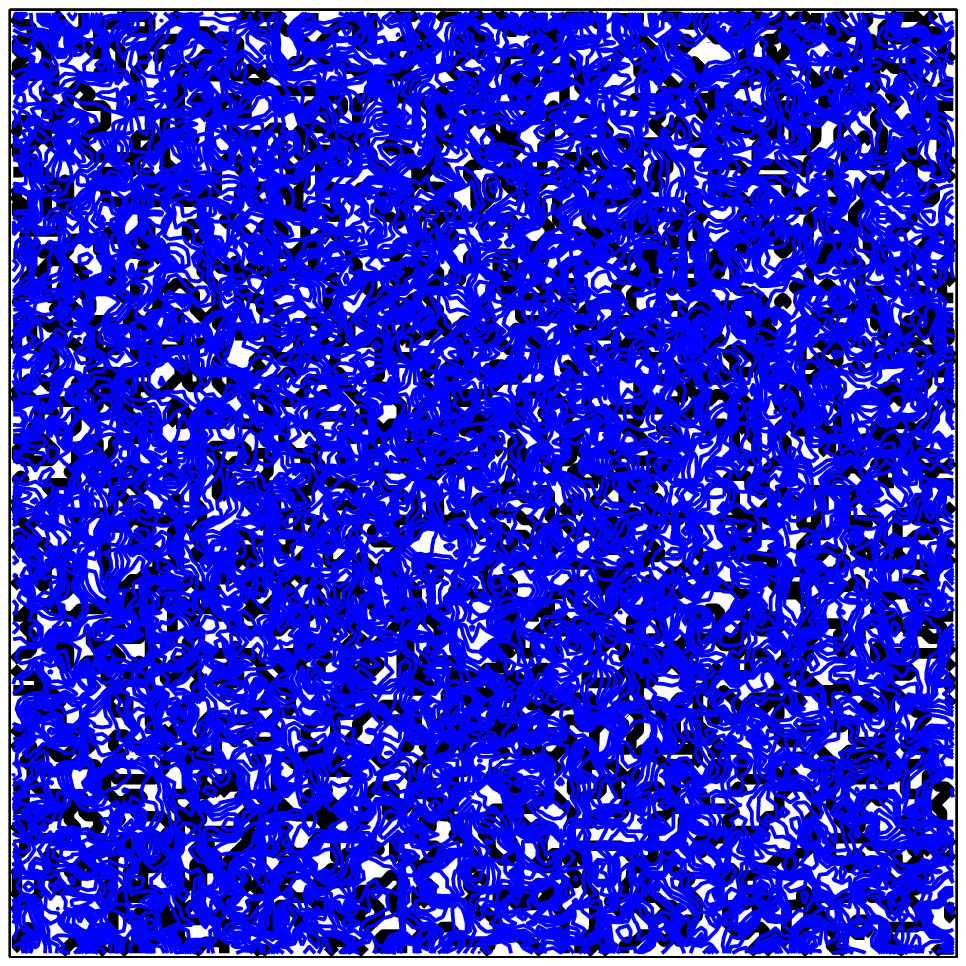} &
      \includegraphics[height=.1200\textheight]{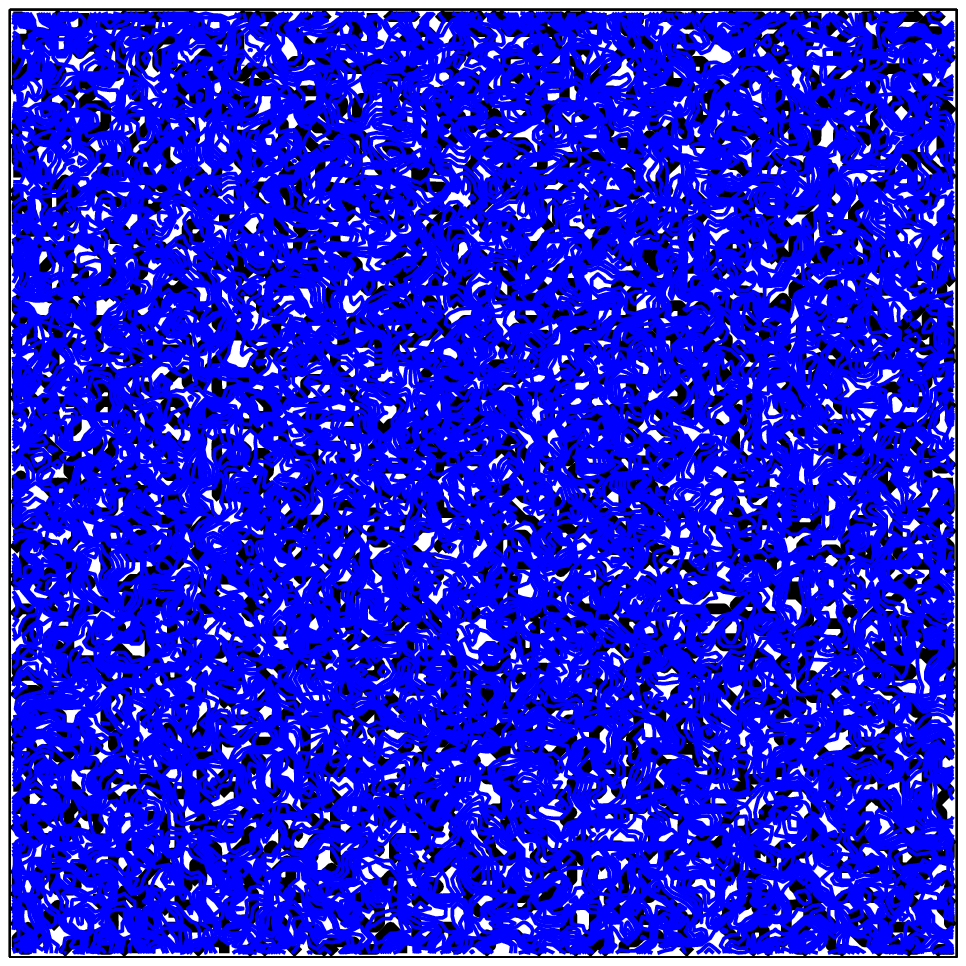} &
      \includegraphics[height=.1200\textheight]{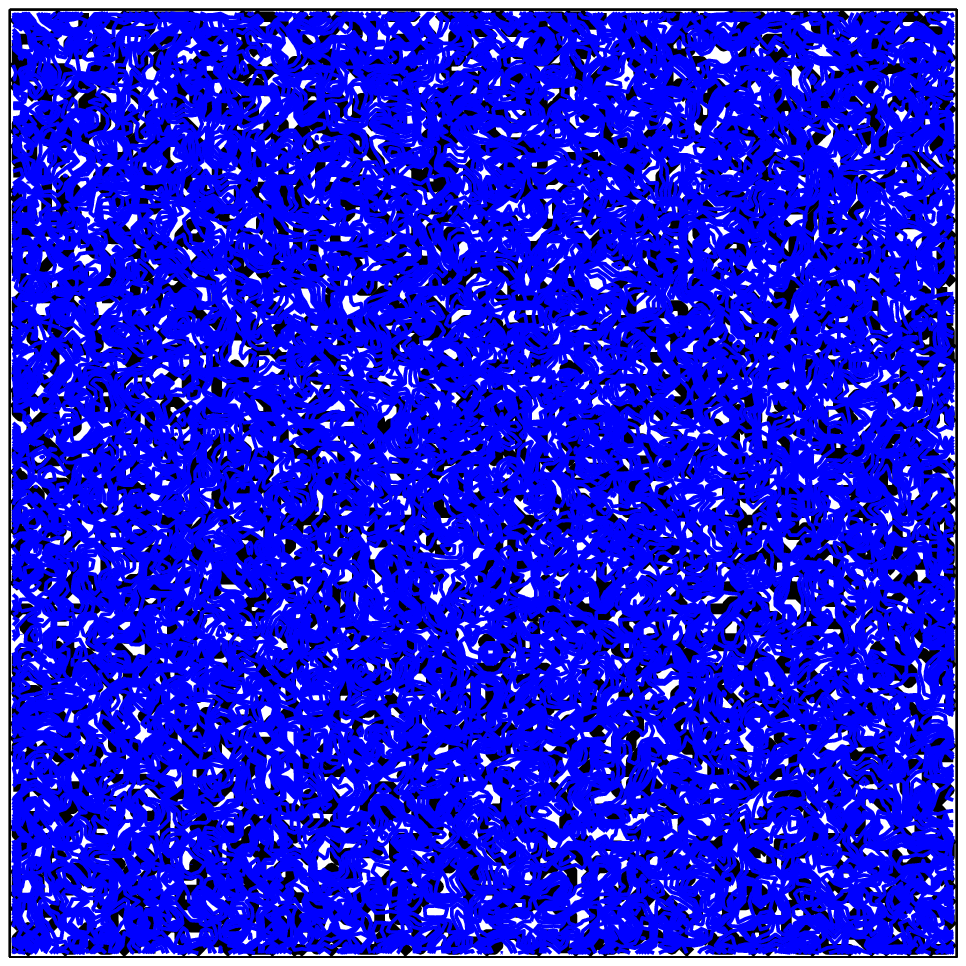} &
      \includegraphics[height=.1200\textheight]{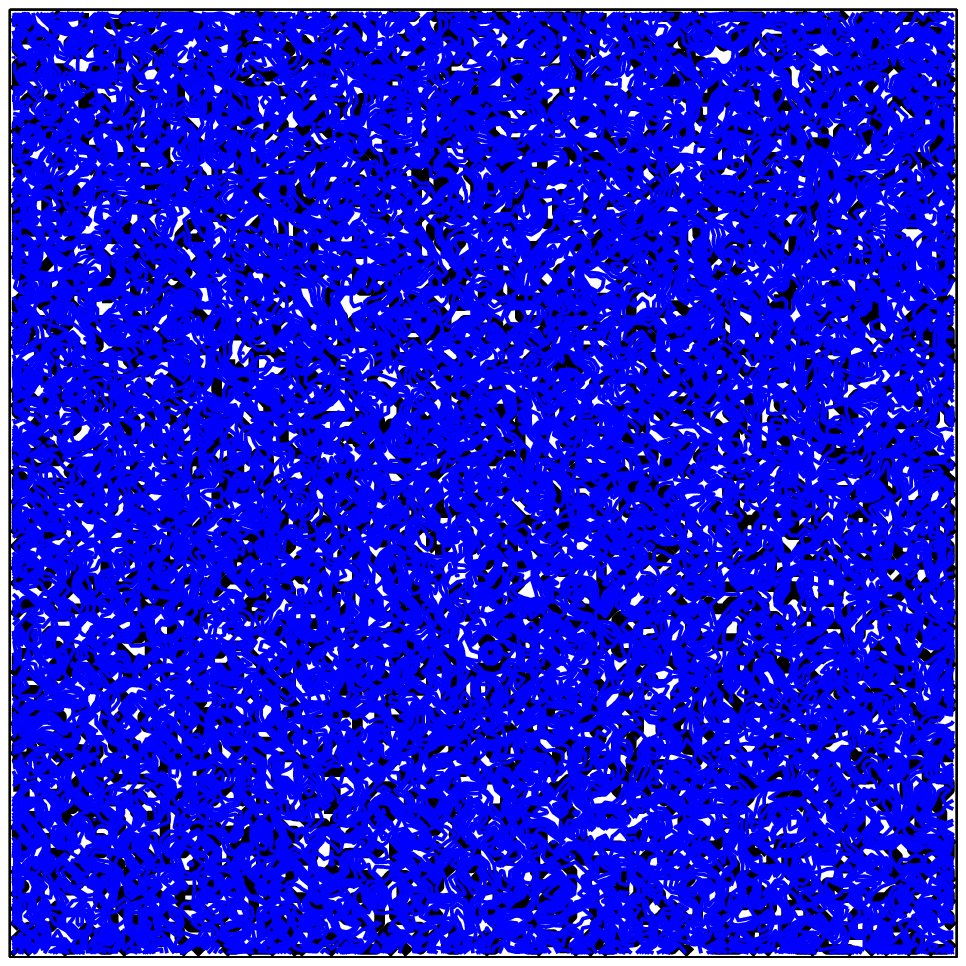} \\[-1ex]
      \rotatebox{90}{\makebox[.1200\textheight][c]{$\beta = 10^{0}$}} &
      \includegraphics[height=.1200\textheight]{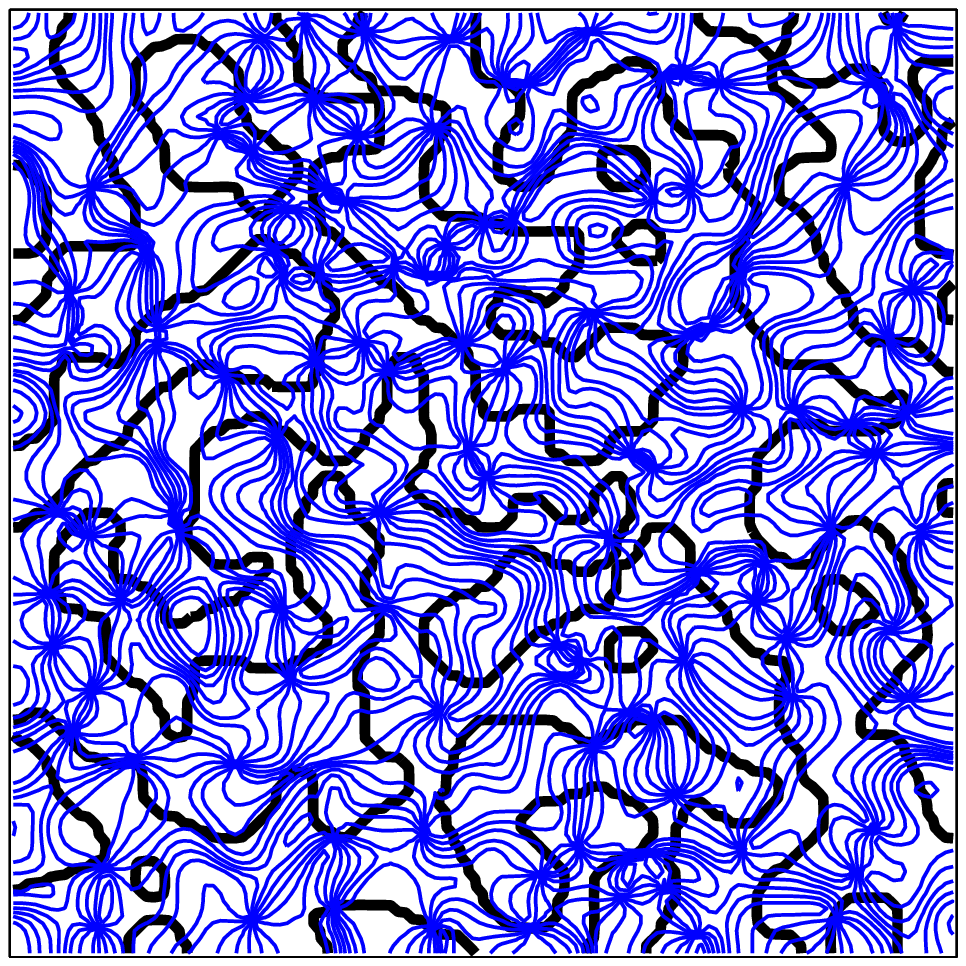} &
      \includegraphics[height=.1200\textheight]{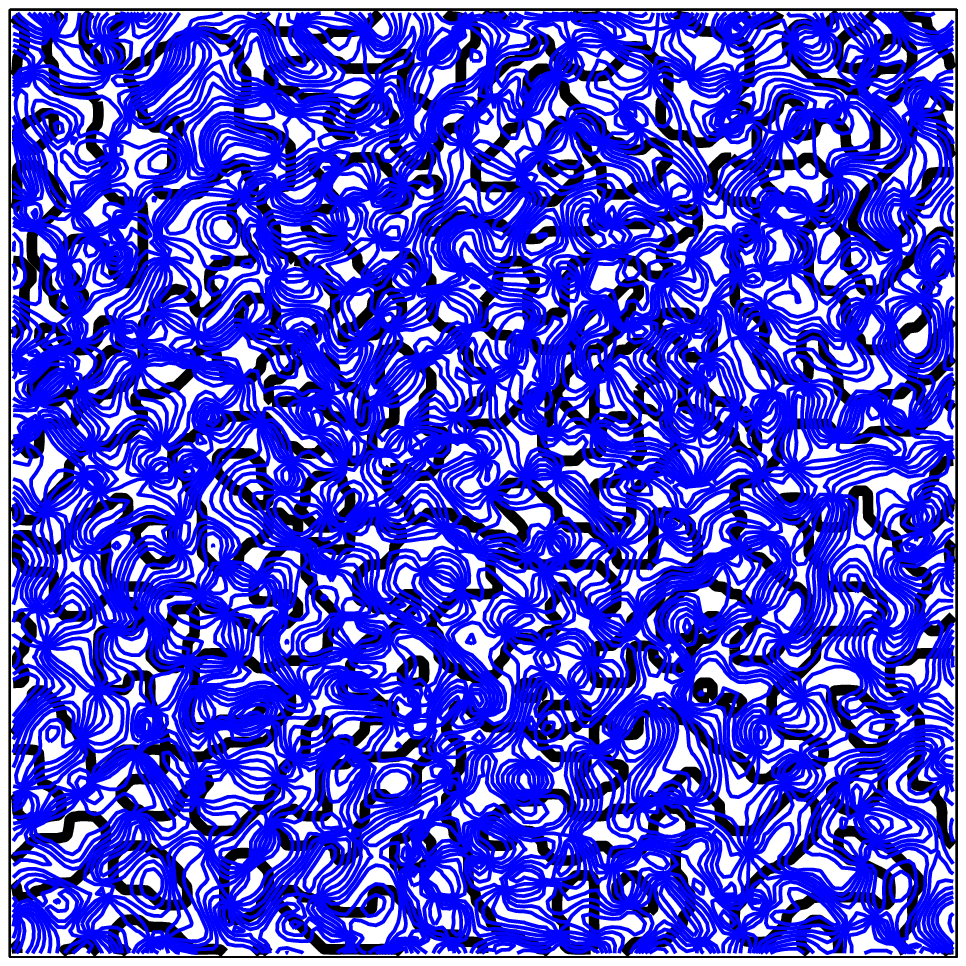} &
      \includegraphics[height=.1200\textheight]{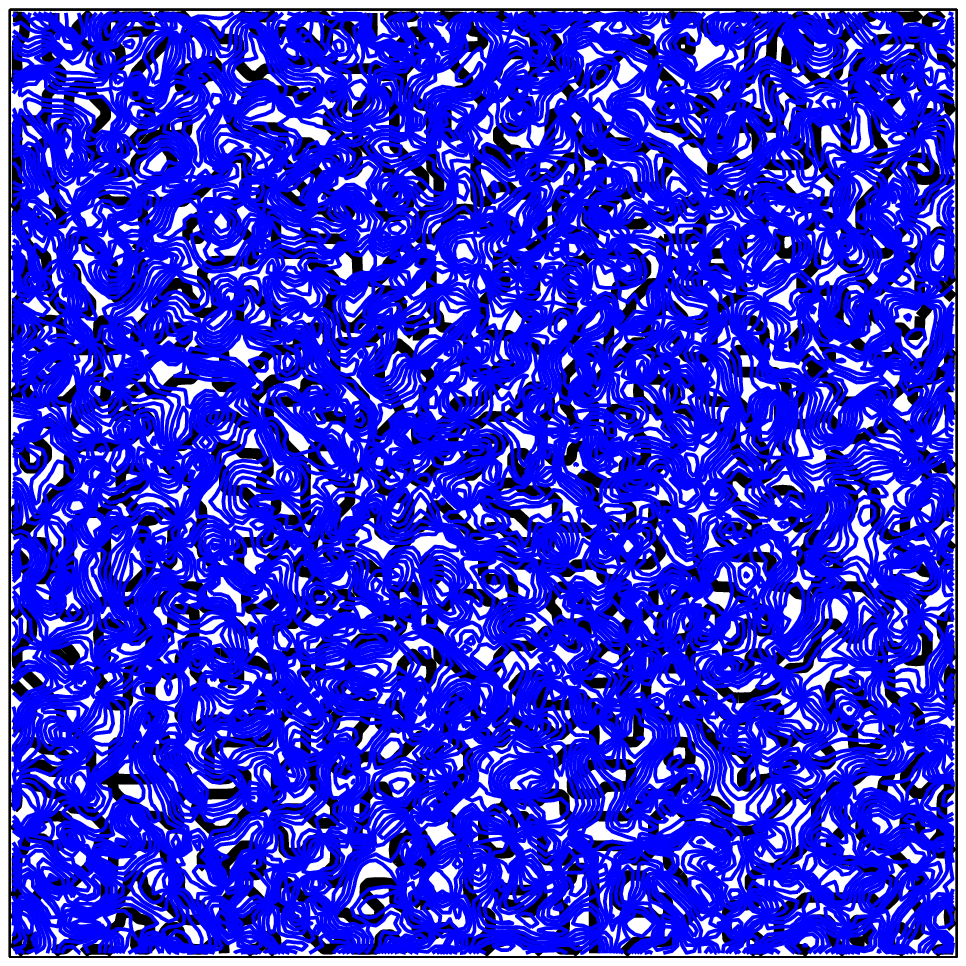} &
      \includegraphics[height=.1200\textheight]{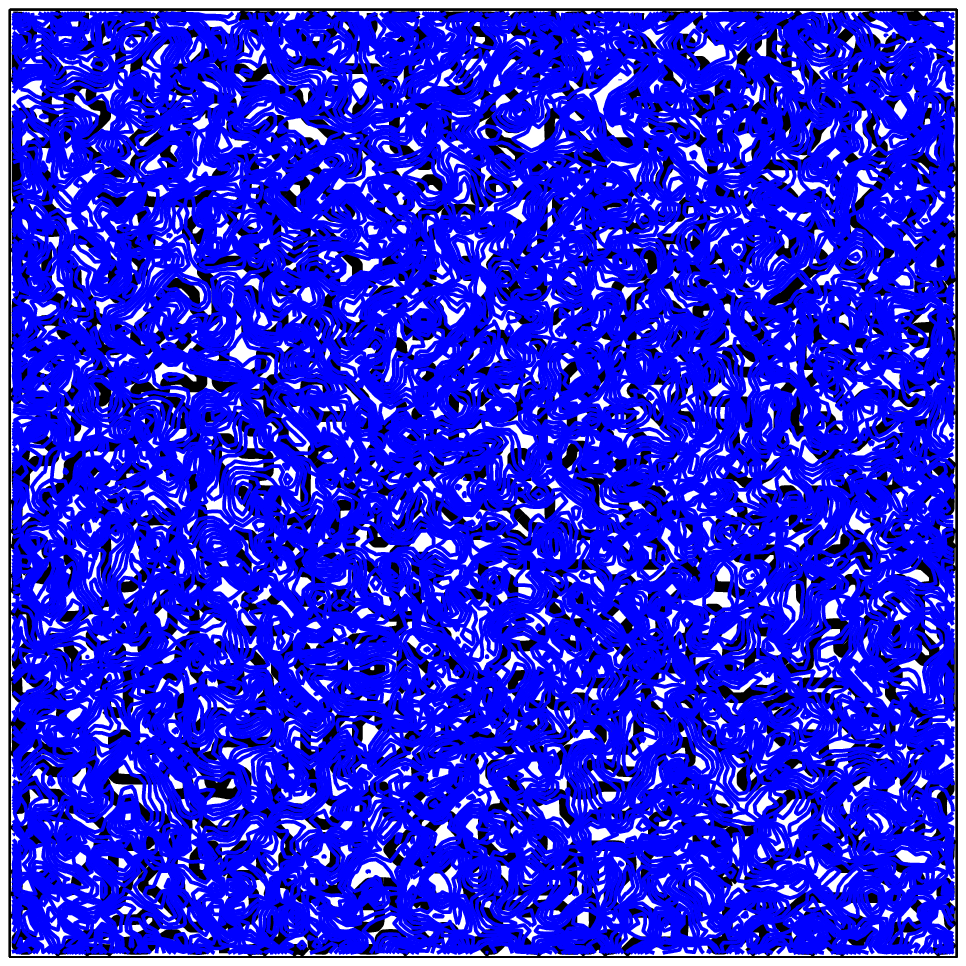} \\[-1ex]
      \rotatebox{90}{\makebox[.1200\textheight][c]{$\beta = 10^{1}$}} &
      \includegraphics[height=.1200\textheight]{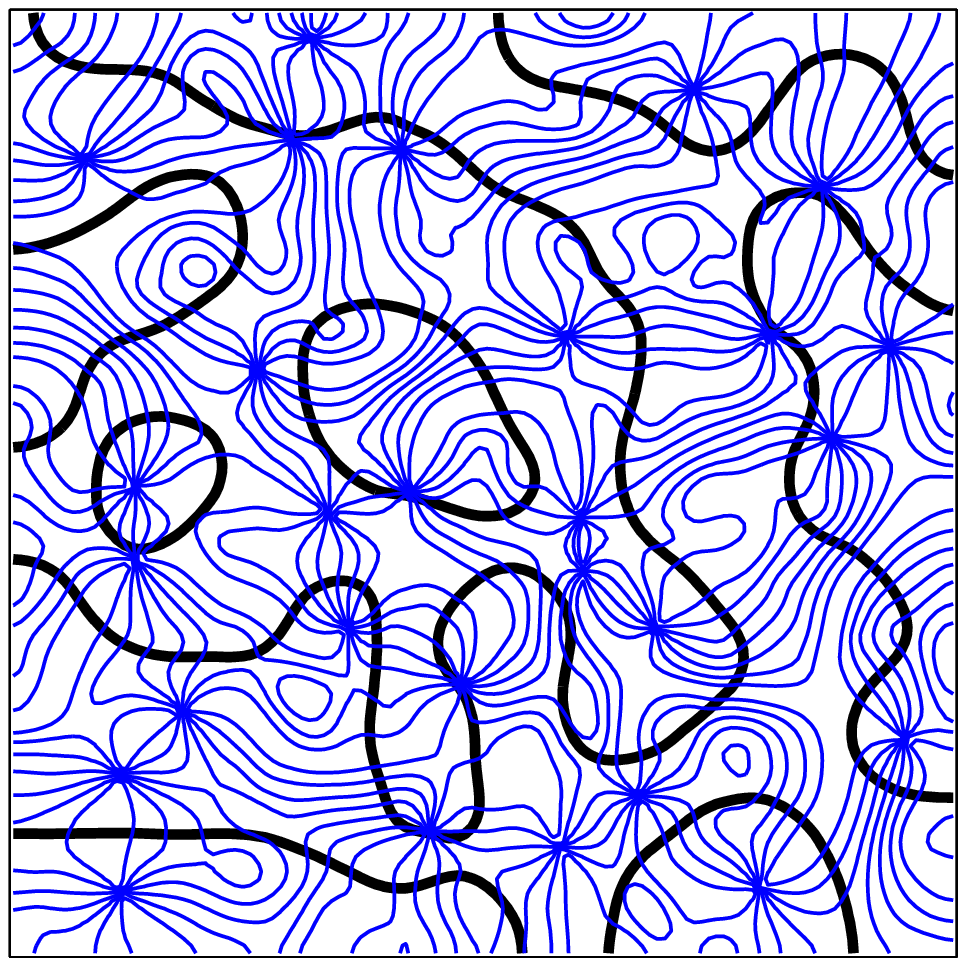} &
      \includegraphics[height=.1200\textheight]{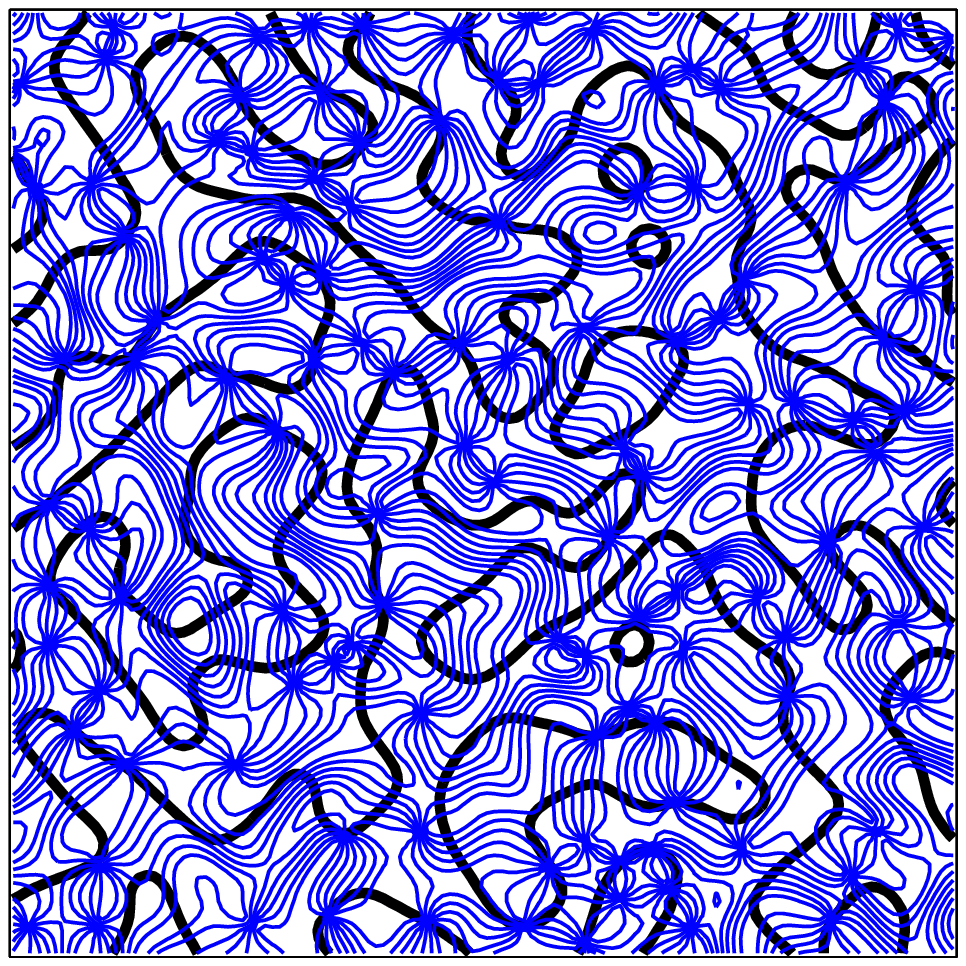} &
      \includegraphics[height=.1200\textheight]{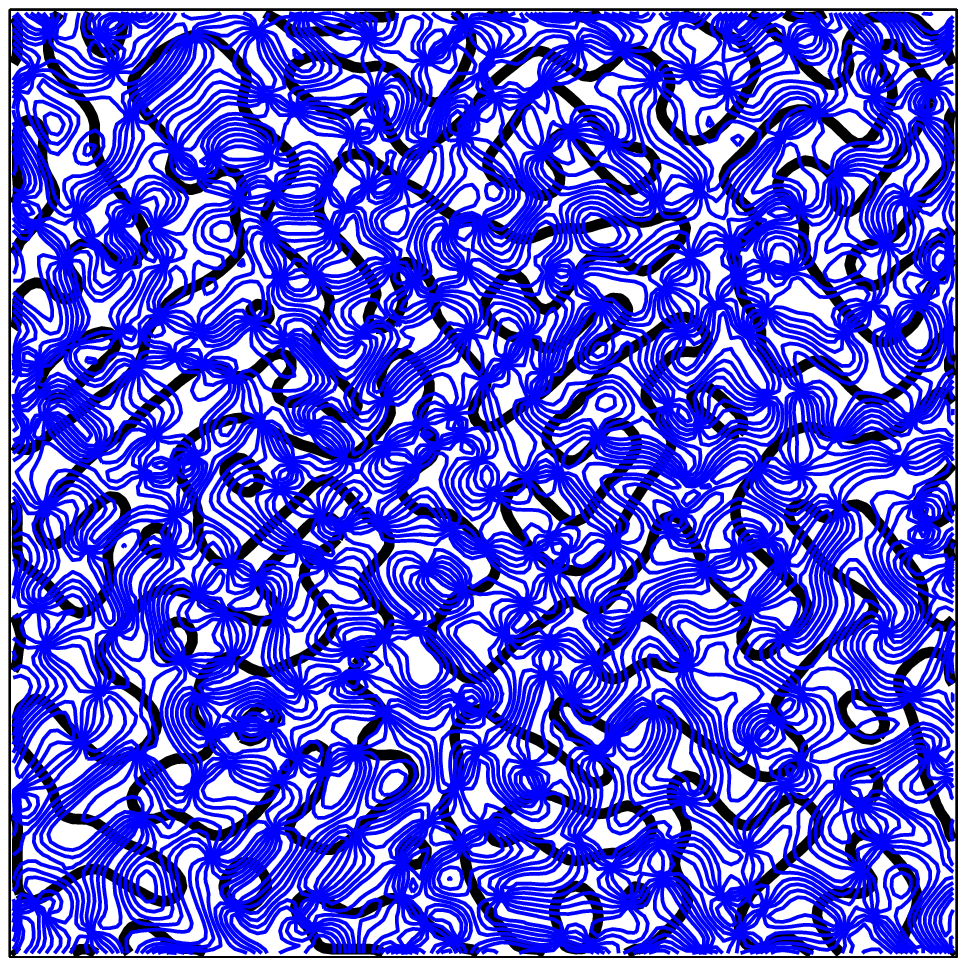} &
      \includegraphics[height=.1200\textheight]{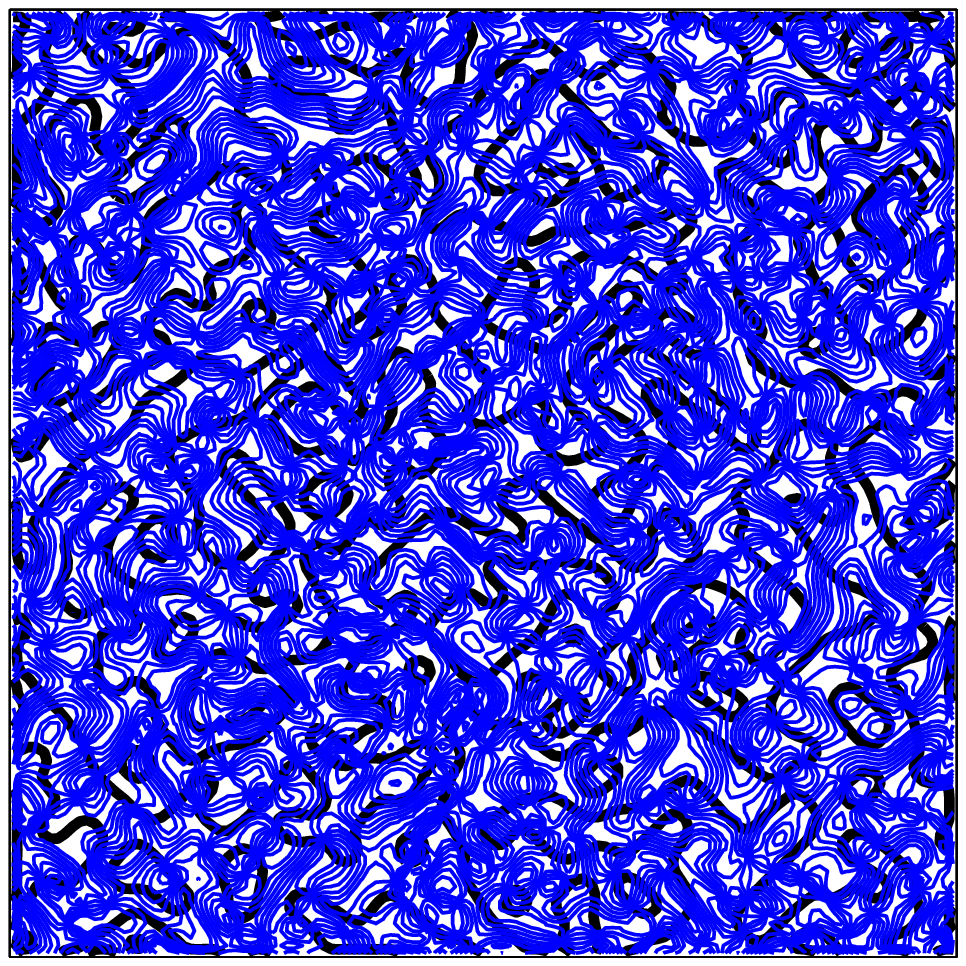} \\[-1ex]
      \rotatebox{90}{\makebox[.1200\textheight][c]{$\beta = 10^{2}$}} &
      \includegraphics[height=.1200\textheight]{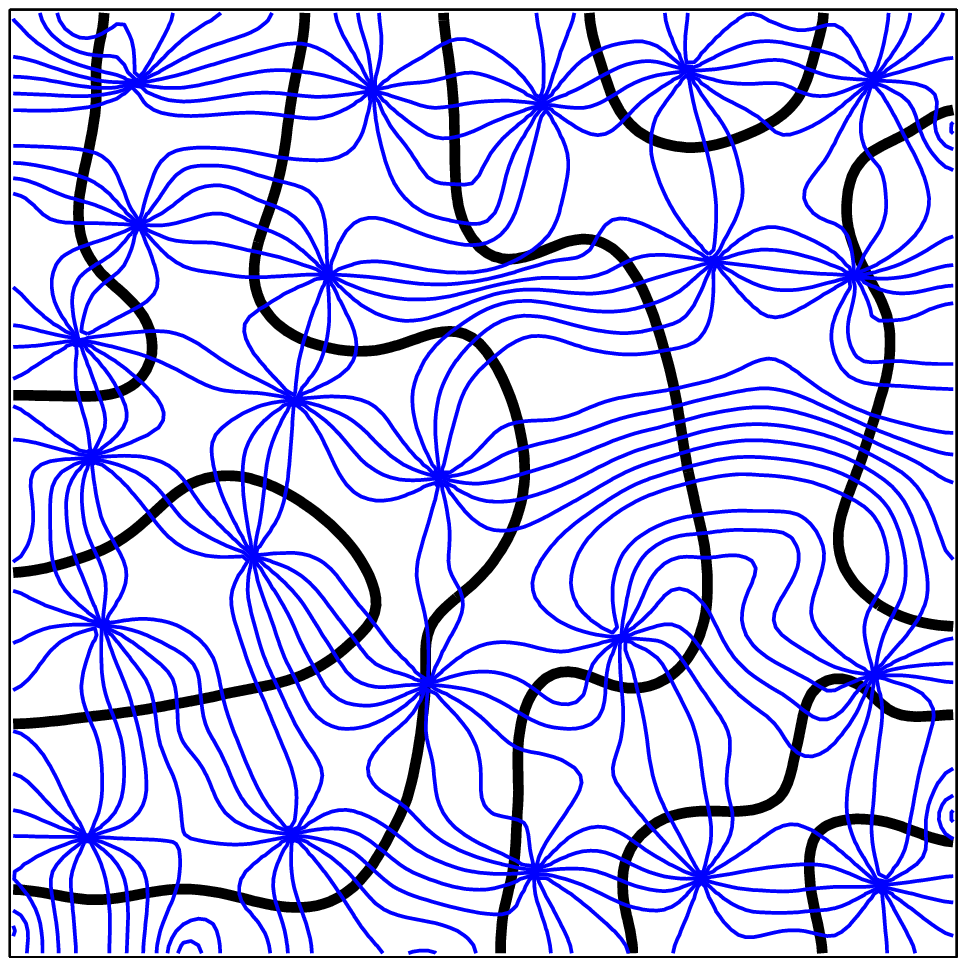} &
      \includegraphics[height=.1200\textheight]{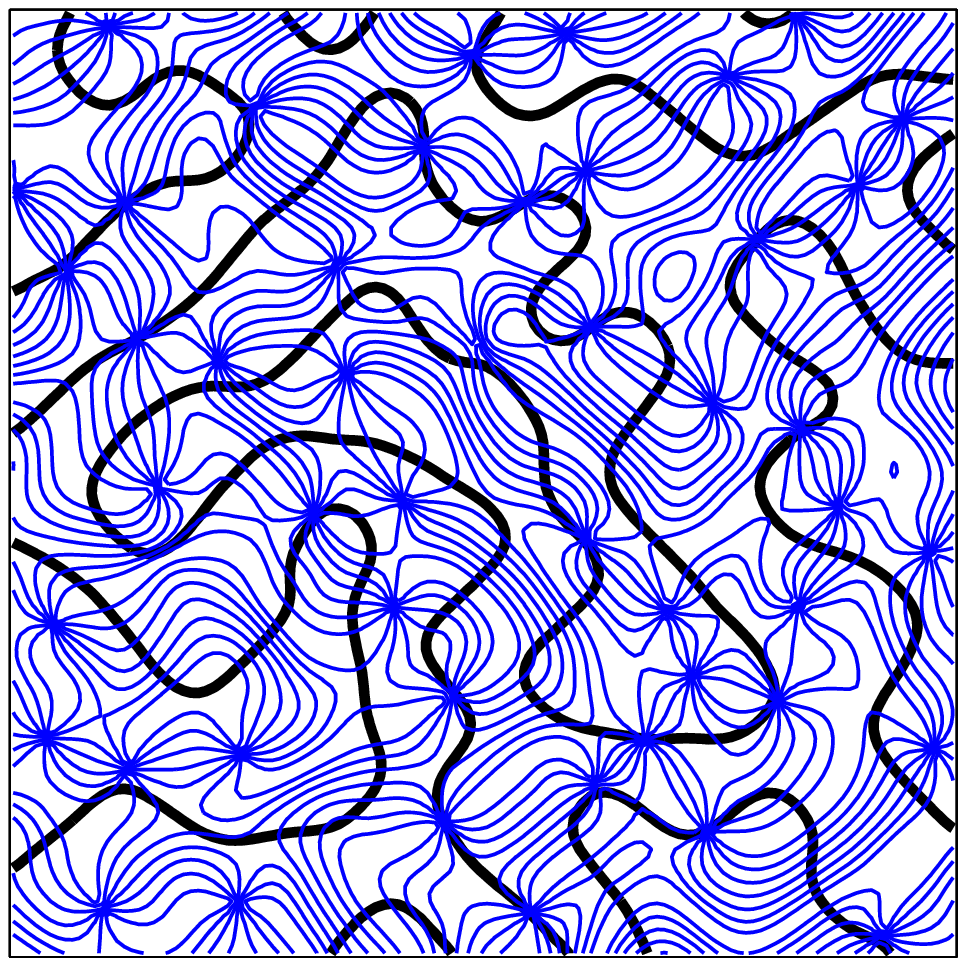} &
      \includegraphics[height=.1200\textheight]{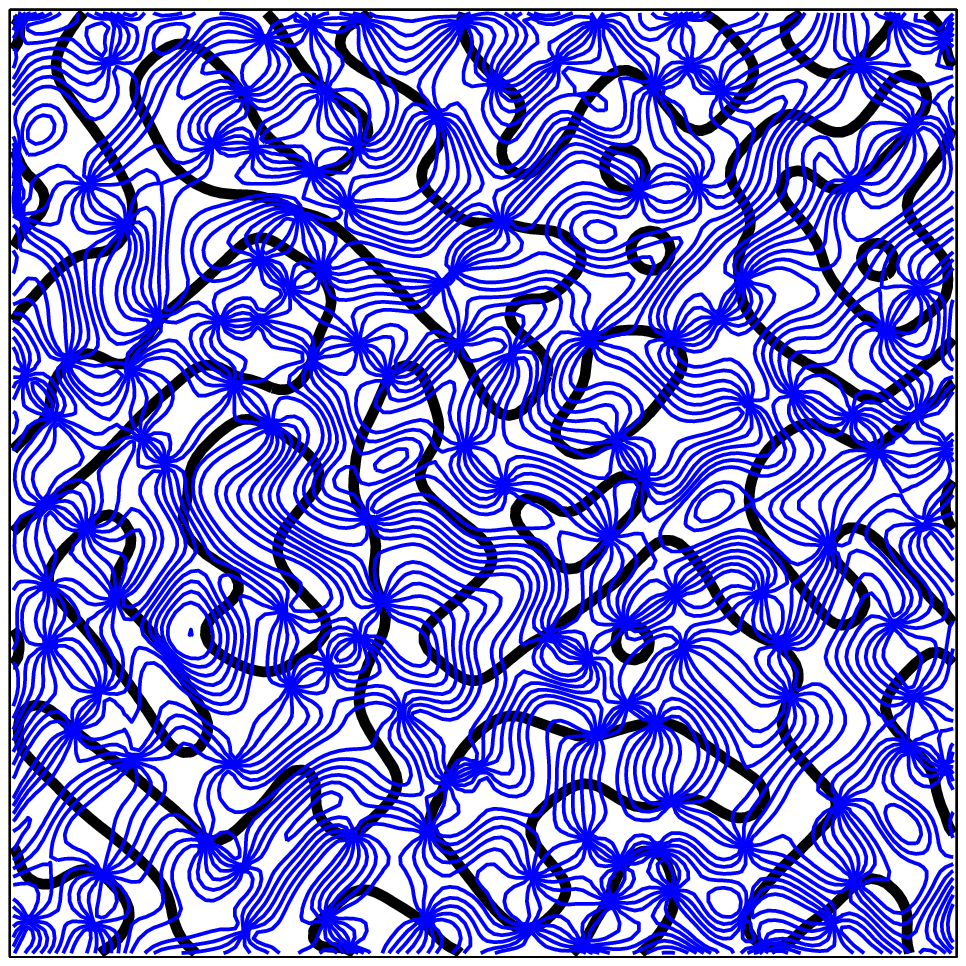} &
      \includegraphics[height=.1200\textheight]{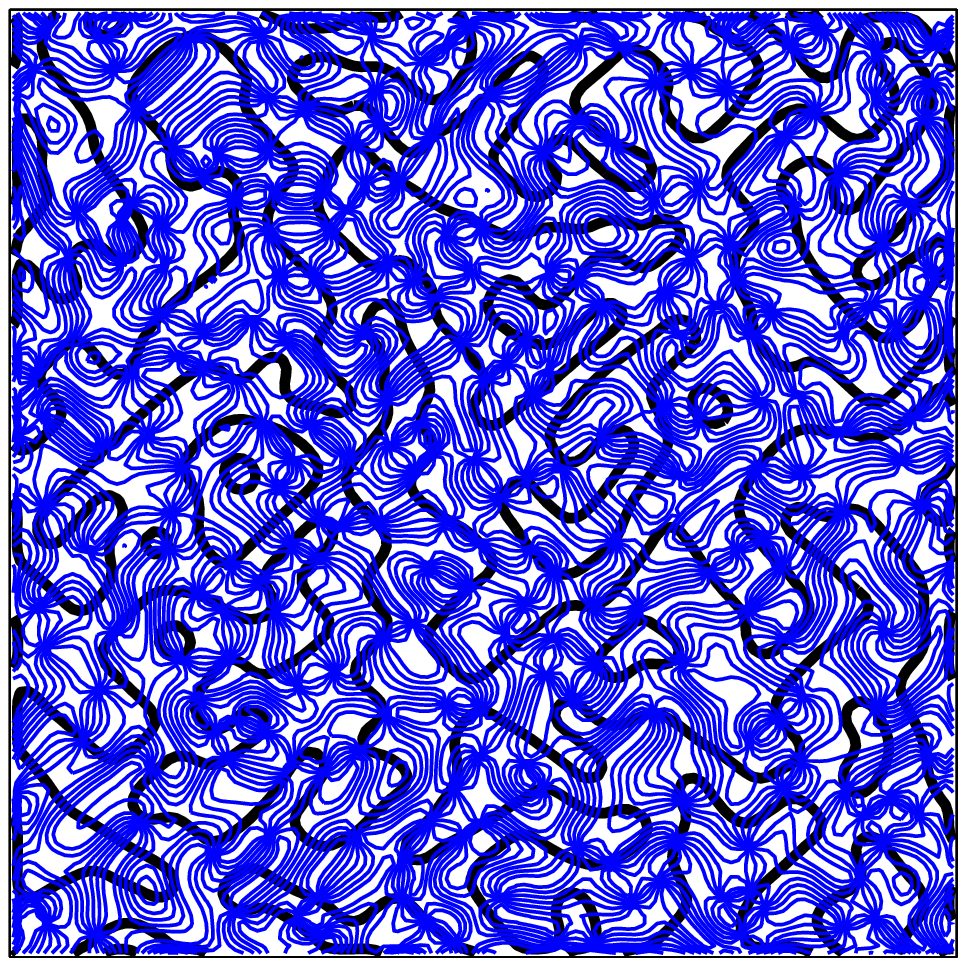} \\[-1ex]
      \rotatebox{90}{\makebox[.1200\textheight][c]{$\beta = 10^{3}$}} &
      \includegraphics[height=.1200\textheight]{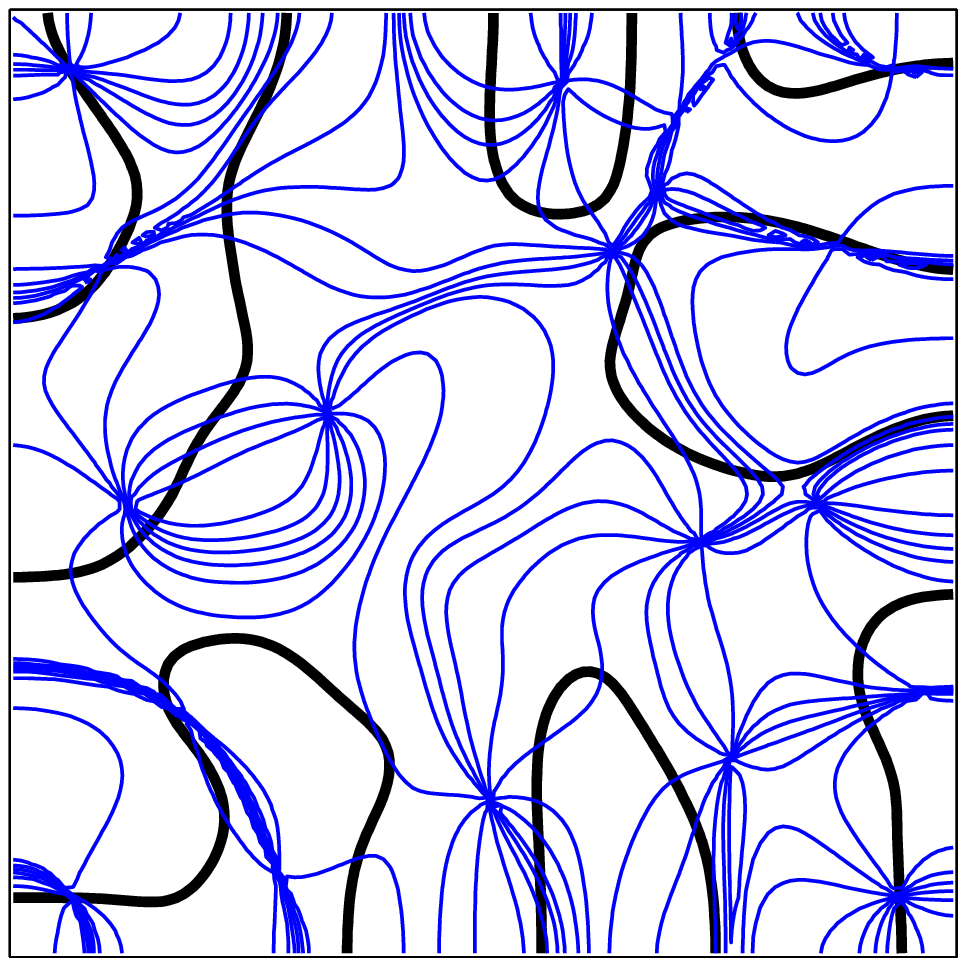} &
      \includegraphics[height=.1200\textheight]{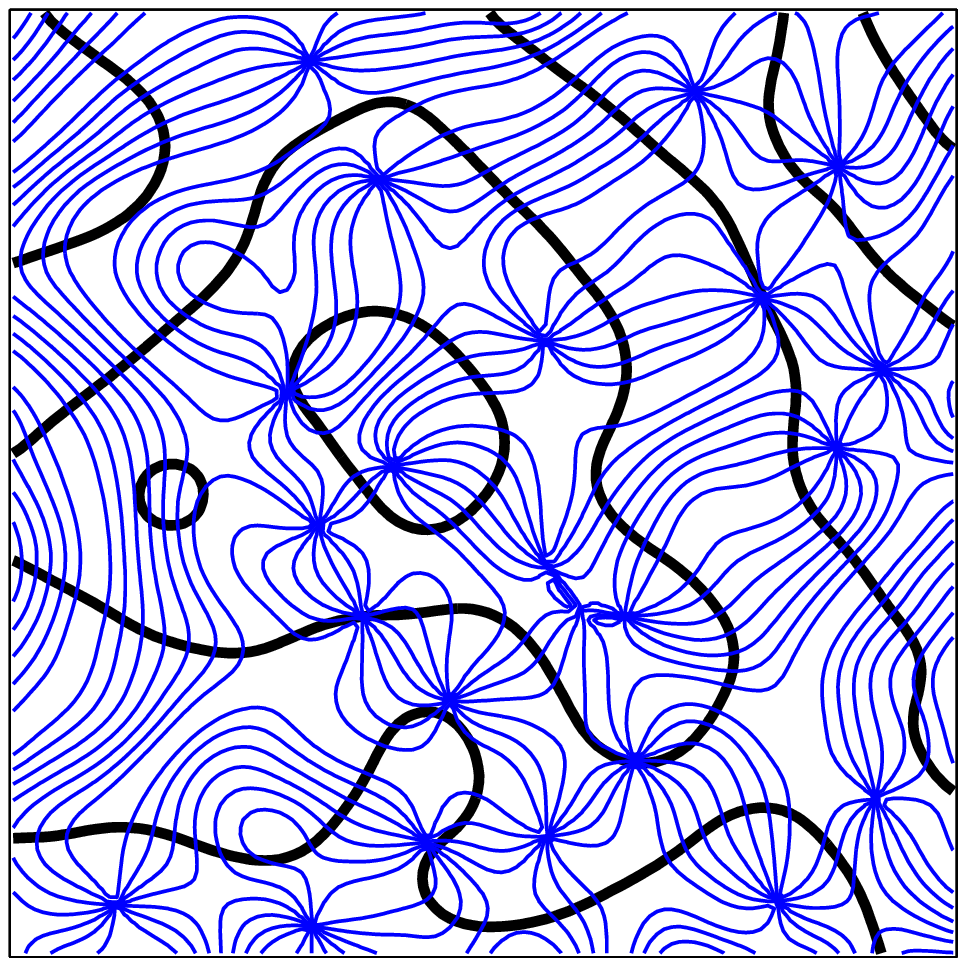} &
      \includegraphics[height=.1200\textheight]{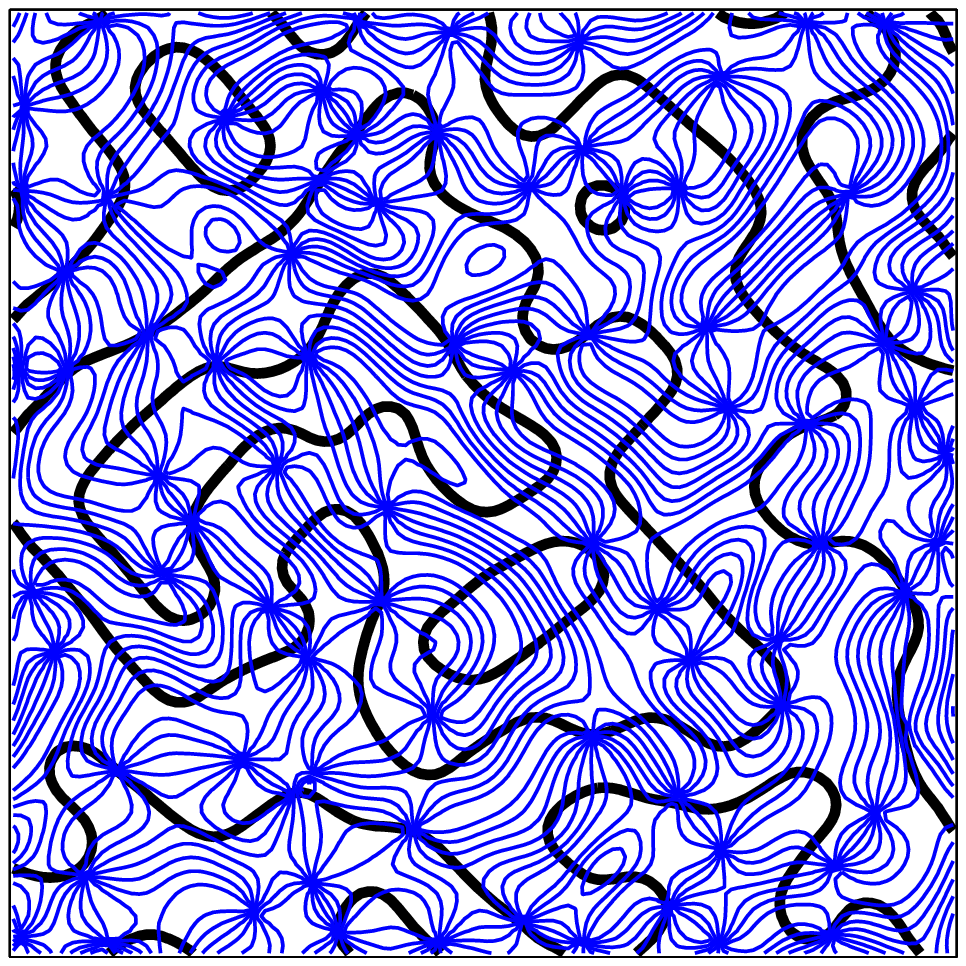} &
      \includegraphics[height=.1200\textheight]{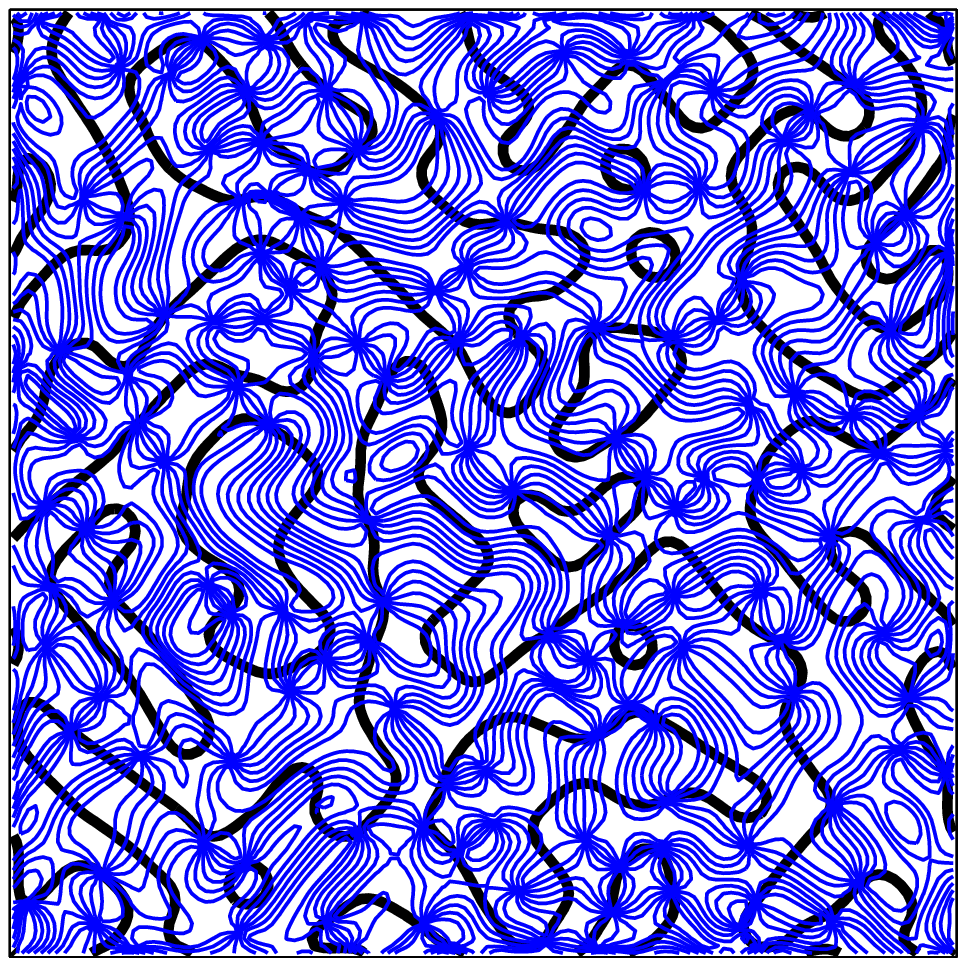} \\[-1ex]
      \rotatebox{90}{\makebox[.1200\textheight][c]{$\beta = 10^{4}$}} &
      \includegraphics[height=.1200\textheight]{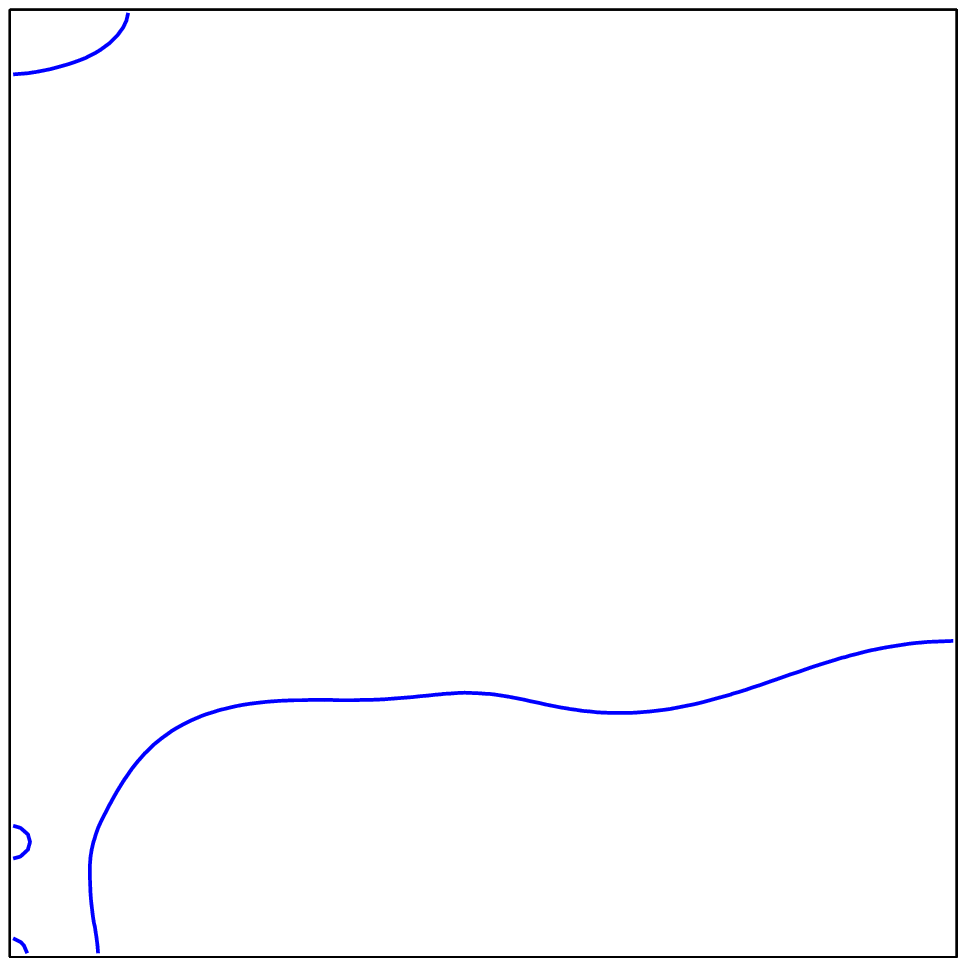} &
      \includegraphics[height=.1200\textheight]{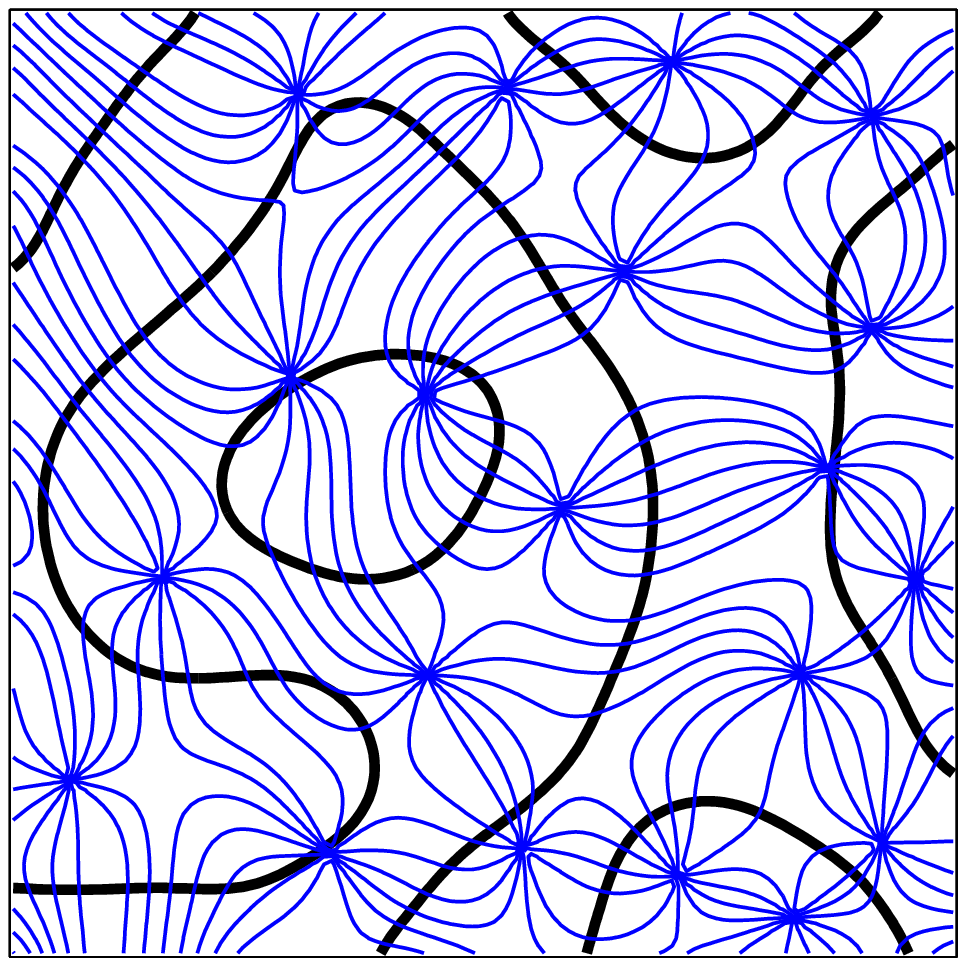} &
      \includegraphics[height=.1200\textheight]{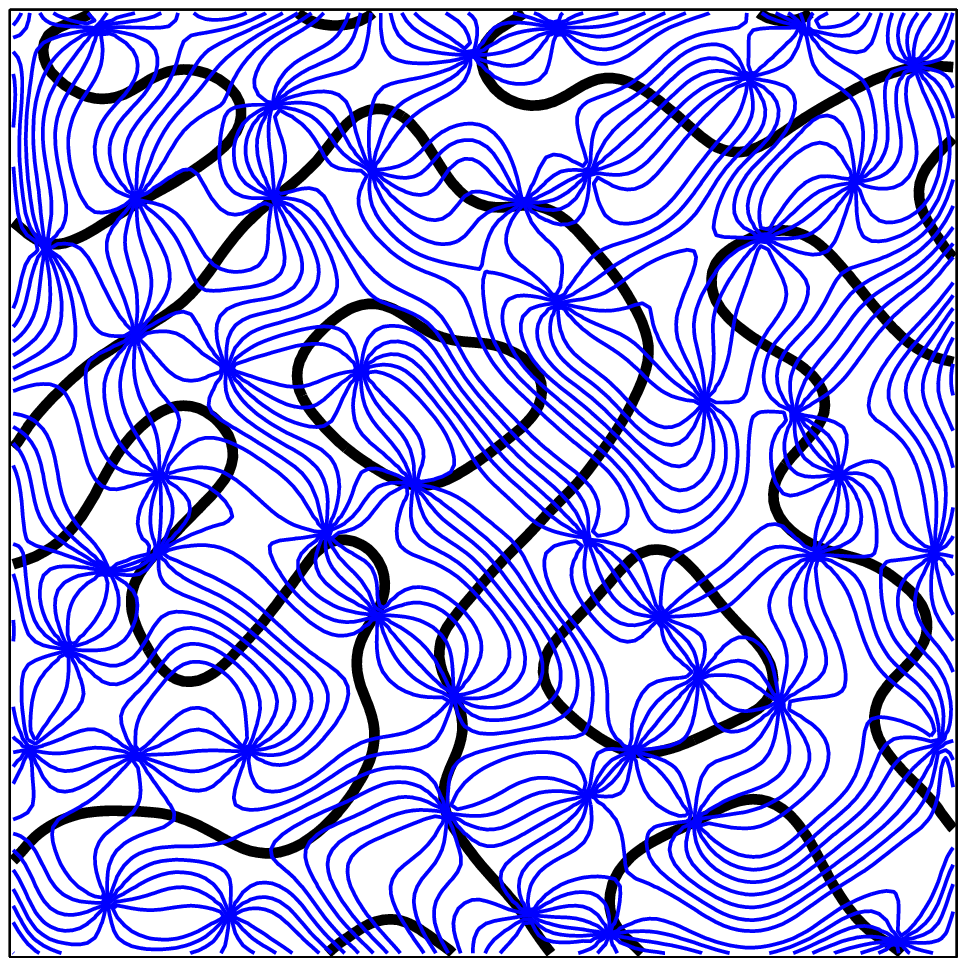} &
      \includegraphics[height=.1200\textheight]{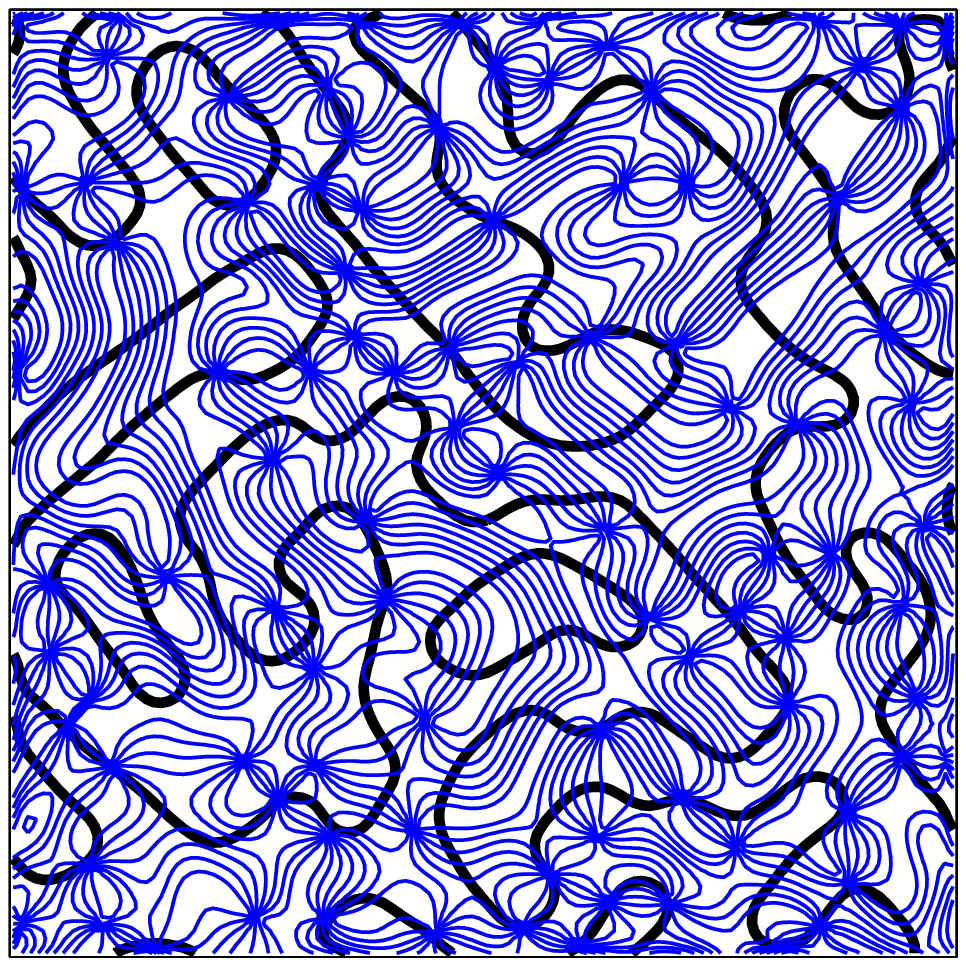} \\[-1ex]
      \rotatebox{90}{\makebox[.1200\textheight][c]{$\beta = 10^{5}$}} &
      \includegraphics[height=.1200\textheight]{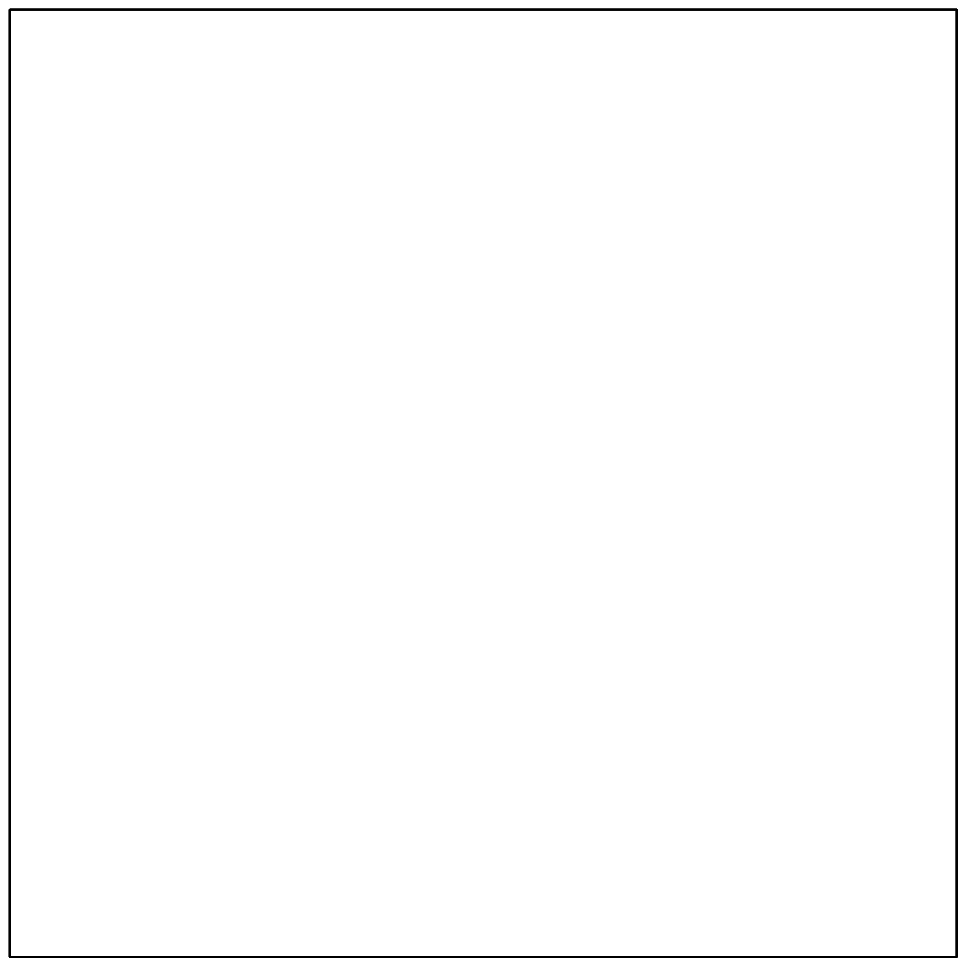} &
      \includegraphics[height=.1200\textheight]{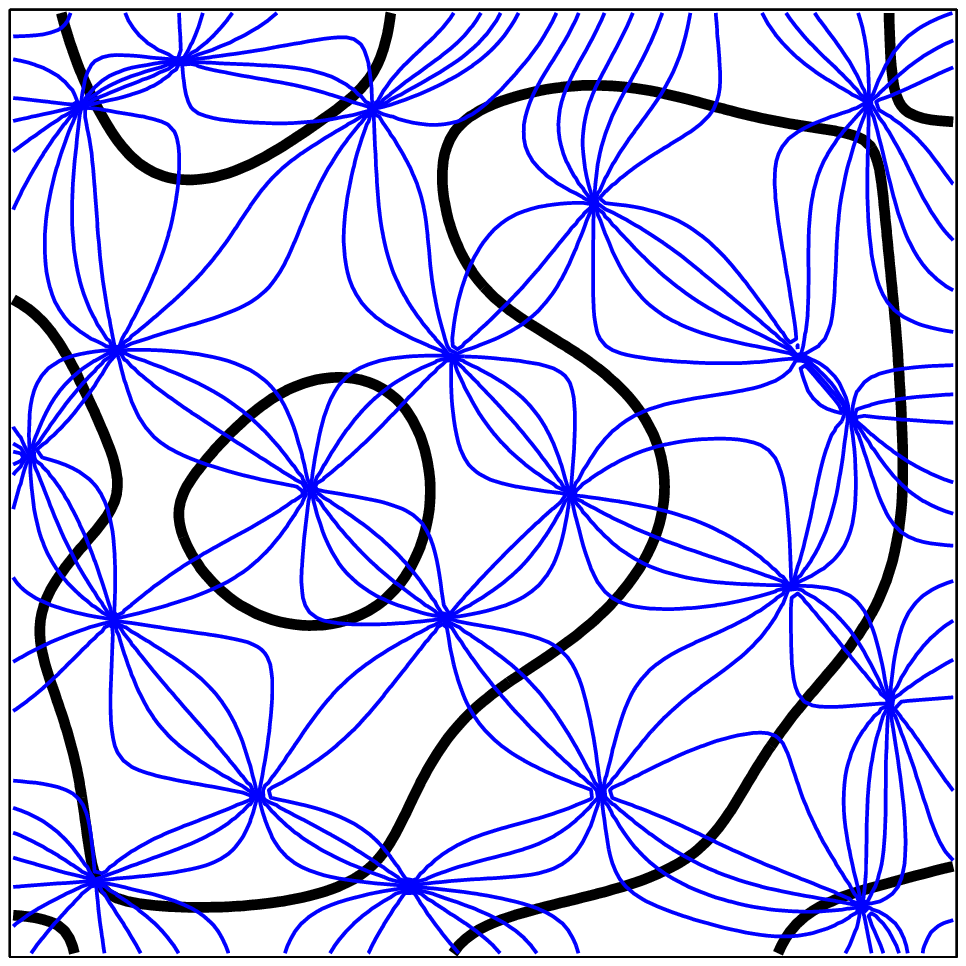} &
      \includegraphics[height=.1200\textheight]{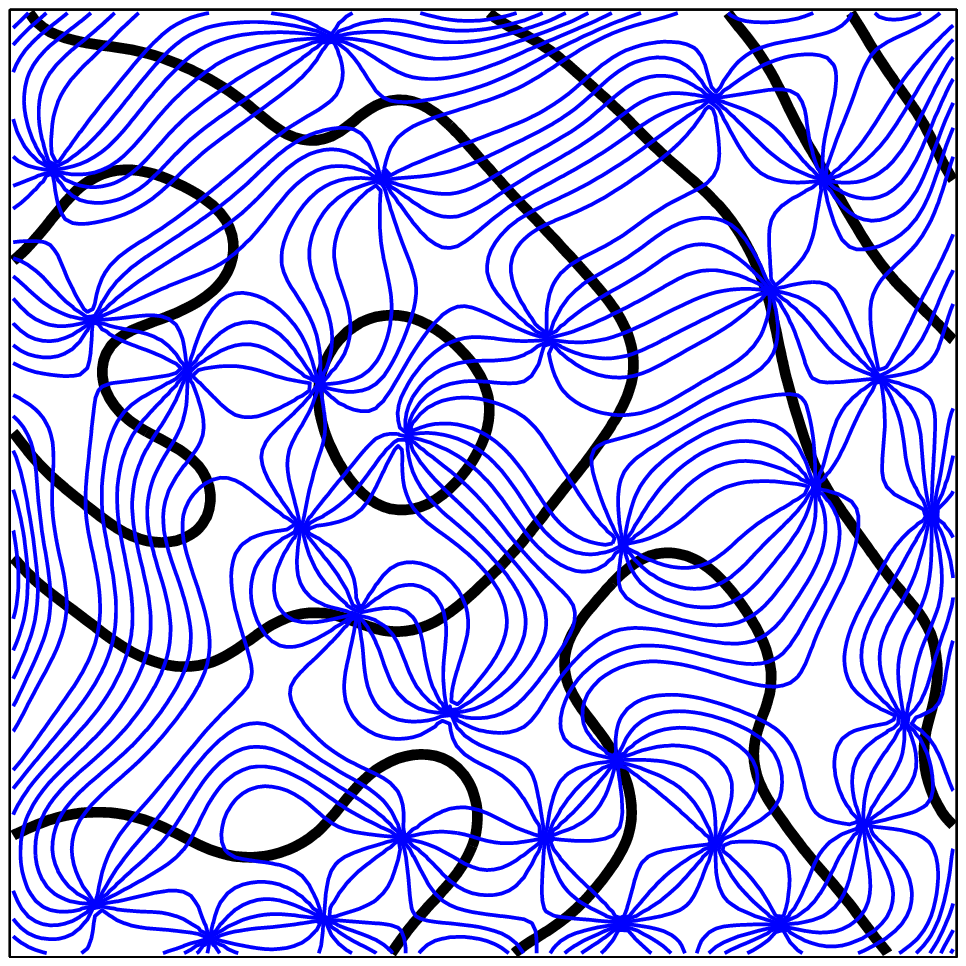} &
      \includegraphics[height=.1200\textheight]{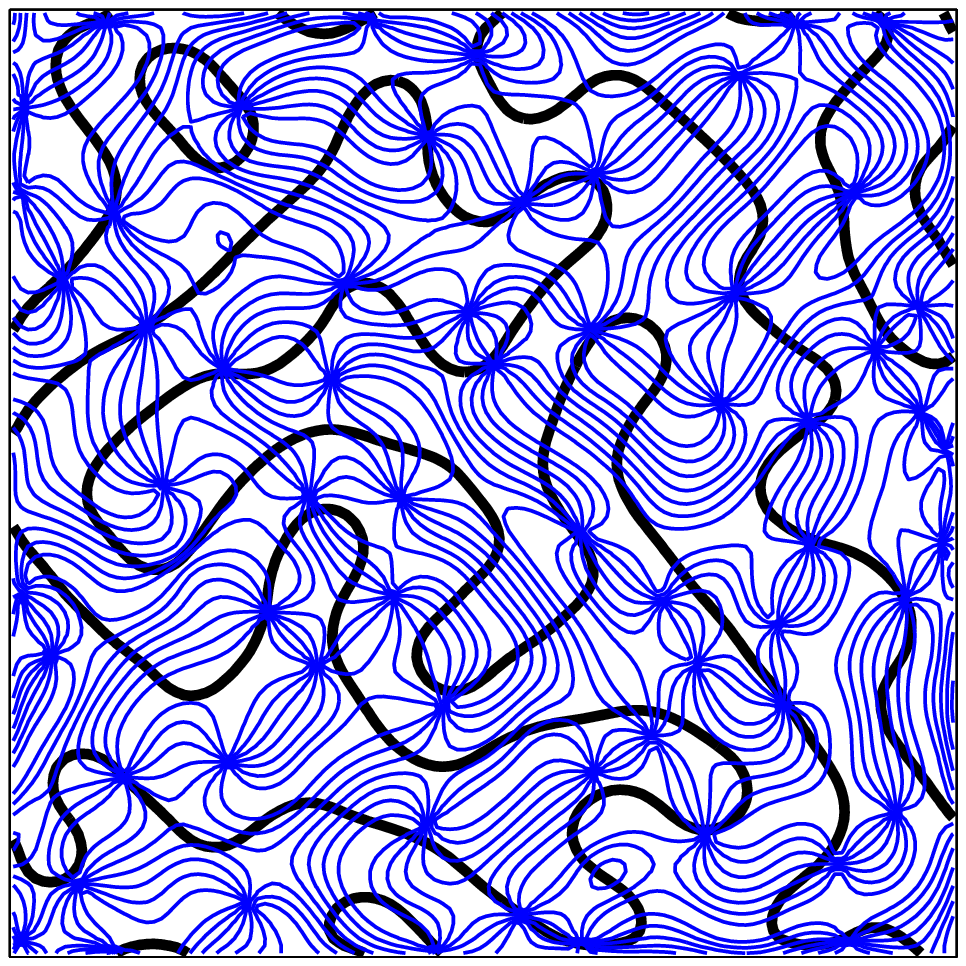} \\[-1ex]
      \rotatebox{90}{\makebox[.1200\textheight][c]{$\beta = 10^{6}$}} &
      \includegraphics[height=.1200\textheight]{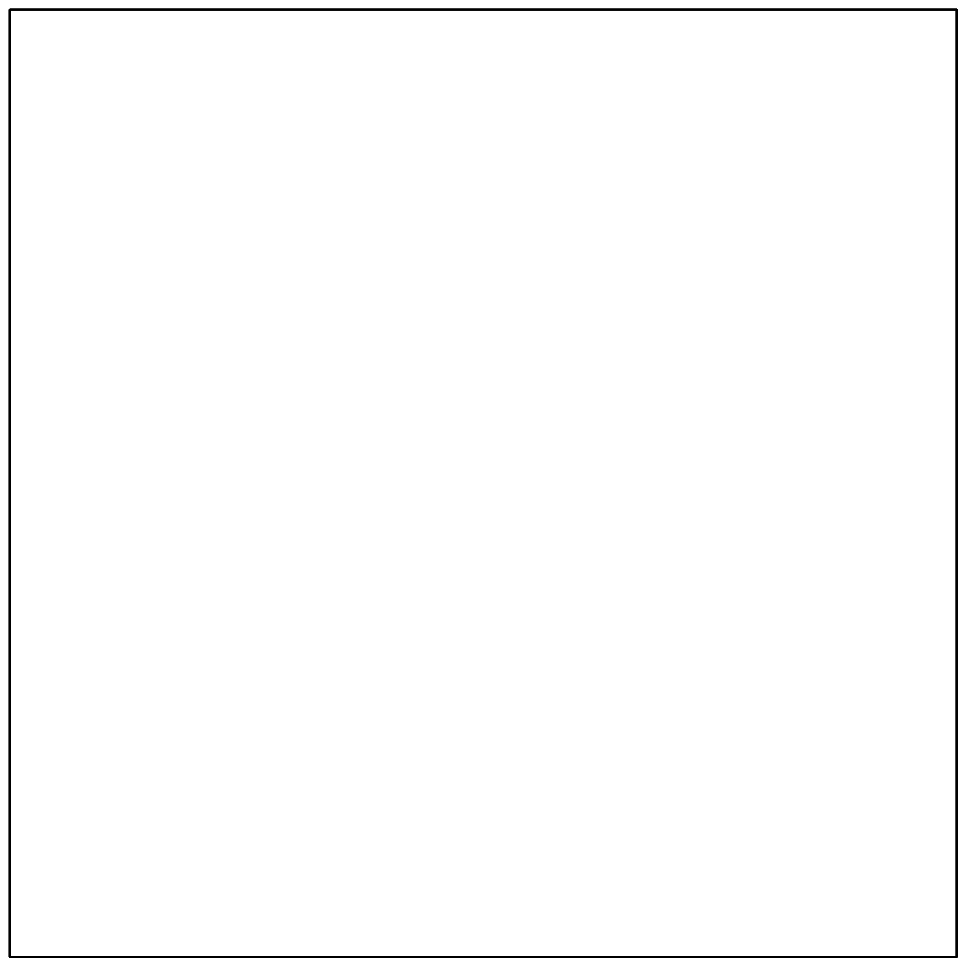} &
      \includegraphics[height=.1200\textheight]{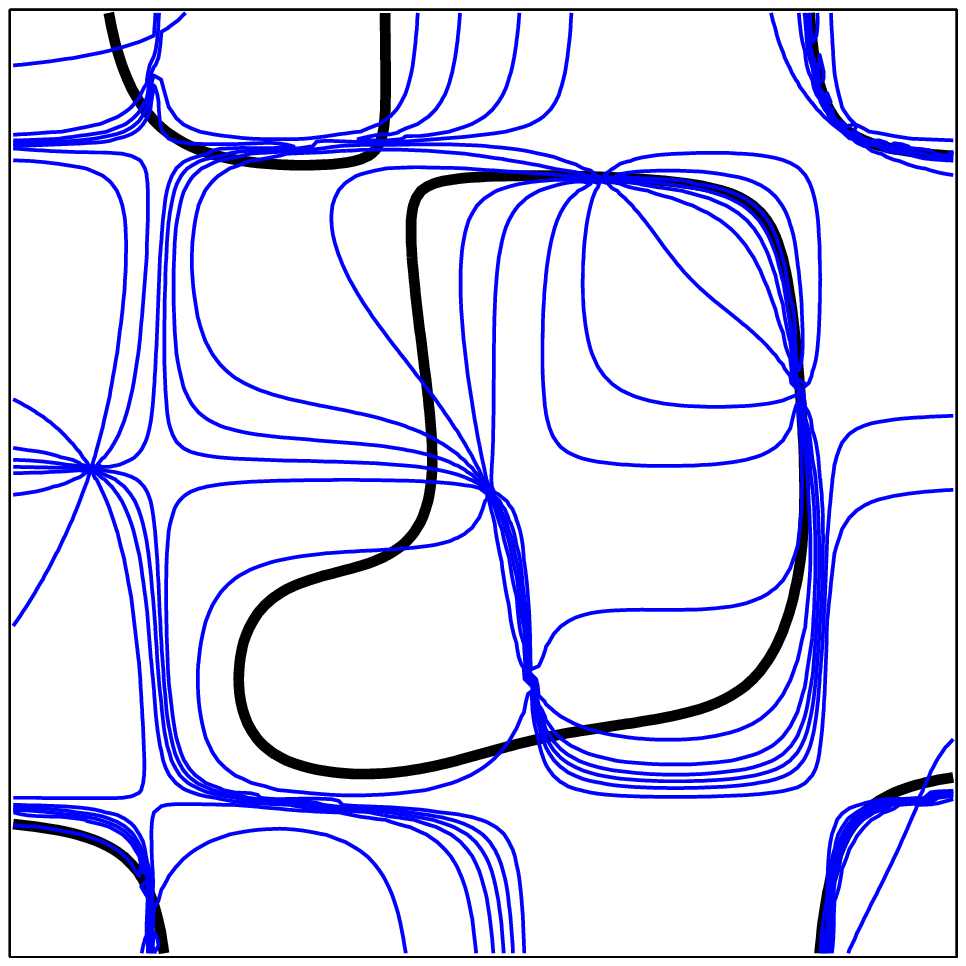} &
      \includegraphics[height=.1200\textheight]{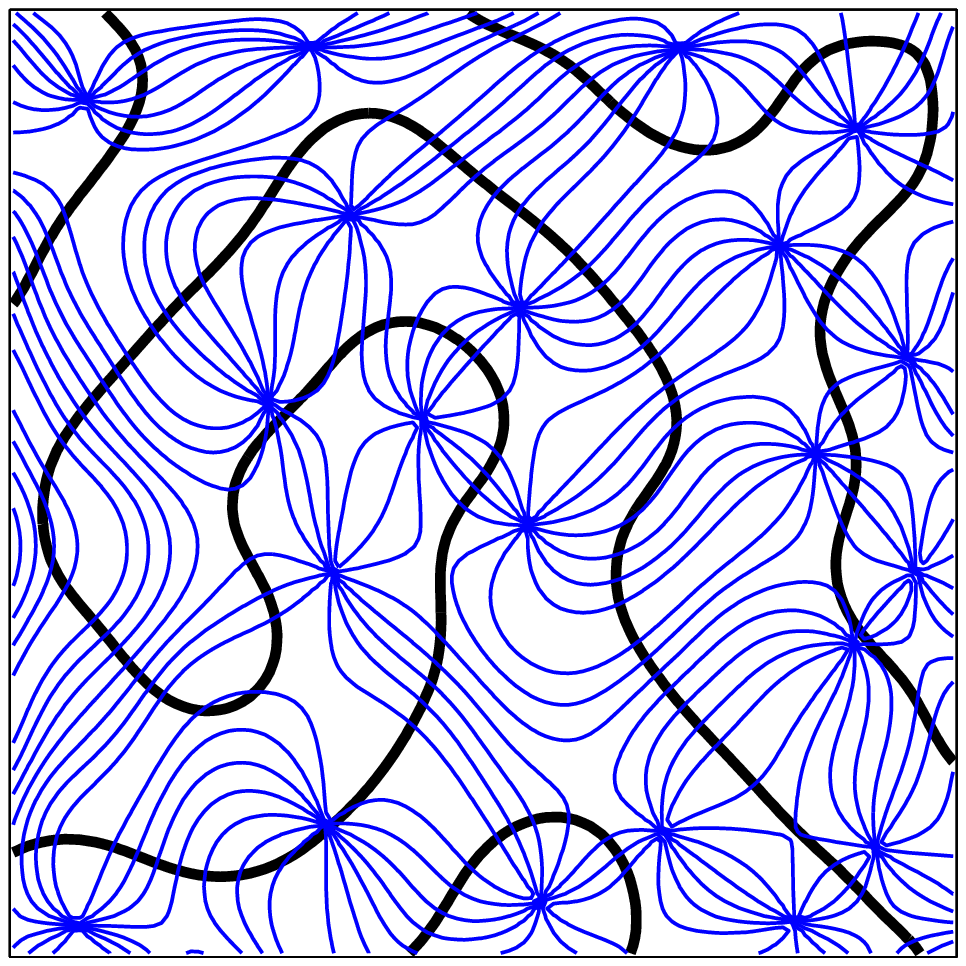} &
      \includegraphics[height=.1200\textheight]{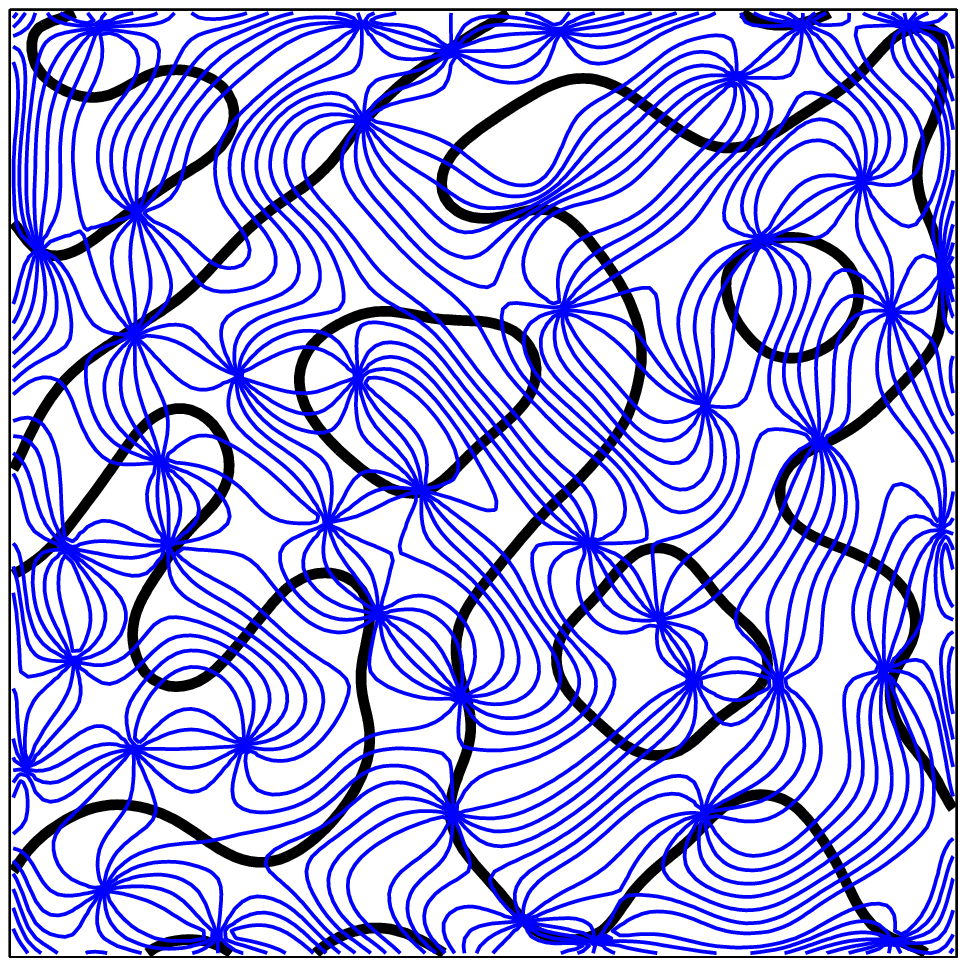} \\[-1ex]
    \end{tabular}
    \caption{As fig.~\ref{f:simul2D:OD} but for contours of ocular dominance and orientation.}
    \label{f:simul2D:ODOR}
  \end{center}
\end{figure}

\section{Discussion and related work}
\label{e:discussion}

\paragraph{Dynamical system analysis of continuous elastic nets}

The most important question in the context of cortical map modelling is, given a stencil, to characterise the emerging net in terms of stripe width, orthogonality of maps, etc. We have given a qualitative explanation of the characteristics of the GEN for different stencil orders and values of $\beta$ (for periodic b.c.). Although our analysis of the frequency spectrum of the stencils is rigorous, the resulting nets depend on the combined effect of the fitness and tension terms. In particular, the interrelations that arise between different maps (such as the crossing angles and matching of gradients in 2D nets) cannot be explained by the tension term alone since different dimensions are independent in our formulation of the tension term. The analysis is simpler if one considers continuous elastic nets, by considering the gradient descent optimisation as a dynamical system $\dot{\y} = -\sigma\, \nabla{E}$ and then linearising it. Some analyses of this type have been done for the continuous first-order 1D net in a 2D space of VF and OD with periodic b.c.\ (self-organising map: \citealp{Obermay_92a}; elastic net: \citealp{WolfGeisel98a,Scherf_99a,Wolf_00a,Dayan01a}). Here one posits an equilibrium solution (inferred by symmetry considerations) consisting of a topographic net unsegregated in ocular dominance, linearises the dynamics at the bifurcation of OD and then studies the eigenvalues of the corresponding eigenfunctions of the linearised operator (which are sine waves). It is then possible to derive a relation for the value $\sigma_{\text{c}}$ at which the OD bifurcation occurs and the wave frequency $k_{\text{c}}$ that dominates.

This analysis is more difficult for the discrete case and depends on the specific stencil used in the tension term. From the power spectrum, the forward-difference family is the closest analog of the continuous case, and thus analysis of the continuous case version should correspond to the limit $M \rightarrow \infty$ in this case.
Linear combinations (in stencil space) of forward-difference stencils should yield to this analysis too. However, for stencils of other types (e.g.\ central-difference ones) is it not obvious how to transform the net into a continuous one. Also, some effects that do not appear at map formation time but in the (much) longer term, such as the loop elimination, may not be explained by the linear approach; and the understanding of the geometric relations between maps requires analysing 2D nets. A theoretical understanding of this topic is left for future research.

\paragraph{Connection with Hebbian models}

Although originally introduced in the TSP context, most work in the
elastic net has been carried out in the cortical map literature.
\citet{Yuille_96a} (see also \citealp{Yuille90a}) sought to unify two types of cortical map models:
the elastic net, and Hebbian models such as that of
\citet{Miller_89a}. They did this by defining a generalised deformable
model whose energy depends on two types of variables: a set of binary
variables \V\ that define constraints on the solution (that each
neuron maps to a unit in either the left or the right eye but not
both, and that each eye unit maps to a neuron); and the centroids \Y.
Eliminating \V\ from the energy leads to the elastic net energy, while
eliminating \Y\ leads (through several approximations, such as a mean
field one) to a model closely related to that of \citet{Miller_89a}.
\citet{Dayan93a} went further and suggested the use of a more general
quadratic form \SS\ for the tension term as we do (though he did not
derive a training algorithm). He then derived an equation similar to
\eqref{e:ann:zerograd} and, based in Yuille et al.'s link to the Hebbian
model, suggested an explicit correspondence between $-\beta \SS$ and
the matrix \B\ in \citet{Miller_89a}, which represents the lateral
connection pattern.

However, the link between both models is rather indirect because of
the approximations involved with the Hebbian model, and also with the
elastic net: in order to obtain the latter from the generalised
deformable model, one needs to assume the binary constraints \V\ hold.
But the OD variable in elastic net simulations is continuous, the
assignment of probability to each centroid is soft (at least for
$\sigma > 0$) and the number of centroids differs in general from that
of stimuli values. This makes it hard to compare directly the lateral
connection matrices in both models. Instead, we propose to compare our \D\ matrix with \B\ (see below).

The effort in comparing elastic and Hebbian matrices was partly
motivated by the fact that in appearance the elastic net defined a
limited lateral connection pattern, where only nearest-neighbouring
cells were connected. We have shown that this nearest-neighbour
connection pattern is exactly equivalent to a certain Mexican-hat
lateral connection pattern (proposition~\ref{p:evenised-fwddiff}). Further, we can draw a formal comparison
between the GEN and the model of
\citet{Miller_89a}. In the GEN, the \D\ matrix
represents the convolution with the neuron preferences (represented by
the centroids \Y); this is then squared and its terms added to obtain the
tension term. In the model of \citet{Miller_89a}, the role of the
convolution with the cortical neuron preferences is played by the
matrix \B. This suggests we should compare \D\ with \B\ (or their
stencils). Also note that in the model of \citet{Miller_89a} the
matrix that is actually taken as a Mexican hat is not the true lateral
connection matrix \B\ but the rather more abstract $(\I - \B)^{-1}$
(which \citet{Miller_89a} call ``cortical interaction function;''
it operates not on the neuron preferences but on the net input
that a neuron receives, roughly equivalent to our \X). However, using
a Mexican hat in $(\I - \B)^{-1}$ results in highly oscillating
stencils in \B, which---given our results---may be one reason why this
model does not appear to reproduce robustly either the orthogonality between OD
and OR or the location of pinwheels at the centre of OD columns
\citep{ErwinMiller98a}.

\paragraph{Diffusion equations}

Consider the continuous version of the tension term of eq.~\eqref{e:tension-cont} for the $p$th-order derivative, with periodic b.c.\ and in 1D for simplicity:
\begin{equation*}
  E_{\text{t}}(y) = \frac{\beta}{2} \int{(y^{(p)}(u))^2 \, du}
\end{equation*}
on the class of periodic functions $y$ with continuous derivatives up to order $2p$. Consider the variational problem of finding the net $y$ in that class that minimises the energy $E = E_{\text{f}}(y) + E_{\text{t}}(y)$ ($E_{\text{f}}(y)$ being the continuous version of the fitness term). The first variation of $E_{\text{t}}(y)$ with respect to $y$ can be found by standard variational calculus.
Let $\delta y$ be an arbitrary function of the same type as $y$. Then:
\begin{equation*}
  \frac{\delta E_{\text{t}}(y)}{\delta y} = \lim_{h \rightarrow 0}{\frac{E_{\text{t}}(y+h\delta y) - E_{\text{t}}(y)}{h}} = \beta \int{y^{(p)}(u) \delta y^{(p)}(u) \, du} = (-1)^p \beta \int{y^{(2p)}(u) \delta y(u) \, du} = (-1)^p \beta y^{(2p)}
\end{equation*}
where we have integrated by parts $p$ times and taken the limit over a test function $\delta y$ that is nonzero near one point $u$ and zero elsewhere.
We now minimise the corresponding Euler-Lagrange equation by gradient descent and obtain:
\begin{equation*}
  \frac{1}{\sigma} \frac{\partial y}{\partial t} = - \frac{\delta E}{\delta y} = (-1)^{p+1} \beta \frac{\partial^{2p} y}{\partial u^{2p}} - \frac{\delta E_{\text{f}}(y)}{\delta y}
\end{equation*}
where $t$ is the continuous iteration time. Disregarding the fitness term, for constant $\sigma$ and $p = 1$ (the original elastic net) this is a diffusion equation for the net $y$ with diffusion constant $\sigma \beta$. That is, a tension term of the form $\frac{\beta}{2} \int{\norm{\nabla y}^2}$ results in a diffusion term $\nabla^2 y$ that opposes the attraction towards the cities coming from the (nonlinear) fitness term. Note that the elastic net was originally derived from the tea-trade model of topographic map formation \citep{MalsburWillsh77a}, based on the diffusion through the brain of hypothesised chemical markers. The learning equation~\eqref{e:ann:zerograd} can also be seen in this light, with the matrix \SS\ acting as a discretised derivative of order $2p$. However, for order $p > 1$, the differential equation has order $2p > 2$, so it no longer corresponds to diffusion.

Note also the correspondence between the continuous and discrete tension terms: in the continuous case, again integrating by parts we obtain $\int{(y^{(p)})^2} = (-1)^p \int{y y^{(2p)}}$, or using linear operators in a Hilbert space, $(\calD_p y, \calD_p y) = (-1)^p (y, \calD_{2p} y)$.
Correspondingly, in the discrete case we have $\norm{\D^p \y}^2 = \y^T \D^{2p} \y$ by representing the first-order matrix \D\ with an evenised stencil (so \D\ is symmetric).

\paragraph{Differential equations}

Difference schemes have, of course, been used extensively in the approximate solution of ordinary and partial differential equations \citep{GodunovRyaben87a,Samars01a}, e.g.\ when approximating a boundary value problem by finite differences over a mesh. In this context, the main concerns regarding the choice of difference scheme are: accuracy (whether the truncation error is of high order, which will mean a faster convergence), stability (whether the approximate solution converges to the true one as the step size tends to zero) and sparsity (whether the scheme has few nonzero coefficients). For example, the central difference (quadratic error) is preferable to the forward difference (linear error) in terms of accuracy, and equivalent in terms of sparsity (both have two nonzero coefficients). Sparsity is desirable to minimise computation time and memory storage, but has to be traded off with accuracy, since a high-order truncation error (as well as a high-order derivative) will require more nonzero coefficients. As for the stability, it depends not only on the scheme itself but on the problem to which it is applied, that is, the combination of type of the differential equation and b.c., difference scheme and step size. For example, given a particular differential equation to be solved, the forward and central difference may both, none or either one or the other result in a stable method \citep{GodunovRyaben87a}. This contrasts with the use of stencils as quadratic penalties, where accuracy and stability do not apply (since for any stencil the matrix \SS\ is positive (semi)definite and so the algorithm converges), while whether the stencil is e.g.\ sawtooth is crucial. Sparsity remains important computationally, although as we showed some penalties admit a dual representation by sparse and nonsparse stencils.

\paragraph{Maximum penalised likelihood estimation (MPLE)}

This consists of the following nonparametric density estimation problem \citep{Silver86a,ThompsTapia90a}:
\begin{equation*}
  \max_{f}{\left( \sum^N_{n=1}{\log{f(\x_n)}} - \beta R(f) \right)}
\end{equation*}
where $\{\x_n\}^N_{n=1}$ is the data sample, $f$ is a density in a suitable function class and $R$ is the smoothness penalty. Note that without the penalty the problem is ill-posed because the likelihood becomes infinite as $f \rightarrow \sum^N_{n=1}{\delta(\x-\x_n)}$. The fact that $f$ must be nonnegative and integrate to $1$ so that it is a density makes this problem mathematically difficult and most of the existing work concerns the 1D case. One approach \citep{GoodGaskin71a} consists of defining $f = \gamma^2$ to remove the nonnegativity constraint. \citet{Scott_80a} consider a discretised MPLE problem where $R$ is based on the first or second forward differences and take $f$ as piecewise constant or linear.

In the GEN, the problem is of density estimation too, but we do not directly penalise the density. If we consider a Gaussian-mixture density $f$ whose centroids are a function of a lower-dimensional latent variable \z:
\begin{equation*}
  f(\x) \propto \int{e^{-\frac{1}{2} \norm{\frac{\x - \bmu(\z)}{\sigma}}^2} \, d\z}
\end{equation*}
then the GEN results from discretising the latent variable \z\ and placing the penalty on the discretised centroid function $\bmu(\z)$.

\paragraph{Other learning algorithms}

The deterministic annealing algorithm of section~\ref{s:ann} has two bottlenecks: one is the computation of the weights \W\ (a very large nonsparse matrix); the other, the solution of the system~\eqref{e:ann:zerograd}. We consider the latter here. In this paper we have dealt with tension terms that result in a sparse, structured, positive semidefinite matrix \SS, as a result of the locality of the stencils implementing differential operators and the spatial structure of the net. For this particular type of \SS\ (or even for a centroid-dependent stencil), the sparse Cholesky factorisation for solving system~\eqref{e:ann:zerograd} works very robustly and efficiently, and the computation time is dominated by the weight computation. However, one may wish to use nonsparse stencils---e.g.\ if one designs the stencil in power space, or uses a stencil derived by discretising a continuous kernel with unbounded support. In this case the Cholesky factorisation will be slower, taking $\calO(M^3)$ operations (where \SS\ is $M \times M$), and it becomes necessary to develop a more sophisticated algorithm to solve the sytem~\eqref{e:ann:zerograd}.

One method that has been used in Bayesian problems and Gaussian processes to invert a covariance matrix is conjugate gradients \citep{Skillin93a,GibbsMackay97a}. This constructs a finite sequence with monotonically decreasing error that converges to $\A^{-1} \uu$ (or in general $\phi(\A) \uu$ for any function $\phi$). An approximate solution is then obtained in $\calO(\text{several} \times M^2)$ rather than $\calO(M^3)$. Unfortunately the number of iterations ``several'' may be difficult to determine in general. Besides, conjugate gradients do not take advantage explicitly of the structure of the matrix \SS, which is in general close%
\footnote{Strictly, the matrix that matters is $\A = \G + \sigma \frac{\beta}{\alpha} \left( \frac{\SS+\SS^T}{2} \right)$. This has the same structure as \SS\ but adds a different value to each element of the diagonal. Thus, even if \SS\ is circulant, \A\ is not circulant or Toeplitz unless $\G \propto \I$.}
to circulant or Toeplitz (by blocks in 2D). Intuitively, such a matrix has only $\calO(M)$ degrees of freedom rather than $\calO(M^2)$. In fact, see \citet[pp.~193ff]{GolubLoan96a}, it is possible to invert Toeplitz matrices of order $M\times M$ in $\calO(M^2)$ using the Trench algorithm and to compute products of a circulant matrix times a vector via the fast Fourier transform in $\calO(M \log{M})$ (this also applies to Toeplitz-vector products by embedding the $M\times M$ Toeplitz matrix in a $(2M-1)\times (2M-1)$ circulant matrix). Extensions of this type of algorithms to near-Toeplitz and other structured matrices are given by \citet{KailatSayed99a}.

\paragraph{Topographic maps and dimensionality reduction}

From a statistical learning point of view, the GEN is both a density estimation and a dimensionality reduction method. As a dimensionality reduction method, the GEN can be seen as a discretisation of a continuous latent variable model with a uniform density in latent space of dimension $L$, a nonparametric mapping from latent to data space (defined by the centroids \Y) and an isotropic Gaussian noise model in data space of dimension $D$ \citep{Carreir01a}. A discretised dimensionality reduction mapping for a data point \x\ can be defined as either the mean or a mode of the posterior distribution $p(m|\x)$. Like the dimensionality reduction mapping, the reconstruction mapping (from latent to data space) is discretised, implicitly defined by the locations of the centroids. Intermediate values can be defined by linear interpolation and we have done so in linking neighbouring centroids in figures~\ref{f:simul1D-nonperiodic}--\ref{f:loop-elim} for visualisation purposes.

Traditionally the elastic net has not been used for density estimation or dimensionality reduction, but rather in applications where generalisation to unseen data was not a concern: the TSP (where the goal is data interpolation with minimal length) and cortical maps (where the goal is to replicate biological maps and the training set is just a computationally convenient scaffolding for a uniform distribution). In fact, in both cases the intrinsic dimensionality of the data (cities and visual stimuli, respectively) is higher than that of the net, which results in the net behaving like a space-filling curve or surface (thus the stripes).

However, the GEN could perfectly be used for data modelling (e.g.\ as in \citealp{Utsugi97a}). It could also be used in computer vision and computer graphics, where models (often implemented as a spline) penalised by curvature terms have been very successful (e.g.\ the snakes of \citealp{Kass_88a} or the functional of \citealp{MumforShah89a})---the GEN having the advantage of defining a density. The GEN is most closely related to the self-organising map (SOM) and the generative topographic mapping (GTM), with which we provide a brief comparison (\citealp{Bishop_98a} give a comparison of the SOM and GTM). Like the GTM and unlike the SOM, the GEN defines a probabilistic model and so an objective function; this gives a principled approach to learning algorithms, determination of convergence and inference. GTM has the advantage over both GEN and SOM of defining a continuous reconstruction mapping (that passes through the Gaussian centroids). The existence of an objective function also allows to define penalties (priors on the mapping or centroids) in a natural way. In the GEN, it is very easy to design a derivative penalty of any order in any dimension of the net, either by passing 1D versions of a forward-difference stencil, or designing a stencil from scratch. In GTM, derivative penalties (e.g.\ to penalise curvature) can in principle be implemented exactly, since the reconstruction mapping is continuous and parametric; but this may be practically cumbersome, so that they may be more conveniently approximated by discrete stencils as we do with the GEN. For the SOM, such penalties might be defined directly in the learning rule (similarly to momentum terms) but this seems more difficult. Alternatively, one could define an objective function that results in a SOM-like algorithm and then apply penalties to it \citep{Luttrel94a,Pascual_01a}; such an approach is likely to end up in something like the GEN. Computationally, all three models suffer from the curse of dimensionality since the latent space is a grid whose number of centroids grows exponentially with its dimensionality.

Two recent methods for manifold learning are Isomap \citep{Tenenb_00a} and locally linear embedding (LLE; \citealp{RoweisSaul00a}), neither of which defines a density. Isomap attempts to model global manifold structure by preserving approximate geodetic distances between all pairs of data points. In contrast, the GEN's tension term uses only local information. LLE is closer to the GEN in that it assigns to each manifold point a squared zero-sum linear combination of neighbouring points. This might be considered as a local metric estimate, similar to the GEN stencil. However, while the latter applies truly to local points in the manifold (by construction) in LLE some neighbouring points may in practice not be local in the manifold.

\paragraph{Combinatorial optimisation and generalised TSPs}

The original elastic net was proposed at first as a continuous optimisation method (by deterministic annealing) to solve the Euclidean travelling salesman problem (TSP; \citealp{Lawler_85a}) of finding the shortest tour of $N$ given cities (although the elastic net tries to minimise the sum of squared distances rather than the sum of distances). The TSP is an NP-complete problem, with a search space of $\frac{(N-1)!}{2}$ different tours. Although the elastic net does find good tours in a reasonable time, it is not really competitive with the leading heuristic TSP methods \citep{Potvin93a,Reinel94a,JohnsonMcgeoc97a,GutinPunnen02a}. These include  $n$\nobreakdash-opt, \citet{LinKernig73a} and variations of them, and their basic idea is to replace $n$ arcs locally in the current tour (``perform a move'') and iterate till a local optimum is found; usually $n$ is $2$ or $3$.

The GEN proposed here can solve a correspondingly generalised TSP (likewise an NP-complete problem, to which $n$\nobreakdash-\hspace{0pt}opt-type algorithms are applicable). We show here how an appropriate definition of the \D\ matrix may be used for other TSP-type problems, in particular the multiple TSP. In the $K$\nobreakdash-TSP \citep[pp.~169--170]{Lawler_85a}, we are looking for a collection of $K$ subtours, each containing a fixed start city, such that each city is in at least one subtour. This problem models the situation where $K$ salesmen work for a company with home office in the start city, and between them each city must be visited at least once. \citet{Lawler_85a} propose the length of the maximum length subtour as quality measure of the collection of routes; and claim that the best algorithm known is a tour-splitting one, i.e., one that starts with a standard travelling salesman tour and proceeds to partition it into $K$ pieces.

The GEN framework defined in this work can be used for a $K$\nobreakdash-TSP where the number of centroids $M_k$ in each subtour $k$ is given in advance. We simply partition the matrix \D\ diagonally in $K$ submatrices; submatrix $k$ is $M_k \times M_k$. For example, if there are $K = 3$ subtours with $k_1 = 3$, $k_2 = 5$ and $k_3 = 4$ centroids, respectively, then:
\begin{equation*}
  \tiny
  \D = \left( \begin{array}{@{}rrr|rrrrr|rrrr@{}}
      -1 &  1 &  0 &  0 &  0 &  0 &  0 &  0 &  0 &  0 &  0 &  0 \\
      0 & -1 &  \phantom{-}1 &  0 &  0 &  0 &  0 &  0 &  0 &  0 &  0 &  0 \\
      0 &  0 &  0 &  0 &  0 &  0 &  0 &  0 &  0 &  0 &  0 &  0 \\
      \hline
      0 &  0 &  0 & -1 &  1 &  0 &  0 &  0 &  0 &  0 &  0 &  0 \\
      0 &  0 &  0 &  0 & -1 &  1 &  0 &  0 &  0 &  0 &  0 &  0 \\
      0 &  0 &  0 &  0 &  0 & -1 &  1 &  0 &  0 &  0 &  0 &  0 \\
      0 &  0 &  0 &  0 &  0 &  0 & -1 &  \phantom{-}1 &  0 &  0 &  0 &  0 \\
      0 &  0 &  0 &  0 &  0 &  0 &  0 &  0 &  0 &  0 &  0 &  0 \\
      \hline
      0 &  0 &  0 &  0 &  0 &  0 &  0 &  0 & -1 &  1 &  0 &  0 \\
      0 &  0 &  0 &  0 &  0 &  0 &  0 &  0 &  0 & -1 &  1 &  0 \\
      0 &  0 &  0 &  0 &  0 &  0 &  0 &  0 &  0 &  0 & -1 &  \phantom{-}1 \\
      0 &  0 &  0 &  0 &  0 &  0 &  0 &  0 &  0 &  0 &  0 &  0
    \end{array} \right)
\end{equation*}
where we have used nonperiodic b.c.\ for all subnets. Clearly, we can create more complex combinations of topology (e.g.\ have each subnet have its own b.c.\ and dimension). As we pointed out in section~\ref{s:cendiff-family}, if the number of centroids $M$ is odd the \D\ matrix associated with the central-difference stencil $\frac{1}{2} (-1,\ 0,\ 1)$ is equivalent to a $2$\nobreakdash-block \D\ matrix (as can be seen by permuting the centroid indices) and results in two disjoint subnets (if $M$ is even we get the usual elastic net but with permuted centroids). These two nets occupy disjoint regions of city space, which can be qualitatively justified as follows. Assume, for a random distribution of $N$ cities over an area $V$ (of any dimensionality $D > 1$), that the expected shortest tour length is approximately proportional to $l \bydef N^p V^q$ with $p, q > 0$, where the actual $p$ and $q$ depend on the distribution of the cities. For example, for uniformly random cities $p = 1-\frac{1}{D}$ and $q = \frac{1}{D}$ (\citealp{Beardw_59a}; \citealp{Lawler_85a}, p.~186ff). If we have $K$ tours each on a disjoint region with the same area and number of cities then the total tour length is proportional to $K (N/K)^p (V/K)^q = K^{1-p-q} l$, while if all tours interleave over the whole area we have a total tour length proportional to $K (N/K)^p V^q = K^{1-p} l$. Thus, it pays better to have disjoint tours, and if $p+q > 1$ it pays better to have as many tours as possible. This is demonstrated in fig.~\ref{f:sawtooth}A, where the two subtours are simply lines along the upper and lower rows of cities, and in fig.~\ref{f:sawtooth-TSP}; and is consistent with the sawteeth linking faraway centroids.

To get the subtours to share a fixed city, we can add a constant centroid $\y_0$ and connect it to every subtour. In general, we can extend the model of section~\ref{s:GEN} by adding ``bias terms'' to the prior on the centroids \Y, i.e., taking $\Y = (\Y_0\ \Y_1)$ with $\Y_0$ constant and $\Y_1$ variable. The \D\ matrix can connect the variable centroids with each other as usual and also any constant centroid in $\Y_0$ to one or more variable centroids in $\Y_1$ (connecting constant centroids with each other clearly has no effect on the optimisation).

The strategy of separating the net into several uncoupled nets is equivalent to allowing discontinuities to occur unpenalised at specific points; see \citet{Terzop86a} for an analogous idea in a continuous setting.

Finally note that, from a TSP perspective, the globally optimal solution of the 1D OD problem of fig.~\ref{f:simul1D-nonperiodic} is a net with a $\sqsubset$ form (i.e., one eye is traversed across VF, then the other eye). In all our simulations, no matter how slow the annealing the original elastic net always results in a wavy net (e.g.\ fig.~\ref{f:simul1D-nonperiodic}), while the second-order forward-difference stencil can obtain the $\sqsubset$ net through loop elimination as in fig.~\ref{f:loop-elim} (if annealing very slowly and using a high value of $\beta$).

\paragraph{Relation with the $C$ measure of \citet{GoodhilSejnow97a}}

\citet{GoodhilSejnow97a} introduced a measure $C$ of the degree of similarity preservation, or neighbourhood preservation, between two discrete spaces mapped by a bijection $H$. $C$ was defined as the following cost functional:
  \begin{equation}
    \label{e:C-measure}
    C \bydef \sum^N_{i=1}{\sum_{j<i}{F(i,j) G(H(i),H(j))}}
  \end{equation}
where $i,j \in \{1,\dots,N\}$ label points in one space and $F$ and $G$ measure the similarity between two points in each of the spaces, respectively. If we take $i$ and $j$ in the cortical space, with $H(i) \bydef \y_i$ being the centroid in city space, and define the similarity functions as follows:
\begin{equation*}
  F(i,j) \bydef 
  \begin{cases}
    1, & \text{$i$ and $j$ are neighbouring} \\ 0, & \text{otherwise}
  \end{cases} \qquad \qquad
  G(H(i),H(j)) \bydef \norm{H(i)-H(j)}^p_p
\end{equation*}
where ``neighbouring'' means adjacent in the cortical space grid, then we get the minimal wiring definition \citep{MitchisDurbin86a,DurbinMitchis90a}. And for $p=2$, minimising $C$ is equivalent to solving the TSP with the original elastic net. Other choices of $F$ and $G$ coincide (approximately) with well-known measures of neighbourhood preservation.

However, eq.~\eqref{e:C-measure} cannot accommodate a GEN tension term $\trace{\Y^T\Y\SS}$ or $\norm{\D\Y^T}^2$ because $C$ is defined by products of binary (pairwise) relations. It is possible to extend $C$ to accommodate it by using multiple relations:
\begin{gather*}
  C \bydef \sum^N_{i=1}{\sum_{j<i}{\sum_{k<j}{\dots F(i,j,k,\dots) G(H(i),H(j),H(k),\dots)}}} \\
  F(i,j,k,\dots) \bydef 
  \begin{cases}
    1, & \text{$i,j,k,\dots$ are neighbouring} \\ 0, & \text{otherwise}
  \end{cases} \qquad
  G(H(i),H(j),H(k),\dots) \bydef \norm{\text{\caja{c}{c}{linear combination of \\ $H(i),H(j),H(k),\dots$}}}^p_p
\end{gather*}
but there seems to be no particular benefit in doing this.

\paragraph{What stencil to use?}

As in other work where differential penalties are used, there is not a general answer to this question---it depends on the particular application. In the discrete case the problem is compounded by the choice not only of differential operator but also of finite difference scheme. In this paper we have studied some properties of two families of stencils in the context of a density estimation problem (the elastic net, applied to cortical map modelling), and shown that the forward-difference family is closer to the continuous case. In general, we think that the choice of stencil should be given by its power spectrum, subject to the stencil being sparse for computational efficiency.

The work done in continuous problems, such as nonparametric regression with roughness penalties \citep{GreenSilver94a,Wahba90a} suggests other conditions that the appropriate operator should satisfy. For example, in dimension $2$ isotropy may be desired, which requires the penalty to be invariant to rotations---though note that this is not naturally defined in the discrete case. Also, the particular application may suggest what functions are natural solutions, and so this defines the nullspace of the operator; for example, if a periodic function is desired, then the operator $\frac{d^2}{du^2} + \omega^2$ should be used (where $\omega$ is a frequency). In practice, most work has restricted itself to 2nd-order derivatives at most, and in fact most differential equations of physics are also 2nd-order or less (with few exceptions, one being the 4th-order biharmonic equation satisfied by the deflection of a loaded beam in elasticity theory).

Higher-order derivatives may be necessary or advisable in higher dimensions, however. The continuous case provides a constraint on the derivative order given by the dimensionality. Consider a $p$th\nobreakdash-order differential operator in the $L$\nobreakdash-dimensional continuous case, such as the following rotationally and translationally invariant penalty:
\begin{equation*}
  R(y) = \int_{\bbR^L}{\sum_{\substack{d_1+\dots+d_L = p \\ d_1,\dots,d_L \ge 0}} \frac{p!}{d_1! \dots d_L!} \left( \frac{\partial^p y}{\partial u^{d_1}_1 \dots \partial u^{d_L}_L} \right)^2 \, du_1 \dots du_L}.
\end{equation*}
Surprisingly, this operator requires $p > L/2$, as the following example from \citet{GreenSilver94a} suggests (the argument can be made rigorously in terms of Beppo-Levy and Sobolev spaces). Consider the ``test function'' $g_s(\uu) = \exp(-\norm{\uu/s}^2)$. Then, by a change of variables $\uu \rightarrow \v/s$, we obtain $R(g_1) = s^{2p-L} R(g_s)$. For $R$ to measure indeed roughness, and since a delta function is maximally rough, $R(g_s)$ must increase as $s \rightarrow 0$ and so we must have $2p-L > 0$. In other words, we have to make sure that the squared magnitude of the derivative term ($s^{2p}$) dominates the loss in volume of the differential element $du_1 \dots du_L$ ($s^L$).

In the discrete case it is difficult to determine a concrete condition (if it exists at all) between the derivative order and the dimensionality of the net, partly because it depends on the stencil chosen to represent the derivative and on the test function, and partly because the finite size of both the stencil and the test function make the calculation complicated. Consider the forward-difference family. The tension of the delta net equals the stencil squared modulus (section~\ref{s:tension:circ}) and grows as $2^p$, so that the higher the derivative the higher the roughness value. But there is a dimensionality effect on the tension term as a measure of roughness. For example, consider a discrete version of the test function that is a triangular bump of height $1$ and width $W$ across all dimensions ($W = 1$ recovers the delta net). Then it is easy to see that for $p = 1$ the tension term is $\calO(W^{L-2})$, so that in 2D the tension is the same for the delta net as for a broad bump, and in $L \ge 3$ the delta net has lower tension. For $p = 2$ we obtain $\calO(W^{L-3})$. This might suggest the use of somewhat higher-order derivatives as the net dimensionality grows. However, this issue is less important than in the continuous case because the net resolution is lower bounded and so is the tension term value. Besides, in practice one is anyway limited to very small $L$ (since the number of centroids $M$ grows exponentially with $L$) and relatively small $p$ (because the size of the stencil grows with $p$ and should be much smaller than $M$). Finally, the use of a 1st-order forward difference in 2D for cortical map models (\citealp{DurbinMitchis90a,GoodhilWillsh90a}; section~\ref{s:cmap}) does not result in rough nets in comparison with higher orders.

\section{Conclusions}

We have studied algorithmically, theoretically and by simulation a generalisation of the elastic net and applied it to cortical map modelling. The original motivation of this work was to understand the role of the tension term, which is an abstract form of lateral interactions, in the development of visual cortical maps.

We have defined the generalised elastic net as a probabilistic model given by a Gaussian mixture (fitness term) whose centroids are subject to a Gaussian prior (tension term). We have given several deterministic annealing algorithms for learning the centroids of the net and argued that the one based on the sparse Cholesky factorisation is the most effective for the type of tension terms we favour. These result from the convolution of the net with a discrete, sparse stencil that approximates a differential operator in the Taylor sense.

We have given a theoretical analysis of this type of tension terms in the case of periodic boundary conditions in 1D. By analogy with the continuous case, the properties of the stencil are apparent in the Fourier power spectrum, or equivalently in the eigenvalues of its circulant matrix. However, in the discrete case the space of stencils is very large, since there are many different stencils that correspond to the same derivative order. We have studied two families obtained by repeated application of a given first-order stencil, forward and central differences, and suggested other ways to design stencils. The forward-difference family is the analogue of the continuous derivative, with stencils behaving as high-pass filters. The central-difference family results in band-pass filters which produce sawtooth patterns in the resulting nets---even as the number of centroids tends to infinity, in stark contrast with the continuous case---and can be seen as competing uncoupled nets. The discrete case is then qualitatively richer in the sense that it can present behaviours that do not have a correlate in the continuous case.

The original elastic net corresponds to the first-order forward-difference stencil. As a corollary of the analysis we have shown that the tension term in this case can be exactly rewritten in a different way to the sum-of-square-distances form. This alternative form has a Mexican-hat shape and results from evenising the stencil subject to having the same power spectrum. Higher-order stencils add more oscillations to this Mexican hat. In the context of previous elastic net work, it demonstrates the equivalence of squared wire length (in stimulus space) and Mexican-hat lateral interaction---and so that what seemingly looks like a sparse pattern of nearest-neighbour connections can be implemented with a denser net of connections (whose strength decays rapidly with distance).
This brings the elastic net closer to most other cortical map models (e.g.\ to that of \citealp{Miller_89a}), in which a dense pattern of connections with a Mexican-hat form is a crucial component.

Finally, we have studied the cortical maps that arise with the forward-difference family. Our Cholesky-based deterministic annealing algorithm allows to explore a wide range of tension terms and so provide an empirical phase diagram for the maps. Such maps result from the interaction of the fitness and tension terms and we have given a qualitative explanation of the striped maps based on a drifting cutoff argument and the stencil spectrum. This argument predicts that the resulting nets show a periodic structure where the frequency increases with the derivative order and decreases with the strength of the tension term. The power spectrum of the stencil explains the anisotropy of the maps. A full theoretical explanation of map development with arbitrary derivatives, in particular the geometric interrelations between maps in 2D nets, is a subject for future research.


\section*{Acknowledgements}

We thank Peter Dayan for many helpful discussions. This work was funded by NIH grant R01~EY12544.

\appendix

\section{Normalisation constant of the prior density}
\label{s:density-norm-cte}

The Gaussian prior on the centroids is proper if \SS\ is positive definite and improper if \SS\ is positive semidefinite (otherwise, \SS\ has a negative eigenvalue and so the prior cannot be normalised since it diverges exponentially along some directions). The full expression of the prior is:
\begin{equation}
  \label{e:tension-termS}
  p(\Y; \beta) = \left(\frac{\beta}{2\pi}\right)^{\frac{Dr}{2}} \abs{\SS}^{\frac{D}{2}}_{+} e^{-\frac{\beta}{2}\trace{\Y^T\Y\SS}}
\end{equation}
where $r \bydef \rank{\SS}$ and $\abs{\SS}_{+}$ is the product of the $r$ positive eigenvalues of \SS. If \SS\ is positive definite then $r = M$ and $\abs{\SS}_{+}$ is just the determinant of \SS. If \SS\ is positive semidefinite, and so has one or more zero eigenvalues, then the covariance $\bSigma = (\beta\SS)^{-1}$ is finite along $r$ eigenvectors and infinite along the rest, so the density is taken as $1$ along those directions. Specifically, calling \y\ a row of \Y\ and spectrally decomposing $\bSigma = \U\bLambda\U^T$ with \U\ orthogonal, $\bLambda = \left(\begin{smallmatrix} \bLambda_1 & \0 \\ \0 & \bLambda_2 \end{smallmatrix}\right)$ with diagonal blocks and $\bLambda^{-1}_2 = \0$, and defining $\z = \U^T\y$:
\begin{equation}
  p(\y) \propto \exp{\left(-\frac{1}{2} \y^T\bSigma^{-1}\y\right)} = \exp{\left(-\frac{1}{2} (\U^T\y)^T\bLambda^{-1}(\U^T\y)\right)} = \exp{\left(-\frac{1}{2} \z^T_1\bLambda^{-1}_1\z_1\right)} \exp{\left(-\frac{1}{2} \z^T_2\bLambda^{-1}_2\z_2\right)}
\end{equation}
where the term in $\bLambda_2$ is 1. Thus, $p(\Y; \beta)$ is given by eq.~\eqref{e:tension-termS} along directions of $\im{\SS}$, with dimension $r$; and $p(\Y; \beta) = 1$ improper along directions of $\ker{\SS}$, with dimension $M-r$.

\section{Stencil normalisation to unit power}
\label{s:stencil-norm}

In order to compare different stencils we need to normalise them, as discussed in section~\ref{s:tension:norm}. Since the eigenvectors of the stencil matrix are plane waves, we will normalise in the Fourier domain. By convention, we will assume that all stencils have unit power, i.e., that the integral of the power spectrum is one. This means that if a stencil favours high frequencies then it must compensate by disfavouring low ones. This is the same as forcing an impulse (delta function) to have unit power in Fourier space after being filtered by the stencil. Therefore we need to calculate the power of a given stencil.

Since the domain is discrete, it would be possible to consider the sum of the powers at each frequency rather than the integral, but by considering the stencil as a sum of delta functions, and thus defined over the real space, we can introduce the discretisation level of the net and so compare nets with different resolutions (e.g.\ the same piece of cortex at different resolutions). Besides, it turns out that both approaches---the discrete sum and the continuous integral---give exactly the same result for the power.

First we consider the continuous approach, i.e., we consider the stencil as a continuous function that is nonzero at the nodes of the grid only. Consider then an $L$\nobreakdash-dimensional net with $M_1 \times M_2 \times \dots \times M_L$ centroids that are a regular sample of an $L$\nobreakdash-dimensional hyperrectangle of lengths $W_1,\dots,W_L$ with $W_l \bydef h M_l$ and $h \in \bbR^+$ the step size. That is, we consider square voxels of size $h^L$ (in net space) for simplicity. The stencil is defined over $\bbR^L$ by the delta-mixture function
\begin{equation*}
  \varsigma(\uu) \bydef \sum^{\M-\1}_{\m=\0}{\varsigma_{\m} \delta(\uu-\uu_{\m})} \qquad \qquad \uu_{\m} = h \m = \frac{W_l}{M_l} \m
\end{equation*}
where we use boldface for real vectors and vectors of indices, e.g.\ $\m = (m_1,\dots,m_L)^T$. The origin of the coordinate system is irrelevant. The continuous Fourier transform of $\varsigma$ is then the plane-wave superposition
\begin{equation*}
  \hat{\varsigma}(\bk) = \int_{\bbR^L}{\varsigma(\uu) e^{-i 2 \pi \bk^T \uu} \, d\uu} = \sum^{\M-\1}_{\m=\0}{\varsigma_{\m} e^{-i 2 \pi \bk^T \uu_{\m}}} = \sum^{\M-\1}_{\m=\0}{\varsigma_{\m} e^{-i 2 \pi \sum^L_{l=1}{\frac{k_l W_l}{M_l} m_l}}}
\end{equation*}
where \uu\ is a length and \bk\ an inverse length (in $L$ dimensions). Both $\varsigma(\uu)$ and $\hat{\varsigma}(\bk)$ are defined in $\bbR^L$ and, while $\varsigma(\uu) = 0$ except at each $\uu_{\m}$, $\hat{\varsigma}(\bk) \neq 0$ in general. If we define $\kappa_l \bydef k_l W_l$ at the integer values $0,\dots,M_l-1$, then $\{\hat{\varsigma}_{\bkappa}\}^{\M-\1}_{\bkappa=\0}$ is the DFT of $\{\varsigma_{\m}\}^{\M-\1}_{\m=\0}$.

The continuous power spectrum, defined in $\bbR^L$, is
\begin{equation}
  \label{e:total-power}
  \begin{split}
    P(\bk) \bydef \abs{\hat{\varsigma}(\bk)}^2 &= \abs{\sum^{\M-\1}_{\m=\0}{\varsigma_{\m} e^{-i 2 \pi \bk^T \uu_{\m}}}}^2 = \sum^{\M-\1}_{\m,\n=\0}{\varsigma_{\m} \varsigma_{\n} e^{-i 2 \pi \bk^T (\uu_{\m} - \uu_{\n})}} \\
    &\stackrel{(*)}{=} \sum^{\M-\1}_{\m,\n=\0}{\varsigma_{\m} \varsigma_{\n} \cos{2 \pi \bk^T (\uu_{\m} - \uu_{\n})}} = \sum^{\M-\1}_{\m=\0}{\varsigma^2_{\m}} + \sum_{\m \neq \n}{\varsigma_{\m} \varsigma_{\n} \cos{2 \pi \bk^T (\uu_{\m} - \uu_{\n})}}
  \end{split}
\end{equation}
where step ($*$) holds because $P(\bk)$ is real. The total power spectrum over $\bbR^L$ diverges, but since $\hat{\varsigma}(\bk)$ is periodic, we can concentrate on the average power in a single period:
\begin{equation*}
  \overline{P} \bydef \int^{\frac{M_1}{W_1}}_0{\dots\int^{\frac{M_L}{W_L}}_0{P(\bk) \, d\bk}} / \int^{\frac{M_1}{W_1}}_0{\dots\int^{\frac{M_L}{W_L}}_0{d\bk}}.
\end{equation*}
To compute $\overline{P}$, note that the cosine term in~\eqref{e:total-power} cancels: changing variables we get
\begin{equation*}
  \int^{\frac{M_1}{W_1}}_0{\dots\int^{\frac{M_L}{W_L}}_0{\cos{2 \pi \bk^T (\uu_{\m} - \uu_{\n})} \, d\bk}} = \text{constant} \times \int^{2\pi}_0{\dots\int^{2\pi}_0{\cos{\left(\sum^L_{l=1}{(m_l - n_l) x_l}\right) \, d\x}}} = 0
\end{equation*}
by using $\cos{(\alpha+\beta)} = \cos{\alpha} \cos{\beta} - \sin{\alpha} \sin{\beta}$ and factoring out variables, and recalling that $m_l - n_l \in \bbZ$. Thus, the expression for the average power in a period, valid for any stencil and any boundary condition, is simply $\overline{P} = \sum^{\M-\1}_{\m=\0}{\varsigma^2_{\m}}$, i.e., the stencil squared modulus (equal to $\frac{1}{M} \trace{\SS}$ for the circulant case). When we have several stencils $\varsigma_j$ (section~\ref{s:lc:power}), since the power is additive we get
\begin{equation*}
  \overline{P} = \sum_j{\overline{P}_j} = \sum_j{\sum^{\M-\1}_{\m=\0}{\varsigma^2_{j,\m}}} = \frac{1}{M} \trace{\textstyle\sum_j{\SS_j}}.
\end{equation*}
As mentioned above, the same result for $\overline{P}$ is obtained by the discrete approach, i.e., approximating the integral by a sum over a grid of voxel $\Delta\bk = \bk_{\bkappa+\1} - \bk_{\bkappa} = \smash{\big(\frac{1}{W_1},\dots,\frac{1}{W_L}\big)}^T$ in Fourier space we obtain
\begin{equation*}
  \frac{1}{\volume{\Delta\bk}} \sum^{\M-\1}_{\bkappa=\0}{\abs{\hat{\varsigma}_{\bkappa}}^2 \volume{\Delta\bk}}  = \sum^{\M-\1}_{\m=\0}{\varsigma^2_{\m}}
\end{equation*}
again by Parseval's theorem.

Finally, we must take into account the step size for the finite difference (see eq.~\eqref{e:D:trunc}) and use $\varsigma_{\m} h^{-p}$ in place of $\varsigma_{\m}$. For example, for eq.~\eqref{e:D:trunc} with the second-order forward difference we have $\frac{1}{h^2} (1,\ -2,\ 1)$. Thus we obtain the final expression for the average power of the stencil $\varsigma$:
\begin{equation*}
  \overline{P} = h^{-2p} \sum^{\M-\1}_{\m=\0}{\varsigma^2_{\m}}
\end{equation*}
which provides the link with a continuous net with physical dimensions (e.g.\ a piece of cortex). However, in this paper we assume $h = 1$ and use a fixed resolution, since we are interested in comparing different values of $\beta$ and the order $p$ over the same net (i.e., at the same discretisation level).

In summary, the normalisation we use amounts to dividing the stencil by its squared modulus, and this is what we do in figures such as~\ref{f:stencil-families1Da} and~\ref{f:simul2D:OD}. For the forward-difference family in 1D (section~\ref{s:fwddiff-family}), the stencil squared modulus is $\binom{2p}{p}$, which grows exponentially as $2^{2p}/\sqrt{\pi p}$. For the central-difference family (section~\ref{s:cendiff-family}), it is $\frac{1}{2^{2p}} \binom{2p}{p}$, which decreases as $1/\sqrt{\pi p}$.

\section{Fourier transforms}
\label{s:fourier}

We use the following formulation for the continuous and discrete Fourier transforms, respectively:
\begin{align*}
  \text{Continuous:}& & \hat{y}(k) &\bydef \int^{\infty}_{-\infty}{y(u) e^{-i 2 \pi k u} \, du} & y(u) &= \int^{\infty}_{-\infty}{\hat{y}(k) e^{i 2 \pi k u} \, dk} & u, k &\in (-\infty,\infty)\\
  \text{Discrete:}& & \hat{y}_k &\bydef \sum^{M-1}_{m=0}{y_m e^{-i 2 \pi \frac{k}{M} m}} & y_m &= \frac{1}{M} \sum^{M-1}_{k=0}{\hat{y}_k e^{i 2 \pi \frac{k}{M} m}} & m, k &= 0,\dots,M-1.
\end{align*}
We use several well-known properties of the Fourier transform, such as Parseval's theorem:
\begin{equation*}
  \text{Continuous: } \int^{\infty}_{-\infty}{\abs{y(u)}^2 \, du} = \int^{\infty}_{-\infty}{\abs{\hat{y}(k)}^2 \, dk} \qquad \qquad \text{Discrete: } \sum^{M-1}_{m=0}{\abs{y_m}^2} = \frac{1}{M} \sum^{M-1}_{k=0}{\abs{\hat{y}_k}^2}
\end{equation*}
or the derivative theorem (continuous case only):
\begin{equation*}
  y^{(p)}(u) \stackrel{\calF}{\rightarrow} (i 2 \pi k)^p \hat{y}(k).
\end{equation*}
More details can be found in standard texts \citep{Bracew00a,Papoul62a}. Note that the mathematical expressions may differ by a constant factor depending on the Fourier transform formulation used.

\section{Proofs}
\label{s:proofs}

In the propositions of section~\ref{s:circ2sten} (and occasionally elsewhere) we encounter some series and integrals whose value can be found in mathematical handbooks, in particular \citet{Prudnik_86a} (abbreviated PBM) and \citet{GradshRyzhik94a} (abbreviated GR).

\begin{proof}[Proof of proposition~\ref{p:eig-circ}]
  By substitution. For $\alpha = 0,\dots,M-1$ and $m = 0,\dots,M-1$:
  \begin{equation*}
    \begin{split}
      (\D\f_m)_{\alpha} &= \sum^{M-1}_{\beta=0}{d_{\alpha\beta} f_{m\beta}} = \sum^{M-1}_{\beta=0}{d_{(\beta - \alpha) \bmod M} \omega^{\beta}_m} = \sum^{\alpha-1}_{\beta=0}{d_{\beta - \alpha + M} \omega^{\beta}_m} + \sum^{M-1}_{\beta=\alpha}{d_{\beta - \alpha} \omega^{\beta}_m} \\
      &= \sum^{M-1}_{\gamma=M-\alpha}{d_{\gamma} \omega^{\gamma+\alpha-M}_m} + \sum^{M-\alpha-1}_{\gamma=0}{d_{\gamma} \omega^{\gamma+\alpha}_m} = \left(\sum^{M-1}_{\gamma=0}{d_{\gamma} \omega^{\gamma}_m}\right) \omega^{\alpha}_m = \lambda_m f_{m\alpha}
    \end{split}
  \end{equation*}
where we have changed $\gamma = \beta - \alpha + M$ and $\gamma = \beta - \alpha$ in the first and second summations, respectively.

It is instructive to see that if a Toeplitz matrix \A\ has plane-wave eigenvectors then it must be circulant. Assuming $a_{nm} = a_{m-n}$ and trying a plane-wave eigenvector $\f = (f_{\beta})$ with $f_{\beta} = e^{i k \beta}$:
\begin{equation*}
  (\A\f)_{\alpha} = \sum^{M-1}_{\beta=0}{a_{\alpha\beta} f_{\beta}} = \sum^{M-1}_{\beta=0}{a_{\beta - \alpha} e^{i k \beta}} = e^{i k \alpha} \sum^{M-1}_{\beta=0}{a_{\beta - \alpha} e^{i k (\beta - \alpha)}} = \lambda_{\alpha} f_{\alpha} \qquad \text{for } \lambda_{\alpha} \bydef \sum^{M-\alpha-1}_{\gamma=-\alpha}{a_{\gamma} e^{i k \gamma}}.
\end{equation*}
Now $\lambda_{\alpha}$ must be independent of $\alpha$, so for all $\alpha = 0,\dots,M-2$ we must have
\begin{equation*}
  \begin{split}
    \lambda_{\alpha} = \lambda_{\alpha+1} &\Rightarrow 0 = \sum^{M-\alpha-1}_{\gamma=-\alpha}{a_{\gamma} e^{i k \gamma}} - \sum^{M-\alpha-2}_{\gamma=-\alpha-1}{a_{\gamma} e^{i k \gamma}} \stackrel{(*)}{=} a_{M-\alpha-1} e^{i k (M-\alpha-1)} - a_{-\alpha-1} e^{- i k (\alpha+1)} \\
    &\Rightarrow k = 2 \pi \frac{m}{M} \text{ for } m = 0,\dots,M-1 \text{ and } a_{M-\alpha-1} = a_{-\alpha-1} \Rightarrow a_{\alpha\beta} = a_{(\beta-\alpha) \bmod M}
  \end{split}
\end{equation*}
where in step $(*)$ the inner summands cancel.
\end{proof}

\begin{proof}[Proof of proposition~\ref{p:eig-circ2}]
  It is straightforward to see that $\D^T$ is circulant with eigenvectors $\f_m$ associated with eigenvalues $\lambda^*_m$. Hence:
  \begin{equation*}
    \SS \f_m = \D^T\D \f_m = \lambda_m \D^T \f_m = \lambda_m \lambda^*_m \f_m = \abs{\lambda_m}^2 \f_m
  \end{equation*}
and likewise $\SS \f^*_m = \abs{\lambda_m}^2 \f^*_m$. So $\f_m$ and $\f^*_m$ are eigenvectors of \SS\ associated with the real eigenvalue $\nu_m = \abs{\lambda_m}^2 \ge 0$. Thus, $\v_m \bydef \frac{1}{2}(\f_m + \f^*_m) = \Re(\f_m)$, with $v_{mn} = \cos{\left( 2 \pi \frac{m}{M} n \right)}$, and $\w_m \bydef \frac{1}{2i}(\f_m - \f^*_m) = \Im(\f_m)$, with $w_{mn} = \sin{\left( 2 \pi \frac{m}{M} n \right)}$, are eigenvectors associated with $\nu_m$.

In terms of the first row of \SS, $(s_0,s_1,\dots,s_{M-1})$, the eigenvalues are (noting that $s_n = s_{M-n}$ because \SS\ is symmetric), for $m=0,\dots,M-1$:
\begin{equation*}
  2 \nu_m = \sum^{M-1}_{n=0}{(s_n + s_{M-n}) \omega^n_m} = \sum^{M-1}_{n=0}{s_n \omega^n_m} + \sum^1_{l=M}{s_l \omega^{M-l}_m} = \sum^{M-1}_{n=0}{s_n (\omega^n_m + \omega^{-n}_m)} = 2 \sum^{M-1}_{n=0}{s_n \cos{\left( 2 \pi \frac{m}{M} n \right)}} \qquad
\end{equation*}
where we have changed $l = M - n$ and $s_0 \equiv s_M$. The first row of \SS\ can be seen to be as follows in terms of that of \D:
\begin{equation}
  \label{e:s-from-d}
  s_{0m} = s_m = \sum^{M-1}_{l=0}{d_{l0} d_{lm}} = \sum^{M-1}_{l=0}{d_l d_{(l-m) \bmod M}} = \sum^{m-1}_{l=0}{d_l d_{l-m+M}} + \sum^{M-1}_{l=m}{d_l d_{l-m}}.
\end{equation}
\end{proof}

\begin{proof}[Proof of proposition~\ref{p:stencil-power}]
  The DFT $\hat{\varsigma}$ of the left-aligned stencil (i.e., padded with zeroes to the right) $\varsigma = (\varsigma_0,\varsigma_1,\dots,\varsigma_{M-1})$ is
  \begin{equation*}
    \hat{\varsigma}_k \bydef \sum^{M-1}_{m=0}{\varsigma_m e^{-i 2 \pi \frac{k}{M} m}} = \sum^{M-1}_{m=0}{\varsigma_m \omega^{*m}_k} \qquad k = 0,1,\dots,M-1.
  \end{equation*}
  Therefore $\hat{\varsigma}_k$ times an offset phase factor $e^{i 2 \pi \frac{k}{M} m'}$ (for some integer $m'$) is equal to $\lambda^*_k$ of the \D\ matrix and the power spectrum is $\abs{\hat{\varsigma}_k}^2 = \abs{\lambda_k}^2 = \nu_k$ of the matrix \SS. Note that the power spectrum is invariant to shifts $n \rightarrow m+m'$, which corresponds to the invariance of \SS\ to row permutations in \D.
\end{proof}

\begin{proof}[Proof of proposition~\ref{p:stencil-sqmodulus}]
  By Parseval's theorem we have $P \bydef \sum^{M-1}_{k=0}{\abs{\hat{\varsigma}_k}^2} = M \sum_m{\varsigma^2_m}$, and since the eigenvalues of \SS\ are $\nu_k = \abs{\hat{\varsigma}_k}^2$, we have $\trace{\SS} = \sum^{M-1}_{k=0}{\nu_k} = M \sum_m{\varsigma^2_m}$. The proposition can also be proven by noting that for $n=0,\dots,M-1$, $s_{nn} = s_{00} = \sum^{M-1}_{m=0}{d^2_{m0}} = \sum^{M-1}_{m=0}{\varsigma^2_m}$ and so $\trace{\SS} = \sum^{M-1}_{n=0}{s_{nn}} = M \sum_m{\varsigma^2_m}$.
\end{proof}

\begin{proof}[Proof of proposition~\ref{p:sawtooth1D}]
  The sawtooth frequency is $m = \frac{M}{2}$ and so $\omega_{\frac{M}{2}} = e^{i \pi} = -1$. We have that its associated eigenvalue in the \D\ matrix is $\lambda_{\frac{M}{2}} = \sum^{M-1}_{n=0}{d_n (-1)^n}$ and its power in the Fourier domain is $\abs{\lambda_{\smash{\frac{M}{2}}}}^2$. Since $(d_0,\dots)$ is the stencil padded with zeroes, $\lambda_{\frac{M}{2}}$ is zero when $\sum_{n \text{ even}}{\varsigma_n} = \sum_{n \text{ odd}}{\varsigma_n}$. Since the stencil is differential it must satisfy $\sum_n{\varsigma_n} = 0$, so $\sum_{n \text{ even}}{\varsigma_n} = \sum_{n \text{ odd}}{\varsigma_n} = 0$.

The same result is obtained in the net domain by forcing $\overleftarrow{\varsigma} \ast f$ to be identically zero, where $f = (1,-1,1,-1,\dots,1,-1)$ is the sawtooth wave.
\end{proof}

\begin{proof}[Proof of proposition~\ref{p:sawtooth2D}]
  Analogous to the 1D case. The 2D DFT of the stencil $\varsigma = (\varsigma_{mn})$ is
  \begin{equation*}
    \hat{\varsigma}_{kl} \bydef \sum^{M-1}_{m=0}{\sum^{N-1}_{n=0}{\varsigma_{mn} e^{-i 2 \pi \left( \frac{k}{M}m + \frac{l}{N}n \right)}}}
  \end{equation*}
which gives
\begin{equation*}
  \hat{\varsigma}_{\frac{M}{2},\frac{N}{2}} = \sum^{M-1}_{m=0}{\sum^{N-1}_{n=0}{\varsigma_{mn} e^{-i \pi (m + n)}}} = \sum^{M-1}_{m=0}{\sum^{N-1}_{n=0}{\varsigma_{mn} (-1)^{(m + n)}}}.
\end{equation*}
Since $\varsigma$ is differential, it also verifies $\sum^{M-1}_{m=0}{\sum^{N-1}_{n=0}{\varsigma_{mn}}} = 0$, hence $\hat{\varsigma}_{\frac{M}{2},\frac{N}{2}} = 0$ if and only if $\sum_{m+n \text{ even}}{\varsigma_{mn}} = \sum_{m+n \text{ odd}}{\varsigma_{mn}} = 0$.
\end{proof}

\begin{proof}[Proof of proposition~\ref{p:stencil-comp}]
  Obviously $\overleftarrow{\varsigma} \ast \overleftarrow{\varrho} = \overleftarrow{\varrho} \ast \overleftarrow{\varsigma}$ (note the respective circulant matrices commute too). Now, from eq.~\eqref{e:D:trunc} with $\alpha_p = 1$:
  \begin{equation*}
    \frac{(\overleftarrow{\varrho} \ast f)(t)}{h^q} = f^{(q)}(t) + \calO(h^{q'})
  \end{equation*}
where $\calO(h^{q'}) = f^{(q+q')}(\xi) \alpha_{q+q'} h^{q'}$ is also a function of $t$ since $\xi$ depends on $t$. Then:
\begin{equation*}
  \begin{split}
    \frac{\Big(\overleftarrow{\varsigma} \ast \frac{\left(\overleftarrow{\varrho} \ast f\right)}{h^q}\Big)(t)}{h^p} &= \frac{((\overleftarrow{\varsigma} \ast \overleftarrow{\varrho}) \ast f)(t)}{h^{p+q}} = \frac{d^p}{dt^p}\left(f^{(q)}(t) + \calO(h^{q'})\right) + \calO(h^{p'}) \\
    &= f^{(p+q)}(t) + \calO(h^{q'}) + \calO(h^{p'}) = f^{(p+q)}(t) + \calO(h^{\min(p',q')}).
  \end{split}
\end{equation*}
\end{proof}

\begin{proof}[Proof of proposition~\ref{p:stencil-comp2}]
  By definition of matrix associated with a stencil, $\D\y = \overleftarrow{\varsigma} \ast \y$ and $\E\y = \overleftarrow{\varrho} \ast \y$ for any $\y \in \bbR^M$. Thus $\D\E\y = \overleftarrow{\varsigma} \ast (\overleftarrow{\varrho} \ast \y) = (\overleftarrow{\varsigma} \ast \overleftarrow{\varrho}) \ast \y$.
\end{proof}

\begin{proof}[Proof of proposition~\ref{p:sawtooth-dom}]
  Call $\overleftarrow{\varpi} = \overleftarrow{\varrho} \ast \overleftarrow{\varsigma}$. Then $\overleftarrow{\varpi}_m = \sum_n{\overleftarrow{\varrho}_{\!\! n} \overleftarrow{\varsigma}_{\!\! m-n}}$ and
\begin{equation*}
  \sum_m{\overleftarrow{\varpi}_{\!\! m} (-1)^m} = \sum_n{\overleftarrow{\varrho}_{\!\! n} \sum_m{\overleftarrow{\varsigma}_{\!\! m-n} (-1)^m}} = \sum_n{\overleftarrow{\varrho}_{\!\! n} (-1)^{-n} \sum_m{\overleftarrow{\varsigma}_{\!\! m} (-1)^m}} = 0.
\end{equation*}
\end{proof}

\begin{proof}[Proof of proposition~\ref{p:fwd-Pascal}]
  By induction. For $p = 0$, necessarily $m = 0$ and so $\leftexp{(0)}{\varsigma_0} = 1$ (the delta function). Assume that the statement holds for $\leftexp{(p)}{\varsigma}$. Then:
  \begin{equation*}
    \leftexp{(p+1)}{\varsigma_m} = 
    \begin{cases}
      m = 0: & (-1)^p \\
      m = 1,\dots,p-1: & -\leftexp{(p)}{\varsigma_m} + \leftexp{(p)}{\varsigma_{m-1}} = - (-1)^{m+p} \binom{p}{m} + (-1)^{m-1+p} \binom{p}{m-1} = \\
      & (-1)^{m+p+1} \left( \binom{p}{m} + \binom{p}{m-1} \right) = (-1)^{m+p+1} \binom{p+1}{m} \\
      m = p: & 1.
    \end{cases}
  \end{equation*}
\end{proof}

\begin{proof}[Proof of proposition~\ref{p:fwd-sqmodulus}]
  $\sum^{M-1}_{m=0}{\leftexp{(p)}{\varsigma^2_m}} = \sum^p_{m=0}{\binom{p}{m}^2} = \binom{2p}{p}$ (GR~0.157:1). 
  Since $\trace{\smash{\leftexp{(p)}{\SS}}} = \sum^{M-1}_{n=0}{\nu_n}$, this also proves the formula $\sum^{M-1}_{n=0}{\left(2 \sin{\left( \pi \frac{n}{M} \right)} \right)^{2p}} = M \binom{2p}{p}$.
\end{proof}

\begin{proof}[Proof of proposition~\ref{p:cen-Pascal}]
  By induction. For $p = 0$, necessarily $m = 0$ and so $\leftexp{(0)}{\varsigma_0} = 1$ (the delta function). Assume that the statement holds for $\leftexp{(p)}{\varsigma}$. Then:
  \begin{equation*}
    \leftexp{(p+1)}{\varsigma_m} = 
    \begin{cases}
      m = 0: & - \frac{1}{2} \leftexp{(p)}{\varsigma_0} = \frac{1}{2^{p+1}} (-1)^{p+1} \\
      m = 1: & - \frac{1}{2} \leftexp{(p)}{\varsigma_1} = 0 \\
      m = 2,\dots,2p: & - \frac{1}{2} \leftexp{(p)}{\varsigma_m} + \frac{1}{2} \leftexp{(p)}{\varsigma_{m-2}} = 
      \begin{cases}
        m = 2n+1 \text{ odd:} & 0 \\
        m = 2n \text{ even:} & - \frac{1}{2} (-1)^{p+n} \binom{p}{n} \frac{1}{2^p} + \frac{1}{2} (-1)^{p+n-1} \binom{p}{n-1} \frac{1}{2^p} = \\
        & (-1)^{p+1+n} \binom{p+1}{n} \frac{1}{2^{p+1}}
      \end{cases}
\\
      m = 2p+1: & \phantom{-} \frac{1}{2} \leftexp{(p)}{\varsigma_{2p-1}} = 0 \\
      m = 2p+2: & \phantom{-} \frac{1}{2} \leftexp{(p)}{\varsigma_{2p}} = \frac{1}{2^{p+1}}.
    \end{cases}
  \end{equation*}
\end{proof}

\begin{proof}[Proof of proposition~\ref{p:sym-circ}]
  Since \SS\ is real circulant, $\nu_{M-k} = \nu^{*}_k$.

  $(\Rightarrow)$ Since $\nu_m = \nu_{M-m} = \nu^{*}_m$, the eigenvalues are real. Diagonalising $\SS = \frac{1}{M} \F \N \F^{*}$, then $\SS^T = \SS^H = (\SS^{*})^T = \frac{1}{M} \F \N \F^{*} = \SS$.

  $(\Leftarrow)$ Since \SS\ is symmetric, its eigenvalues are real, so $\nu^{*}_k = \nu_k$. Thus $\nu_k = \nu_{M-k}$.
\end{proof}

\begin{proof}[Proof of proposition~\ref{p:evenised-fwddiff-inf}]
  Consider the integral $I \bydef \int^{\pi}_0{f(x) \, dx}$ where $f$ is a continuous function in $[0,\pi]$. We can construct a Riemann sum that converges to the integral \citep{Apostol57a} 
by dividing the interval $[0,\pi]$ into $M$ bins of width $\Delta x = \frac{\pi}{M}$, with $x_k \bydef k \Delta x$:
  \begin{equation*}
    I = \lim_{M \rightarrow \infty}{\sum^{M-1}_{k=0}{f(x_k) \Delta x}} = \lim_{M \rightarrow \infty}{ \frac{\pi}{M} \sum^{M-1}_{k=0}{f(x_k)}}.
  \end{equation*}
  Thus, for $\leftexp{(p)}{d_m}$ as given by eq.~\eqref{e:evenised-fwddiff} we have
  \begin{equation*}
    \leftexp{(p)}{d^{\infty}_m} = \frac{2^p}{\pi} I(p,m) \qquad \text{with} \qquad I(p,m) \bydef \int^{\pi}_0{ \sin^p{x} \cos{2 m x} \, dx} \qquad m = 0,1,2,\dots,\infty.
  \end{equation*}
  \citet{GradshRyzhik94a} give%
  \footnote{\citet{Prudnik_86a} give a correct value for the first integral (PBM1~2.5.12:38--39) but an incorrect one for the second (PBM1~2.5.12:13--16,37).}:
  \begin{equation*}
    \begin{array}{ll}
      \text{GR~3.631:12:} & \int^{\pi}_0{ \sin^{2n}{x} \cos{2 m x} \, dx} =
      \begin{cases}
        (-1)^m \, 2^{-2n} \, \binom{2n}{n - m} \, \pi, & n \ge m \\
        0, & n < m
      \end{cases} \\[3ex]
      \text{GR~3.631:13:} & \int^{\pi}_0{ \sin^{2n+1}{x} \cos{2 m x} \, dx} =
      \begin{cases}
        \displaystyle
        \frac{(-1)^m \, 2^{n+1} \, n! \, (2n+1)!!}{(2n-2m+1)!! \, (2m+2n+1)!!}, & n \ge m-1 \\[2ex]
        \displaystyle
        \frac{(-1)^{n+1} \, 2^{n+1} \, n! \, (2n-2m-3)!! \, (2n+1)!!}{(2m+2n+1)!!}, & n < m-1.
      \end{cases}
    \end{array}
  \end{equation*}
  We can simplify GR~3.631:13 (using induction on $n$ for the case $n \ge m-1$) to give
  \begin{equation*}
    I(p,m) = \frac{(-1)^{n+1} \, 2^{n+1} \, n! \, (2n+1)!!}{\prod^n_{k=0}{( 4m^2 - (2k+1)^2 )}} \qquad m = 0,1,2,\dots
  \end{equation*}
  The statement now follows trivially.
\end{proof}

\begin{proof}[Proof of proposition~\ref{p:evenised-fwddiff}]
  By summing the following geometric series for $n = 0,\dots,M-1$:
  \begin{equation*}
    \sum^{M-1}_{m=0}{\left(e^{i \pi \frac{n}{M}}\right)^m} = \frac{1 - e^{i \pi n}}{1 - e^{i \pi \frac{n}{M}}} =
    \begin{cases}
      M, & n = 0 \\
      0, & n > 0 \text{ even} \\
      \frac{2}{1 - e^{i \pi \frac{n}{M}}}, & n \text{ odd}
    \end{cases}
  \end{equation*}
and taking its imaginary part we obtain
\begin{equation*}
  \sum^{M-1}_{m=0}{\sin{\left( \pi \frac{m}{M} n \right)}} =
  \begin{cases}
    0, & n \text{ even} \\
    \frac{\sin{\left( \pi \frac{n}{M} \right)}}{1 - \cos{\left( \pi \frac{n}{M} \right)}} = \frac{1 + \cos{\left( \pi \frac{n}{M} \right)}}{\sin{\left( \pi \frac{n}{M} \right)}} = \cot{\left( \pi \frac{n}{2M} \right)}, & n \text{ odd}.
  \end{cases}
\end{equation*}
The statement follows by applying this to
\begin{equation*}
  \sum^{M-1}_{k=0}{\sin{\left( \pi \frac{k}{M} \right)} \cos{\left( 2 \pi \frac{m}{M} k \right)}} = \sum^{M-1}_{k=0}{\frac{1}{2} \left( \sin{\left( \pi \frac{k}{M} (2m+1) \right)} - \sin{\left( \pi \frac{k}{M} (2m-1) \right)} \right)}
\end{equation*}
and simplifying.
\end{proof}

\begin{proof}[Proof of proposition~\ref{p:evenised-fwddiff-signs}]
  Trivial by induction on $n$, noting that
  \begin{equation*}
    \leftexp{(2n+1)}{d^{\infty}_m} = \frac{-K}{4m^2 - (2n+1)^2} \leftexp{(2n-1)}{d^{\infty}_m}
  \end{equation*}
  (for some $K > 0$ from proposition~\ref{p:evenised-fwddiff-inf}) and so $\leftexp{(2n+1)}{d^{\infty}_m}$ has the same sign as $\leftexp{(2n-1)}{d^{\infty}_m}$ for $m \le n$ and the opposite sign for $m > n$.
\end{proof}

\begin{proof}[Proof of proposition~\ref{p:evenised-cendiff-inf}]
  Analogously to the proof of proposition~\ref{p:evenised-fwddiff-inf}, for $\leftexp{(p)}{d_m}$ as given by eq.~\eqref{e:evenised-cendiff} we have
  \begin{equation*}
    \leftexp{(p)}{d^{\infty}_m} = \frac{1}{\pi} I(p,m) \qquad \text{with} \qquad I(p,m) \bydef \int^{\pi}_0{ \abs{\sin^p{2x}} \cos{2 m x} \, dx} \qquad m = 0,1,2,\dots,\infty.
  \end{equation*}
  For odd $m$, the integral vanishes because the integrand is an odd function of $x$ around $x = \frac{\pi}{2}$. For $m$ even, the integrand is even, and changing variables $x$ to $\frac{y}{2}$ the integral becomes $\int^{\pi}_0{ \sin^p{y} \cos{m y} \, dy}$ which can be calculated with GR~3.631:12--13 as before.
\end{proof}

\ifx\undefined\allcaps\def\allcaps#1{#1}\fi

\end{document}